\documentclass[11pt,a4paper,twoside,british]{book}

\usepackage[paperwidth=210mm, paperheight=297mm,width=160mm,top=25mm,bottom=25mm,inner=25mm, outer=25mm]{geometry}
\usepackage[mathletters]{ucs}
\usepackage{layout}
\usepackage[utf8x]{inputenc}
\usepackage{pgfplots}
\usepackage{natbib} %
\usepackage{comment}
\usepackage{multirow}
\usepackage{multicol}
\usepackage{emptypage}
\usepackage{amsmath}
\usepackage{amsthm}
\usepackage{amssymb} %
\usepackage{datetime2}
\usepackage{tikz}
\usepackage{tikz-cd}
\usepackage{tikz-qtree}
\usepackage[british]{babel}
\usepackage{titlesec} %
\usepackage[ruled,noend,noline,linesnumbered]{algorithm2e}
\usepackage{galois}   %
\usepackage{paralist} %
\usepackage{float}
\usepackage[labelfont=bf,font={small,it}]{caption}
\usepackage{stmaryrd}
\usepackage{booktabs}
\usepackage{tabu}
\usepackage{tablefootnote}
\usepackage{multirow}
\usepackage{xcolor}
\usepackage{verbatim}
\usepackage[inline]{enumitem}%
\usepackage[os=mac, mackeys=symbols]{menukeys} %
\usepackage[colorlinks=true, linkcolor=blue!80, citecolor=blue!80, urlcolor=blue!80,bookmarks=true, bookmarksnumbered=true]{hyperref}
\usepackage{anyfontsize}%
\usepackage{libertine}  
\usepackage[libertine]{newtxmath}
\usepackage{ascii} %
\usepackage{changepage} %
\usepackage{etoolbox} %
\usepackage{imakeidx} %

\makeindex[title=Alphabetical Index,options= -s index_style.ist, intoc]

\newcolumntype{?}{!{\vrule width 2pt}}

\usepackage{frontmatteren}
\usepackage{secondpage}

\usepackage{chapters} 
\usepackage{fancyhdr}
\makeatletter
\newcommand{\myyear}{\@dtm@year\xspace}
\newcommand{\mydate}{\DTMenglishmonthname{\@dtm@month} \@dtm@year\xspace}
\makeatother
\newcommand{\mytitle}{ON THE USE OF QUASIORDERS IN FORMAL LANGUAGE THEORY}
\newcommand{\myauthor}{Pedro Valero Mejía\xspace}
\newcommand{\myadvisor}{Dr.~Pierre Ganty\xspace}

\ifx\draft\undefined\else
 \usepackage{draftwatermark}
 \SetWatermarkLightness{0.9}
 \SetWatermarkScale{3.5}
\fi

\title{\mytitle}
\author{\myauthor}
\date{}
\usetikzlibrary{shapes,arrows,positioning,automata,patterns}

\makeatletter
\newcommand{\RemoveAlgoNumber}{\renewcommand{\fnum@algocf}{\AlCapSty{\AlCapFnt\algorithmcfname}}}
\makeatother

\newcommand*{\AlgRegularW}{\ensuremath{\hyperref[alg:RegIncW]{\textsc{FAIncW}}}\xspace}
\newcommand*{\AlgRegularWr}{\ensuremath{\hyperref[alg:RegIncWr]{\textsc{FAIncWr}}}\xspace}
\newcommand*{\AlgRegularA}{\ensuremath{\hyperref[alg:RegIncA]{\textsc{FAIncS}}}\xspace}
\newcommand*{\AlgGrammarW}{\ensuremath{\hyperref[alg:CFGIncW]{\textsc{CFGIncW}}}\xspace}
\newcommand*{\AlgGrammarA}{\ensuremath{\hyperref[alg:CFGIncA]{\textsc{CFGIncS}}}\xspace}
\newcommand*{\AlgRegularGfp}{\ensuremath{\hyperref[alg:RegIncGfp]{\textsc{FAIncGfp}}}\xspace}

\newcommand*{\AlgCountLines}{\ensuremath{\hyperref[alg:CountLines]{\textsc{CountLines}}}\xspace}
\newcommand*{\AlgSLPIncS}{\ensuremath{\hyperref[alg:SLPIncS]{\textsc{SLPIncS}}}\xspace}

\newtheorem{theorem}{\textsc{Theorem}}[section]
\newtheorem{remark}[theorem]{Remark}
\newtheorem*{conjecture*}{Conjecture}
\newtheorem{example}[theorem]{Example}
\newtheorem{definitionNI}[theorem]{Definition}
\newtheorem*{definitionNI*}{Definition}
\newtheorem*{example*}{Example}
\newtheorem{exampleC}{Example}[chapter]
\newtheorem{corollary}[theorem]{\textsc{Corollary}}
\newtheorem{lemma}[theorem]{\textsc{Lemma}}
\newtheorem{lemmaC}{\textsc{Lemma}}[chapter]
\usepackage{mdframed}
\renewenvironment{proof}{\begin{mdframed}[
backgroundcolor=verylightgray,
skipabove=0,
skipbelow=0,
usetwoside=true,
outermargin=0pt,
innermargin=0pt,
innertopmargin=3pt,
innerleftmargin=2pt,
innerrightmargin=3pt,
align=center,
leftline=false,
rightline=false,
topline=false,
bottomline=true,
linewidth=2pt,
linecolor=lipicsdarkgray]{\noindent\bfseries\color{lipicsdarkgray}{Proof.}}}{\end{mdframed}\medskip}

\makeatletter
\newenvironment{myAlign}[2]{%
\@ifnotmtarg{#1}{\setlength{\abovedisplayskip}{#1}}%
\@ifnotmtarg{#2}{\setlength{\belowdisplayskip}{#2}}%
\csname align*\endcsname}{\csname endalign*\endcsname}
\makeatother

\newenvironment{myAlignEP}{\setlength{\belowdisplayskip}{0pt}\csname align*\endcsname}{\csname endalign*\endcsname}    

\newtheoremstyle{mydef}
{10pt}%
{10pt}%
{\itshape}%
{0pt}%
{}%
{.}%
{ }%
{\textbf{\thmname{#1}\thmnumber{ #2}} (\thmnote{#3})\index{#3}}%
\theoremstyle{mydef}

\newtheorem{definition}[theorem]{Definition}
\newtheorem*{definition*}{Definition}
\newcommand{\demph}[1]{\emph{#1}\index{#1}}
\newcommand{\mindex}[1]{#1\index{\(#1\)}}
\usepackage{esvect} 
\usepackage{bm}
\newcommand{\vect}[1]{\vv{\bm{#1}}}
\newcommand{\vectarg}[2]{\vv{\bm{#1}}^{{\!{\scriptstyle {#2}}}}}
\newcommand{\minor}[1]{\lfloor {#1}\rfloor}
\newcommand{\tuple}[1]{\langle {#1}\rangle}

\newcommand{\lang}[1]{{\mathcal{L}(#1)}} %
\DeclareMathOperator{\Pre}{{Pre}}
\DeclareMathOperator{\ctx}{{ctx}}
\DeclareMathOperator{\Post}{{Post}}
\DeclareMathOperator{\CPre}{{CPre}}
\DeclareMathOperator{\cpre}{{cpre}}
\DeclareMathOperator{\cpost}{{cpost}}
\DeclareMathOperator{\Suf}{\cS}
\DeclareMathOperator{\Pref}{\cP}
\DeclareMathOperator{\absincl}{{Incl^{\sharp}}}
\DeclareMathOperator{\abseq}{{AbsEq}}

\DeclareMathOperator{\base}{{\textit{b}}}
\DeclareMathOperator{\minim}{{min}}

\DeclareMathOperator{\AC}{{AC}}
\DeclareMathOperator{\Kleene}{{\normalfont \textsc{Kleene}}}
\DeclareMathOperator{\KleeneQO}{{\normalfont \widehat{\textsc{Kleene}}}}
\DeclareMathOperator{\wsqsubseteq}{{\widetilde{\sqsubseteq}}}
\DeclareMathOperator{\wsqcup}{{\widetilde{\sqcup}}}
\DeclareMathOperator{\wsqcap}{{\widetilde{\sqcap}}}
\DeclareMathOperator{\wbigsqcup}{{\widetilde{\bigsqcup}}}

\newcommand{\nullable}[3]{{\ensuremath{\psi^{#2}_{#3}(#1)}}}

\newcommand{\fp}{\mathcal{FP}}

\newcommand{\bN}{\mathbb{N}}
\newcommand{\wt}{\widetilde}
\newcommand{\cA}{\mathcal{A}}
\newcommand{\cB}{\mathcal{B}}
\newcommand{\cC}{\mathcal{C}}
\newcommand{\cD}{\mathcal{D}}
\newcommand{\cGr}{\mathcal{G}}
\newcommand{\cH}{\mathsf{H}}
\newcommand{\cK}{\mathcal{K}}
\newcommand{\cM}{\mathcal{M}}
\newcommand{\cN}{\mathcal{N}}
\newcommand{\cO}{\mathcal{O}}
\newcommand{\cT}{\mathcal{T}}
\newcommand{\cV}{\mathcal{V}}
\newcommand{\cP}{\mathcal{P}}
\newcommand{\cS}{\mathcal{S}}
\newcommand{\cR}{\mathsf{R}}

\newcommand{\ud}{\stackrel{\rm\scriptscriptstyle{\fontfamily{ptm}\selectfont def}}{=}}
\newcommand{\Lra}{\Leftrightarrow}
\newcommand{\Ra}{\Rightarrow}
\newcommand{\La}{\Leftarrow}
\newcommand{\ra}{\rightarrow}

\DeclareMathOperator{\uco}{uco}

\DeclareMathOperator{\pre}{pre}
\DeclareMathOperator{\post}{post}

\DeclareMathOperator{\lfp}{lfp}
\DeclareMathOperator{\gfp}{gfp}

\DeclareMathOperator{\Eqn}{{Eqn}}
\DeclareMathOperator{\Eqnr}{{Eqn}^{r}}
\DeclareMathOperator{\Fn}{{Fn}}
\newcommand{\eox}{\hfill{\ensuremath{\Diamond}}}
\newcommand{\eod}{\hfill\rule{0.5em}{0.5em}}

\DeclareFontFamily{U}{mathx}{\hyphenchar\font45}
\DeclareFontShape{U}{mathx}{m}{n}{<-> mathx10}{}

\newcommand{\udr}{\stackrel{\rm\scriptscriptstyle{\fontfamily{ptm}\selectfont def}}{\Longleftrightarrow}}
\newcommand{\udrshort}{\stackrel{\rm\scriptscriptstyle{\fontfamily{ptm}\selectfont def}}{\Leftrightarrow}}
\newcommand{\udiff}{\stackrel{\rm\scriptscriptstyle{\fontfamily{ptm}\selectfont def}}{\Leftrightarrow}}

\newcommand{\goes}[1]{\stackrel{#1}{\leadsto}}
\newcommand{\ggoes}[1]{\stackrel{#1}{\rightarrow}}

\usepgfplotslibrary{groupplots}

\usetikzlibrary{positioning,automata}

\usetikzlibrary{arrows.meta,patterns}

\makeatletter
\newcommand{\verbatimfont}[1]{\def\verbatim@font{#1}}%
\makeatother

\newcommand{\len}[1]{{\vert{#1}\vert}}
\newcommand{\pr}{P}
\newcommand{\produces}{\Rightarrow}
\def\tool#1{{{\asciifamily #1}}}
\newcommand{\winner}[1]{\textcolor{blue}{\textbf{#1}}}
\newcommand{\second}[1]{\textcolor{black}{\textbf{#1}}}
\newcommand{\loser}[1]{\textcolor{red}{\textbf{#1}}}

\newcommand{\True}{\text{\emph{true}}}
\newcommand{\False}{\text{\emph{false}}}

\SetSymbolFont{stmry}{bold}{U}{stmry}{m}{n}
\newcommand*{\NL}{\return} %

\newcommand*{\varleft}[1]{{\text{{\asciifamily L}}_{#1}}}
\newcommand*{\varright}[1]{{\text{{\asciifamily R}}_{#1}}}
\newcommand*{\varnl}[1]{{\text{{\asciifamily N}}_{#1}}}
\newcommand*{\varcount}[1]{{\text{{\asciifamily M}}_{#1}}}
\newcommand*{\counting}[1]{{\mathcal{C}_{#1}}}
\newcommand*{\edges}[1]{{\text{{\asciifamily E}}_{#1}}}

\newcommand*{\algvarleft}[1]{{\text{{\asciifamily L}}_{#1}}}
\newcommand*{\algvarright}[1]{{\text{{\asciifamily R}}_{#1}}}
\newcommand*{\algvarnl}[1]{{\text{{\asciifamily N}}_{#1}}}
\newcommand*{\algvarcount}[1]{{\text{{\asciifamily M}}_{#1}}}

\usetikzlibrary{ipe} %

\tikzstyle{ipe stylesheet} = [
  ipe import,
  even odd rule,
  line join=round,
  line cap=butt,
  ipe pen normal/.style={line width=0.4},
  ipe pen heavier/.style={line width=0.8},
  ipe pen fat/.style={line width=1.2},
  ipe pen ultrafat/.style={line width=2},
  ipe pen normal,
  ipe mark normal/.style={ipe mark scale=3},
  ipe mark large/.style={ipe mark scale=5},
  ipe mark small/.style={ipe mark scale=2},
  ipe mark tiny/.style={ipe mark scale=1.1},
  ipe mark normal,
  /pgf/arrow keys/.cd,
  ipe arrow normal/.style={scale=7},
  ipe arrow large/.style={scale=10},
  ipe arrow small/.style={scale=5},
  ipe arrow tiny/.style={scale=3},
  ipe arrow normal,
  /tikz/.cd,
  ipe arrows, %
  <->/.tip = ipe normal,
  ipe dash normal/.style={dash pattern=},
  ipe dash dashed/.style={dash pattern=on 4bp off 4bp},
  ipe dash dotted/.style={dash pattern=on 1bp off 3bp},
  ipe dash dash dotted/.style={dash pattern=on 4bp off 2bp on 1bp off 2bp},
  ipe dash dash dot dotted/.style={dash pattern=on 4bp off 2bp on 1bp off 2bp on 1bp off 2bp},
  ipe dash normal,
  ipe node/.append style={font=\normalsize},
  ipe stretch normal/.style={ipe node stretch=1},
  ipe stretch normal,
  ipe opacity 10/.style={opacity=0.1},
  ipe opacity 30/.style={opacity=0.3},
  ipe opacity 50/.style={opacity=0.5},
  ipe opacity 75/.style={opacity=0.75},
  ipe opacity opaque/.style={opacity=1},
  ipe opacity opaque,
]

\definecolor{red}{rgb}{1,0,0}
\definecolor{green}{rgb}{0,1,0}
\definecolor{blue}{rgb}{0,0,1}
\definecolor{yellow}{rgb}{1,1,0}
\definecolor{orange}{rgb}{1,0.647,0}
\definecolor{gold}{rgb}{1,0.843,0}
\definecolor{purple}{rgb}{0.627,0.125,0.941}
\definecolor{gray}{rgb}{0.745,0.745,0.745}
\definecolor{brown}{rgb}{0.647,0.165,0.165}
\definecolor{navy}{rgb}{0,0,0.502}
\definecolor{pink}{rgb}{1,0.753,0.796}
\definecolor{seagreen}{rgb}{0.18,0.545,0.341}
\definecolor{turquoise}{rgb}{0.251,0.878,0.816}
\definecolor{violet}{rgb}{0.933,0.51,0.933}
\definecolor{darkblue}{rgb}{0,0,0.545}
\definecolor{darkcyan}{rgb}{0,0.545,0.545}
\definecolor{darkgray}{rgb}{0.663,0.663,0.663}
\definecolor{verydarkgray}{gray}{0.4}
\definecolor{lipicsdarkgray}{rgb}{0.31,0.31,0.33}
\definecolor{darkgreen}{rgb}{0,0.392,0}
\definecolor{darkmagenta}{rgb}{0.545,0,0.545}
\definecolor{darkorange}{rgb}{1,0.549,0}
\definecolor{darkred}{rgb}{0.545,0,0}
\definecolor{lightblue}{rgb}{0.678,0.847,0.902}
\definecolor{lightcyan}{rgb}{0.878,1,1}
\definecolor{lightgray}{rgb}{0.827,0.827,0.827}
\definecolor{verylightgray}{rgb}{0.95,0.95,0.95}
\definecolor{lightgreen}{rgb}{0.565,0.933,0.565}
\definecolor{lightyellow}{rgb}{1,1,0.878}
\definecolor{black}{rgb}{0,0,0}
\definecolor{white}{rgb}{1,1,1}

\newcommand*\pct{\scalebox{.8}{\%}}

\newenvironment{nscenter}
 {\parskip=0pt\par\nopagebreak\centering}
 {\par\noindent\ignorespacesafterend}

\newcommand{\rr}{\sim^{r}}
\newcommand{\qo}{\leqslant}
\newcommand{\qon}{<}
\newcommand{\ql}{\leqslant^{\ell}}
\newcommand{\qr}{\leqslant^{r}}
\newcommand{\qln}{<^{\ell}}
\newcommand{\qrn}{<^{r}}

\newcommand{\rrN}{\rr_{\cN}}

\newcommand{\rrL}{\rr_{L}}
\newcommand{\qlN}{\ql_{\cN}}
\newcommand{\qrN}{\qr_{\cN}}
\newcommand{\qlL}{\ql_{L}}
\newcommand{\qrL}{\qr_{L}}

\newcommand{\qA}{\qr_{L_{\Suf}}}
\newcommand{\qAn}{\qrn_{L_{\Suf}}}

\newcommand{\cF}[1]{\mathsf{Can}^{#1}}
\newcommand{\cG}[1]{\mathsf{Res}^{#1}}
\newcommand{\cL}{\mathsf{L}}

\DeclareMathOperator{\row}{r}
\DeclareMathOperator{\Rows}{Rows}

\newenvironment{myEnum}{
\smallskip
  \begin{compactenum}
}{
  \end{compactenum}
 \medskip
}

\newenvironment{myEnumA}{
  \begin{enumerate}[label=(\alph*), topsep=0pt, parsep=1pt, itemsep=1pt, font=\normalfont\bfseries\color{lipicsdarkgray}]
}{
  \end{enumerate}
}

\newenvironment{myEnumI}{
  \begin{enumerate}[label=(\roman*), topsep=0pt, parsep=1pt, itemsep=1pt, font=\normalfont\bfseries\color{lipicsdarkgray}]
}{
  \end{enumerate}
}

\newenvironment{myEnumAL}{
  \begin{enumerate*}[label=(\alph*), topsep=0pt, parsep=0pt, itemsep=0pt, partopsep=0pt, font=\normalfont\bfseries\color{lipicsdarkgray}]
}{
  \end{enumerate*}
}

\newenvironment{myEnumIL}{
  \begin{enumerate*}[label=(\roman*), topsep=0pt, parsep=0pt, itemsep=0pt, partopsep=0pt, font=\normalfont\bfseries\color{lipicsdarkgray}]
}{
  \end{enumerate*}
}

\newenvironment{myItem}{
  \begin{itemize}[label=\textendash, topsep=0pt, parsep=1pt, itemsep=1pt, font=\normalfont\bfseries\color{lipicsdarkgray}]
}{
  \end{itemize}
}

\titleformat{\section}
{\normalfont\fontsize{14}{10}\selectfont\bfseries}{\thesection}{1em}{}
\titleformat{\subsection}
{\normalfont\fontsize{13}{10}\selectfont\bfseries}{\thesubsection}{1em}{}
\titleformat{\paragraph}[hang]{\normalfont\large\bfseries\color{lipicsdarkgray}}{\theparagraph}{1em}{}
\titlespacing*{\paragraph}{0pt}{3.25ex plus 1ex minus .2ex}{0.5em}

\makeatletter
\newcommand{\specialcell}[1]{\ifmeasuring@#1\else\omit$\displaystyle#1$\ignorespaces\fi}
\makeatother

\setcitestyle{square,aysep={},yysep={;},notesep={, }}

\newcommand*{\citen}[2][]{{%
\ifthenelse{\equal{#1}{}}{[{\citeyear{#2}}]}{[{\citeyear[#1]{#2}}]}%
}}
\renewcommand{\cite}[2][]{{%
\ifthenelse{\equal{#1}{}}{\citep{#2}}{\citep[#1]{#2}}%
}}

\newcommand{\eval}[1]{\ensuremath{[\![#1]\!]}}
\newcommand{\cBp}{\cB^{+}}

\allowdisplaybreaks

\begin{document}
\verbatimfont{\asciifamily}

\frontmatter
\maketitleen

\pagenumbering{roman}
{
  \pagebreak
  \thispagestyle{empty}
  \setcounter{page}{1}
  \hspace{0pt}\vfill
  \begin{center}
  \end{center}
}

{
\let\cleardoublepage\clearpage

  \begin{titlepage}
    
    \par\centering\fontencoding{OT1}\fontfamily{ppl} 
        \large{\textsc{Departamentamento de Lenguajes y Sistemas Inform\'aticos e Ingenieria de Software}
      \par}\vskip 1em
    {\par\centering{\fontencoding{OT1}\fontfamily{ppl} \large{\textsc{Escuela T\'ecnica Superior de Ingenieros Inform\'aticos}}}\par}
    {\par\centering{\fontencoding{OT1}\fontfamily{ppl} {
          \vskip 1em}}\par}
    
    \vspace{2.08cm} 

    {\par\centering{\fontencoding{OT1}\fontfamily{ppl} 
        \begin{center}
          \LARGE{\textbf{\mytitle}}
        \end{center}
      }}
    {\par\centering{\fontencoding{OT1}\fontfamily{ppl} \small{\textsc{Submitted in partial fulfillment of the requirements for the degree of:}}}\par}
    {\par\centering{\fontencoding{OT1}\fontfamily{ppl} \small{\textsc{\textit{Doctor of Philosophy in Software, Systems and Computing}}}}\par}

    \vspace{6em}
    {\fontencoding{OT1}\fontfamily{ppl}
    \begin{tabular}{ll}
    Author: \quad &\textbf{\myauthor} \\[5pt] 
    & \textit{Double Degree in Computer Science and Mathematics} \\[0.7cm]
    Advisor: \quad & \textbf{\myadvisor}
    \\[5pt] 
    & \textit{Ph.D. in Computer Science}
    \end{tabular}}
    \vspace{1.5cm}

    {\par\centering{\fontencoding{OT1}\fontfamily{ppl} \Large{\mydate}}\par}
    \vspace{5cm}

    \begin{flushleft}
     {\par 
       \fontencoding{OT1}\fontfamily{ppl} \large{\emph{Thesis Committee:}}\par}
    {
      \vspace{5pt}

  \hspace{10pt} Prof. Javier Esparza,
  \emph{Technische Universität München}, Germany\\\vspace{2pt}
  \hspace{10pt} Prof. Manuel Hermenegildo, \emph{Instituto IMDEA Software}, Spain\\\vspace{2pt}
  \hspace{10pt} Prof. Ricardo Peña,
  \emph{Universidad Complutense de Madrid}, Spain\\\vspace{2pt}
  \hspace{10pt} Prof. Samir Genaim, \emph{Universidad Complutense de Madrid}, Spain\\\vspace{2pt}
  \hspace{10pt} Prof. Parosh Aziz Abdulla, \emph{Uppsala Universitet}, Sweden\\
%

    }

    \end{flushleft}

  \end{titlepage}

}

\newpage
\chapter*{Abstract of the Dissertation}

In this thesis we use \emph{quasiorders} on words to offer a new perspective on two well-studied problems from \emph{Formal Language Theory}: deciding language inclusion and manipulating the finite automata representations of regular languages.

First, we present a generic quasiorder-based framework that, when instantiated with different quasiorders, yields different algorithms (some of them new) for deciding \emph{language inclusion}.
We then instantiate this framework to devise an efficient algorithm for \emph{searching with regular expressions on grammar-compressed text}.
Finally, we define a framework of quasiorder-based automata constructions to offer a new perspective on \emph{residual automata}.

\paragraph*{The Language Inclusion Problem}
First, we study the \emph{language inclusion problem} \(L_1 \subseteq L_2\) where \(L_1\) is regular or context-free and \(L_2\) is regular.
Our approach relies on checking whether an over-approximation  of \(L_1\), obtained by successively over-approximating the Kleene iterates of its least fixpoint characterization, is included in \(L_2\). 
We show that a language inclusion problem is decidable whenever the over-approximating function satisfies a completeness condition (i.e.\ its loss of precision causes no false alarm) and prevents infinite ascending chains (i.e.\ it guarantees termination of least fixpoint computations).

Such over-approximation of \(L_1\) can be defined using \emph{quasiorder} relations on words where the over-approximation gives the language of all words ``greater than or equal to'' a given input word for that quasiorder. 
We put forward a range of quasiorders that allow us to systematically design decision procedures for different language inclusion problems such as regular languages into regular languages or into trace sets of one-counter nets and  context-free languages into regular languages.

Some of the obtained inclusion checking procedures correspond to well-known algorithms like the so-called \emph{antichains} algorithms. 
On the other hand, our quasiorder-based framework allows us to derive an equivalent greatest fixpoint language inclusion check which relies on quotients of languages and which, to the best of our knowledge, was \emph{not previously known}.  

\paragraph*{Searching on Compressed Text}
Secondly, we instantiate our quasiorder-based framework for the scenario in which \(L_1\) consists on a single word generated by a context-free grammar and \(L_2\) is the regular language generated by an automaton. 
The resulting algorithm can be used for deciding whether a grammar-compressed text contains a match for a regular expression.

We then extend this algorithm in order to count the number of lines in the uncompressed text that contain a match for the regular expression.
We show that this extension runs in time \emph{linear} in the size of the \emph{compressed} data, which might be exponentially smaller than the uncompressed text.

Furthermore, we propose efficient data structures that yield \emph{optimal} complexity bounds and an implementation --\tool{zearch}-- that outperforms the state of the art, offering up to $40\pct$ speedup with respect to \emph{highly optimized} implementations of the decompress and search approach.

\paragraph*{Residual Finite-State Automata}
Finally, we present a framework of finite-state automata constructions based on quasiorders over words to provide new insights on residual finite-state automata (RFA for short).

We present a new residualization operation and show that the residual equivalent of the double-reversal method holds, i.e.\ our residualization operation applied to a co-residual automaton generating the language \(L\) yields the canonical RFA for \(L\).
We then present a generalization of the double-reversal method for RFAs along the lines of the one for deterministic automata.

Moreover, we use our quasiorder-based framework to offer a new perspective on NL\(^*\), an on-line learning algorithm for RFAs.

We conclude that \emph{quasiorders} are fundamental to \emph{residual automata} in the same way \emph{congruences} are fundamental for \emph{deterministic automata}.

\newpage
\chapter*{Resumen de la Tesis Doctoral}

En esta tesis, usamos \emph{preórdenes} para dar un nuevo enfoque a dos problemas fundamentales en \emph{Teoría de Lenguajes Formales}: decidir la inclusión entre lenguajes y manipular la representación de lenguajes regulares como autómatas finitos.

En primer lugar, presentamos un esquema que, dado un preorden que satisface ciertos requisitos, nos permite derivar de manera sistemática algoritmos de decisión para la inclusión entre diferentes tipos de lenguajes.
Partiendo de este esquema desarrollamos un algoritmo de búsqueda con expresiones regulares en textos comprimidos mediante gramáticas.
Por último, presentamos una serie de autómatas, cuya definición depende de un preorden, que nos permite ofrecer un nuevo enfoque sobre la clase de autómatas residuales. 

\paragraph*{El Problema de la Inclusión de Lenguajes}
En primer lugar, estudiamos el problema de decidir \(L_1 \subseteq L_2\), donde \(L_1\) es un lenguaje independiente de contexto y \(L_2\) es un lenguaje regular.
Para resolver este problema, sobre-aproximamos los sucesivos pasos de la iteración de punto fijo que define el lenguaje \(L_1\).
Con ello, obtenemos una sobre-aproximación de \(L_1\) y comprobamos si está incluida en el lenguaje \(L_2\). 
Esta técnica funciona siempre y cuando la sobre-aproximación sea completa (es decir, la imprecisión de la aproximación no produzca falsas alarmas) y evite cadenas infinitas ascendentes (es decir, garantice que la iteración de punto fijo termina).

Para definir una sobre-aproximación que cumple estas condiciones, usamos un preorden.
De este modo, la aproximación del lenguaje \(L_1\) contiene todas las palabras ``mayores o iguales que'' alguna palabra de \(L_1\).
En concreto, definimos una serie de preórdenes que nos permiten derivar, de manera sistemática, algoritmos de decisión para diferentes problemas de inclusión de lenguajes como la inclusión entre lenguajes regulares o la inclusión de lenguajes independientes de contexto en lenguajes regulares.

Algunos de los algoritmos obtenidos mediante esta técnica coinciden con algoritmos bien conocidos como los llamados \emph{antichains algorithms}.
Por otro lado, nuestra técnica también nos permite derivar algoritmos de punto fijo que, hasta donde sabemos, \emph{no han sido descritos anteriormente}.  

\paragraph*{Búsqueda en textos comprimidos}
En segundo lugar, aplicamos nuestro algoritmo de decisión de inclusión entre lenguajes al problema \(L_1 \subseteq L_2\), donde \(L_1\) es un lenguaje descrito por una gramática que genera una única palabra y \(L_2\) es un lenguaje regular definido por un autómata o expresión regular.
De esta manera, obtenemos un algoritmo que nos permite decidir si un texto comprimido mediante una gramática contiene, o no, una coincidencia de una expresión regular dada.

Posteriormente, modificamos este algoritmo para contar las líneas del texto comprimido que contienen coincidencias de la expresión regular.
De este modo, obtenemos un algoritmo que opera en tiempo \emph{linear} respecto del tamaño del texto \emph{comprimido} el cual, por definición, puede ser exponencialmente más peque-ño que el texto original.

Además, describimos las estructuras de datos necesarias para que nuestro algoritmo opere en tiempo \emph{óptimo} y presentamos una implementación --\tool{zearch}-- que resulta hasta un $40\pct$ más rápida que las mejores implementaciones del método estándar de descompresión y búsqueda.

\paragraph*{Autómatas Residuales}
Finalmente presentamos una serie de autómatas parametrizados por preórdenes que nos permiten mejorar nuestra compresión de la clase de autómatas residuales (abreviados como RFA).

Estos autómatas parametrizados nos permiten definir una nueva operación de residualization y demostrar que el método de \emph{double-reversal} funciona para RFAs, es decir, residualizar un autómata cuyo reverso es residual da lugar al canonical RFA (un RFA de tamaño mínimo).
Tras esto, generalizamos este método de forma similar a su generalización para el caso de autómatas deterministas. 
Por último, damos un nuevo enfoque a NL\(^*\), un algoritmo de aprendizaje de RFAs.

Como conclusión, encontramos que los \emph{preórdenes} juegan el mismo papel para los \emph{autómatas residuales} que las \emph{congruencias} para los \emph{deterministas}.

{
\newpage
\thispagestyle{empty}
\vspace*{10em}
\begin{flushright}
  \textit{To my parents and my wife, for their endless love and support}
\end{flushright}
}

\chapter*{Acknowledgments}

Tras un proyecto tan largo e intenso como un doctorado, la lista de personas a las que quiero dar las gracias es muy extensa.
En general, quiero dar las gracias a todas aquellas personas que, de un modo u otro, han formado parte de mi vida durante estos últimos años.
En las siguientes lineas trataré de nombrarlos a todos, aunque seguramente me deje nombres en el tintero.

En primer lugar, quiero dar las gracias a Pierre quien comenzó siendo mi director de tesis y a quien a día de hoy considero un amigo. 
Pierre, gracias por darme la oportunidad de realizar mis primeras prácticas en IMDEA y por ayudarme a realizar mi primera estancia fuera de casa.
Aquella experiencia me hizo descubrir que quería hacer un doctorado y fue tu interés y confianza en mi lo que me llevó a hacerlo en IMDEA.
Gracias por guiarme con paciencia y apoyarme en mis decisiones durante estos 4 años, especialmente en mi interés por realizar estancias para conocer gente y lugares.
Gracias a eso tuve el placer de trabajar con Rupak en Kaiserslatuern, con Javier en Munich y con Yann en San Francisco. 
Gracias también a ellos tres, y a los compañeros que tuve en esos viajes, en especial a Isa, Harry, Filip, Dmitry, Rayna, Marko, Bimba, Nick y Felix, por hacer de mis visitas grandes experiencias llenas de buenos recuerdos.

Quiero dar las gracias, también, a todo el personal del Instituto IMDEA Software.
Ha sido un placer llevar a cabo mi trabajo rodeado de grandes profesionales en todos los ámbitos.
Gracias Paloma, Álvaro, Felipe, Miguel, Isabel, Kyveli, Joaquín, Germán, Platón y Srdjan, entre otros, por ser los artífices de tantos buenos recuerdos. 
Especialmente, quiero agradecer a Ignacio su humor, su ayuda prestada durante estos últimos años y su paciencia al leer múltiples versiones de la introducción de este trabajo.
Gracias por ser ese amigo del despacho de al lado al que ir a molestar siempre que quería comentar alguna idea, por tonta que fuera.

Elena, creo que ha sido una experiencia estupenda haber compartido mis años de universidad y de doctorado con una amiga como tú.
He disfrutado muchísimo de todas las ocasiones en que hemos podido trabajar juntos y creo que hacíamos un equipo estupendo.

A mis profesores de bachillerato Soraya y Mario.
Con vosotros entendí que estudiar era mucho más que aprobar un examen y me hicisteis disfrutar aprendiendo.
Despertasteis en mi la pasión por aprender y por afrontar nuevos retos y fue esa pasión la que me llevó a estudiar el Doble Grado de Matemáticas con Informática y a realizar posteriormente un doctorado.

A mi familia, que recientemente creció en número, por el mero hecho de estar ahí. 
Gracias en especial a mi prima, la Dra. Gámez, por ser la pionera, la primera investigadora y Dra. en la familia, que me ahorró el esfuerzo de explicar a todos cómo funciona el mundo de la investigación en que nos movemos.

A mis amigos de siempre y a los más recientes.
Gracias por tantos buenos momentos, por visitarme cuando estaba fuera y por los viajes y planes que aún quedan por hacer.
Creo firmemente que haber sido feliz en mi vida personal ha sido una pieza clave de mis éxitos profesionales.
Quiero dar las gracias por ello a Alberto, Carlos, David, Rubén, Antonio, Victor, Eduardo, Álvaro, Guillermo, Cristina, Lara e Irene, entre muchos otros.  

A mis padres, gracias por hacerme ser quien soy y por apoyarme siempre aún sin terminar de entender la aventura en la que me embarcaba al iniciar el doctorado. 
Gracias a vosotros he tenido una vida llena de facilidades, que me ha permitido centrarme siempre en mis estudios y mi trabajo.
Cada uno de mis logros es resultado de vuestro esfuerzo.

Por último, quiero dar las gracias mi mujer.
Jimena, gracias por apoyarme durante este tiempo, por acompañarme en mis viajes siempre que fue posible y por soportar la distancia cuando no. Gracias, en definitiva, por estar ahí.

\hypersetup{citecolor=black, urlcolor=black,linkcolor=black}
{
\newpage
\thispagestyle{empty}
\tableofcontents
} 

{
\newpage
\thispagestyle{empty}
\listoffigures 
}

\newpage{
 \thispagestyle{empty}
 \listoftables
}
\hypersetup{citecolor=blue!80, urlcolor=blue!80,linkcolor=blue!80}

\markboth{List of Publications}{List of Publications}
\chapter*{List of Publications}
This thesis comprises the following four papers for which I am the main author.
The first two have been published in top peer-reviewed academic conferences while the last two have recently been submitted and have not been published yet:
\bigskip

\begin{myEnum}
\item \myauthor and \myadvisor\\
      \textbf{Regular Expression Search on Compressed Text}\\
      Published in \emph{Data Compression Conference}, March 2019.
\medskip
\item \myauthor, \myadvisor and Prof. Francesco Ranzato\\
      \textbf{Language Inclusion Algorithms as Complete Abstract Interpretations}\\
      Published in \emph{Static Analysis Symposium}, October 2019.
\medskip
\item \myauthor, \myadvisor and Elena Guti\'errez\\
      \textbf{A Quasiorder-based Perspective on Residual Automata}\\
      Published in \emph{Mathematical Foundations of Computer Science}, August 2020.
\medskip
\item \myauthor, \myadvisor and Prof. Francesco Ranzato\\
      \textbf{Complete Abstractions for Checking Language Inclusion}\\
      Submitted to \emph{Transactions on Computational Logic}, \mydate.
\end{myEnum}
\medskip

Using the techniques presented in first of the above mentioned papers, I developed a tool for searching with regular expressions in compressed text.
The implementation is available on-line at \url{https://github.com/pevalme/zearch}.
\medskip

I have also contributed to the following papers which are not part of this thesis.
\bigskip
\begin{myEnum}
\item Elena Guti\'errez, \myauthor and \myadvisor\\
      \textbf{A Congruence-based Perspective on Automata Minimization Algorithms}\\
      Published in \emph{International Symposium on Mathematical Foundations of Computer Science}, August 2019 
\medskip
\item \myauthor, \myadvisor and Boris Köpf\\
  \textbf{A Language-theoretic View on Network Protocols}\\
  Published in \emph{Automated Technology for Verification and Analysis}, October 2017
\end{myEnum}

\mainmatter
\pagenumbering{arabic}
\setcounter{page}{1}

\clearpage{}%
%
%
%
%
%
%
%
%
%
%
%
%
%
%
\chapter{Introduction}
\label{chap:introduction}

\emph{Formal languages}, i.e. languages for which we have a \emph{finite formal description}, are used to model possibly infinite sets so that their finite descriptions can be used to reason about these sets.
As a consequence, \emph{Formal Language Theory}, i.e. the study of formal languages and the techniques for manipulating their finite representations, finds applications in several domains in computer science.

For example, the possibly infinite set of assignments that satisfy a given formula in some logic can be seen as a formal language whose finite description is the formula itself.
In some logics, the set of values that satisfy any formula is regular and, therefore, it can be described by means of a finite-state automaton (automaton for short).
When this is the case, it is possible to reason in that logic by manipulating automata as shown in Example~\ref{example:DecisionProcedure}.

\begin{figure}[!ht]
\centering
\begin{minipage}[l]{0.6\textwidth}
\begin{tikzpicture}[->,>=stealth',shorten >=1pt,auto,node distance=5mm and 1cm,thick,initial text=]
\tikzstyle{every state}=[scale=0.75,fill=customblue!60,draw=blue!60,text=black]
      
\node[initial,state, accepting] (0) {\(0\)};
\node[state] (1) [right=of 0] {\(1\)};
\node[state] (2) [right=of 1] {\(2\)};
\node (a) at (-1,1.1) {\(A_4:\)};

\path (0) edge [loop above] node {\small\(0\)} (0)
(0) edge [bend left] node {\small\(1\)} (2)
(2) edge [loop above] node {\small\(0{,}1\)} (2)
(2) edge node {\small\(0\)} (1)
(1) edge node {\small\(0\)} (0);
\end{tikzpicture}
\end{minipage}\hfill
\begin{minipage}[r]{0.35\textwidth}
\begin{tikzpicture}[->,>=stealth',shorten >=1pt,auto,node distance=5mm and 1cm,thick,initial text=]
\tikzstyle{every state}=[scale=0.75,fill=customblue!60,draw=blue!60,text=black]
      
\node[initial,state, accepting] (0) {\(0\)};
\node[state] (1) [right=of 0] {\(1\)};
\node (a) at (-1,1.1) {\(A_2:\)};

\path (0) edge [loop above] node {\small\(0\)} (0)
(0) edge [bend left] node {\small\(1\)} (1)
(1) edge [loop above] node {\small \(0{,}1\)} (1)
(1) edge [bend left] node {\small\(0\)} (0);

\end{tikzpicture}
\end{minipage}

\begin{tikzpicture}[->,>=stealth',shorten >=1pt,auto,node distance=5mm and 1cm,thick,initial text=]
\tikzstyle{every state}=[scale=0.75,fill=customblue!60,draw=blue!60,text=black]
      
\node[initial,state, accepting] (0) {\(0\,0\)};
\node[state] (1) [right=of 0] {\(2\,1\)};
\node[state] (2) [right=of 1] {\(2\,0\)};
\node[state] (3) [above=of 2] {\(1\,0\)};
\node[state] (4) [below=of 2] {\(1\,1\)};
\node (a) at (-1,1.1) {\(A_{42}:\)};

\path 
(0) edge [loop above] node {\small\(0\)} (0)
(0) edge node {\small\(1\)} (1)
(1) edge [loop above] node [yshift=-3pt] {\small \(0{,}1\)} (1)
(1) edge [bend left=20] node [yshift=-3pt] {\small\(0\)} (2)
(1) edge node [yshift=-2pt, xshift=2pt] {\small\(0\)} (3)
(1) edge node [below] [xshift=-2pt, yshift=2pt] {\small\(0\)} (4)
(2) edge [bend left=20] node [above, yshift=-3pt] {\small\(1\)} (1)
(2) edge [loop right] node {\small \(0\)} (2)
(2) edge node {\small \(0\)} (3)
(3) edge [bend right] node [above] {\small\(0\)} (0)
(4) edge [bend left] node [xshift=-1pt, yshift=3pt] {\small\(0\)} (0);
\end{tikzpicture}
\caption{Automata accepting the set of binary encodings of numbers divisible by 4 (top left), divisible by two (top right) and the product of these two automata (bottom).}\label{fig:FAdivisible}
\end{figure}
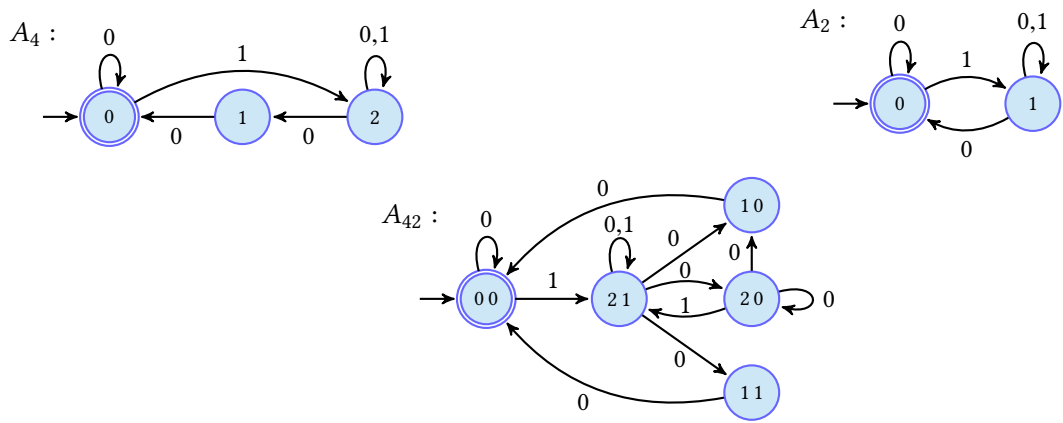

\begin{exampleC}\label{example:DecisionProcedure}
Consider the formulas \(f_{2} : `` x \text{ mod } 2 = 0"\) and \(f_{4} : `` x \text{ mod } 4 = 0 "\).
Next we show how to reason about the formula \(f_{42} : `` f_4 \land f_2"\) by means of automata.

A binary sequence ``\(x\)'' encodes a number divisible by 4 if{}f the last two digits are 0's.
Similarly, ``\(x\)'' encodes a number divisible by 2 if{}f the last digit is 0.
Therefore, the automata \(A_4\) and \(A_2\) from Figure~\ref{fig:FAdivisible} accept the binary encodings of numbers ``\(x\)'' that satisfy the formulas \(f_4: ``x \text{ mod } 4 = 0"\) and \(f_2: `` x \text{ mod } 2 = 0"\), respectively.

Since the numbers satisfying the formula \(f_{42}\) are, by definition, the ones satisfying both \(f_4\) and \(f_2\), the automaton for \(f_{42}\) is \(A_{42} = A_2 \times A_4\), shown in Figure~\ref{fig:FAdivisible}, which recognizes exactly the encodings accepted by both \(A_2\) and \(A_4\).
Thus, there exists a number satisfying \(f_{42}\) if{}f the language accepted by \(A_{42}\) is not empty.

On the other hand, since the automaton \(A_4\) accepts a language that is included in the one of \(A_2\), we conclude that the encodings satisfying \(f_4\) also satisfy \(f_2\).
Thus, the automaton for \(A_4\) is equivalent to, i.e. it accepts the same language as, the automaton \(A_{42}\) and both are automata for \(f_{42}\).\eox
\end{exampleC}

This idea led to the development of \emph{automata-based decision procedures} for logical theories such as Presburger arithmetic~\cite{wolper1995automata} and the Weak Second-order theory of One or Two Successors (WS1S/WS2S)~\cite{Klarlund:mona95,Klarlund:mona99} among others~\cite{allouche2003automatic,schaeffer2013deciding}.

A similar idea is used in \emph{regular model checking}~\cite{bouajjani2000regular,abdulla2012regular,clarke18}, where formal languages are used to describe the possibly infinite sets of states that a system might reach during its execution.

A different use of formal languages in computer science is the \emph{lossless compression of textual data} \cite{Charikar2005Smallest,Hucke2016Smallest}.
In this scenario the data is seen as a language consisting of a single word and its finite formal description as a grammar is seen as a \emph{succinct representation} of the language it generates.
As the following example evidences, the grammar might be exponentially smaller than the data.

\begin{exampleC}\label{example:GrammarCompression}
Let \(k\) be an integer greater than 1 and let \(\cGr\) be the grammar with the set of variables \(\{X_i \mid 0 \leq i \leq k\}\), alphabet \(\{a\}\), axiom \(X_k\) and set of rules \(\{X_i \to X_{i{-}1}X_{i{-}1} \mid 1 \leq i \leq k\} \cup \{X_0 \to a\}\).

Clearly, \(\cGr\) has size \emph{linear} in \(k\) and produces the word \(a^{2^k}\).
Therefore, the grammar is \emph{exponentially smaller} than the word it generates. \eox
\end{exampleC}

The idea of using grammars to compress textual data has led to the development of several grammar-based compression algorithms~\cite{ziv1978compression,nevill1997compression,larsson2000off}.
These algorithms offer some advantages with respect to other classes of compression techniques, such as the ones based on the well-known \text{LZ77} algorithm~\cite{ziv1977compression}, in terms of the structure of the compressed representation of the data (which is a grammar).
In particular, they allow us to analyze the uncompressed text, i.e. the language, by looking at the compressed data, i.e. the grammar~\cite{lohrey2012algorithmics}.
\medskip

\section{The Contributions of This Dissertation}
In this dissertation we focus on three problems from \emph{Formal Language Theory}: deciding language inclusion, searching on grammar-compressed text and building residual automata.
As we describe next, these are well-studied and important problems in computer science for which there are still challenges to overcome.

\paragraph*{The Language Inclusion Problem}

In the first two scenarios described before, i.e. \emph{automata-based decision procedures} and \emph{regular model checking}, the \emph{language inclusion problem}, i.e. deciding whether the language inclusion \(L_1 \subseteq L_2\) holds, is a fundamental operation.

For instance, in Example~\ref{example:DecisionProcedure}, deciding the language inclusion between the languages generated by automata \(A_4\) and \(A_2\) allows us to infer that all values satisfying \(f_4\) also satisfy \(f_2\).
Similarly, in the context of regular model checking, we can define a possibly infinite set of ``good'' states that the system should never leave and solve a language inclusion problem to decide whether the system is confined to the set of good states.

As a consequence, the \emph{language inclusion problem} is a fundamental and classical problem in computer science~\cite[Chapter 11]{HU79}. 
In particular, language inclusion problems of the form \(L_1 \subseteq L_2\), where both \(L_1\) and \(L_2\) are regular languages, appear naturally in different scenarios as the ones previously described.

The standard approach for solving such problems consists on reducing them to \emph{emptiness} problems using the fact that \(L_1 \subseteq L_2\Lra L_1 \cap L_2^c = \varnothing\).
However, algorithms implementing this approach suffer from a worst case exponential blowup when computing \(L_2^c\) since it requires determinizing the automaton for \(L_2\).
The state of the art approach to overcome this limitation is to keep the computation of the automaton for \(L_2^c\) implicit, thus preventing the exponential blowup for many instances of the problem.

For instance, \citet{DBLP:conf/cav/WulfDHR06} developed an algorithm for deciding language inclusion between regular languages that uses \emph{antichains}, i.e. sets of incomparable elements, to reduce the blowup resulting from building the complement of a given automaton.
Their work was later improved by \citet{Abdulla2010} and \citet{DBLP:conf/popl/BonchiP13} who used \emph{simulations} between the states of the automata to further reduce the blowup associated to the complementation step. %
Then, \citet{Holk2015} adapted the use of antichains to decide the inclusion of context-free languages into regular ones.

However, even though these algorithms have a common foundation, i.e. they all reduce the language inclusion problem to an emptiness one through complementation and use antichains to keep the complementation implicit, the relation between them is not well understood.
This is evidenced by the fact that the generalization by \citet{Holk2015} of the antichains algorithm of \citet{DBLP:conf/cav/WulfDHR06} was obtained by rephrasing the inclusion problem as a data flow analysis problem over a relational domain.

\bigskip\noindent\textbf{Our Contribution.}
We use \emph{quasiorders}, i.e. reflexive and transitive relations, to define a framework from which we systematically derive algorithms for deciding language inclusion such as the ones of \citet{DBLP:conf/cav/WulfDHR06} and \citet{Holk2015}.
Indeed, we show that these two algorithms are conceptually equivalent and correspond to two instantiations of our framework using different quasiorders.
Moreover, by using a quasiorder based on simulations between the states of an automata, we derive an improved antichains algorithm that partially matches the one of \citet{Abdulla2010}.

Furthermore, our framework goes beyond inclusion into regular languages and allows us to derive an algorithm for deciding the language inclusion \(L_1 \subseteq L_2\) when \(L_1\) is regular and \(L_2\) is the set of traces of a \emph{one counter net}, i.e. an automaton equipped with a counter that cannot test for 0.
Finally, we also derive a \emph{novel} algorithm for deciding inclusion between regular languages.

\paragraph*{Searching on Compressed Text}
The growing amount of information handled by modern systems demands for efficient techniques both for compression, to reduce the storage cost, and for regular expression searching, to speed up querying.

Therefore, the problem of searching on compressed text is of practical interest as evidenced by the fact that state of the art tools for searching with regular expressions, such as \tool{grep}\footnote{\url{https://www.gnu.org/software/grep/manual/grep.html}.} and \tool{ripgrep}\footnote{\url{https://github.com/BurntSushi/ripgrep}.}, provide a method for searching on compressed files by decompressing them on-the-fly.

Due to the high performance of state of the art compressors such as \tool{zstd}\footnote{\url{https://github.com/facebook/zstd}} and \tool{lz4}\footnote{\url{https://github.com/lz4/lz4}}, the performance of searching on the decompressed data as it is recovered by the decompressor is comparable with that of searching on the uncompressed data.
Therefore, the parallel decompress-and-search approach is the state of the art for searching on compressed text.

However, when using a grammar-based compression technique it is possible to manipulate the compressed data, i.e. the grammar, to analyze the uncompress data, i.e. the language generated by the grammar.
Intuitively, this means that the information about repetitions in the text present in its compressed version can be used to enhance the search.
Therefore, \emph{searching on grammar-compressed text} could be even faster than searching on the uncompressed text. 

This idea is exploited by multiple algorithms that perform certain operations directly on grammar-compressed text, i.e. without having to recover the uncompressed data, such as finding given words \cite{navarro2005lzgrep}, finding words that match a given regular expression~\cite{navarro2003regular,bille2009improved} or finding approximate matches~\cite{navarro2001approximateSurvey}.

Nevertheless, the implementations of \citet{navarro2003regular} and \citet{navarro2005lzgrep} (to the best of our knowledge, the only existing tools for searching on compressed text) are not faster than the state of the art decompress and search approach. Partly, this due to the fact that these algorithms only apply to data compressed with one specific grammar-based compressor, namely \tool{LZ78}~\cite{ziv1978compression}, which, as shown by \citet{Hucke2016Smallest}, cannot achieve exponential compression ratios\footnote{The compression ratio for a file of size \(T\) compressed into size \(t\) is \(T/t\).}.

\bigskip\noindent\textbf{Our Contribution.}
We improve this situation by rephrasing the problem of searching on compressed text as a language inclusion problem between a context-free language (the text) and a regular one (the expression).
Then, we instantiate our quasiorder-based framework for solving language inclusion and adapt it to the specifics of this scenario, where the context-free grammar generates a single word: the uncompressed text.
The resulting algorithm is not restricted to any class of grammar-based compressors and it reports the number of lines in the text containing a match for a given expression in time \emph{linear} with respect to the size of the compressed data.

We implement this algorithm in a tool --\tool{zearch}\footnote{\url{https://github.com/pevalme/zearch}}-- for searching with regular expressions in\linebreak grammar-compressed text. 
The experiments evidence that compression can be used to enhance the search and, therefore, the performance of \tool{zearch} improves with the compression ratio of the data.
Indeed, our tool is as fast as searching on the uncompressed data when the data is well-compressible, i.e. it results in compression ratio above 13, which occurs, for instance, when considering automatically generated \emph{log files}.

\paragraph*{Building Residual Automata}
Clearly, %
the problem of finding a concise representation of a regular language is also a fundamental problem in computer science.

There exists two main classes of automata representations for regular languages, both having the same expressive power: non-deterministic (NFA for short) and deterministic (DFA for short) automata.
While DFAs are simpler to manipulate than NFAs\footnote{For instance, in order to build the complement of a DFA it suffices to switch final and non-final states while complementing an NFA requires determinizing it.} they are, in the worst case, exponentially larger. 

\begin{exampleC}
The minimal DFA for the set of words of length \(2n{+}2\) with two 1's separated by \(n\) symbols has size exponential in \(n\) since any DFA for that language must have one state for each of the \(2^n\) possible prefixes of length \(n\).
Figure~\ref{fig:NFAvsDFA} shows the minimal DFA and an exponentially smaller NFA for \(n=2\).\eox
\end{exampleC}

\begin{figure}[!ht]
\begin{minipage}[l]{0.58\textwidth}
\resizebox{!}{150pt}{
%
%
\begin{tikzpicture}[->,>=stealth',shorten >=1pt,auto,node distance=5mm and 1cm,thick,initial text=]
\tikzstyle{every state}=[fill=customblue!60,draw=blue!60,text=black, inner sep=0pt, minimum size=12pt]

    \node[state,initial] (A0) {};

    \node[state] (B1) [xshift=1.2*1cm, yshift=1*0.5cm] {};
    \node[state] (B2) [xshift=1.2*1cm, yshift=-1*0.5cm] {};
    
    \node[state] (C1) [xshift=1.2*2cm, yshift=1*1.5cm] {};
    \node[state] (C2) [xshift=1.2*2cm, yshift=1*0.5cm] {};
    \node[state] (C3) [xshift=1.2*2cm, yshift=-1*0.5cm] {};
    \node[state] (C4) [xshift=1.2*2cm, yshift=-1*1.5cm] {};

    \node[state] (D1) [xshift=1.2*3cm, yshift=1*3cm] {};
    \node[state] (D2) [xshift=1.2*3cm, yshift=1*2cm] {};
    \node[state] (D3) [xshift=1.2*3cm, yshift=1*1cm] {};
    \node[state] (D4) [xshift=1.2*3cm] {};
    \node[state] (D5) [xshift=1.2*3cm, yshift=-1*1cm] {};
    \node[state] (D6) [xshift=1.2*3cm, yshift=-1*2cm] {};
    \node[state] (D7) [xshift=1.2*3cm, yshift=-1*3cm] {};

    \node[state] (E1) [xshift=1.2*4cm, yshift=1*2.5cm] {};
    \node[state] (E2) [xshift=1.2*4cm, yshift=1*1.5cm] {};
    \node[state] (E3) [xshift=1.2*4cm, yshift=1*0.5cm] {};
    \node[state] (E4) [xshift=1.2*4cm, yshift=-1*0.5cm] {};
    \node[state] (E5) [xshift=1.2*4cm, yshift=-1*1.5cm] {};
    \node[state] (E6) [xshift=1.2*4cm, yshift=-1*2.5cm] {};

    \node[state] (F1) [xshift=1.2*5cm, yshift=1*0.5cm] {};
    \node[state] (F2) [xshift=1.2*5cm, yshift=-1*0.5cm] {};

    \node[state,accepting] (G1) [xshift=1.2*6cm] {};

\path (A0) edge [blue, above] node[yshift=-2pt] {$1$} (B1);
\path (A0) edge [red, above] node[yshift=-2pt] {$0$} (B2);

\path (B1) edge [blue, above] node[yshift=-2pt] {$1$} (C1);
\path (B1) edge [red, above] node[yshift=-2pt] {$0$} (C2);
\path (B2) edge [blue, above] node[yshift=-2pt] {$1$} (C3);
\path (B2) edge [red, above] node[yshift=-2pt] {$0$} (C4);

\path (C1) edge [red, above] node[yshift=-2pt] {$0$} (D1);
\path (C1) edge [blue, above] node[yshift=-2pt] {$1$} (D2);
\path (C2) edge [red, above] node[yshift=-2pt] {$0$} (D3);
\path (C2) edge [blue, below] node[near end, xshift=-3pt, yshift=2pt] {$1$} (D6);
\path (C3) edge [red, above] node[yshift=-2pt] {$0$} (D4);
\path (C3) edge [blue, above] node[yshift=-2pt] {$1$} (D5);
\path (C4) edge [blue, above] node[yshift=-2pt] {$1$} (D7);

\path (D1) edge [blue, above] node[yshift=-2pt] {$1$} (E1);
\path (D1) edge [red, below] node[near start, xshift=-2pt, yshift=2pt] {$0$} (E2);
\path (D2) edge [blue, above] node[yshift=-2pt] {$1$} (E1);
\path (D2) edge [red, above] node[yshift=-2pt] {$0$} (E3);
\path (D3) edge [blue, above] node[yshift=-2pt] {$1$} (E1);
\path (D4) edge [black, above] node[near start, xshift=-3pt,yshift=-3pt] {$0{,}1$} (E2);
\path (D5) edge [blue, above] node[yshift=-1pt] {$1$} (E2);
\path (D5) edge [red, above] node[yshift=-2pt] {$0$} (E3);
\path (D6) edge [red, above] node[yshift=-2pt] {$0$} (E4);
\path (D6) edge [blue, above] node[yshift=-2pt] {$1$} (E5);
\path (D7) edge [blue, above] node[yshift=-2pt] {$1$} (E5);
\path (D7) edge [red, above] node[yshift=-2pt] {$0$} (E6);

\path (E1) edge [black, above] node[xshift=3pt,yshift=-2pt] {$0{,}1$} (F1);
\path (E2) edge [blue, above] node[yshift=-2pt] {$1$} (F1);
\path (E3) edge [blue, above] node[yshift=-2pt] {$1$} (F1);
\path (E3) edge [red, above] node[yshift=-2pt] {$0$} (F2);
\path (E4) edge [blue, above] node[yshift=-2pt] {$1$} (F2);
\path (E5) edge [black, above] node[xshift=-2pt, yshift=-2pt] {$0{,}1$} (F2);
\path (E6) edge [red, above] node[yshift=-2pt] {$0$} (F2);

\path (F1) edge [black, above] node[yshift=-2pt] {$0{,}1$} (G1);
\path (F2) edge [blue, above] node[yshift=-2pt] {$1$} (G1);

\end{tikzpicture} 
 }
\end{minipage}%
\hfill
\begin{minipage}[r]{0.35\textwidth}
\resizebox{!}{160pt}{
%
%
\begin{tikzpicture}[->,>=stealth',shorten >=1pt,auto,node distance=5mm and 1cm,thick,initial text=]
\tikzstyle{every state}=[fill=customblue!60,draw=blue!60,text=black, inner sep=0pt, minimum size=12pt]
    \node[state,initial,initial where=above] (T0) {};
    \node[state] (T1) [below =of T0] {};
    \node[state] (T2) [below =of T1] {};
    \node[state] (T3) [below =of T2] {};
    \node[state] (T4) [below =of T3] {};

  \node[state] (S0) [left =of T1] {};
    \node[state] (S1) [below =of S0] {};
    \node[state] (S2) [below =of S1] {};
    \node[state] (S3) [below =of S2] {};
    \node[state] (S4) [below =of S3] {};

    \node[state] (R0) [left =of S1] {};
    \node[state] (R1) [below =of R0] {};
    \node[state] (R2) [below =of R1] {};
    \node[state] (R3) [below =of R2] {};
    \node[state,accepting] (R4) [below =of R3] {};

    \path (T0) edge [blue, left] node {$1$} (T1);
    \path (T0) edge [above] node {$0{,}1$} (S0);
    \path (S0) edge [above] node {$0{,}1$} (R0);
    \path (S0) edge [blue, left] node {$1$} (S1);
    \path (R0) edge [blue, left] node {$1$} (R1);

    \path (T1) edge [left] node {$0{,}1$} (T2);
    \path (T2) edge [left] node {$0{,}1$} (T3);
    \path (T3) edge [blue, thick, left] node {$1$} (T4);
    \path (T4) edge [above] node {$0{,}1$} (S4);
    \path (S4) edge [above] node {$0{,}1$} (R4);

    \path (S1) edge [left] node {$0{,}1$} (S2);
    \path (S2) edge [left] node {$0{,}1$} (S3);
    \path (S3) edge [blue, thick, left] node {$1$} (S4);

    \path (R1) edge [left] node {$0{,}1$} (R2);
    \path (R2) edge [left] node {$0{,}1$} (R3);
    \path (R3) edge [blue, thick, left] node {$1$} (R4);
\end{tikzpicture} 
 }
\end{minipage}
\caption{Minimal DFA (left) and NFA (right) accepting the words in the alphabet \(\{0,1\}\) of length 6 that contains two 1's separated by two symbols.
For clarity, we use colors red, blue and black for transitions with labels ``0'', ``1'' and ``0,1'', respectively.
}\label{fig:NFAvsDFA}
\end{figure}
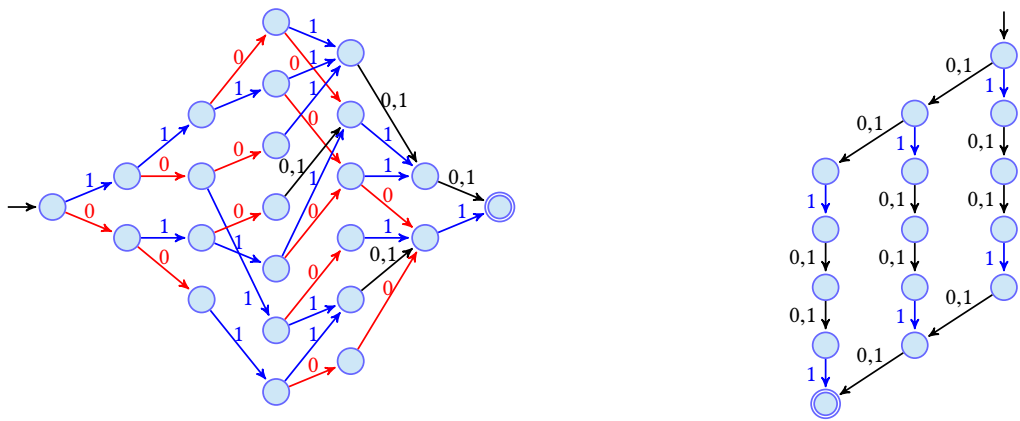

Therefore, algorithms relying on determinized automata, such as the standard algorithm for building the complement of an NFA, do not scale despite the existence of different techniques for reducing the size of DFAs~\cite{hopcroft1971,moore1956} and for building DFAs of minimal size~\cite{Sakarovitch,Adamek2012,Brzozowski2014}.

This has led to the introduction of \emph{residual automata}~\cite{denis2001residual,denis2002residual} (RFA for short) as a generalization of DFAs that breaks determinism in favor of conciseness of the representation.
Therefore, RFAs are easier to manipulate than NFAs (there exists a canonical minimal RFA for every regular language, which makes learning easier) and more concise than DFAs (both automata from Figure~\ref{fig:NFAvsDFA} are RFAs).
These properties make RFAs specially appealing in certain domains such us Grammatical Inference~\cite{denis2004learning,bollig2009angluin}.

There exists a clear relationship between RFAs and DFAs as evidenced by the similarities between the \emph{residualization} and \emph{determinization} operations and the fact that a straightforward modification of the double-reversal method for building minimal DFAs yields a method for building minimal RFAs.
However, the connection between these two formalisms is not fully understood as evidenced by the fact that the relation between the generalization of the double-reversal methods for DFAs~\cite{Brzozowski2014} and RFAs~\cite{tamm2015generalization} is not immediate. 

\bigskip\noindent\textbf{Our Contribution.}
We present a framework of quasiorder-based automata constructions that yield residual and co-residual automata.
We find that one of these constructions defines a residualization operation that produces smaller automata than the one of \citet{denis2002residual} and for which the double-reversal method holds: residualizing a co-residual automaton yields the canonical RFA.
Moreover, we derive a generalization of this double-reversal method for RFAs, along the lines of the one of \citet{Brzozowski2014} for DFAs that is more general than the one of \citet{tamm2015generalization}.

Incidentally, we also evidence the connection between the generalized double-reversal method for RFAs of \citet{tamm2015generalization} and the one of \citet{Brzozowski2014} for DFAs. 
Finally, we offer a new perspective of the NL\(^*\) algorithm of \citet{bollig2009angluin} for learning RFAs as an algorithm that iteratively refines a quasiorder and uses our automata constructions to build RFAs.

\section{Methodology}

The contributions of this thesis, described in the previous section, are the result of using \emph{monotone well-quasiorders}, i.e. quasiorders that satisfy certain properties with respect to concatenation of words and for which there is no infinite decreasing sequence of elements, as building blocks for tackling problems from \emph{Formal Language Theory}.

Monotone well-quasiorders have proven useful for reasoning about formal languages from a theoretical perspective (see the survey of \citet{d2008well}).
For instance, \citet{ehrenfeucht_regularity_1983} showed that a language is regular if{}f it is closed for a monotone well-quasiorder and \citet{deLuca1994} extended this result by showing that a language is regular if{}f it is closed for a left monotone and for a right monotone well-quasiorders.
On the other hand, \citet{kunc2005regular} used well-quasiorders to show that all maximal solutions of certain systems of inequalities on languages are regular.

Our work evidences that monotone well-quasiorders also have practical applications by placing them at the core of some well-known algorithms.

\paragraph*{Monotone Well-Quasiorders}

\emph{Quasiorders} are binary relations that are \emph{reflexive}, i.e. every word is related to itself, and \emph{transitive}, i.e. if a word ``u'' is related to ``v'' which is related to ``w'' then ``u'' is related to ``w''.

Intuitively, we use quasiorders to group words that behave ``similarly'' (in a certain way) with respect to a given regular language.
This naturally leads to the use of \emph{monotone quasiorders} so that ``similarity'' between words is preserved by concatenation, i.e. when concatenating two ``similar'' words with the same letter the resulting words remain ``similar''.

\begin{exampleC}\label{example:lengthQO}
Consider the \emph{length quasiorder}, which says that ``u'' is related to ``\(v\)'' if{}f \(\len{u} \leq \len{v}\) where \(\len{u}\) denotes the length of a word ``\(u\)''.

It is straightforward to check that this is a monotone quasiorder since 
\smallskip
\begin{myEnumI}
\item \(\len{u} \leq \len{u}\) for every word \(u\), hence it is \emph{reflexive};
\item if \(\len{u} \leq \len{v}\) and \(\len{v} \leq \len{w}\) then \(\len{u} \leq \len{w}\), hence it is \emph{transitive};
\item if \(\len{u} \leq \len{v}\) then \(\len{ua} \leq \len{va}\) for every letter \(a\), hence it is \emph{monotone}.\eox
\end{myEnumI}
\end{exampleC}

The most basic sets of words that can be formed by using a quasiorder are the so called \emph{principals}, i.e. sets of words that are related to a single one which we refer to as the \emph{generating word} of the principal.
For example, given the length quasiorder, the principal with generating word ``u'' is the set of all words ``w'' with \(\len{\text{u}} \leq \len{\text{w}}\).

Finally, when considering \emph{well-quasiorders} we find that the union of the principals of any (possibly infinite) set of words coincides with the union of the principals of a \emph{finite} subset of words. For instance, the quasiorder from Example~\ref{example:lengthQO} is a monotone well-quasiorder since the union of the principals of any infinite set of words coincides with the principal of the shortest word in the set. 

Next, we offer a high-level description on how we use \emph{monotone well-quasi-orders} and their induced principals in each of the contributions of this thesis.

\subsection{Quasiorders for Deciding Language Inclusion}

Consider the language inclusion problem \(L_1 \subseteq L_2\) where \(L_1\) is context-free and \(L_2\) is regular.
The principals of a given monotone well-quasiorder can be used to compute an \emph{over-approximation} of \(L_1\) that consists of a \emph{finite} number of elements.
If the quasiorder is such that a principal is included in \(L_2\) if{}f its generating word is in \(L_2\), then we can reduce the language inclusion problem \(L_1 \subseteq L_2\) to the simpler problem of deciding a finite number of membership queries for \(L_2\).
To do that it suffices to compute the over-approximation of \(L_1\) and check membership in \(L_2\) for the generating words of the finitely many principals that form the over-approximation.
This approach is illustrated in Figure~\ref{fig:LangIncl}.

\begin{figure}[!ht]
\begin{minipage}[l]{0.44\textwidth}
\includegraphics[width=\textwidth]{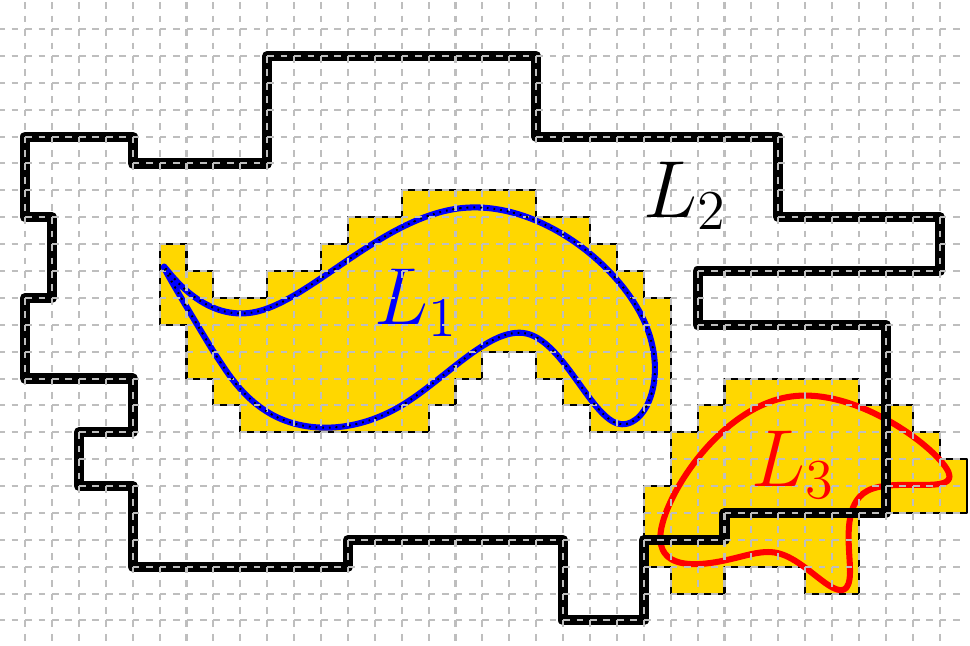}
\end{minipage}%
\hfill
\begin{minipage}[r]{0.53\textwidth}
Given a monotone well-quasiorder whose principals are the dashed squares shown on the image on the left, we compute over-approximations (colored areas) of the languages \(L_1\) and \(L_3\).
Since \(L_2\) is a union of principals, the over-approximation of a language is included in \(L_2\) if{}f the language is included in \(L_2\).
Therefore, we find that \(L_1 \subseteq L_2\) but \(L_3 \nsubseteq L_2\).
\end{minipage}
\caption{Illustration of our quasiorder-based approach for deciding the language inclusion problems \(L_1 \subseteq L_2\) and \(L_3 \subseteq L_2\).}
\label{fig:LangIncl}
\end{figure}

In order to compute the over-approximation of \(L_1\) we successively over-approximate the Kleene iterates of its least fixpoint characterization.
The following example shows the language equations for a context-free language and the first steps of the Kleene iteration, which converges to the least fixpoint of the equations.

\begin{exampleC}\label{example:KleeneIterate}
Consider the language equations \(\{X = aX \cup Ya  \cup bY ,\; Y = a\}\), whose Kleene iterates converge to their least fixpoint:
\[\begin{cases}
    X = \varnothing \\
    Y = \varnothing 
  \end{cases} \hspace{-12pt}\Ra \begin{cases}
    X = \varnothing \\
    Y = \{a\} 
  \end{cases} \hspace{-12pt}\Ra \begin{cases}
    X = \{aa, ba\} \\
    Y = \{a\} 
  \end{cases} \hspace{-12pt}\Ra \ldots \Ra \begin{cases}
    X = a^*(aa\,|\,ba) \\
    Y = \{a\} 
  \end{cases} \enspace \tag*{\eox}\]
\end{exampleC}

This approach for solving language inclusion problems is studied in Chapter~\ref{chap:LangInc}.
In that chapter we present a quasiorder-based framework which, by instantiating it with different monotone well-quasiorders, allows us to systematically derive well-known decision procedures for different language inclusion problems such as the antichains algorithms of \citet{DBLP:conf/cav/WulfDHR06} and \citet{Holk2015}.

Moreover, by switching from least fixpoint equations for computing the over-approximation of \(L_1\) to greatest fixpoint equations, we are able to obtain a \emph{novel} algorithm for deciding language inclusion between regular languages.

\subsection{Quasiorders for Searching on Compressed Text}
Searching with a regular expression in a grammar-compressed text\footnote{By ``searching'' we mean finding subsequences of the uncompressed text that match a regular expression, i.e. that are included in a given regular language.} amounts to deciding whether the language generated by a grammar, which consists of a single word, is included in a regular language.
Therefore, we can apply the quasiorder-based framework described in the previous section, i.e. we can compute an over-approximation of the language generated by the grammar and check inclusion of the over-approximation into the regular language.

However this approach would only indicate whether there is a subsequence in the text that matches the expression and it would not produce enough information to count the matches let alone recover them.

In order to report the exact lines\footnote{We use the standard definition of \emph{line} as a sequence of characters delimited by ``new line'' symbols.} that contain a match (either count them or recover the actual lines), we need to compute some extra information for each variable of the grammar, beyond the over-approximation of the generated language.
Indeed, we need to compute the following information regarding the language generated by each variable, which consists of a single word\footnote{Recall that, in the context of grammar-based compression, the grammar is a compressed representation of a text, hence it generates a single word: the text. As a consequence, each variable of the grammar generates a single word.}, namely \(w\):
\smallskip
\begin{myEnumI}
\item The number of lines that contain a match.%
\item Whether there is a ``new line'' symbol in \(w\).
\item Whether the prefix of \(w\) contains a match.
\item Whether the suffix of \(w\) contains a match. 
\end{myEnumI}
\medskip

This quasiorder-based approach is presented in Chapter~\ref{chap:zearch} where we show that the above mentioned extra information for each variable of the grammar is trivially computed for the terminals and then propagated through all the variables until the axiom.
Furthermore, Chapter~\ref{chap:zearch} includes a detailed description of the implementation and evaluation of the resulting algorithm which, as the experiments show, outperforms the state of the art.

\subsection{Quasiorders for Building Residual Automata}

It is well-known that the construction of the minimal DFA for a language is related to the use of \emph{congruences}, i.e. symmetric monotone quasiorders~\cite{Buchi89,Khoussainov2001}.

Recently, \citet{ganty2019congruence} generalized this idea and offered a congruence-based perspective on minimization algorithms for DFAs.
Intuitively, they build automata by using the principals induced by congruences as states and define the transitions according to inclusions between the principals and the sets obtained by concatenating them with letters.
When the congruence has finite index then it induces a finite number of principals and, therefore, the resulting automata have finitely many states.
Figure~\ref{fig:rho_automata} illustrates this automata construction. 

\begin{figure}[!ht]
\begin{minipage}[l]{0.5\textwidth}
\resizebox{0.8\textwidth}{!}{
%
%
\begin{tikzpicture}[ipe stylesheet]
  \draw[ipe pen fat]
    (394.0097, 540.0653)
     .. controls (404.627, 521.6085) and (407.2196, 488.8923) .. (395.0591, 465.3737)
     .. controls (382.8986, 441.8552) and (355.9849, 427.5341) .. (343.454, 428.8922)
     .. controls (330.9231, 430.2502) and (332.775, 447.2873) .. (337.096, 462.0404)
     .. controls (341.417, 476.7935) and (348.2071, 489.2627) .. (350.1207, 506.3615)
     .. controls (352.0343, 523.4603) and (349.0713, 545.1888) .. (356.9109, 553.9542)
     .. controls (364.7504, 562.7197) and (383.3924, 558.5221) .. cycle;
  \draw[ipe pen fat]
    (173.3333, 480)
     .. controls (165.3333, 442.6667) and (186.6667, 405.3333) .. (214.1168, 384.9429)
     .. controls (241.567, 364.5525) and (275.1339, 361.1049) .. (288.3989, 376.9738)
     .. controls (301.6638, 392.8427) and (294.6267, 428.028) .. (285.8433, 467.0851)
     .. controls (277.0598, 506.1422) and (266.5299, 549.0711) .. (242.5983, 551.8689)
     .. controls (218.6667, 554.6667) and (181.3333, 517.3333) .. cycle;
  \draw[ipe pen fat]
    (199.8822, 505.5973)
     .. controls (218.4311, 488.5279) and (244.8622, 481.0558) .. (265.979, 494.2951)
     .. controls (287.0959, 507.5343) and (302.8984, 541.4851) .. (313.6828, 574.5545)
     .. controls (324.4673, 607.6238) and (330.2336, 639.8119) .. (303.7835, 639.906)
     .. controls (277.3333, 640) and (218.6667, 608) .. (194.6667, 578.6667)
     .. controls (170.6667, 549.3333) and (181.3333, 522.6667) .. cycle;
  \draw[ipe pen fat]
    (189.3971, 580.1623)
     .. controls (158.2806, 537.4317) and (151.8899, 475.5161) .. (173.4553, 428.1731)
     .. controls (195.0206, 380.8302) and (244.5418, 348.0598) .. (292.292, 359.1264)
     .. controls (340.0421, 370.193) and (386.021, 425.0965) .. (403.6132, 473.9229)
     .. controls (421.2053, 522.7492) and (410.4106, 565.4984) .. (393.777, 597.1624)
     .. controls (377.1434, 628.8264) and (354.6709, 649.4051) .. (315.5134, 647.9217)
     .. controls (276.356, 646.4383) and (220.5136, 622.8928) .. cycle;
  \node[ipe node, font=\huge]
     at (256, 592) {\(\rho(a)\)};
  \node[ipe node, font=\huge]
     at (222.916, 416) {\(\rho(b)\)};
  \node[ipe node, font=\huge]
     at (217.489, 515.337) {\(\rho(c)\)};
  \node[ipe node, font=\huge]
     at (360.818, 498.11) {\(\rho(\varepsilon)\)};
  \node[ipe node, font=\huge]
     at (325.183, 583.593) {\(\rho(aa)\)};
  \draw[->, ipe pen fat]
    (309.812, 621.3618)
     .. controls (344.6269, 680.6212) and (401.2937, 645.436) .. (368.701, 600.621);
  \node[ipe node, font=\huge]
     at (343.407, 661.63) {\(a,b,c\)};
  \node[ipe node, font=\huge]
     at (272.666, 325.333) {\(a,b,c\)};
  \draw[->, ipe pen fat]
    (263.1452, 388.3982)
     .. controls (269.8119, 346.5463) and (287.9601, 342.6574) .. (303.1453, 344.787)
     .. controls (318.3305, 346.9166) and (330.5528, 355.0648) .. (326.4787, 391.7316);
  \draw[->, ipe pen fat]
    (375.7381, 462.8429)
     .. controls (441.2938, 473.954) and (447.5901, 526.9171) .. (397.2196, 549.1394);
  \node[ipe node, font=\huge]
     at (438.593, 500.518) {\(a,b,c\)};
  \draw[->, ipe pen fat]
    (354.9973, 494.6948)
     .. controls (326.4787, 474.3244) and (291.6638, 477.6577) .. (260.923, 513.2134);
  \node[ipe node, font=\huge]
     at (309.704, 487.925) {\(c\)};
  \draw[->, ipe pen fat]
    (362.0343, 520.9912)
     .. controls (324.9972, 519.8801) and (302.7749, 538.3986) .. (292.7749, 568.7691);
  \node[ipe node, font=\huge]
     at (308.222, 544.592) {\(a,c\)};
  \node[ipe node, font=\huge]
     at (294.149, 438.294) {\(b,c\)};
  \draw[->, ipe pen fat]
    (340.9232, 446.5465)
     .. controls (326.4787, 423.9539) and (308.3305, 419.8798) .. (274.6267, 438.7687);
  \draw[->, ipe pen fat]
    (203.145, 521.7319)
     .. controls (177.9598, 572.1024) and (159.6264, 570.0654) .. (149.1634, 559.1703)
     .. controls (138.7004, 548.2752) and (136.1078, 528.5221) .. (139.9967, 516.3307)
     .. controls (143.8856, 504.1393) and (154.256, 499.5096) .. (173.5153, 510.2504);
  \node[ipe node, font=\huge]
     at (109.703, 574.592) {\(a,b,c\)};
  \draw[->, ipe pen fat]
    (344.2565, 399.139)
     .. controls (359.8121, 378.3982) and (374.627, 376.9167) .. (383.9788, 387.2871)
     .. controls (393.3307, 397.6575) and (397.2196, 419.8798) .. (351.6639, 422.102);
  \node[ipe node, font=\huge]
     at (393.408, 396.073) {\(a,b,c\)};
\end{tikzpicture}
 }
\end{minipage}\hfill%
\begin{minipage}[r]{0.45\textwidth}
\begin{tikzpicture}[->,>=stealth',shorten >=1pt,auto,node distance=3mm and 7mm,thick,initial text=]
 \tikzstyle{every state}=[scale=0.85,fill=customblue!60,draw=blue!60,text=black,style={draw,ellipse,inner sep=0pt}]
  \node[initial, state] (0) {\(ρ(\varepsilon)\)};
  \node[state] (2) [right=of 0] {\(ρ(c)\)};
  \node[state] (1) [above=of 2] {\(ρ(a)\)};
  \node[state] (3) [below=of 2] {\(ρ(b)\)};
  \node[state] (4) [right=of 2, xshift=0.2cm] {\(ρ(aa)\)};
  
  \path (0) edge [in=60, out=120, distance=70pt] node {\small \(a,b,c\)} (4)  
      (0) edge [bend left] node {\small\(a,c\)} (1)
      (0) edge node {\small\(c\)} (2)
      (0) edge [bend right] node [yshift=-5pt] {\small\(b,c\)} (3);

  \path (1) edge [bend left] node {\small\(a,b,c\)} (4)
      (2) edge node {\small\(a,b,c\)} (4)
      (3) edge [bend right] node [xshift=5pt, yshift=-4pt] {\small\(a,b,c\)} (4)
      (4) edge [loop right] node [above, xshift=-8pt, yshift=3pt] {\small \(a,b,c\)} (4);
  \end{tikzpicture}
\end{minipage}
\caption{The image on the left shows the principals induced by a quasiorder. Each arrow of the form \(ρ(x)\ggoes{a}ρ(y)\) indicates that \(ρ(x) a \subseteq ρ(y)\).
For clarity, we show on the right the automaton resulting from the relation between the principals.}\label{fig:rho_automata}
\end{figure}
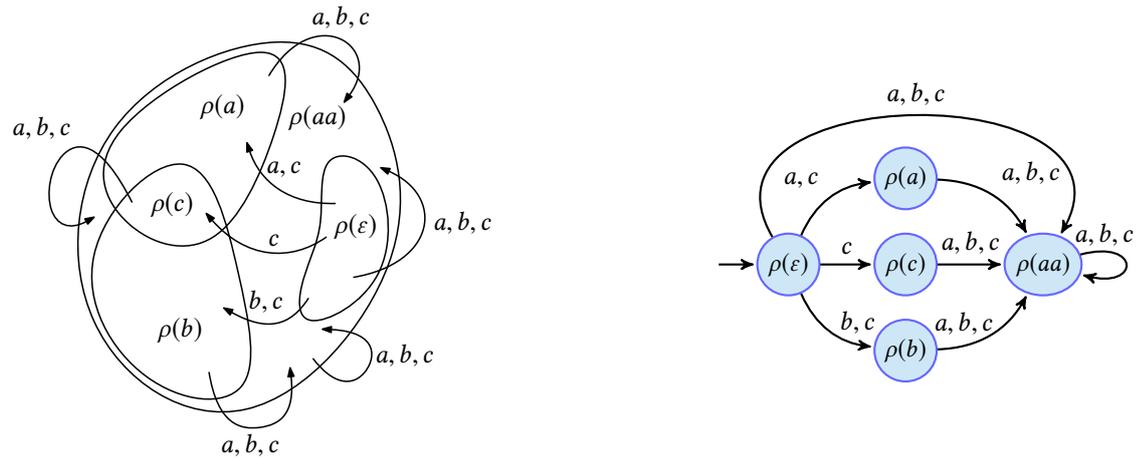

Let \(ρ(u)\) denote the principal for a word \(u\).
The monotonicity of congruences ensures that every set \(ρ(u)a\) is included in a principal \(ρ(v)\) and, since congruences are symmetric, the principals induced by a congruence are disjoint and, therefore, the resulting automata is deterministic.
By switching from congruences to quasiorders we obtain possibly overlapping principals which enables non-determinism and allows us to obtain \emph{residual automata} which, recall, are a generalization of DFAs.
Clearly, the principals shown in Figure~\ref{fig:rho_automata} correspond to a quasiorder rather than a congruence since they are not disjoint.

This quasiorder-based perspective on RFAs is presented in Chapter~\ref{chap:RFA} where we define quasiorder-based automata constructions that yield RFAs or co-RFAs, depending on the properties of the input quasiorder.
Moreover, given two comparable quasiorders, our automata construction instantiated with the coarser quasiorder yields a smaller automaton.
This is to be expected since a coarser quasiorder induces fewer principals which, recall, are the states of the automata.

As a consequence, building the canonical minimal RFA for a given language amounts to instantiating our automata construction with the coarsest quasiorder that satisfies certain requirements.
Interestingly, building the minimal DFA amounts to instantiating the framework of \citet{ganty2019congruence} with the coarsest congruence that satisfies the same requirements.
As we shall see in Chapter~\ref{chap:RFA}, the congruence and the quasiorder used for building the minimal DFA and RFA, respectively, are closely related.

We conclude that \emph{monotone quasiorders} are fundamental for RFAs as \emph{congruences} are fundamental for DFAs, which evidences the relationship between these two
classes of automata.

%

%
%
%
%
%

%
%
%
%
%
%
%
%
%
%
%
%

%
%
%
%
\clearpage{}%
\clearpage{}%
%
\chapter{State of the Art}
\label{chap:related}

In this dissertation, we present two quasiorder-based frameworks that allow us to systematically derive algorithms for solving different language inclusion problems and manipulating residual automata, respectively.
Moreover, we show that our algorithms for deciding language inclusion can be adapted for searching on compressed text.

Our theoretical framework allows us to devise some novel algorithms and offer new insights on existing ones.
Therefore, most of the works related to ours are briefly discussed within the following chapters, when explaining them within our quasiorder-based perspective.
This is the case, specially, in Chapters~\ref{chap:LangInc} and~\ref{chap:RFA}.

However, we present in this chapter a detailed description of some previous works in order to provide an overview of the state of the art for these problems before writing this Ph.D. Thesis.

\section{The Language Inclusion Problem}\label{sec:TheLanguageInclusionProblem}
Consider the language inclusion problem \(L_1 \subseteq L_2\).
When the underlying representations of \(L_1\) and \(L_2\) are regular expressions, one can check language inclusion using some rewriting techniques~\cite{Antimirov1995,keil_et_al:LIPIcs:2014:4841}, thus avoiding the translation of the regular expression into an equivalent automaton.

On the other hand, when the languages are given through finite automata, a well known and standard method to solve the language inclusion problem is to reduce it to a disjointness problem via the construction of the language complement: \(L_1 \subseteq L_2\) if{}f \(L_1 \cap L_2^c = \varnothing\).
The bottleneck of this approach is the language complementation since it involves a determinization step which entails a worst case exponential blowup. 

In order to alleviate this bottleneck, \citet{DBLP:conf/cav/WulfDHR06} put forward a breakthrough result where complementation was sidestepped by a lazy construction of the determinized NFA, which provided a huge performance gain in practice.
Their algorithm, deemed the \emph{antichains} algorithm, was subsequently enhanced with simulation relations by \citet{Abdulla2010}.
The current state of the art for solving the language inclusion problem between regular languages is the bisimulation up-to approach proposed by \citet{DBLP:conf/popl/BonchiP13}, of which the antichains algorithm and their enhancement with simulations can be viewed as particular cases. 

\subsection{Antichains Algorithms}\label{sec:antichainsAlgorithms}
The \emph{antichains} algorithm of \citet{DBLP:conf/cav/WulfDHR06} was originally designed as an algorithm for solving the universality problem for regular languages, i.e. deciding whether \(Σ^* \subseteq L\) holds when \(L\) is regular.

Before the introduction of this algorithm, the standard approach for deciding universality of a regular language given its automaton was to determinize the automaton and check whether all states are final.
The \emph{antichains} algorithm improved this situation by keeping the determinization step implicit.

In their work, \citet{DBLP:conf/cav/WulfDHR06} also adapted their antichains algorithm for solving the language inclusion problem \(L_1 \subseteq L_2\) when both \(L_1\) and \(L_2\) are regular.
Next, we describe this antichains algorithm for solving language inclusion.

Consider the inclusion problem \(L_1 \subseteq L_2\) and let \(\cN_1\) and \(\cN_2\) be finite-state automata generating the languages \(L_1\) and \(L_2\) respectively.
The intuition behind the \emph{antichains} algorithm is to compute, for each state \(q\) of \(\cN_1\), the set \(S_q\) of sets of states of \(\cN_2\) from which no final state of \(\cN_2\) is reachable by reading words generated from \(q\) in \(\cN_1\).\footnote{Note that this is equivalent to finding states of the complement of the determinized version of \(\cN_2\) from which a final state is reachable by reading a word generated from \(q\) in \(\cN_1\).}
Clearly, the inclusion \(L_1 \subseteq L_2\) holds if{}f none of the sets of states computed for the initial states of \(\cN_1\) contain some initial state of \(\cN_2\).

In order to prevent the computation of all possible subsets of \(\cN_2\) from which the final states are non-reachable, which would be equivalent to determinizing \(\cN_2\), the \emph{antichains} algorithm ensures that the set \(S_q\) for each state \(q\) in \(\cN_1\) is an antichain, i.e. \(\forall s,s' \in S_q, \; s \nsubseteq s' \land s' \nsubseteq S\).
The idea behind the use of \emph{antichains} is that, given two sets of states of \(\cN_2\), namely \(s\) and \(s'\), if \(s \subseteq s'\) then if no final state of \(\cN_2\) is reachable from \(s'\) by reading words in a certain set then the same holds for \(s\).
Therefore, discarding the set \(s\) and keeping the set \(s'\) preserves the correctness of the algorithm.
The resulting algorithm is refer to as the \demph{backward antichains algorithm}.

Furthermore, \citet{DBLP:conf/cav/WulfDHR06} also defined a dual of the antichains algorithm described above.
In this case, the algorithm computes the set \(\widetilde{S}_q\) of sets of states of \(\cN_2\) reachable from an initial state by reading a word generated from \(q\) in \(\cN_1\).
In this case, the inclusion \(L_1 \subseteq L_2\) holds if{}f for every initial state \(q\) of \(\cN_1\), all the sets in \(\wt{S}_q\) contain a final state.
Again, by ensuring that \(\widetilde{S}_q\) is an \emph{antichain}, we can reduce the number of sets of states of \(\cN_2\) that need to be computed since, whenever \(s \subseteq s'\), if a final state is reachable from \(s\) by a word in a given language, the same holds for \(s'\) and, therefore, it is possible to discard \(s'\).
The resulting algorithm is referred to as the \demph{forward antichains algorithm}.

The proof of the correctness of the \emph{antichains} algorithm, as presented by \citet{DBLP:conf/cav/WulfDHR06}, heavily depends on the automata representation of the languages.
We believe that our quasiorder-based framework, presented in Chapter~\ref{chap:LangInc}, offers a better understanding on the \emph{antichains} algorithm and its correctness proof by offering a new explanation of the algorithm from a language perspective.

\paragraph*{Improvements on the Antichains Algorithm}
The \emph{antichains} algorithm of \citet{DBLP:conf/cav/WulfDHR06} was later improved by \citet{Abdulla2010}, who used simulations (between states and between sets of states) for reducing the amount of sets of states considered by the algorithm.

In particular, they found that, for the \emph{forward antichains algorithm}, there is no need to add the set \(s\) of states of \(\cN_2\) to the set \(\wt{S}_q\) for a certain state \(q\) of \(\cN_1\) if there exists a state \(q'\) of \(\cN_1\) such that \(q\) simulates \(q'\) and whose associated set \(\wt{S}_{q'}\) contains a set \(s'\) that simulates \(s\).
The idea behind this approach is that simulation is a sufficient condition for language inclusion to hold, i.e. if the set of states \(s'\) simulates the set \(s\) then the language generated from \(s'\) is a subset of the language generated from \(s\).

As we show in Chapter~\ref{chap:LangInc}, this improvement on the \emph{antichains} algorithm can be partially accommodated by our quasiorder-based framework by using simulations in the definition of the quasiorder.
By doing so, the resulting algorithm matches the behavior of the one of ~\citet{Abdulla2010} when \(q = q'\).

On the other hand, \citet{DBLP:conf/popl/BonchiP13} defined a new type of relation between sets of states, denoted \emph{bisimulation up to congruence}, and used it to define a new algorithm for deciding language equivalence between sets of states of a given automaton.

Intuitively, bisimulations up to congruence are enhanced bisimulations (and, therefore, if they relate two sets of states then both sets generate the same language) that might relate sets of states that are not explicitly related by the underlying bisimulation but are related by its implicit congruence closure.
Since \(L_1 \subseteq L_2 \Lra L_1 \cup L_2 = L_2\), the algorithm of \citet{DBLP:conf/popl/BonchiP13} can be used to decide the inclusion \(L_1 \subseteq L_2\) by considering the union automaton \(\cN_1 \cup \cN_2\) and checking whether the bisimulations up to congruence holds between the union of the initial states of \(\cN_1\) and \(\cN_2\), which generate \(L_1 \cup L_2\), and the initial states of \(\cN_2\), which generate \(L_2\).

Finally, \citet{Holk2015} used \emph{antichains} to solve the language inclusion problem \(L_1 \subseteq L_2\) when \(L_1\) is a context-free language and \(L_2\) is regular.
To do that, they reduced the language inclusion problem to a data flow analysis one.
This allowed them to rephrase the language inclusion problem as an inclusion problem between sets of relations on the states of the automaton.
Then, they applied the antichains principle to reduce the number of relations that need to be manipulated.

As we show in Chapter~\ref{chap:LangInc}, our quasiorder-based framework for deciding the language inclusion \(L_1 \subseteq L_2\) also applies to the case in which \(L_1\) is a context-free grammar.
Indeed, when \(L_1\) is regular we instantiate our framework with left or right monotone quasiorders and obtain the antichains algorithm of \citet{DBLP:conf/cav/WulfDHR06} and its variants, among other algorithms.
Similarly, when \(L_1\) is context-free, we use a left and right monotone quasiorders and obtain the antichains algorithm of \citet{Holk2015}, among others.

Therefore, our framework allows us to offer a more direct presentation of the \emph{antichains} algorithm for grammars of \citet{Holk2015} as a straightforward extension of the \emph{antichains} algorithm for regular languages.

\subsection{Solving Language Inclusion through Abstractions}
Our approach draws inspiration from the work of \citet{Hofmann2014}, who considered the language inclusion problem on \emph{infinite words} \(L_1 \subseteq L_2\) where \(L_1\) is  represented by a B{\"u}chi automata and \(L_2\) is regular.

They defined a language inclusion algorithm based on fixpoint computations and a language abstraction based on an equivalence relation between states of the underlying automata representation.   
Although the equivalence relation is folklore (you find it in several textbooks on language theory \cite{Khoussainov2001,Sakarovitch}), \citet{Hofmann2014} were the first, to the best of our knowledge, to use it as an abstraction and, in particular, as a complete domain in abstract interpretation.

As we show in Chapter~\ref{chap:LangInc}, our framework for solving the language inclusion problem also relies on computing the language abstraction of a fixpoint computation. 
However, we focus on languages on finite words and generalize the language abstractions by relaxing their equivalence relations to quasiorders.  
Moreover, by considering quasiorders instead of equivalences, we are able to generalize the fixed point-based approach to check $L_1\subseteq L_2$ when $L_2$ is non-regular.

\section{Searching on Compressed Text}\label{sec:related:Search}

The problem of searching with regular expressions on grammar-compressed text has been extensively studied for the last decades.
Results in this topic can be divided in two main groups:
\begin{myEnumA}
\item[a)] Characterization of the problem's complexity from a theoretical point of view~\cite{plandowski1999complexity,markey2004ptime,Amir2018FineGrained}.
\item[b)] Development of algorithms and data structures to efficiently solve different versions of the problem such as pattern matching~\cite{navarro2005lzgrep,de1998direct,navarro2007compressedIndex}, approximate pattern matching~\cite{bille2009improved,karkkainen2003approximate}, multi-pattern matching~\cite{kida1998multipattern,gawrychowski2014simple}, regular expression matching~\cite{navarro2003regular,bille2009improved} and
subsequence matching~\cite{bille2017compressed}.
\end{myEnumA}
\smallskip

To characterize the complexity of search problems on grammar-compressed text it is common to use \emph{straight line programs} (grammars generating a single string) to represent the output of the compression.
Straight line programs are a natural model for algorithms such as \textsc{LZ78}~\cite{ziv1978compression}, \textsc{LZW}~\cite{welch1984technique}, Recursive Pairing~\cite{larsson2000off} or Sequitur~\cite{nevill1997compression} and, as proven by \citet{Rytter2004Equivalent}, polynomially equivalent to \textsc{LZ77}~\cite{ziv1977compression}.
However, algorithms for searching with regular expressions on grammar-compressed text are typically designed for a specific compression scheme~\cite{navarro2005lzgrep,navarro2003regular,bille2009improved}.

The first algorithm to solve this problem is due to \citet{navarro2003regular} and it is defined for \textsc{LZ78}/\textsc{LZW} compressed text.
His algorithm reports all positions in the uncompressed text at which a substring that matches the expression ends and exhibits $\mathcal{O}(2^s+s\cdot T+\text{occ}\cdot s\cdot \log{s})$  worst case time complexity using $\mathcal{O}(2^s+t\cdot s)$ space, where ``occ'' is the number of occurrences, $s$ is the size of the expression and $T$ is the length of the text compressed to size $t$.
To the best of our knowledge this is the only algorithm for regular expression searching on compressed text that has been implemented and evaluated in practice.

\citet{bille2009improved} improved the result of Navarro by defining a relationship between the time and space required to perform regular expression searching on compressed text.
They defined a data structure of size $o(t)$ to represent \textsc{LZ78} compressed texts and an algorithm that, given a parameter $\tau$, finds all occurrences of a regular expression in a \textsc{LZ78} compressed text in $\mathcal{O}(t\cdot s\cdot (s+\tau)+\text{occ}\cdot s\cdot \log{s})$ time using $\mathcal{O}(t\cdot s^2/\tau + t\cdot s)$ space.
To the best of our knowledge, no implementation of this algorithm was carried out.

We tackle the problem of searching in grammar-compressed text by using our algorithms for deciding language inclusion.
We \emph{adapt} these algorithms to efficiently handle straight line programs and \emph{enhance} them with additional information, that is computed for each variable of the grammar, in order to find the exact matches.

Our approach, presented in Chapter~\ref{chap:zearch}, differs from the previous ones in the generality of its definition since, by working on straight line programs, our algorithm and its complexity analysis apply to any grammar-based compression scheme.
This is a major improvement since, as shown by \citet{Hucke2016Smallest}, the \textsc{LZ78} representation of a text of length $T$ has size $t=\Theta((T/\log(T))^{2/3})$ while its representation as a straight line program has size $t=Ω(\log(T)/(\log\log(T)))$ and $t=\mathcal{O}((T/\log(T))^{2/3})$.
Therefore, our approach allows us to handle much more concise representations of the data.

Moreover, the definition of ``occurrence'' used in previous works, i.e. positions in the uncompressed text from which we can read a match of the expression, is of limited practical interest.
As an evidence, state of the art tools for regular expression searching, such as \tool{grep} or \tool{ripgrep}, define an occurrence as a line of text containing a match of the expression and so do us.

As a consequence, our algorithm reports the number of occurrences of a \emph{fixed} regular expression in a compressed text in $\mathcal{O}(t)$ time while previous algorithms require $\mathcal{O}(T)$ since $\text{occ}=\mathcal{O}(T)$.
Even when there are no matches ($\text{occ}=0$), so previous approaches operate in $\mathcal{O}(t)$ time, the result of \citet{Hucke2016Smallest} shows that our algorithm behaves potentially better than the others.

\paragraph*{Deciding the Existence of a Match}
It is worth to remark that the problem of deciding language inclusion between the languages generated by a straight line program and an automaton has been studied before.
In particular \citet{plandowski1999complexity} reduced this problem to a series of matrix multiplications, showing that it can be solved in $\mathcal{O}(t\cdot s^3)$ time ($\mathcal{O}(t\cdot s)$ for deterministic automata) where $t$ is the size of the grammar and $s$ is the size of the automaton.
Note that this problem corresponds to deciding whether a grammar-compressed text contains a match for a given regular expression.

On the other hand, \citet{esparza00} defined an algorithm to solve a number of decision problems involving automata and context-free grammars which, when restricted to grammars generating a single word, results in a particular implementation of Plandowsky's approach.
Indeed, this implementation coincides with our Algorithm \AlgSLPIncS, presented in Chapter~\ref{chap:zearch} as a straightforward adaptation of the algorithm given in Chapter~\ref{chap:LangInc} for deciding the inclusion of a context-free language into a regular one.

\section{Building Residual Automata}
Residual automata (RFA for short) were first introduced by \citet{denis2000residual,denis2001residual,denis2002residual}.
We deliberatively use the notation RFA for residual automata, instead of the standard RFSA, in order to be consistent with the notation used in this thesis for deterministic (DFA) and non-deterministic (NFA) automata.

When introducing RFAs, \citet{denis2000residual}  defined an algorithm for \emph{residualizing} an automaton, which is an adaptation of the well-known subset construction used for determinization.
Moreover, they showed that there exists a \emph{unique} \emph{canonical} RFA, which is minimal in number of states, for every regular language.
Finally, they showed that the residual-equivalent of the double-reversal method holds, i.e.\ residualizing an automaton \(\cN\) whose reverse is residual yields the canonical RFA for the language generated by \(\cN\).

Later, \citet{tamm2015generalization} generalized the double-reversal method for RFAs by giving a sufficient and necessary condition that guarantees that the residualization operation defined by \citet{denis2002residual} yields the canonical RFA.
This generalization comes in the same lines as that of \citet{Brzozowski2014} for the double-reversal method for DFAs.

In Chapter~\ref{chap:RFA}, we present a quasiorder-based framework of automata constructions inspired by the work of \citet{ganty2019congruence}, who defined a framework of automata constructions based on \emph{equivalences} over words to provide new insights on the relation between well-known methods for computing the minimal \emph{deterministic} automaton of a language.
Intuitively, the shift from equivalences to quasiorders allows us to move from deterministic automata to residual ones.

In their work, \citet{ganty2019congruence} used \emph{congruences}, i.e. monotone equivalences, over words that induce finite partitions over \(\Sigma^*\).
Then, they used well-known automata constructions that yield automata generating a given language \(L\)~\cite{Buchi89,Khoussainov2001} to derive new automata constructions parametrized by a congruence.
As a result, when using the Nerode's congruence for \(L\), their automata construction yields the minimal DFA for \(L\)~\cite{Buchi89,Khoussainov2001} while, when using the so-called \emph{automata-based equivalence} relative to an NFA their construction yields the determinized version of the input NFA.
They also obtained counterpart automata constructions that yield, respectively, the minimal co-deterministic and a co-deterministic automaton for the language.

The relation between the automata constructions resulting from the Nerode's and the automata-based congruences allowed them to relate determinization and minimization operations.
Finally, they re-formulated the generalization of the double-reversal method presented by \citet{Brzozowski2014}, which gives a sufficient and necessary condition that guarantees that determinizing an NFA yields the minimal DFA for the language generated by the NFA.

Our quasiorder-based framework allows us to extend the work of \citet{ganty2019congruence} and devise automata constructions that result in residual automata.
Moreover, we derive a residual-equivalent of the generalized double-reversal method from \citet{Brzozowski2014} that is more general than the one presented by \citet{tamm2015generalization}.
%
%

%

\clearpage{}%
\clearpage{}%
%
\chapter{Background}
\label{chap:prel}

In this section, we introduce all the concepts and notation that will be used throughout the rest of the thesis.

\section{Words and Languages}

Let \(\Sigma\) be a finite nonempty \demph{alphabet} of symbols.
A \demph{string} or \demph{word} \(w\) is a finite sequence of symbols of \(Σ\) where the empty sequence is denoted \(ε\).
We denote \(w^R\) the \demph{reverse} of \(w\) and use \(\len{w}\) to denote the \demph{length} of \(w\) that we abbreviate to \(†\) when \(w\) is clear from the context.
We define \( (w)_i \) as the \(i\)-th symbol of \(w\) if \(1 ≤ i ≤ †\) and \(ε\) otherwise.
Similarly, \((w)_{i,j}\) denotes the substring, also called \demph{factor}, of $w$ between the $i$-th and the $j$-th symbols, both included.
Clearly, \(w = (w)_{1,\dag}\).

We write \(\Sigma^*\) to denote the set of all finite words on $\Sigma$ and write \(\wp(S)\) to denote the set of all subsets of \(S\), i.e. \(\wp(S) \ud \{S' \mid S' \subseteq S\}\).
Given a language \(L \in \wp(\Sigma^*)\), \(L^R \ud \{w^R \mid w \in L\}\) denotes the \demph{reverse} of \(L\) while \(L^c \ud \{w \in Σ^* \mid w \notin L\}\) denotes its \demph{complement}.
Concatenation in \(\Sigma^*\) is simply denoted by juxtaposition, both for concatenating words \(uv\), languages \(L_1L_2\) and words with languages such as \(uLv\).
We sometimes use the symbol \(\cdot\) to refer explicitly to concatenation.

\begin{definition*}[Quotient]
Let \(L \subseteq Σ^*\) and \(u \in Σ^*\).
The \emph{left quotient} of \(L\) by the word \(u\) is the set of suffixes of the word \(u\) in \(L\), i.e. 
\[u^{-1}L \ud \{w \in Σ^* \mid uw \in L\}\enspace .\]
Similarly, the \emph{right quotient} of \(L\) by the word \(u\) is the set of all prefixes of \(u\) in \(L\), i.e. 
\[Lu^{-1} \ud \{w \in Σ^* \mid wu \in L\}\enspace . \]
Finally, we lift the notions of left and right quotients by a word to sets \(S\subseteq Σ^*\) as:
\[S^{-1}L \ud \{w \in Σ^* \mid \forall s\in S,\; sw \in L\} \;\text{ and }\; LS^{-1} \ud \{w \in Σ^* \mid \forall s\in S,\; ws \in L\} \tag*{\hfill\rule{0.5em}{0.5em}}\]
\end{definition*}

Note that the definition of quotient by a set is unconventional as it uses the universal quantifier instead of existential. 
We use this definition since it guarantees that the quotient by a set is the adjoint of concatenation, i.e.
\[XY \subseteq L \Lra Y \subseteq X^{-1}L \Lra X \subseteq LY^{-1}\enspace .\]

\begin{definitionNI*}[Composite and Prime Quotients]\index{Quotient!composite}\index{Quotient!prime}
A left (resp. right) quotient \(u^{-1}L\) is \emph{composite} if{}f it is the union of all the left (resp. right) quotients that it strictly contains, i.e. 
\[u^{-1}L = \hspace{-15pt}\bigcup_{x \in Σ^*, \; x^{-1}L \subsetneq u^{-1}L}\hspace{-15pt} x^{-1}L\qquad\qquad (\text{resp. }Lu^{-1} = \hspace{-15pt}\bigcup_{x \in Σ^*, \; Lx^{-1} \subsetneq Lu^{-1}}\hspace{-15pt} Lx^{-1})\enspace .\]
When a quotient is not composite, we say it is \emph{prime}.\eod
\end{definitionNI*}

\section{Finite-state Automata}\label{sec:FSA}
Throughout this dissertation we consider three different classes of automata: non-deterministic, deterministic and residual.
Next, we define these classes of automata and introduce some basic notions related them.

\paragraph*{Non-Deterministic Finite-State Automata}

\begin{definition*}[NFA]\index{non-deterministic automaton}
A \emph{non-deterministic finite-state automaton} (NFA for short) is a tuple \(\cN=\tuple{Q,Σ,\delta,I,F}\) where \(\Sigma\) is the \emph{alphabet}, \(Q\) is the finite \emph{set of states}, \(I\subseteq Q\) is the subset of \emph{initial states}, \(F\subseteq Q\) is the subset of \emph{final states}, and \(\delta\colon  Q\times \Sigma \ra \wp(Q)\) is the \emph{transition relation}. \eod
\end{definition*}

We sometimes use the notation \(q\ggoes{a} q'\) to denote that \(q'\in \delta(q,a)\). 
If \(u\in \Sigma^*\) and \(q,q'\in Q\) then \(q \stackrel{u}{\leadsto} q'\) means that the state \(q'\) is reachable 
from \(q\) by following the string \(u\). 
Formally, by induction on the length of $u\in \Sigma^*$:
\begin{myEnumI}
\item if $u=\epsilon$ then \(q \goes{\epsilon} q'\) if{}f \(q=q'\);
\item if $u=av$ with $a\in \Sigma,v\in \Sigma^*$ then \(q \goes{av} q'\) if{}f $\exists q''\in \delta(q,a),\; q''\goes{v}q'$.
\end{myEnumI}
\medskip

The \emph{language} generated by an NFA \(\cN\), often referred to as the \demph{language accepted} by \(\cN\) is \(\lang{\cN}\ud\{u \in \Sigma^* \mid \exists q_i\in I, \exists q_f \in F, \; q_i\goes{u}q_f\}\).
We define the successors and the predecessors of a set \(S \subseteq Q\) by a word \(w \in Σ^*\) as: 
\begin{align*}
\mindex{\post_w^{\cN}}(S) & \ud \{q \in Q \mid \exists q' \in S, \; q' \goes{w} q\} &
\mindex{\pre_w^{\cN}}(S) & \ud \{q \in Q \mid \exists q' \in S, \; q \goes{w}q'\} \enspace . 
\end{align*}
In general, we omit the automaton \(\cN\) from the superscript when it is clear from the context.
Figure~\ref{fig:NFA} shows an example of an NFA.

\begin{figure}[!ht]
\centering
\begin{tikzpicture}[->,>=stealth',shorten >=1pt,auto,node distance=5mm and 1cm,thick,initial text=]
\tikzstyle{every state}=[scale=0.75,fill=customblue!60,draw=blue!60,text=black]
      
\node[initial,state] (0) {\(0\)};
\node[state] (1) [right=of 0] {\(1\)};
\node[state] (2) [right=of 1] {\(2\)};
\node[state, accepting] (3) [right=of 2] {\(3\)};

\path (0) edge node {\(a\)} (1)
        (0) edge[loop above] node {\(a,b\)} (0)
      (1) edge node {\(a,b\)} (2)
      (2) edge node {\(a\)} (3)
      (2) edge[bend right=45] node[above] {\(b\)} (0)
      (3) edge[loop above] node {\(a,b\)} (3);
\end{tikzpicture}
\caption{An NFA \(\cN\) with \(Σ = \{a,b\}\) and \(\lang{\cN}= Σ^*aΣaΣ^*\).}
\label{fig:NFA}
\end{figure}

Given \(S,T \subseteq Q\), define 
\[\mindex{W^{\cN}_{S,T}} \ud \{w \in \Sigma^* \mid \exists q \in S, q' \in T, \; q \goes{w} q')\}\enspace .\]

When \(S\) or \(T\) are singletons, we abuse of notation and write \(W_{q,T}^{\cN}\), \(W_{S,q'}^{\cN}\) or even \(W_{q,q'}^{\cN}\).
In particular, when \(S = \{q\}\) and \(T = F\), we say that \(W^{\cN}_{q,F}\) is the \demph{right language} of \(q\).
Likewise, when \(S = I\) and \(T = \{q\}\), we say that \(W^{\cN}_{I,q}\) is the \demph{left language} of \(q\).
We say that a state \(q\) is \demph{unreachable} if{}f \(W^{\cN}_{I,q} = \varnothing\) and we say that \(q\) is \demph{empty} if{}f \(W^{\cN}_{q,F} = \varnothing\).
Finally, note that 
\[\lang{\cN} = \bigcup_{q \in I} W_{q,F}^{\cN} = \bigcup_{q \in F} W_{I,q}^{\cN} = W_{I,F}^{\cN} \enspace . \]

\begin{definition*}[Sub-automaton]
Let \(\cN = \tuple{Q,Σ,δ,I,F}\) be an NFA.
A \emph{sub-automaton} of \(\cN\) is an NFA \(\cN' = \tuple{Q', Σ, δ', I', F'}\) for which \(Q' \subseteq Q\), \(F' \subseteq F\), \(I' \subseteq I\) and for every \(q,q' \in Q\) and \(a \in Σ\) we have that \(q' \in δ'(q,a) \Ra q' \in δ(q,a)\).\eod
\end{definition*}

Clearly, if \(\cN'\) is a sub-automaton of \(\cN\) then \(\lang{\cN'}\subseteq \lang{\cN}\).

\begin{definition*}[Reverse Automaton]
Let \(\cN = \tuple{Q,Σ,δ,I,F}\) be an NFA.
The \emph{reverse} of \(\cN\) is the NFA \(\cN^R \!\ud\! \tuple{Q, \Sigma, \delta^R, F, I}\) where for every \(q,q' \in Q\) and \(a \in Σ\) we have that \(q \!\in\! \delta^R (q',a) \Lra q' \!\in\! \delta(q,a)\).\eod
\end{definition*}
It is straightforward to check that \(\lang{\cN}^R = \lang{\cN^R}\).

\paragraph*{Deterministic Finite-State Automata}

\begin{definitionNI*}[DFA and co-DFA]\index{DFA}\index{co-DFA}\index{deterministic automaton}\index{co-deterministic automaton}
A \emph{deterministic finite-state automaton} (DFA for short) is an NFA such that \(I = \{q_0\}\) and, for every state \(q \in Q\) and every symbol \(a \in \Sigma\), there exists \emph{at most} one state \(q' \in Q\) such that \(\delta(q,a) = q'\).

A \emph{co-deterministic finite-state automaton} (co-DFA for short) is an NFA \(\cN\) such that \(\cN^R\) is a DFA.\eod
\end{definitionNI*}

\begin{definition*}[Subset Construction]\index{determinization}
Let \(\cN = \tuple{Q, \Sigma, \delta, I, F}\) be an NFA.
The \emph{subset construction} builds a DFA \(\cN^D \ud \tuple{Q^D, \Sigma, \delta^D, I^D, F^D}\) where
\begin{align*}
Q^D & \ud \{\post_u^{\cN}(I) \mid u \in Σ^*\} \\ 
I^D & \ud \{I\} \\ 
F^D & \ud \{S \in \wp(Q) \mid S \cap F \neq \varnothing\}\\ 
\delta^D(S,a) & \ud \{q' \mid \exists q \in S, \; q' \in \delta(q,a)\} \text{ for every \(S \in Q\) and \(a \in \Sigma\)} \tag*{\eod}
\end{align*}
\end{definition*}

Given an NFA \(\cN\), we denote by \(\cN^D\) the DFA that results from applying the subset construction to \(\cN\) where only subsets that are reachable from the initial states of \(\cN^D\) are used.
As shown by \citet{Ullman2003}, \(\lang{\cN^D} = \lang{\cN}\) for every automaton \(\cN\).
Figure~\ref{fig:DFA} shows the DFA obtained when applying the subset construction to the NFA from Figure~\ref{fig:NFA}.

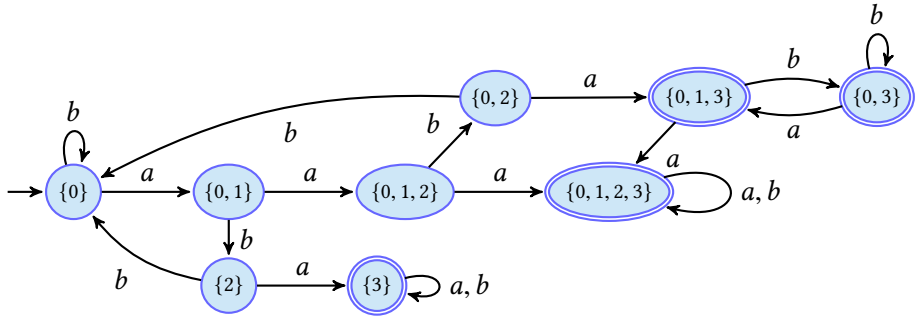
\begin{figure}[!ht]
\centering
\begin{tikzpicture}[->,>=stealth',shorten >=1pt,auto,node distance=5mm and 1.2cm,thick,initial text=]
\tikzstyle{every state}=[scale=0.75,fill=customblue!60,draw=blue!60,text=black,style={draw,ellipse,inner sep=0pt}]
\node[initial,state] (0) {\(\{0\}\)};
\node[state] (1) [right=of 0] {\(\{0,1\}\)};
\node[state] (2) [right=of 1] {\(\{0,1,2\}\)};
\node[state, accepting] (3) [right=of 2] {\(\{0,1,2,3\}\)};
\node[state] (4) [above=of 2, xshift=45pt] {\(\{0,2\}\)};
\node[state, accepting] (5) [above=of 3, xshift=45pt] {\(\{0,1,3\}\)};
\node[state, accepting] (6) [right=of 5] {\(\{0,3\}\)};
\node[state] (7) [below=of 1] {\(\{2\}\)};
\node[state, accepting] (8) [right=of 7] {\(\{3\}\)};

\path (0) edge node {\(a\)} (1)
      (0) edge[loop above] node[above] {\(b\)} (0)
      (1) edge node {\(a\)} (2)
      (1) edge node {\(b\)} (7)
      (2) edge node {\(a\)} (3)
      (2) edge node {\(b\)} (4)
      (3) edge[loop right] node[right] {\(a,b\)} (3)
      (4) edge node {\(a\)} (5)
      (4) edge[bend right=15] node {\(b\)} (0)
      (5) edge node {\(a\)} (3)
      (5) edge[bend left=15] node {\(b\)} (6)
      (6) edge[bend left=15] node {\(a\)} (5)
      (6) edge[loop above] node[above] {\(b\)} (6)
      (7) edge node {\(a\)} (8)
      (7) edge[bend left=20] node {\(b\)} (0)
      (8) edge[loop right] node[right] {\(a,b\)} (8)   
;
\end{tikzpicture}
\caption{DFA \(\cN^D\) obtained by determinizing the NFA from Figure~\ref{fig:NFA}.}
\label{fig:DFA}
\end{figure}

A DFA for the language \(\lang{\cN}\) is \emph{minimal}, denoted by \(\cN^{DM}\), if it has no unreachable states and no two states have the same right language.
For instance, the DFA from Figure~\ref{fig:DFA} is not minimal since the states \(\{0,1,3\},\{0,3\},\{0,1,2,3\}\) and \(\{3\}\) all have the same right language.
The minimal DFA for a regular language is \emph{unique} modulo isomorphism and is determined by the right quotients of the generated language.

\begin{definition*}[Minimal DFA]
Let \(L\) be a regular language.
The \emph{minimal DFA} for \(L\) is the DFA \(\cD \ud \tuple{Q^D, Σ, δ^D, I^D, F^D}\) where 
\begin{align*}
Q^D & \ud \{u^{-1}L \mid u \in Σ^*\}\\
I^D & \ud \{u^{-1}L \in Q \mid u^{-1}L \subseteq L\}\\
F^D & \ud \{u^{-1}L \in Q \mid \varepsilon \in u^{-1}L\}\\
δ^D(u^{-1}L, a) & \ud \{v^{-1}L \in Q \mid v^{-1}L = a^{-1}(u^{-1}L)\} \text{ for every \(u^{-1}L \in Q\) and \(a \in Σ\)} \tag*{\eod}
\end{align*}
\end{definition*}

\paragraph*{Residual Finite-State Automata}

\begin{definitionNI*}[RFA and co-RFA]\index{RFA}\index{co-RFA}\index{residual automaton}\index{co-residual automaton}
A \emph{residual finite-state automaton} (RFA for short) is an NFA such that the right language of each state is a left quotient of the language generated by the automaton.

A \emph{co-residual automaton} (co-RFA for short) is an NFA \(\cN\) such that \(\cN^R\) is residual, i.e. the left language of each state is a right quotient of the language generated by the automaton.\eod
\end{definitionNI*}

Formally, an RFA is an NFA \(\cN = \tuple{Q, Σ, δ, I, F}\) satisfying 
\[\forall q \in Q, \exists u \in Σ^*, \; W_{q,F} = u^{-1}\lang{\cN}\enspace .\]
Similarly, \(\cN\) is a co-RFA if{}f it satisfies 
\[\forall q \in Q, \exists u \in Σ^*, \; W_{I,q} = Lu^{-1}\enspace .\]

The right quotients of the form \(u^{-1}L\), where \(L \subseteq Σ^*\) is a language and \(u \in Σ^*\), are also known as \demph{residuals}, which gives name to RFAs.
We say \(u \in Σ^*\) is a \demph{characterizing word} for \(q \in Q\) if{}f \(W_{q,F}^{\cN} = u^{-1}\lang{\cN}\) and we say \(\cN\) is \emph{consistent}\index{consistent RFA} if{}f each state \(q\) is reachable by a characterizing word for \(q\).
Moreover, \(\cN\) is \emph{strongly consistent}\index{strongly consistent RFA} if{}f every state \(q\) is reachable by every characterizing word of \(q\).

Similarly to the case of DFAs, there exists a \emph{residualization} operation~\cite{denis2002residual} that, given an NFA \(\cN\), builds an RFA \(\mindex{\cN^{\text{res}}}\) such that \(\lang{\cN^{\text{res}}} = \lang{\cN}\).
This construction can be seen as a determinization followed by the removal of coverable states and the addition of new transitions.
We say that the set \(\post_u^{\cN}(I)\) is \demph{coverable} if{}f 
\[\post_u^{\cN}(I) = \hspace{-25pt}\bigcup_{x\in\Sigma^*, \; \post_x^{\cN}(I) \subsetneq \post_u^{\cN}(I)}\hspace{-25pt}\post_{x}^{\cN}(I) \enspace .\]

\begin{definition*}[Residualization]\index{\(\cN^{\text{res}}\)}
Let \(\cN = \tuple{Q, Σ, δ, I, F}\) be an NFA.
Then the \emph{residualization} operation builds the RFA \(\cN^{\text{res}} \ud \tuple{\widetilde{Q}, Σ, \widetilde{δ}, \widetilde{I}, \widetilde{F}}\) with 
\begin{align*}
\widetilde{Q} & \ud \{ \post_u^{\cN}(I) \mid u \in Σ^* \land \post_u^{\cN}(I) \text{ is not coverable}\}\\
\widetilde{I} & \ud \{S \in \widetilde{Q} \mid S \subseteq I\}\\
\widetilde{F} & \ud \{S \in \widetilde{Q} \mid S \cap F \neq \varnothing\}\\
\widetilde{δ}(S, a) & = \{S' \in \widetilde{Q} \mid S' \subseteq δ(S, a)\} \text{ for every \(S \in \widetilde{Q}\) and \(a \in Σ\)}\tag*{\eod}
\end{align*}  
\end{definition*}

Figure~\ref{fig:RFA} shows the RFA obtained by applying the residualization operation to the NFA from Figure~\ref{fig:NFA}.

\begin{figure}[!ht]
\centering
\begin{tikzpicture}[->,>=stealth',shorten >=1pt,auto,node distance=7mm and 1.4cm,thick,initial text=]
\tikzstyle{every state}=[scale=0.75,fill=customblue!60,draw=blue!60,text=black,style={draw,ellipse,inner sep=0pt}]
\node[initial,state] (0) {\(\{0\}\)};
\node[state] (1) [right=of 0] {\(\{0,1\}\)};
\node[state] (7) [below=of 1] {\(\{2\}\)};
\node[state, accepting] (8) [right=of 7] {\(\{3\}\)};

\path (0) edge[bend left=15] node {\(a\)} (1)
      (0) edge[loop above] node[above] {\(b\)} (0)
      (1) edge[bend left=15] node {\(a\)} (0)
      (1) edge[loop above] node[above] {\(a\)} (1)
      (1) edge node {\(a,b\)} (7)
      (7) edge node {\(a\)} (8)
      (7) edge[bend left=20] node {\(b\)} (0)
      (8) edge[loop right] node[right] {\(a,b\)} (8)   
;
\end{tikzpicture}
\caption{RFA \(\cN^{\text{res}}\) obtained when residualizing to the NFA from Figure~\ref{fig:NFA}.}
\label{fig:RFA}
\end{figure}
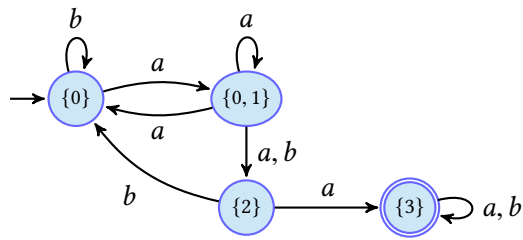

Similarly, to the case of DFAs, there exists an RFA for every regular language that is minimal in the number of states and is \emph{unique} modulo isomorphism: the \emph{canonical RFA}.

\begin{definition*}[Canonical RFA]
Let \(L\) be a regular language.
The \emph{canonical RFA} for \(L\) is the RFA \(\cC \ud \tuple{Q^C, Σ, δ^C, I^C, F^C}\) with 
\begin{align*}
Q^C & \ud \{u^{-1}L \mid u \in Σ^*, \; u^{-1}L \text{ is prime}\}\\
I^C & \ud \{u^{-1}L \in Q \mid u^{-1}L \subseteq L\}\\
F^C & \ud \{u^{-1}L \in Q \mid \varepsilon \in u^{-1}L\}\\
δ^C(u^{-1}L, a) & \ud \{v^{-1}L \in Q \mid v^{-1}L \subseteq a^{-1}(u^{-1}L)\} \text{ for every \(u^{-1}L \in Q\) and \(a \in Σ\)} \tag*{\eod}
\end{align*}
\end{definition*}
The canonical RFA is a strongly consistent RFA and it is the \emph{minimal} (in number of states) RFA such that \(\lang{\cC} = L\)~\cite{denis2002residual}.
Moreover, by definition, the canonical RFA has the \emph{maximal} number of transitions.

Finally, it is straightforward to check that any DFA \(\cD\) is also an RFA since \(W_{q,F}^{\cD} = u^{-1}L\) for all \(u \in W_{I,q}^{\cD}\).
Therefore, we have the following relations between these classes of automata:

\[\text{{\large DFA}} \subsetneq \text{{\large RFA}} \subsetneq \text{{\large NFA}} \enspace .\]

\section{Context-free Grammars}

\begin{definitionNI*}[CFG]\index{grammar}\index{Context-Free Grammar}\index{CFG}
A \emph{context-free grammar} (grammar or CFG for short) is a tuple \(\cGr \ud \tuple{\cV,\Sigma, P}\) where \(\cV=\{X_0,\ldots,X_n\}\) is a finite set of \emph{variables} including the \emph{start symbol} \(X_0\) (also denoted \demph{axiom}), \(\Sigma\) is a \emph{finite alphabet} of terminals and \(P\) is the set of \emph{rules} \(X_i\ra \beta\) where \(\beta\in (\cV\cup\Sigma)^*\)\eod
\end{definitionNI*}

In the following we assume, for simplicity and without loss of generality, that grammars are always given in Chomsky Normal Form (CNF)\index{Chomsky Normal Form}\index{CNF}~\cite{DBLP:journals/iandc/Chomsky59a}, that is, every rule \(X_i \ra \beta\in P\) is such that \(\beta\in (\cV\times \cV) \cup \Sigma \cup \{\epsilon\}\) and if $\beta=\epsilon$ then \(i=0\).
We also assume that for all \(X_i \in \cV\) there exists a rule \(X_i \ra \beta\in P\), otherwise \(X_i\) can be safely removed from \(\cV\).

Given two strings \(w, w' \in (\cV \cup \Sigma)^*\) we write \(w \Ra w'\) if{}f there exists two strings \(u, v \in (\cV \cup \Sigma)^*\) and a grammar rule \(X \to \beta \in P\) such that \(w = u X v\) and \(w' = u\beta v\).
We denote by \(\Ra^*\) the reflexive-transitive closure of \(\Ra\).

The \emph{language} generated by a \(\cGr\) is \(\lang{\cGr} \ud \{w \in \Sigma^* \mid X_0 \Ra^* w\}\).

\paragraph*{Straight-line Programs}
In the context of grammar-based compression we are interested in straight line programs, i.e. grammars generating exactly one word.

\begin{definitionNI*}[SLP]\index{Straight Line Program}\index{SLP}
A \emph{straight line program} (SLP for short), is a CFG $\cP=\tuple{\cV,Σ,P}$ where the set of rules is of the form 
\[P \ud \{X_i → α_i β_i \mid 1 \qo i \qo \len{V}, \; α_i,β_i \in (Σ \cup \{X_1,…,X_{i-1}\}\}\enspace . \]
We refer to $X_{\len{V}} → \alpha_\len{V}β_\len{V}$ as the \emph{axiom rule}.\eod
\end{definitionNI*}

It is straightforward to check that the language generated by an SLP consists of a single string $w ∈ Σ^*$ and, by definition, $\len{w} > 1$.
Since \(\lang{P}=\{w\}\) we identify \(w\) with \(\lang{P}\).

\section{Quasiorders}
Let $f:X\ra Y$ be a function between sets and let $S\in \wp(X)$.
We denote the image of \(f\) on $S$ by $f(S) \ud \{f(x) \in Y \mid x\in S\}$. 
The composition of two functions $f$ and $g$ is denoted by $fg$ or $f\comp g$.

A \demph{quasiordered set} (qoset for short) is a tuple \(\tuple{D,\mathord{\leqslant}}\) such that \(\mathord{\leqslant}\) is a \demph{quasiorder} (qo for short) relation on $D$, i.e.\ a reflexive and transitive binary relation. 
Given a qoset \(\tuple{D,\mathord{\leqslant}}\) we denote by \(\sim_D\) the equivalence relation induced by \(\qo\):
\[d \sim_D d' \udrshort d\leqslant d' \:\wedge\: d' \leqslant d, \quad\text{for all $d,d'\in D$}\enspace .\]
Moreover, given a qo \(\qo\) we denote its strict version by \(<\):
\[u < v \udiff u \qo v \land  v \not\qo u \enspace. \]

We say that a qoset satisfies the \emph{ascending} (resp.\ \emph{descending}) \emph{chain condition} (\demph{ACC}, resp.\ \demph{DCC}) if there is no countably infinite  sequence of distinct elements \(\{x_i\}_{i \in \mathbb{N}}\) such that, for all $i\in\bN$, \(x_i \leqslant x_{i{+}1}\) (resp. \(x_{i{+}1} \leqslant x_{i}\)).
If a qoset satisfies the ACC (resp. DCC) we say it is ACC (resp. DCC).

\begin{definitionNI*}[Closure and Principals]\index{Closure}\index{Principal}
Let \({\qo}\) be a quasiorder on \(Σ^*\) and let \(S \subseteq Σ^*\).
The \emph{closure} of \(S\) is 
\[ρ_{\qo}(S) \ud \{w \in Σ^* \mid \exists x \in S, \; x \qo w\}\enspace .\]
We say \(ρ_{\qo}(S)\) is a \emph{principal} if \(S\) is a singleton.
In that case, we abuse of notation and write \(ρ_{\qo}(u)\) instead of \(ρ_{\qo}(\{u\})\).\eod
\end{definitionNI*}

Given two quasiorders \(\mathord{\qo}\) and \(\mathord{\qo'}\) we say that \(\mathord{\qo}\) is finer than \(\mathord{\qo'}\) (or \(\mathord{\qo'}\) is coarser than \(\mathord{\qo}\)) and write  \(\mathord{\qo} \subseteq \mathord{\qo'}\) if{}f \(ρ_{\qo}(S)\subseteq ρ_{\qo'}(S)\) for every set \(S \subseteq \Sigma^*\).

\begin{definitionNI*}[Left and Right Quasiorders]\index{Left quasiorder}\index{Right quasiorder}
Let \(\qo\) be a quasiorder.
We say \(\qo\) is \emph{right monotone} (or equivalently, \(\qo\) is a \emph{right quasiorder}), and denote it by \(\qr\), if{}f 
\[u \qr v \Ra ua \qr va, \quad \text{ for all \(u,v \in Σ^*\) and \(a \in Σ\)} \enspace .\]

Similarly, we say \(\qo\) is a \emph{left quasiorder}, and denote it by \(\ql\), if{}f 
\[u \ql v \Ra au \ql av, \quad \text{ for all \(u,v \in Σ^*\) and \(a \in Σ\)}\enspace \tag*{\eod}\]
\end{definitionNI*}

A qoset \(\tuple{D,\mathord{\leqslant}}\) is a \demph{partially ordered set} (poset for short) when \(\mathord{\leqslant}\) is antisymmetric.
A subset $X\subseteq D$ of a poset is \demph{directed} if{}f $X$ is nonempty and every pair of elements in $X$ has an upper bound in $X$.

\begin{definition*}[Least Upper Bound]\index{lub}
Let \(\tuple{D,\qo}\) be a partially ordered set and let \(x, y \in D\).
The \emph{least upper bound} of \(x\) and \(y\) is the element \(z \in D\) such that 
\[x \qo z \land  y \qo z \land  \left(\forall d \in D, \; (x \qo d \land y \qo d)\Ra z \qo d \right)\enspace . \tag*{\eod}\]
\end{definition*}

\begin{definition*}[Greatest Lower Bound]\index{glb}
Let \(\tuple{D,\qo}\) be a partially ordered set and let \(x, y \in D\).
The \emph{greatest lower bound} of \(x\) and \(y\) is the element \(z \in D\) such that 
\[z \qo x \land z \qo y \land \left(\forall d \in D, \; (d \qo x \land d \qo y)\Ra d \qo z \right)\enspace . \tag*{\eod}\]
\end{definition*}

A poset \(\tuple{D,\mathord{\leqslant}}\) is a \demph{directed-complete partial order} (CPO for short) if{}f it has the least upper bound (lub for short) of all its directed subsets. 
A poset is a \demph{join-semilattice} if{}f it has the lub of all its nonempty finite subsets (therefore binary lubs are enough). 
A poset is a \demph{complete lattice} if{}f it has 
the lub of all its arbitrary (possibly empty) subsets; in this case, let us recall that 
it also has the greatest lower bound (glb for short) of all its arbitrary subsets. 

\paragraph{Well-quasiorders}

\begin{definition*}[Antichain]
Let \(\tuple{D,\mathord{\leqslant}}\) be a qoset.
A subset $X \subseteq D$ is an \emph{antichain} if{}f any two distinct elements in $X$ are incomparable. \eod
\end{definition*}

We denote the set of antichains of a qoset $\tuple{D,\mathord{\leqslant}}$ by 
\[\AC_{\tuple{D,\mathord{\leqslant}}} \ud \{X\subseteq D \mid X \text{ is an antichain}\}\enspace .\]

\begin{definition*}[Well-quasiorder]
Let \(\tuple{D,\mathord{\leqslant}}\) be a quasiordered set.
We say it is a \emph{well-quasiordered set} (wqoset for short), and $\mathord{\leqslant}$ is a \emph{well-quasiorder} (wqo for short), if{}f for every countably infinite  sequence of elements \(\{x_i\}_{i\in \bN}\) there exist \(i,j\in \bN\) such that \(i<j\) and \(x_i\leqslant x_j\). 

Equivalently, we say \(\tuple{D,\mathord{\leqslant}}\) is a well-quasiordered set if{}f $D$ is DCC and $D$  has no infinite antichain. \eod
\end{definition*}

For every qoset \(\tuple{D,\mathord{\leqslant}}\), we shift the quasiorder \(\qo\) to a binary relation $\sqsubseteq_{\leqslant}$ on the powerset as follows.
Given \(X,Y\in \wp(D)\), 

\[X\mindex{\sqsubseteq_{\leqslant}} Y \udr \forall x\in X, \exists y\in Y,\; y\leqslant x \enspace .
\] 

When the quasiorder is clear from the context, we drop the subindex and write simply \(\sqsubseteq\).
Given a qoset \(\tuple{D, \leqslant}\), we define the set of \demph{minimal elements} of a subset \(X \subseteq D\):
\[\mindex{\minim_{\leqslant}}(X) \ud \{x \in X \mid \forall y \in X, y \leqslant x \Ra y=x\}\enspace .\]

\begin{definition*}[Minor]
Let \(\tuple{D, \leqslant}\) be a qoset.
A \emph{minor} of a subset \(X \subseteq D\), denoted by \(\minor{X}\), is a subset of the minimal elements of \(X\) w.r.t.\ \(\leqslant\), i.e. \(\minor{X}\subseteq\minim_{\leqslant}(X)\), such that 
\(X \sqsubseteq \minor{X}\) holds.\eod
\end{definition*}

Clearly, a minor $\minor{X}$ of some set \(X\) is always an antichain. 

Let us recall that every subset $X$ of a wqoset \(\tuple{D,\leqslant}\) has at least one minor set, all minor sets of $X$ are finite, 
$\minor{\{x\}}=\{x\}$, $\minor{\varnothing}=\varnothing$, and
if \(\tuple{D,\mathord{\leqslant}}\) is additionally a poset then there exists exactly one minor set of $X$.
It turns out that \(\tuple{\AC_{\tuple{D,\mathord{\leqslant}}},\sqsubseteq}\) is a qoset which is ACC if \(\tuple{D,\leqslant}\) is a wqoset and is a poset if \(\tuple{D,\leqslant}\) is a poset.

\paragraph{Nerode Quasiorders}
\begin{definition*}[Nerode's Quasiorders]
Let \(L \subseteq Σ^*\) be a language.
The left and right \emph{Nerode's quasiorders} on \(\Sigma^*\) are, respectively
\begin{align*}\
      u\mindex{\qlL} v &\udrshort\; L u^{-1} \subseteq L v^{-1} \,,&
      u\mindex{\qrL} v &\udrshort\; u^{-1} L \subseteq v^{-1} L \tag*{\eod} 
\end{align*}
\end{definition*}

As shown by \citet{deLuca1994}, \(\mathord{\qlL}\) and \(\mathord{\qrL}\) are, respectively, left and right monotone and, if $L$ is regular then both \(\mathord{\qlL}\) and \(\mathord{\qrL}\) are wqos \citep[Theorem~2.4]{deLuca1994}. 

Furthermore, \citet{deLuca1994} showed that \(\mathord{\qlL}\) is maximum in the set of all left monotone quasiorders \(\ql\) that satisfy \(ρ_{\ql}(L) = L\).
Therefore, for every left quasiorder \(\ql\), if \(ρ_{\ql}(L) = L\) then \(x \ql y \Ra x \qlL y\).
Similarly holds for right quasiorders and the right Nerode quasiorder.

\section{Kleene Iterates}

Let \(\tuple{X,\leqslant}\) be a qoset and \(f:X \ra X\) be a function. 
The function $f$ is \emph{monotone} if{}f $x\qo y$ implies $f(x) \leqslant f(y)$. 
Given \(b\in X\), the trace of values of the variable \(x\in X\) computed by the following iterative procedure: 
\[
\mindex{\Kleene}(f,b) \ud \left\{ \begin{array}{l}
x:=b; \\
\textbf{while~} f(x) \neq x \textbf{~do~} x:=f(x);\\
\textbf{return~} x;
\end{array}
\right.
\]
provides the possibly infinite sequence of so-called 
\demph{Kleene iterates} of the function \(f\) starting from the basis \(b\). 

Whenever \(\tuple{X,\leqslant}\) is an ACC (resp. DCC) CPO, \(b\leqslant f(b)\) (resp. \(f(b)\leqslant b\)) and \(f\) is monotone
then, by Knaster-Tarski-Kleene fixpoint theorem, \(\Kleene(f,b)\) terminates and returns the least (resp.\ greatest) fixpoint of the function \(f\) which is greater (resp.\ lower) than or equal to \(b\). 
In particular,  if $\bot_X$ (resp.\ $\top_X$) is the least (resp. greatest) element of $X$ then 
\(\Kleene(f,\bot_X)\) (resp.\ \(\Kleene(f,\top_X)\)) computes 
the sequence of Kleene iterates that finitely converges to the least (resp.\ greatest) 
fixpoint of $f$, denoted by \(\lfp(f)\) (resp.\ \(\gfp(f)\)).

\begin{theorem}\label{theorem:Kleene}
Let \(\tuple{X, \qo}\) be an ACC CPO and let \(f:X \ra X\) be a monotone function.
Then \(\Kleene(f,\bot_X)\) terminates and returns the least fixpoint of \(f\).
\end{theorem}
\begin{proof}
To simplify the notation, we use \(\bot\) to denote the least element of \(X\), \(\bot_X\).
Next, we show by induction that \(f^n(\bot) \qo f^{n{+}1}(\bot)\) for all \(n \geq 0\).

\begin{myItem}
\item \emph{Base case:} The relation \(\bot \qo f(\bot)\) holds since \(\bot\) is the least element in \(X\).

\item \emph{Inductive step:} Assume \(f^n(\bot) \qo f^{n{+}1}(\bot)\) for some value \(n\).
Then, since \(f\) is a monotone function, we have that \(f^{n{+}1}(\bot) \qo f^{n{+}2}(\bot)\).
\end{myItem}
\medskip

We conclude that \(f^n(\bot) \qo f^{n{+}1}(\bot)\) holds for all \(n \geq 0\).
Since the qoset \(\tuple{X,\qo}\) is an ACC, there is no infinite sequence of ascending elements and, as a consequence, \(\Kleene(f,\bot)\) terminates and returns a fixpoint of function \(f\).

Next, we show that if \(f^n(\bot) = f^{n{+}1}(\bot)\) for some \(n\) then \(f^n(\bot) = \lfp(f)\). 
To do that, we show that \(f^{i}(\bot) \qo p\) for every \(i \geq 0\) and for every fixpoint \(p\) of \(f\).
Therefore, the fixpoint \(f^{n}(\bot)\) is below (for the quasiorder \(\qo\)) than any other fixpoint, hence \(f^n(\bot)\) is the least fixpoint of \(f\), i.e. \(f^n(\bot) =\lfp(f)\).

Again, we proceed by induction on \(n\).
Let \(p\) be a fixpoint of \(f\), i.e. \(f(p) = p\).

\begin{myItem}
\item \emph{Base case:} The relation \(\bot \qo p\) trivially holds by definition of \(\bot\).

\item \emph{Inductive step:} Assume \(f^n(\bot) \qo p\) for some value \(n\).
Then, since \(f\) is a monotone function, we have that \(f^{n{+}1}(\bot) \qo f(p) = p\), where the last equality follows from the fact that \(p\) is a fixpoint.
\end{myItem}
\medskip

Clearly, \(f^n(\bot) \qo p\) for all \(n \geq 0\) and for all fixpoint \(p\) of \(f\).
Therefore \(\Kleene(f, \bot) = \lfp(f)\).
\end{proof}

For the sake of clarity, we overload the notation and use the same symbol for a function/relation 
and its componentwise (i.e.\ pointwise) extension on product domains.
For instance, if $f:X\ra Y$ then $f$ also denotes the standard product function $f:X^n \ra Y^n$ defined by 
$\lambda\tuple{x_1,...,x_n}\in X^n.\tuple{f(x_1),...,f(x_n)}$. 
A vector \(\vect{Y}\) in some product domain \(D^{|S|}\) is also denoted by \(\tuple{Y_i}_{i \in S}\) and, for some $i\in S$, 
\(\vect{Y}_{\!\! i}\) denotes its component \(Y_i\).

\section{Closures and Galois Connections}
We conclude this chapter by recalling some basic notions on closure operators and Galois Connections commonly used in abstract interpretation (see, e.g., \cite{CC79,mine17}).  

Closure operators and Galois Connections are equivalent notions \cite{Cousot78-1-TheseEtat} and, therefore, they are both used for defining the notion of \emph{approximation} in abstract interpretation, where closure operators allow us to define and reason on abstract domains independently of a specific representation which is required by 
Galois Connections. 

\begin{definition*}[Upper Closure Operator]
Let \(\tuple{C,\mathord{\qo_C},\vee,\wedge}\) be a complete lattice, where $\vee$ and $\wedge$ denote, respectively, the lub and glb. 
An \emph{upper closure operator}, or simply \demph{closure}, on \(\tuple{C,\mathord{\qo_C}}\) is a function \(\rho:C\to C\) which is:
\begin{myEnumI}
\item \emph{monotone}, i.e. \(x \qo_C y \Ra ρ(x) \qo_C ρ(y) \) for all \(x,y \in C\);
\item \emph{idempotent}, i.e. \(\rho(\rho(x)) = \rho(x)\) for all \(x \in C\), and
\item \emph{extensive}, i.e. \(x \qo_C \rho(x)\) for all \(x \in C\). \eod
\end{myEnumI}
\end{definition*}
The set of all upper closed operators on \(C\) is denoted by \(\uco(C)\).
We often write \(c \in \rho(C)\), or simply \(c \in \rho\), to denote that  
there exists \(c' \in C\) such that \(c = \rho(c')\), and 
recall that this happens if{}f $\rho(c) = c$. 
If $\rho\in \uco(C)$ then for all  \(c_1\in C\), \(c_2\in \rho\)  and \(X \subseteq C\), 
it turns out that:
\begin{align}
&c_1 \qo_C c_2 \Lra \rho(c_1)\qo_C \rho(c_2) \Lra \rho(c_1)\qo_C c_2 \label{equation:abstractcheck}\\
&\rho ({\textstyle\vee} X) = \rho({\textstyle\vee}\rho(X)) \quad \text{and}\quad {\textstyle\wedge}\rho (X) = \rho({\textstyle\wedge}\rho(X))\enspace. \label{equation:lubAndGlb}
\end{align}
In abstract interpretation, a closure operator \(\rho\in \uco(C)\) on a concrete domain $C$ plays
the role of abstraction function for objects of $C$. Given two closures \(\rho,\rho' \in \uco(C)\), \(\rho\) is a 
\demph{coarser abstraction} 
than \(\rho'\) (or, equivalently, 
$\rho'$ is a more precise abstraction than $\rho$) if{}f the image of 
\(\rho\) is a subset of the image of \(\rho'\), i.e. \(\rho(C) \subseteq \rho'(C)\), and this happens if{}f for any $x\in C$, 
$\rho'(x) \qo_C \rho(x)$.

\begin{definition*}[Galois Connection]
A \emph{Galois Connection} (GC for short) or \emph{adjunction} between two posets \(\tuple{C,\qo_C}\) (a concrete domain) and \(\tuple{A,\qo_A}\) (an abstract domain) consists of two monotone functions \(\alpha\colon C\ra A\) and \(\gamma \colon A\ra C\) such that 
\[\alpha(c)\qo_A a \:\Lra\: c\qo_C \gamma(a), \quad \text{for all } a\in A, c \in C \enspace .\] 
A Galois Connection is denoted by \( \tuple{C,\qo_C} \galois{\alpha}{\gamma} \tuple{A,\qo_A}\).\eod
\end{definition*}
\pagebreak
\begin{lemma}\label{lemma:propertiesGC}
Let \(\tuple{C,\qo_C}\galois{\alpha}{\gamma}\tuple{A,\qo_A}\) be a GC.
The following properties hold:
\begin{myEnumA}
\item \(x \qo_C \gamma \comp \alpha (x)\) and \(\alpha \comp \gamma(y) \qo_A y\).\label{lemma:propertiesGC:gammaalpha}
\item \(\alpha\) and \(\gamma\) are monotonic functions.\label{lemma:propertiesGC:monotone}
\item \(\alpha = \alpha \comp \gamma \comp \alpha\) and \(\gamma = \gamma \comp \alpha \comp \gamma\).
\end{myEnumA}
\end{lemma} 
\begin{proof}

\begin{myEnumA}
\item Since \(\qo_A\) is reflexive, we have that for all \(x \in A\) \(\alpha(x) \qo_A \alpha(x)\) holds and, by definition of GC, \(\alpha(x) \qo_A \alpha(x) \Lra x \qo_C \gamma(\alpha(x))\).
Therefore, \(x \qo_C \gamma(\alpha(x))\).

Similarly, since \(\alpha(\gamma(y)) \qo_A y \Lra \gamma(y) \qo_C \gamma(y)\) and \(\gamma(y) \qo_C \gamma(y)\), we conclude that \(\alpha(\gamma(y)) \qo_A y\).

\item Let \(c,c' \in C\) be such that $c\qo_C c'$. Then, by Lemma~\ref{lemma:propertiesGC}~\ref{lemma:propertiesGC:gammaalpha}, we have that $c' \qo_C \gamma(\alpha(c'))$ and, by definition of GC, $c \qo_C  \gamma(\alpha(c')) \Ra \alpha(c) \qo_A \alpha(c')$.

Similarly, let \(a,a' \in A\) be such that $a\qo_A a'$. Then, by Lemma~\ref{lemma:propertiesGC}~\ref{lemma:propertiesGC:gammaalpha}, we have that $α(γ(a)) \qo_A a'$, hence $α(γ(a) \qo_A a' \Ra γ(a) \qo_A γ(a')$.

\item Let \(c \in C\).
By Lemma~\ref{lemma:propertiesGC}~\ref{lemma:propertiesGC:gammaalpha}, we have that \(c\qo_C \gamma(\alpha(c))\) which, by Lemma~\ref{lemma:propertiesGC}~\ref{lemma:propertiesGC:monotone}, implies that $\alpha(c) \qo_A \alpha(\gamma(\alpha(c)))$. 
Moreover, since \(\gamma(\alpha(c))\qo_C \gamma(\alpha(c))\), it follows from the definition of GC that \(\alpha(\gamma(\alpha(c))) \qo_A \alpha(c)\).
Therefore \(\alpha(\gamma(\alpha(c))) = \alpha(c)\).

Similarly, let \(a \in A\).
By Lemma~\ref{lemma:propertiesGC}~\ref{lemma:propertiesGC:gammaalpha} and~\ref{lemma:propertiesGC:monotone}, we have that $γ(α(γ(a))) \qo_A γ(a)$ and, since \(α(γ(a))\qo_C α(γ(a))\), it follows from the definition of GC that \(γ(a) \qo_A γ(α(γ(a)))\).
Therefore \(γ(a) = γ(α(γ(a)))\).
\end{myEnumA}
\end{proof}

The function $\alpha$ is called the \demph{left-adjoint} of $\gamma$, and, dually, $\gamma$ is called the \demph{right-adjoint} of $\alpha$. 
This terminology is justified by the fact that if a function $\alpha:C\ra A$ admits a right-adjoint $\gamma:A\ra C$ then this is unique (and this dually holds for left-adjoints).

It turns out that, in a GC, \(\gamma\) is always \demph{co-additive}, i.e. it preserves arbitrary glb's, while \(\alpha\) is always \demph{additive}, i.e. it preserves arbitrary lub's. 
Moreover, an additive function \(\alpha : C\ra A\) uniquely determines its right-adjoint by 
\[\gamma\ud \lambda a\ldotp \bigvee_C\{c\in C \mid \alpha(c)\qo_A a\}\enspace .\] 
Dually, a co-additive function \(\gamma: A\ra C\) uniquely determines its left-adjoint by 
\[\alpha \ud \lambda c\ldotp \bigwedge_A\{a\in A \mid c\qo_C \gamma(a)\}\enspace .\] 
We conclude this chapter with the following lemma, which is folklore in abstract interpretation yet we provide a proof for the sake of completeness. 
\begin{lemma}\label{lemma:alpharhoequality}
Let \( \tuple{C,\qo_C} \galois{\alpha}{\gamma} \tuple{A,\qo_A}\) be a GC between complete lattices and 
\(f\colon C\rightarrow C\) be a monotone function. Then,
\(
	\gamma( \lfp (\alpha  f \gamma )) = \lfp (\gamma \alpha f)
\).	
\end{lemma}
\begin{proof}
Let us first show that \(\gamma( \lfp (\alpha f\gamma) ) \geqslant_C \lfp (\gamma\alpha f) \):
\begin{align*}
\gamma(\lfp(\alpha f \gamma)) \qo_C \gamma(\lfp(\alpha f\gamma)) &\Lra \quad\text{[Since \(g(\lfp(g))=\lfp(g)\)]}\\
\gamma\alpha f(\gamma(\lfp(\alpha f \gamma))) \qo_C \gamma(\lfp(\alpha f \gamma))&\Rightarrow \quad\text{[Since $g(x)\qo x \Ra \lfp(g)\qo x$]}\\
\lfp(\gamma\alpha f)\qo_C \gamma(\lfp(\alpha f \gamma)) &
\end{align*}
	Then, let us prove that \(\gamma( \lfp (\alpha f\gamma) ) \qo_C \lfp (\gamma\alpha f) \):
\begin{align*}
\lfp(\gamma\alpha f)\qo_C \lfp(\gamma\alpha f) & \Lra \quad\text{[Since $g(\lfp(g))=\lfp(g)$]}\\
\gamma \alpha f(\lfp(\gamma\alpha f)) \qo_C \lfp(\gamma\alpha f) & \Ra \quad\text{[Since $\alpha$ is monotone]}\\
\alpha \gamma \alpha f(\lfp(\gamma\alpha f)) \qo_A \alpha (\lfp(\gamma\alpha f)) & \Lra \quad\text{[Since $\alpha\gamma\alpha=\alpha$ in GCs]}\\
 \alpha f(\lfp(\gamma\alpha f)) \qo_A \alpha (\lfp(\gamma\alpha f)) & \Lra \quad\text{[Since $g(\lfp(g))=\lfp(g)$]} \\
\alpha f \gamma(\alpha(\lfp(\gamma\alpha f)))\qo_A \alpha(\lfp(\gamma\alpha f))&\Ra\quad\text{[Since $g(x)\qo x \Ra \lfp(g)\qo x$]}\\
\lfp(\alpha f\gamma) \qo_A \alpha(\lfp (\gamma\alpha f)) &\Ra\quad\text{[Since $\gamma$ is monotone]}\\
\gamma(\lfp(\alpha f \gamma)) \qo_C \gamma\alpha(\lfp(\gamma\alpha f)) &\Lra\quad\text{[Since $g(\lfp(g))=\lfp(g)$]}\\
\gamma(\lfp (\alpha f\gamma))	\qo_C \lfp(\gamma\alpha f)
\end{align*}
\end{proof}

\section{Complexity Notation}
In this thesis we analyze the time and space complexity of some algorithms and constructions.
To do that, we use the standard small-O, big-O and big-Omega notation to compare functions.
Next, we define these notations for the shake of completeness, where, given a real number \(k\), we write \(\len{k}\) to denote its absolute value.

\begin{definitionNI*}[Small-O, Big-O, Big-Omega]\index{Small-O}\index{Big-O}\index{Big-Omega}
Let \(f\) and \(g\) be two functions on the real numbers.
Then
\begin{align*}
f(n) = o(g(n)) &  \udiff  \forall k > 0, \exists n_0, \; \forall n > n_0, f(n) \leq k\cdot g(n)   \udiff  \lim_{n \to \infty} \frac{f(n)}{g(n)} = 0\\
f(n) = \mathcal{O}(g(n)) &  \udiff  \exists k > 0, \exists n_0, \; \forall n > n_0, f(n) \leq k\cdot g(n)   \udiff  \limsup_{n \to \infty} \frac{f(n)}{g(n)} < \infty \\
f(n) =  Ω(g(n)) &  \udiff  \exists k >0, \exists n_0, \; \forall n > n_0, f(n) \geq k \cdot g(n)   \udiff  \liminf_{n \to \infty} \frac{f(n)}{g(n)} > 0 
\end{align*}

Intuitively, \(f(n) = o(g(n))\) indicates that \(f\) is \emph{asymptotically dominated} by \(g\); \(f(n) = \mathcal{O}(g(n))\) indicates that \(f\) is \emph{asymptotically bounded above} by \(g\) and \(f(n) = Ω(g(n))\) indicates that \(f\) is \emph{asymptotically bounded below} by \(g\).\eod
\end{definitionNI*}

These notations allow us to simplify the complexity analysis by removing all components of low impact in a complexity function. 
For instance, let the number of operations performed by an algorithm on an input of size \(n\) be given by a function \(f(n)\) that satisfies
\[n^2 + n\cdot \log n+k \leq f(n) \leq n^3+n^2+\log n + k'\enspace ,\]
where \(k\) and \(k'\) are constants.
Since, by definition, \(n^2 = o(n^3)\), \(\log(n) = o (n^3)\) and \(k' = o(n^3)\) we find that the components \(n^2\), \(\log n\) and \(k'\) have low impact in the behavior of the upper bound of \(f(n)\) for large values of \(n\).
Similarly, the components \(n \cdot \log(n)\) and \(k\) have low impact in the lower bound of \(f(n)\) for large values of \(n\).
Therefore, we find that \(\mathcal{O}(f(n)) = \mathcal{O}(n^3+n^2+\log n + k') = \mathcal{O}(n^3)\) and \(Ω(f(n)) = Ω(n^2 + n\cdot \log n+k) = Ω(n^2)\).
Intuitively, this means that for large values of the parameter \(n\) the function \(f(n)\) is \emph{below} \(n^3\) and \emph{above} \(n^2\).
\clearpage{}%
\clearpage{}%
%

\chapter{Deciding Language Inclusion}
\label{chap:LangInc}

In this chapter, we present a quasiorder-based framework for deciding language inclusion which is a fundamental and classical problem~\cite[Chapter~11]{HU79} with applications to different areas of computer science.

The basic idea of our approach for solving a language inclusion problem $L_1\subseteq L_2$ is to leverage
Cousot and Cousot's abstract interpretation \cite{CC77,CC79} for checking the inclusion of an over-approximation (i.e.\ a superset) of \(L_1\) into \(L_2\). 
This idea draws inspiration from the work of \citet{Hofmann2014}, who used abstract interpretation to decide language inclusion between languages of infinite words.

Assuming that \(L_1\) is specified as least fixpoint of an equation system on $\wp(\Sigma^*)$, an over-approximation of $L_1$ is obtained by applying an over-approximating abstraction function for sets of words \(\rho:\wp(\Sigma^*)\ra \wp(\Sigma^*)\) at each step of the Kleene iterates converging to the least fixpoint $L_1$. 
This abstraction map \(\rho\) is an upper closure operator which is used in standard abstract interpretation for approximating an input language by adding words (possibly none) to it. 

This abstract interpretation-based approach provides an abstract inclusion check $\rho(L_1) \subseteq L_2$ which is always \emph{sound} by construction because $L_1 \subseteq \rho(L_1)$.
We then give conditions on \(\rho\) which ensure a \emph{complete} abstract inclusion 
check, namely, the answer to $\rho(L_1) \subseteq L_2$ is always exact (no ``false alarms'' in abstract interpretation terminology).
These conditions are:%
\begin{myEnumIL}
\item \(\rho(L_2)=L_2\) and\label{condition:rhoL}
\item \(\rho\) is a complete abstraction for symbol concatenation $\lambda X\in \wp(\Sigma^*).\,aX$, for all $a\in \Sigma$, 
according to the standard notion of completeness in abstract interpretation \cite{CC77,GiacobazziRS00,Ranzato13}.\label{condition:rhobw}%
\end{myEnumIL}%
This approach leads us to design in  Section~\ref{sec:an_algorithmic_framework_for_language_inclusion_based_on_complete_abstractions} 
two general algorithmic frameworks for language inclusion problems which are parameterized by an underlying language abstraction (see Theorems~\ref{theorem:FiniteWordsAlgorithmGeneral} and~\ref{theorem:EffectiveAlgorithm}).
Intuitively, the first of these frameworks allows us to decide the inclusion \(L_1 \subseteq L_2\) by manipulating finite sets of words, even if the languages \(L_1\) and \(L_2\) are infinite.
On the other hand, the second framework allows us to decide the inclusion by working on an abstract domain.

We then focus on over-approximating abstractions $\rho$ which are 
induced by a quasiorder relation $\mathord{\leqslant}$ on words in $\Sigma^*$. Here, a language \(L\) is over-approximated by adding all the words which are ``greater than or equal to'' some word of \(L\) for $\mathord{\leqslant}$. This allows us to 
instantiate the above conditions~\ref{condition:rhoL} and~\ref{condition:rhobw} 
for having a complete abstract inclusion check in terms of the quasiorder $\mathord{\leqslant}$.
Termination, which corresponds to having finitely many Kleene iterates in the fixpoint computations, 
is guaranteed by requiring that 
the relation $\mathord{\leqslant}$ is a well-quasiorder.

We define quasiorders satisfying the above conditions which are directly derived from the standard Nerode equivalence relations on words.
These quasiorders have been first investigated by \citet{ehrenfeucht_regularity_1983} and have been later generalized and extended by \citet{deLuca1994,deluca2011}.
In particular, drawing from a result by \citet{deLuca1994}, we show that the language abstractions induced by the Nerode's quasiorders are the most general ones (thus, intuitively optimal) which fit in our algorithmic framework for checking language inclusion.  

While these quasiorder abstractions do not depend on some language representation (e.g., some class of 
automata), 
we provide quasiorders which instead exploit an underlying language representation given by a finite automaton.
In particular, by selecting suitable well-quasiorders for the class of language inclusion problems at hand, we are able to systematically derive 
decision procedures for the inclusion problem 
$L_1\subseteq L_2$ when:\begin{myEnumIL}
\item both \(L_1\) and \(L_2\) are regular,
\item \(L_1\) is regular and \(L_2\) is the trace language of a one-counter net and
\item \(L_1\) is context-free and \(L_2\) is regular.
\end{myEnumIL}

These decision procedures that we systematically derive here by instantiating our framework 
are then related to existing language inclusion checking algorithms. 
We study in detail the case where both languages $L_1$ and $L_2$ are regular and represented by finite-state automata. 
When our decision procedure for $L_1\subseteq L_2$ is derived from 
a well-quasiorder on $\Sigma^*$ by exploiting the automaton-based representation of \(L_2\), it turns out that 
we obtain the well-known ``antichains algorithm'' by \citet{DBLP:conf/cav/WulfDHR06}. 
Also, by including a simulation relation in the definition of the well-quasiorder we derive a decision procedure that partially matches the language inclusion algorithm by \citet{Abdulla2010}, and in turn also that by \citet{DBLP:conf/popl/BonchiP13}.
For the case in which \(L_1\) is regular and \(L_2\) is the set of traces of a one-counter net we derive an alternative proof for the decidability of the language inclusion problem~\cite{JANCAR1999476}.
Moreover, for the case in which \(L_1\) is context-free and \(L_2\) is regular, we derive a decision procedure that matches the ``antichains algorithm'' for context-free languages presented by \citet{Holk2015}.
\\
\indent
Finally, we leverage a standard duality result~\cite{cou00} and put forward a \emph{greatest} fixpoint approach (instead of the above \emph{least} fixpoint-based procedures) for the case where both \(L_1\) and \(L_2\) are regular languages.
In this case, we exploit the properties of the over-approximating abstraction induced by the quasiorder  in order to show that the Kleene iterates of this greatest fixpoint computation are finitely many.
Interestingly, the Kleene iterates of the greatest fixpoint are finitely many whether you apply the over-approximating abstraction or not, which we show by relying on so-called forward complete abstract interpretations~\cite{gq01}. 

\section{Inclusion Check by Complete Abstractions}%
\label{sec:inclusion_checking_by_complete_abstractions}

The language inclusion problem consists in checking whether \(L_1 \!\subseteq\! L_2\) holds where \(L_1\) and \(L_2\) are two languages over a common alphabet \(\Sigma\).
In this section, we show how complete abstractions $\rho$ of $\wp(\Sigma^*)$ can be used to compute 
an over-approximation \(\rho(L_1)\) of \(L_1\) such that \(\rho(L_1) \!\subseteq\! L_2 \Lra L_1 \!\subseteq\! L_2\).

Closure-based abstract interpretation can be applied to solve a generic inclusion problem by leveraging backward complete abstractions~\cite{CC77,CC79,GiacobazziRS00,Ranzato13}. 
An upper closure \(\rho\in \uco(C)\) is called \demph{backward complete}
for a concrete 
monotone function \(f:C\ra C\) when \( \rho f=\rho f \rho \) holds. Since $\rho f(c) \leq_C \rho f \rho(c)$ always holds for all $c\in C$, 
the intuition is that 
backward completeness models an ideal situation where no loss of precision
is accumulated in the computations of $\rho f$ when 
its concrete input objects $c$ are over-approximated by $\rho(c)$. 
It is well known~\cite{CC79} 
that backward completeness implies completeness of least fixpoints, namely 
\begin{equation} \label{eqn:lfpcompleteness}
\rho f=\rho f \rho \;\Ra\;
\rho(\lfp(f))=\lfp(\rho f) = \lfp(\rho  f \rho)
\end{equation}
provided that these least fixpoints exist (this is the case, for instance, when $C$ is a CPO).  
Theorem~\ref{theorem:inc-check-comp-abs} states how  
a concrete inclusion check \(\lfp(f) \leq_C c_2\) can be equivalently performed 
in a backward complete abstraction \(\rho\) when  \(c_2\in \rho\).

\begin{theorem}\label{theorem:inc-check-comp-abs}
If $C$ is a CPO, \(f: C\ra C\) is monotone, $\rho\in \uco(C)$ is backward complete for \(f\) and \(c_2\in \rho\), then
\[\lfp(f) \leq_C c_2 \Lra \lfp(\rho f) \leq_C c_2 \enspace .\]
In particular, if \(\tuple{C,\leq_C}\) is ACC then the Kleene iterates of \ \(\lfp(\rho f)\) are finitely many. 
\end{theorem}
\begin{proof}
First, we show that \(\lfp(f) \leq_C c_2 \Lra \lfp(\rho f) \leq_C c_2\).
\begin{align*}
\lfp(f) \leq_C c_2 &\Lra \quad\text{[Since $c_2\in \rho$]}\\
\lfp(f) \leq_C \rho(c_2) &\Lra \quad\text{[Since $x\leq \rho(y) \Lra \rho(x)\leq \rho(y)$]}\\
\rho(\lfp(f)) \leq_C \rho(c_2) &\Lra \quad\text{[By Equation~\eqref{eqn:lfpcompleteness}]}\\
\lfp(\rho f) \leq_C \rho(c_2) &\Lra \quad\text{[Since $c_2\in \rho$]}\\
\lfp(\rho f) \leq_C c_2 &
\end{align*}
It remains to prove that the Kleene iterates of \(\lfp(ρf)\) are finitely many.
Observe that, since \(ρ\) and \(f\) are monotone and \(\bot \leq_C ρf(\bot)\), we have that 
\[(ρf)^n(\bot) \leq_C (ρf)^{n{+}1}(\bot) \text{ for all \(n \geq 1\)}\enspace .\]
If \(\tuple{C,\leq_C}\) is ACC then, by definition, there are no infinite ascending chains, hence the sequence of Kleene iterates \[\bot \leq_C ρf(\bot) \leq_C (ρf)^2(\bot) \leq_C \ldots \leq_C (ρf)^n(\bot)\]
converges in finitely many steps. 
\end{proof}

In the following, we will apply this general abstraction scheme to a number of
different language inclusion problems, by designing inclusion algorithms which rely on 
several different backward complete abstractions of \(\wp(\Sigma^*)\).

\section{An Algorithmic Framework for Language Inclusion}%
\label{sec:an_algorithmic_framework_for_language_inclusion_based_on_complete_abstractions}

\subsection{Languages as Fixed Points}%
\label{sub:languages_as_fixpoints}

Let \(\cN=\tuple{Q,Σ,\delta,I,F}\) be an NFA.
Recall that the language accepted by \(\cN\) is given by \(\lang{\cN} \ud W^{\cN}_{I,F}\) and, therefore,
\begin{equation}%
\label{eq:unionofrightlg}
\lang{\cN}={\textstyle\bigcup_{q\in I}} W^\cN_{q,F}={\textstyle\bigcup_{q\in F}} W^\cN_{I,q}\enspace
\end{equation}
\noindent
where, as usual, \(\textstyle{\bigcup \varnothing} = \varnothing\).

Let us recall how to define the language accepted by an automaton as a solution of a set of equations~\cite{Schutzenberger63}.
To do that, given a generic boolean predicate \(p(x)\) (typically a membership predicate) on some set and two generic sets $T$ and $F$, we define the following parametric 
choice function:\[
\nullable{p(x)}{T}{F} \ud \begin{cases}
		T & \text{if \(p(x)\) holds} \\
		F & \text{otherwise}
\end{cases} \enspace .\]

The NFA \(\cN\) induces the following set of equations, where the $X_q$'s 
are variables of type $X_q\in \wp(\Sigma^*)$ and are indexed by states $q\in Q$:
\begin{equation}\label{leftEqn}
	\Eqn(\cN) \ud \{ X_q = \nullable{q \in F}{\lbrace\epsilon\rbrace}{\varnothing} \cup {\textstyle \bigcup_{a\in \Sigma, q'\in\delta(q,a)}} a X_{q'} \mid  q\in Q\} \enspace .
\end{equation}

It follows that the functions in the right-hand side of the equations in
\(\Eqn(\cN)\) have
type \(\wp(\Sigma^*)^{|Q|} \ra \wp(\Sigma^*)\).
Since \(\tuple{\wp(\Sigma^*)^{|Q|},\subseteq}\) is a (product) complete lattice (because \(\tuple{\wp(\Sigma^*),\subseteq}\) is a complete lattice) and all the right-hand side functions in \(\Eqn(\cN)\) are clearly monotone, 
the least solution \(\tuple{Y_q}_{q\in Q}\in \wp(\Sigma^*)^{|Q|}\) of \(\Eqn(\cN)\) does exist and it is easy to check  
that for every \(q\in Q\), \(Y_q = W^{\cN}_{q,F}\) holds, hence, by Equation~\eqref{eq:unionofrightlg}, \(\lang{\cN} = {\textstyle\bigcup_{q_i \in I}} Y_{q_i}\). 

It is worth noticing that, by relying on right concatenations rather than left ones 
$aX_{q'}$ used 
in \(\Eqn(\cN)\), one could also define a
set of symmetric equations whose least solution coincides with \(\tuple{W_{I,q}^{\cN}}_{q\in Q}\) instead of \(\tuple{W_{q,F}^{\cN}}_{q\in Q}\).

\begin{figure}[!ht]
\centering
\begin{tikzpicture}[->,>=stealth',shorten >=1pt,auto,node distance=5mm and 1cm,thick,initial text=]
	\tikzstyle{every state}=[scale=0.75,fill=customblue!60,draw=blue!60,text=black]
	
	\node[initial,accepting,state] (1) {\(1\)};
	\node[state] (2) [right=of 1] {\(2\)};
	
	\path (1) edge[bend left] node {\(b\)} (2)
	      (2) edge[bend left] node {\(a\)} (1)
	      (2) edge[loop above] node {\(b\)} (2)
	      (1) edge[loop above] node {\(a\)} (1)
	          ;
	\end{tikzpicture}
	\caption{An NFA \(\cN\) with \(\lang{\cN}= (a + (b^+ a))^*\).}
		\label{fig:A}
\end{figure}

\begin{example}\label{ex-first}
	Let us consider the automaton \(\cN\) in Figure~\ref{fig:A}. 	
The set of equations induced by \(\cN\) are as follows: 
\[
	\Eqn(\cN)=\begin{cases}
		X_1 = \{\epsilon\} \cup aX_1 \cup bX_2\\
		X_2 = \varnothing \cup aX_1 \cup b X_2 
	\end{cases} \enspace . \tag*{\eox}
\]
\end{example}

It is convenient 
to state  the equations in \(\Eqn(\cN)\) by exploiting
an ``initial'' vector \(\vectarg{\epsilon}{F} \in \wp(\Sigma^*)^{|Q|}\) and a predecessor
function \(\Pre_\cN \colon \wp(\Sigma^*)^{|Q|} {\ra} \wp(\Sigma^*)^{|Q|}\) de-fined as follows:
\begin{align*}
\vectarg{\epsilon}{F} &\ud \tuple{\nullable{q \in F}{\{\epsilon\}}{\varnothing}}_{q\in Q}\,, &\qquad
\Pre_\cN(\tuple{X_{q}}_{q\in Q}) &\ud \tuple{ {\textstyle \bigcup_{a\in \Sigma, q'\in\delta(q,a)}} aX_{q'}}_{q\in Q} \enspace .
\end{align*}

The intuition for the function \(\Pre_{\cN}\) is that given the language \(W_{q',F}^{\cN}\) and a transition \(q' \!\in\! \delta(q,a) \), we have that \(aW^{\cN}_{q',F} \subseteq W^{\cN}_{q,F}\) holds, i.e. given a subset $X_q'$ of the language generated by \(\cN\) from some state \(q'\), the function \(\Pre_{\cN}\) computes a subset $X_q$ of the language generated by \(\cN\) for its predecessor state \(q\).

Since \(\epsilon \in W_{q,F}^{\cN}\) for all \(q \in F\), the least fixpoint computation can start from the vector 
\(\vectarg{\epsilon}{F}\) and iteratively apply $\Pre_\cN$.
Therefore
\begin{equation}\label{eq:WqFAequalslfp}
\tuple{W^{\cN}_{q,F}}_{q\in Q} = \lfp(\lambda \vect{X}\ldotp \vectarg{\epsilon}{F} \cup \Pre_\cN(\vect{X})) \enspace .
\end{equation}

Together with Equation~\eqref{eq:unionofrightlg}, it follows that \(\lang{\cN}\) equals the union of the component languages of the vector 
\(\lfp(\lambda \vect{X}\ldotp \vectarg{\epsilon}{F} \cup \Pre_\cN(\vect{X}))\) indexed by the initial states in $I$. 

\begin{example}[Continuation of Example~\ref{ex-first}]
The fixpoint characterization of \(\tuple{W_{q,F}^{\cN}}_{q\in Q}\) is:
\[
	\left( \begin{array}{c}
		 W^{\cN}_{q_1,q_1} \\ W^{\cN}_{q_2,q_1}
	\end{array} \right) =
	\lfp\biggl(\lambda \left( \begin{array}{c}
		X_1 \\ X_2
	\end{array} \right) .
	\left(\begin{array}{c}
			\{\epsilon\} \cup a X_1 \cup b X_2 \\
			\varnothing \cup a X_1 \cup b X_2
		\end{array}\right)\biggr) = \left( \begin{array}{c}
		 (a+(b^+ a))^* \\ (a+b)^*a
	\end{array} \right) \enspace .\tag*{\eox}
\]
\end{example}

\paragraph{Fixpoint-based Inclusion Check}

Consider the language inclusion problem \(L_1 \subseteq L_2\), where \(L_1=\lang{\cN}\) for some NFA \(\cN=\tuple{Q,Σ,\delta,I,F}\).
The language \(L_2\) can be formalized as a vector in \(\wp(\Sigma^*)^{|Q|}\) as follows:
\begin{equation}\label{eq:elltwo}
\vectarg{L_2}{I} \ud \tuple{\nullable{q \in I}{L_2}{\Sigma^*}}_{q\in Q}
\end{equation} 
whose components indexed by initial states are $L_2$ and those indexed by non-initial states are $\Sigma^*$. Then, as a consequence of Equations~\eqref{eq:unionofrightlg},~\eqref{eq:WqFAequalslfp} and~\eqref{eq:elltwo},  we have that
\begin{equation}\label{eq:lfp}
\lang{\cN}\subseteq L_2 \Lra
\lfp(\lambda \vect{X}\ldotp\vectarg{\epsilon}{F} \cup \Pre_{\cN}(\vect{X})) \subseteq \vectarg{L_2}{I} \enspace .
\end{equation}

\subsection{Abstract Inclusion Check using Closures}
In what follows, we will apply Theorem~\ref{theorem:inc-check-comp-abs} for solving the language inclusion problem where:
\(C=\tuple{\wp(\Sigma^*)^{|Q|},\subseteq}\), \(f=\lambda \vect{X}\ldotp\vectarg{\epsilon}{F} \cup \Pre_{\cN}(\vect{X})\) and
\(\rho\in \uco(\wp(\Sigma^*))\), so that $\rho\in \uco\left(\wp(\Sigma^*)^{|Q|}\right)$.

\begin{theorem}\label{theorem:backComplete}
Let \(\rho \in \uco(\wp(\Sigma^*))\) be backward complete for \(\lambda X\in \wp(\Sigma^*)\ldotp aX\) for all \(a\in \Sigma\) and let \(\cN=\tuple{Q,Σ,\delta,I,F}\) be an NFA. Then the extension of \(ρ\) to vectors, $\rho\in \uco\left(\wp(\Sigma^*)^{|Q|}\right)$, is backward complete for \(\Pre_{\cN}\) and \(\lambda \vect{X}\ldotp\vectarg{\epsilon}{F} \cup \Pre_{\cN}(\vect{X})\).
\end{theorem}
\begin{proof}
First, it turns out that:
\begin{align*}
\rho( \Pre_\cN(\tuple{X_q}_{q\in Q})) &= \quad\text{[By definition of \(\Pre_{\cN}\)]}\\
\rho ({\textstyle \bigcup_{a\in \Sigma, q'\in \delta(q,a)}} aX_{q'}) &= 
\quad\text{[By Equation~\eqref{equation:lubAndGlb}]}\\
\rho ({\textstyle \bigcup_{a\in \Sigma, q'\in \delta(q,a)}} \rho(aX_{q'})) &= \quad\text{[By backward completeness of $\rho$ for \(\lambda X\ldotp aX\)]}\\
\rho ({\textstyle \bigcup_{a\in \Sigma, q'\in \delta(q,a)}} \rho(a \rho(X_{q'}))) &=\quad\text{[By Equation~\eqref{equation:lubAndGlb}]}\\
\rho ({\textstyle \bigcup_{a\in \Sigma, q'\in \delta(q,a)}} a \rho(X_{q'})) &=\quad\text{[By definition of \(\Pre_{\cN}\)]}\\
\rho( \Pre_\cN(\rho(\tuple{X_q}_{q\in Q}))) & \enspace .
\end{align*}

Next, we show backward completeness of \(\rho\) for \(\lambda \vect{X}\ldotp \vectarg{\epsilon}{F} \cup  {\Pre}_\cN (\vect{X})\):
\begin{align*}
\rho (\vectarg{\epsilon}{F} \cup \Pre_\cN (\rho(\vect{X}))) & = 
\quad\text{[By Equation~\eqref{equation:lubAndGlb}]}\\
\rho (\rho (\vectarg{\epsilon}{F}) \cup  \rho (\Pre_\cN (\rho (\vect{X})))) & = 
\quad\text{[By backward completeness of $\rho$ for \(\Pre_{\cN}\)]}\\
\rho (\rho (\vectarg{\epsilon}{F}) \cup  \rho (\Pre_\cN (\vect{X}))) & = 
\quad\text{[By Equation~\eqref{equation:lubAndGlb}]}\\
\rho (\vectarg{\epsilon}{F} \cup\: \Pre_\cN (\vect{X})) \enspace .& 
\end{align*}
\end{proof}

Then, by Equation~\eqref{eqn:lfpcompleteness}, we obtain the following result.

\begin{corollary}\label{corol:rholfp}
If \(\rho \in \uco(\wp(\Sigma^*))\) 
is backward complete for \(\lambda X\in \wp(\Sigma^*)\ldotp aX\) for all \(a\in\Sigma\) then
\[\rho (\lfp(\lambda \vect{X}\ldotp\vectarg{\epsilon}{F} \cup \Pre_\cN(\vect{X}))) = \lfp(\lambda \vect{X}\ldotp \rho (\vectarg{\epsilon}{F} \cup \Pre_\cN(\vect{X})))\enspace . \]
\end{corollary}

Note that if \(\rho\) is backward complete for \(\lambda X. aX\) for all \(a \in \Sigma\) and \(L_2\in \rho\) then, as a 
consequence of Theorem~\ref{theorem:inc-check-comp-abs} and Corollary~\ref{corol:rholfp}, we find that Equivalence~\eqref{eq:lfp} becomes
\begin{equation}\label{equation:checklfpRLintoRL}
\lang{\cN}\subseteq L_2 \Lra \lfp(\lambda \vect{X}\ldotp\rho (\vectarg{\epsilon}{F} \!\cup \Pre_{\cN}(\vect{X}))) \subseteq \vectarg{L_2}{I} \enspace.
\end{equation}

\subsubsection{Right Concatenation}\label{rightwqo}
Let us consider the symmetric case of right concatenation.  
Recall that, given an NFA \(\cN = \tuple{Q,Σ,δ,I,F}\), we have that

\[W^{\cN}_{I,q} = \nullable{q\in I}{\{\epsilon\}}{\varnothing} \cup {\textstyle\bigcup_{a\in\Sigma,a\in W^{\cN}_{q',q}}} W^{\cN}_{I,q'}a \enspace .\] 

Correspondingly, we can define a set of fixpoint equations on \(\wp(\Sigma^*)\) which is based on right concatenation 
and is symmetric to Equation~\eqref{leftEqn}:

\[\Eqnr(\cN) \ud \{X_q = \nullable{q\in I}{\{\epsilon\}}{\varnothing} \cup {\textstyle\bigcup_{a\in\Sigma,q\in \delta(q',a)}} X_{q'}a \mid q \in Q\} \enspace .\]

In this case, if \(\vect{Y}=\tuple{Y_q}_{q\in Q}\) is the 
least fixpoint solution of \(\Eqnr(\cN)\) then we have that \(Y_q = W^{\cN}_{I,q}\) for every \(q\in Q\).
Also, by defining \(\vectarg{\epsilon}{I} \in \wp(\Sigma^*)^{|Q|}\) and \(\Post_\cN \colon \wp(\Sigma^*)^{|Q|} {\ra} \wp(\Sigma^*)^{|Q|}\) as follows: 
\begin{align*}
\vectarg{\epsilon}{I} &\ud \tuple{\nullable{q \in I}{\{\epsilon\}}{\varnothing}}_{q\in Q}, \quad & \Post_\cN(\tuple{X_q}_{q\in Q}) &\ud \tuple{ {\textstyle \bigcup_{a\in \Sigma, q\in \delta(q',a)}} X_{q'}a}_{q\in Q} \enspace ,
\end{align*}

we have that

\begin{equation}\label{eq:WIqAequalslfp}
\tuple{W_{I,q}}_{q\in Q} = \lfp(\lambda \vect{X}\ldotp \vectarg{\epsilon}{I} \cup \Post_\cN(\vect{X})) \enspace .
\end{equation}

Since, by Equation~\eqref{eq:unionofrightlg}, we have that \(\lang{\cN} = {\textstyle\bigcup_{q \in F}} W_{I,q}\), it follows that \(\lang{\cN}\) is the union of the component languages of the vector \(\lfp(\lambda \vect{X}\ldotp \vectarg{\epsilon}{I} \cup \Post_\cN(\vect{X}))\) indexed by the final states in  \(F\). 

\begin{example}\label{ex-firstPost}
Consider again the NFA \(\cN\) in Figure~\ref{fig:A}.
The set of right equations 	for \(\cN\) is:
\[
	\Eqnr(\cN)=\begin{cases}
		X_1 = \{\epsilon\} \cup X_1a \cup X_2 a\\
		X_2 = \varnothing \cup X_1b \cup X_2b
	\end{cases}
\]
so that 
\[
	\left( \begin{array}{c}
		 W_{q_1,q_1} \\ W_{q_1,q_2}
	\end{array} \right)=
	\lfp\biggl(\lambda \left( \begin{array}{c}
		X_1 \\ X_2
	\end{array} \right) .
	\left(\begin{array}{c}
			\{\epsilon\} \cup X_1 a \cup X_2 a \\
			\varnothing \cup X_1 b\cup X_2 b
		\end{array}\right)\biggr) = \left( \begin{array}{c}
		 (a+(b^+ a))^* \\ a^*b(b+a^+b)^*
	\end{array} \right) \enspace .\tag*{\eox}
\]
\end{example}
\smallskip

Finally, given a language inclusion problem \(\lang{\cN} \subseteq L_2\), 
the language \(L_2\) can be formalized as the vector
\[\vectarg{L_2}{F} \ud \tuple{\nullable{q \in F}{L_2}{\Sigma^*}}_{q\in Q} \in \wp(\Sigma^*)^{|Q|} \enspace ,\]
so that, by Equation~\eqref{eq:WIqAequalslfp}, it turns out that
\begin{align*}
\lang{\cN}\subseteq L_2 \:\Lra\: 
\lfp(\lambda \vect{X}\ldotp\vectarg{\epsilon}{I} \cup \Post_{\cN}(\vect{X})) \subseteq \vectarg{L_2}{F}
\end{align*}
We therefore have the following symmetric version 
of Theorem~\ref{theorem:backComplete} for right concatenation.

\begin{theorem}\label{theorem:backCompleteRight}
Let \(\rho \in \uco(\wp(\Sigma^*))\) be backward complete for \(\lambda X\in \wp(\Sigma^*)\ldotp Xa\) for all \(a\in \Sigma\) and let $\cN=\tuple{Q,Σ,\delta,I,F}$ be an NFA.
Then the extension of \(ρ\) to vectors, \(\rho\in \uco\left(\wp(\Sigma^*)^{|Q|}\right)\), is backward complete for \(\Post_{\cN}(\vect{X})\) and \(\lambda \vect{X}\ldotp\vectarg{\epsilon}{I} \cup \Post_{\cN}(\vect{X})\).
\end{theorem}
\begin{proof}
First, it turns out that:
\begin{align*}
\rho( \Post_\cN(\tuple{X_q}_{q\in Q})) &= \quad\text{[By definition of \(\Post_{\cN}\)]}\\
\rho ({\textstyle \bigcup_{a\in \Sigma, q\in \delta(q',a)}} X_{q'}a) &= 
\quad\text{[By Equation~\eqref{equation:lubAndGlb}]}\\
\rho ({\textstyle \bigcup_{a\in \Sigma, q\in \delta(q',a)}} \rho(X_{q'}a)) &= \quad\text{[By backward completeness of $\rho$ for \(\lambda X\ldotp Xa\)]}\\
\rho ({\textstyle \bigcup_{a\in \Sigma, q\in \delta(q',a)}} \rho(\rho(X_{q'})a)) &=\quad\text{[By Equation~\eqref{equation:lubAndGlb}]}\\
\rho ({\textstyle \bigcup_{a\in \Sigma, q\in \delta(q',a)}} \rho(X_{q'})a) &=\quad\text{[By definition of \(\Post_{\cN}\)]}\\
\rho( \Post_\cN(\rho(\tuple{X_q}_{q\in Q}))) & \enspace .
\end{align*}

Next, we show backward completeness of \(\rho\) for \(\lambda \vect{X}\ldotp \vectarg{\epsilon}{I} \cup  {\Post}_\cN (\vect{X})\):
\begin{myAlignEP}
\rho (\vectarg{\epsilon}{I} \cup \Post_\cN (\rho(\vect{X}))) & = 
\quad\text{[By Equation~\eqref{equation:lubAndGlb}]}\\
\rho (\rho (\vectarg{\epsilon}{I}) \cup  \rho (\Post_\cN (\rho (\vect{X})))) & = 
\quad\text{[By backward completeness of $\rho$ for \(\Post_{\cN}\)]}\\
\rho (\rho (\vectarg{\epsilon}{I}) \cup  \rho (\Post_\cN (\vect{X}))) & = 
\quad\text{[By Equation~\eqref{equation:lubAndGlb}]}\\
\rho (\vectarg{\epsilon}{I} \cup\: \Post_\cN (\vect{X})) \enspace .& %
\end{myAlignEP}%
\end{proof}

\subsection{Solving the Abstract Inclusion Check}\label{sec:SolvingAbstractInclusionCheck}
In this section we present two techniques for solving the language inclusion problem $\lang{\cN} \subseteq L_2$ by relying on Equivalence~\eqref{equation:checklfpRLintoRL}.

The first of these techniques leads to algorithms for solving the inclusion problem by using \emph{finite languages}.
Intuitively, given a closure \(ρ\), we show that it is possible to work on the domain \(\tuple{\{ρ(S) \mid S \in \wp(Σ^*)\}, \subseteq}\) while considering only languages \(S\) that are finite.

On the other hand, we present a second technique that relies on the use of Galois Connections in order to solve the language inclusion problem in a different domain.
This technique allows us to decide the inclusion \(\lang{\cN} \subseteq L_2\) by manipulating the underlying automata representation of the language \(L_2\).

\subsubsection{Using Finite Languages}

The following result shows that the successive steps of the fixpoint iteration for computing the \(\lfp(\rho (\vectarg{\epsilon}{F} \!\cup \Pre_{\cN}(\vect{X})))\) can be replicated by iterating on a function \(f\), instead of \(ρ(\vectarg{\epsilon}{F} \!\cup \Pre_{\cN}(\vect{X}))\), and then abstracting the result, provided that \(f\) meets a set of requirements.

\begin{lemma}\label{lemma:FiniteWordsAlgorithm}
Let \(\cN\!=\!\tuple{Q,Σ,\delta,I,F}\) be an NFA, let \(ρ \!\in\! \uco(Σ^*)\) be backward complete for \(\lambda X\!\in\!\wp(\Sigma^*)\ldotp aX\) for all \(a\in \Sigma\) and let \(f: \wp(Σ^*)^{|Q|} \to \wp(Σ^*)^{|Q|}\) be a function such that
\(\rho (\vectarg{\epsilon}{F} \!\cup \Pre_{\cN}(\vect{X})) = ρ(f(\vect{X}))\).
Then, for all \(0 \leq n\),
\[(ρ(\vectarg{\epsilon}{F} \!\cup \Pre_{\cN}(\vect{X}))^n = ρ(f^n(\vect{X})) \enspace .\]
\end{lemma}
\begin{proof}
We proceed by induction on \(n\).

\begin{myItem}
\item \emph{Base case:} Let \(n = 0\).
Then \(f^0(\vect{X}) = (ρ(\vectarg{\epsilon}{F} \!\cup \Pre_{\cN}(\vect{X}))^0 = \vect{\varnothing}\).

\item \emph{Inductive step:} Assume that \(ρ(f^n(\vect{X})) = (ρ(\vectarg{\epsilon}{F} \!\cup \Pre_{\cN}(\vect{X}))^n\) holds for some value \(n \geq 0\).
To simplify the notation, let \(\cP(\vect{X}) = \vectarg{\epsilon}{F} \!\cup \Pre_{\cN}(\vect{X})\) so that \(ρ f^n = (ρ\cP)^n\).
Then
\begin{align*}
ρf^{n{+}1}(\vect{X}) & = \quad \text{[Since \(f^{n{+}1} = f^n f\)]} \\
ρf^nf(\vect{X}) & = \quad \text{[By Inductive Hypothesis]} \\
(ρ\cP)^nf(\vect{X}) & = \quad \text{[By Theorem~\ref{theorem:backComplete}, \(ρ\) is bw. complete for \(\cP\)]}\\
(ρ\cP)^nρf(\vect{X})  & = \quad \text{[Since \(ρ f = ρ \cP\)]} \\
(ρ\cP)^nρ\cP(\vect{X})  & = \quad \text{[Since \((ρ \cP)^{n{+}1} = (ρ  \cP)^n  ρ  \cP\)]} \\
(ρ\cP)^{n{+}1}(\vect{X})
\end{align*}%
\end{myItem}%
We conclude that \((ρ(\vectarg{\epsilon}{F} \!\cup \Pre_{\cN}(\vect{X}))^n = ρ(f^n(\vect{X}))\) for all \(0 \leq n\).
\end{proof}

Lemma~\ref{lemma:FiniteWordsAlgorithm} shows that the iterates of \(\lfp(\rho (\vectarg{\epsilon}{F} \!\cup \Pre_{\cN}(\vect{X})))\) can be computed by abstracting the iterates of a function \(f\), which might manipulate only finite languages.
Moreover, its straightforward to check that Lemma~\ref{lemma:FiniteWordsAlgorithm} remains valid when considering a different function \(f\) at each step of the iteration as long as all the considered functions satisfy the requirements.

To simplify the notation, given a set of functions \(\mathcal{F}\) and a function \(f\), we write \(\mathcal{F}f\) to denote the composition of one arbitrary function from \(\mathcal{F}\) with \(f\). 
Similarly, \(f\mathcal{F} \) denotes the composition of \(f\) with an arbitrary function from \(\mathcal{F}\).
Finally, we write \(\mathcal{F}^2 = f\), for instance, to indicate that any composition of two functions in \(\mathcal{F}\) equals \(f\).

\begin{corollary}\label{corol:FiniteWordsAlgorithm}
Let \(\cN=\tuple{Q,Σ,\delta,I,F}\) be an NFA, let \(ρ \in \uco(Σ^*)\) be backward complete for \(\lambda X\in \wp(\Sigma^*)\ldotp aX\) for all \(a\in \Sigma\) and let \(\mathcal{F}\) be a set of functions such that every function \(f \in \mathcal{F}\) is of the form \(f: \wp(Σ^*)^{|Q|} \to \wp(Σ^*)^{|Q|}\) and satisfies \(\rho (\vectarg{\epsilon}{F} \!\cup \Pre_{\cN}(\vect{X})) = ρ(f(\vect{X}))\).
Then, for all \(0 \leq n\),
\[(ρ(\vectarg{\epsilon}{F} \!\cup \Pre_{\cN}(\vect{X}))^n = ρ(\mathcal{F}^n(\vect{X})) \enspace .\]
\end{corollary}

Observe that, in particular, Corollary~\ref{corol:FiniteWordsAlgorithm} holds when considering the set \(\mathcal{F} = \{f\}\) with \(f = \vectarg{\epsilon}{F} \!\cup \Pre_{\cN}(\vect{X})\).
Intuitively, this means that we can compute the least fixpoint for \(\rho (\vectarg{\epsilon}{F} \!\cup \Pre_{\cN}(\vect{X}))\) by iterating on \(\vectarg{\epsilon}{F} \!\cup \Pre_{\cN}(\vect{X})\) until we reach an \emph{abstract fixpoint}, i.e. the abstraction of two consecutive steps coincide.

The idea of recursively applying a function \(f\) until its abstraction reaches a fixpoint is captured by the following definition of the \emph{abstract Kleene procedure}:
\[
\KleeneQO(\abseq,f,b) \ud \left\{ \begin{array}{l}
x:=b; \\
\textbf{while~} \neg \abseq(f(x), x) \textbf{~do~} x:=f(x);\\
\textbf{return~} x;
\end{array}
\right. \enspace ,
\] 
where \(\abseq(x,y)\) is a function that returns \(\True\) if{}f the abstraction of \(x\) and \(y\) coincide, i.e. \(ρ(x) = ρ(y)\).
Clearly, \(\KleeneQO(id,f,b) = \Kleene(f,b)\) where \(id(x,y)\) returns \(\True\) if{}f \(x = y\).
For simplicity, we abuse of notation and write \(\KleeneQO(\abseq,\mathcal{F},b)\) to denote the abstract \(\KleeneQO\) iteration where, at each step, an arbitrary function from the set \(\mathcal{F}\) is applied.

As the following lemma shows, whenever the domain \(\tuple{\{ρ(S) \mid S \in \wp(Σ^*)\}, \subseteq}\) is ACC and the abstraction \(ρ\) is backward complete for all the functions in the set \(\mathcal{F}\), i.e. \(ρ \mathcal{F} = ρ  \mathcal{F}  ρ\), the procedure \(\KleeneQO(\abseq,\mathcal{F},b)\) can be used to compute \(\lfp(\lambda \vect{X}\ldotp\rho (\vectarg{\epsilon}{F} \!\cup \Pre_{\cN}(\vect{X})))\).

\begin{lemma}\label{lemma:KleeneQOLfp}
Let \(\cN=\tuple{Q,Σ,\delta,I,F}\) be an NFA, let \(ρ \in \uco(Σ^*)\) be backward complete for \(\lambda X\in \wp(\Sigma^*)\ldotp aX\) for all \(a\in \Sigma\) such that \(\tuple{\{ρ(S) \mid S \in \wp(Σ^*)\}, \subseteq}\) is an ACC CPO.
Let \(\mathcal{F}\) be a set of monotone functions such that every \(f \in \mathcal{F}\) is of the form \(f: \wp(Σ^*)^{|Q|} \to \wp(Σ^*)^{|Q|}\) and satisfies
\(\rho (\vectarg{\epsilon}{F} \!\cup \Pre_{\cN}(\vect{X})) = ρ(f(\vect{X}))\).
Then, 
\[\lfp(\lambda \vect{X}\ldotp\rho (\vectarg{\epsilon}{F} \!\cup \Pre_{\cN}(\vect{X}))) = ρ\left(\KleeneQO(\abseq,\mathcal{F},\vect{\varnothing})\right) \enspace .\]
Moreover, the iterates of \(\,\Kleene(\lambda \vect{X}\ldotp\rho (\vectarg{\epsilon}{F} \!\cup \Pre_{\cN}(\vect{X})), \vect{\varnothing})\) coincide in lockstep with the abstraction of the iterates of \(\,\KleeneQO(\abseq,\mathcal{F},\vect{\varnothing})\).
\end{lemma}
\begin{proof}
Since \(\tuple{\{ρ(S) \mid S \in \wp(Σ^*)\}, \subseteq}\) is an ACC CPO, by Theorem~\ref{theorem:Kleene}, we have that
\[\lfp(\lambda \vect{X}\ldotp\rho (\vectarg{\epsilon}{F} \!\cup \Pre_{\cN}(\vect{X}))) = \Kleene(\lambda \vect{X}\ldotp\rho (\vectarg{\epsilon}{F} \!\cup \Pre_{\cN}(\vect{X})), \vect{\varnothing})\]
On the other hand, by Corollary~\ref{corol:FiniteWordsAlgorithm}, the iterates of the above Kleene iteration coincide in lockstep with the abstraction of the iterates of \(\KleeneQO(\abseq,\mathcal{F},\vect{\varnothing})\) and, therefore,
\[\Kleene(\lambda \vect{X}\ldotp\rho (\vectarg{\epsilon}{F} \!\cup \Pre_{\cN}(\vect{X})) = ρ\left(\KleeneQO(\abseq,\mathcal{F},\vect{\varnothing})\right)\]
As a consequence,
\[\lfp(\lambda \vect{X}\ldotp\rho (\vectarg{\epsilon}{F} \!\cup \Pre_{\cN}(\vect{X}))) = ρ\left(\KleeneQO(\abseq,\mathcal{F},\vect{\varnothing})\right) \enspace .\]
\end{proof}

The following result relies on the \(\KleeneQO\) procedure to design an algorithm that solves the language inclusion problem \(\lang{\cN} \subseteq L_2\) whenever the abstraction \(ρ\) and the set of functions \(\mathcal{F}\) satisfy a list of requirements in terms of backward completeness and computability.

\begin{theorem}\label{theorem:FiniteWordsAlgorithmGeneral}
Let \(\cN=\tuple{Q,Σ,\delta,I,F}\) be an NFA, let \(L_2\) be a regular language, let \(ρ \in \uco(Σ^*)\) and let \(\mathcal{F}\) be a set of functions.
Assume that the following properties hold:
\begin{myEnumI}
\item The abstraction \(ρ\) is backward complete for \(\lambda X\in \wp(\Sigma^*)\ldotp aX\) for all \(a\in \Sigma\) and satisfies \(ρ(L_2) = L_2\).\label{theorem:FiniteWordsAlgorithmGeneral:rho}
\item The set \(\tuple{\{ρ(S) \mid S \in \wp(Σ^*)\}, \subseteq}\) is an ACC CPO.\label{theorem:FiniteWordsAlgorithmGeneral:ACC}
\item Every function \(f\) in the set \(\mathcal{F}\) is of the form \(f: \wp(Σ^*)^{|Q|} \to \wp(Σ^*)^{|Q|}\), it is computable and satisfies \(\rho (\vectarg{\epsilon}{F} \!\cup \Pre_{\cN}(\vect{X})) = ρ(f(\vect{X}))\).\label{theorem:FiniteWordsAlgorithmGeneral:F}
\item There is an algorithm, say \(\abseq^{\sharp}(\vect{X}, \vect{Y})\), which decides the abstraction equivalence \(ρ(\vect{X}) = ρ(\vect{Y})\), for all \(\vect{X}, \vect{Y} \in \wp(Σ^*)^{|Q|}\).\label{theorem:FiniteWordsAlgorithmGeneral:EQ}
\item There is an algorithm, say \(\absincl(\vect{X})\), which decides the inclusion \(ρ(\vect{X}) \subseteq \vectarg{L_2}{I}\), for all \(\vect{X} \in \wp(Σ^*)^{|Q|}\).\label{theorem:FiniteWordsAlgorithmGeneral:INC}
\end{myEnumI}
\medskip
Then, the following is an algorithm which decides whether \(\lang{\cN} \subseteq L_2\):

\medskip

\(\tuple{Y_q}_{q\in Q} := \KleeneQO (\abseq^{\sharp},\mathcal{F}, \vect{\varnothing})\)\emph{;}

\emph{\textbf{return}} \(\absincl(\tuple{Y_q}_{q\in Q})\)\emph{;}

\end{theorem}

\begin{proof}
It follows from hypotheses~\ref{theorem:FiniteWordsAlgorithmGeneral:rho},~\ref{theorem:FiniteWordsAlgorithmGeneral:ACC} and~\ref{theorem:FiniteWordsAlgorithmGeneral:F}, by Lemma~\ref{lemma:KleeneQOLfp}, that 
\begin{equation}\label{eq:lfpKleeneQO}
\lfp(\lambda \vect{X}\ldotp\rho (\vectarg{\epsilon}{F} \!\cup \Pre_{\cN}(\vect{X}))) = ρ\left(\KleeneQO(\abseq,\mathcal{F},\vect{\varnothing})\right)
\end{equation}
The function \(\abseq\) can be replaced by function \(\abseq^{\sharp}\) due to hypothesis~\ref{theorem:FiniteWordsAlgorithmGeneral:EQ}.
Moreover, by Equivalence~\eqref{equation:checklfpRLintoRL}, which holds by hypothesis~\ref{theorem:FiniteWordsAlgorithmGeneral:rho}, and Equation~\eqref{eq:lfpKleeneQO} we have that
\[\lang{\cN}\subseteq L_2 \Lra ρ\left(\KleeneQO (\abseq^{\sharp}, \mathcal{F}, \vect{\varnothing})\right) \subseteq \vectarg{L_2}{I}\enspace .\]

Finally, hypotheses~\ref{theorem:FiniteWordsAlgorithmGeneral:EQ} and~\ref{theorem:FiniteWordsAlgorithmGeneral:INC} guarantee, respectively, the decidability of the inclusion \(ρ\mathcal{F}(X) \subseteq ρ(X)\) performed at each step of the \(\KleeneQO\) iteration and the decidability of the inclusion of the abstraction of the lfp in \(\vectarg{L_2}{I}\).
\end{proof}

Note that Theorem~\ref{theorem:FiniteWordsAlgorithmGeneral} can also be stated in a symmetric version for right concatenation similarly to Theorem~\ref{theorem:backCompleteRight}.

\subsubsection{Using Galois Connections}
The next result reformulates Equivalence~\eqref{equation:checklfpRLintoRL} by using Galois Connections rather than closures, and shows how to design an algorithm that solves a language inclusion problem $\lang{\cN} \subseteq L_2$ on 
an \emph{abstraction} $D$ of the concrete domain $\tuple{\wp(\Sigma^*),\subseteq}$ whenever $D$ satisfies a list of requirements related to backward completeness and computability.

\begin{theorem}\label{theorem:EffectiveAlgorithm}
Let \(\cN=\tuple{Q,Σ,\delta,I,F}\) be an NFA and \(L_2\) be a language over \(\Sigma\).
Let \(\tuple{\wp(\Sigma^*),\subseteq} \galois{\alpha}{\gamma}\tuple{D,\leq_D}\) be a GC where \( \tuple{D,\leq_D}\) is a poset.
Assume that the following properties hold:
\begin{myEnumI}
\item \(L_2\in\gamma(D)\) and for every \( a \in \Sigma\) and \(X \in \wp(\Sigma^*)\), \(\gamma\alpha(a X) = \gamma\alpha(a \gamma\alpha(X))\).\label{theorem:EffectiveAlgorithm:prop:rho}
\item \((D,\leq_D,\sqcup,\bot_D)\) is an effective domain, meaning that: \((D,\leq_D,\sqcup,\bot_D)\) is an ACC join-semilattice with bottom $\bot_D$, 
every element of \(D\) has a finite representation, the binary relation 
\(\leq_D\) is decidable and the binary lub \(\sqcup\) is computable.\label{theorem:EffectiveAlgorithm:prop:absdecidable}
\item There is an algorithm, say \(\Pre^{\sharp}(\vect{X}^\sharp)\), which computes \(\alpha(\Pre_{\cN}(\gamma(\vect{X}^\sharp)))\),
for all \(\vect{X}^\sharp\in \alpha(\wp(\Sigma^*))^{|Q|}\).\label{theorem:EffectiveAlgorithm:prop:abspre}
\item There is an algorithm, say \(\epsilon^{\sharp}\), which computes \(\alpha(\vectarg{\epsilon}{F})\).\label{theorem:EffectiveAlgorithm:prop:abseps}
\item There is an algorithm, say \(\absincl(\vect{X}^\sharp)\), which decides the abstract inclusion 
\(\vect{X}^\sharp \leq_D \alpha(\vectarg{L_2}{I})\), for all \(\vect{X}^\sharp\in \alpha(\wp(\Sigma^*))^{|Q|}\).
\label{theorem:EffectiveAlgorithm:prop:absincl}
\end{myEnumI}
\medskip
Then, the following is an algorithm which decides whether \(\lang{\cN} \subseteq L_2\):

\medskip

\(\tuple{Y^\sharp_q}_{q\in Q} := \Kleene (\lambda \vect{X}^\sharp\ldotp\epsilon^{\sharp} \sqcup \Pre^{\sharp}(\vect{X}^\sharp), \vect{\bot_D})\)\emph{;}

\emph{\textbf{return}} \(\absincl(\tuple{Y^\sharp_q}_{q\in Q})\)\emph{;}
\end{theorem}
\begin{proof}
Let \(\rho \ud \gamma \alpha\in \uco(\wp(\Sigma^*))\), so that hypothesis~\ref{theorem:EffectiveAlgorithm:prop:rho} can be stated as 
\(L_2 \in \rho\) and \(\rho(aX) = \rho(a\rho(X))\).
It turns out that:
\begin{align*}
	\lang{\cN}\subseteq L_2 &\Lra\quad
	\text{[By Equivalence~\eqref{equation:checklfpRLintoRL}]}\\
	\lfp(\lambda \vect{X}\ldotp\rho (\vectarg{\epsilon}{F} \cup \Pre_{\cN}(\vect{X}))) \subseteq \vectarg{L_2}{I} &\Lra \quad
	\text{[By Lemma~\ref{lemma:alpharhoequality}]}\\
		\gamma(\lfp (\lambda \vect{X}^\sharp\ldotp \alpha (\vectarg{\epsilon}{F} \cup \Pre_{\cN}(\gamma(\vect{X}^\sharp))))) \subseteq \vectarg{L_2}{I} &\Lra 
	\quad \text{[By GC]}\\
	\gamma(\lfp (\lambda \vect{X}^\sharp\ldotp\alpha(\vectarg{\epsilon}{F}) \sqcup \alpha(\Pre_{\cN}(\gamma(\vect{X}^\sharp))))) \subseteq \vectarg{L_2}{I} &\Lra 
	\quad \text{[By GC since $L_2\in \gamma(D)$]}\\
	\lfp (\lambda \vect{X}^\sharp\ldotp\alpha(\vectarg{\epsilon}{F}) \sqcup \alpha(\Pre_{\cN}(\gamma(\vect{X}^\sharp)))) \leq_D \alpha(\vectarg{L_2}{I}) &
\end{align*}

By hypotheses~\ref{theorem:EffectiveAlgorithm:prop:absdecidable}, \ref{theorem:EffectiveAlgorithm:prop:abspre} and \ref{theorem:EffectiveAlgorithm:prop:abseps},
\(\Kleene (\lambda \vect{X}^\sharp\ldotp\epsilon^{\sharp} \sqcup \Pre^{\sharp}(\vect{X}^\sharp), \vect{\bot_D})\) is an algorithm computing the least fixpoint
$\lfp (\lambda \vect{X}^\sharp\ldotp\alpha(\vectarg{\epsilon}{F}) \sqcup \alpha(\Pre_{\cN}(\gamma(\vect{X}^\sharp))))$.  In particular, the hypotheses
\ref{theorem:EffectiveAlgorithm:prop:absdecidable}, \ref{theorem:EffectiveAlgorithm:prop:abspre} and \ref{theorem:EffectiveAlgorithm:prop:abseps} ensure that the Kleene iterates of 
$\lambda \vect{X}^\sharp\ldotp\epsilon^{\sharp} \sqcup \Pre^{\sharp}(\vect{X}^\sharp)$
starting from $\vect{\bot_D}$
are in $\alpha(\wp(\Sigma^*))^{|Q|}$, computable and finitely many and that 
it is decidable whether the iterates have reached the fixpoint.

Finally, hypothesis~\ref{theorem:EffectiveAlgorithm:prop:absincl} ensures the 
decidability of the  $\leq_D$-inclusion check of this least fixpoint 
in $\alpha(\vectarg{L_2}{I})$.  
\end{proof}

It is also worth noticing that, analogously to what has been done in Theorem~\ref{theorem:backCompleteRight}, 
the above Theorem~\ref{theorem:EffectiveAlgorithm} can be also stated in a symmetric version
for right (rather than left) concatenation.

\section{Instantiating the Framework}%
\label{sec:instantiating_the_framework_language_based_well_quasiorders}

We instantiate the general algorithmic framework of 
Section~\ref{sec:an_algorithmic_framework_for_language_inclusion_based_on_complete_abstractions} to the class of closure operators induced by quasiorder relations on words. 

\subsection{Word-based Abstractions}
Let \(\mathord{\leqslant} \subseteq \Sigma^* \times \Sigma^*\) be a quasiorder relation on words in $\Sigma^*$.  
Recall that the corresponding closure operator \(\rho_\leqslant \in \uco(\wp(\Sigma^*))\) is defined as follows: 
\begin{equation}\label{eq:qo-up-closure}
\rho_\leqslant(X) \ud  %
\{v\in \Sigma^* \mid \exists u\in X, \;u \leqslant v \} \enspace .
\end{equation}
Thus, $\rho_\leqslant(X)$ is the $\leqslant$-upward closure
of $X$ and it is easy to check that $\rho_\leqslant$ is indeed a closure
on the complete lattice $\tuple{\wp(\Sigma^*),\subseteq}$.

As described in Chapter~\ref{chap:prel}, the quasiorder \(\leqslant\) is left-monotone (resp.\ right-monotone) if{}f 
\begin{equation}\label{def-leftmon}
\forall x_1,x_2 \in \Sigma^*,\forall a\in \Sigma,\: x_1\leqslant x_2 \,\Ra\, a x_1 \leqslant a x_2 \quad \text{(resp.\ $x_1 a \leqslant x_2 a$)}
\end{equation}

In fact, if $x_1\leqslant x_2$ then Equation~\eqref{def-leftmon} implies 
that for all $y\in \Sigma^*$, $y x_1 \leqslant y x_2$ since, by induction 
on the length $|y|\in\bN$, we have that:
\begin{myEnumI}
\item if $y=\epsilon$ then $y x_1 \leqslant y x_2$; 
\item if $y=av$ with $a\in \Sigma,v\in \Sigma^*$ then, by inductive hypothesis, $v x_1 \leqslant v x_2$, so that by \eqref{def-leftmon}, $yx_1=av x_1 \leqslant av x_2=yx_2$
\end{myEnumI} 

\begin{definition}[$L$-Consistent Quasiorder]\label{def:LConsistent}\rm
Let $L\in$ $\wp(\Sigma^*)$. A quasiorder \(\mathord{\leqslant_L}\) on \(\Sigma^*\) is called \emph{left} (resp.\ \emph{right}) \(L\)\emph{-consistent} if{}f   
\begin{myEnumA}
\item \(\mathord{\leqslant}_L \cap (L\times \neg L) = \varnothing \);\label{eq:LConsistentPrecise}
\item \(\mathord{\leqslant}_L\) is left-monotone (resp.\ right-monotone). \label{eq:LConsistentmonotone}
\end{myEnumA} 
\medskip
Also, \(\mathord{\leqslant}_L\) is called \emph{\(L\)-consistent} when it is both left and right \(L\)-consistent.\hfill{\rule{0.5em}{0.5em}}
\end{definition}

As the following lemma shows, it turns out that a quasiorder is $L$-consistent if{}f it induces a closure which includes $L$ in its image and it is 
backward complete for concatenation. 
\begin{lemma}\label{lemma:properties}
Let \(L\in \wp(\Sigma^*)\) and \(\mathord{\leqslant_L}\) be a quasiorder on \(\Sigma^*\).
Then, \(\mathord{\leqslant_L}\) is a 
left (resp.\ right) \(L\)-consistent quasiorder on \(\Sigma^*\) if and only if
\begin{myEnumA}
\item \(\rho_{\leqslant_L}(L) = L\), and \label{lemma:properties:L}
\item \(\rho_{\leqslant_L}\) is backward complete for \(\lambda X\ldotp a X\) (resp.\ \(\lambda X\ldotp Xa\)) for all \(a\in \Sigma\).\label{lemma:properties:bw}
\end{myEnumA}
\end{lemma}
\begin{proof}
We consider the left case, the right case is symmetric. 

\begin{myEnumA}
\item The inclusion 
\(L\subseteq \rho_{\leqslant_L}(L)\) always 
holds because \(\rho_{\leqslant_L}\) is an upper closure. 
For the reverse inclusion we have that 
\begin{align*}
\rho_{\leqslant_L}(L)\subseteq L & \Lra \quad \text{[By definition of \(ρ_{\qo_L}(L)\)]}\\
\forall v\in \Sigma^*, \; (\exists u\in L,\, u \leqslant_L v) \:\Ra\: v\in L & \Lra \quad \\
\mathord{\leqslant}_L \cap (L\times \neg L) = \varnothing \enspace .
\end{align*}
Thus, $\rho_{\leqslant_L}(L)= L$ iff
condition~\ref{eq:LConsistentPrecise} of Definition~\ref{def:LConsistent} holds. 

\item We first prove that if \(\mathord{\leqslant}_L\) is left-monotone then for all $X\in \wp(\Sigma^*)$ we have that 
\(\rho_{\leqslant_L}(a X) = \rho_{\leqslant_L}(a \rho_{\leqslant_L}(X))\) for all \(a\in\Sigma\). 

Monotonicity of concatenation together with monotonicity and extensivity of 
$\rho_{\leqslant_L}$ imply that the inclusion \(\rho_{\leqslant_L}(a X) \subseteq \rho_{\leqslant_L}(a \rho_{\leqslant_L}(X))\) holds.
For the reverse inclusion, we have that:
\begin{align*}
	\rho_{\leqslant_L}(a \rho_{\leqslant_L}(X)) %
	&= \quad \text{[By definition of \(\rho_{\leqslant_L}\)]}\\
	\rho_{\leqslant_L}\left( \{ a y \mid \exists x\in X, x \leqslant_L y \} \right)
	&= \quad \text{[By definition of \(\rho_{\leqslant_L}\)]}\\
	\{ z \mid \exists x\in X, y\in \Sigma^*,\, x\leqslant_L y \land a y \leqslant_L z \}
	&\subseteq \quad \text{[By monotonicity of \(\leqslant_L\)]}\\
	\{ z \mid \exists x\in X, y\in \Sigma^*,\, ax\leqslant_L ay \land a y \leqslant_L z \}
	&= \quad \text{[By transitivity of \(\leqslant_L\)]}\\
	\{ z \mid \exists x\in X , a x\leqslant_L z\}
	&= \quad \text{[By definition of \(\rho_{\leqslant_L}\)]}\\
	\rho_{\leqslant_L}(a X) &\enspace . 
\end{align*}

Next, we show that if \(\rho_{\leqslant_L}(a X) = \rho_{\leqslant_L}(a \rho_{\leqslant_L}(X))\) for all $X\in \wp(\Sigma^*)$ and \(a\in\Sigma\) then \(\leqslant_L\) is left-monotone.

Let $x_1,x_2\in \Sigma^*$ and $a\in \Sigma$. 
If $x_1 \leqslant_L x_2$ then
$\{x_2\} \subseteq \rho_{\leqslant_L}(\{x_1 \})$, hence
$a\{x_2\} \subseteq  a\rho_{\leqslant_L}(\{x_1 \})$.
Then, by applying the monotone function 
$\rho_{\leqslant_L}$ we have that
$\rho_{\leqslant_L}(a\{x_2\}) \subseteq  \rho_{\leqslant_L}(a\rho_{\leqslant_L}(\{x_1 \}))$, so that, by backward completeness, 
$\rho_{\leqslant_L}(a\{x_2\}) \subseteq  \rho_{\leqslant_L}(a\{x_1 \})$.
Thus, $a\{x_2\} \subseteq \rho_{\leqslant_L}(a\{x_1 \})$, namely, 
$ax_1 \leqslant_L ax_2$. By~\eqref{def-leftmon}, this shows that $\leqslant_L$ is left-monotone. 
\end{myEnumA}
\end{proof}

Since \(ρ_{\leqslant} (\vectarg{\epsilon}{F} \!\cup \Pre_{\cN}(\vect{X})) = ρ_{\leqslant}(\minor{\vectarg{\epsilon}{F} \! \cup \Pre_{\cN}(\vect{X})})\) for every quasiorder then, by Lemma~\ref{lemma:properties}, we can apply Theorem~\ref{theorem:FiniteWordsAlgorithmGeneral} with the abstraction \(ρ_{\leqslant_{L_2}}\) induced by a left \(L_2\)-consistent well-quasiorder and \(\mathcal{F} =\minor{\vectarg{\epsilon}{F} \! \cup \Pre_{\cN}(\vect{X})}\) interpreted as the set of functions of the form \(f_i = \minor{\vectarg{\epsilon}{F} \! \cup \Pre_{\cN}(\vect{X})}_i\) where each \(\minor{\cdot}_i\) is a function mapping each set \(X \in \wp(Σ^*)\) into a minor \(\minor{X}_i\).
Intuitively, this means that we can manipulate \(\qo\)-upward closed sets in \(\wp(Σ^*)\) using their finite minors, as already shown by \citet{ACJT96}.

As a consequence, we obtain Algorithm~\AlgRegularW which, given a left \(L_2\)-consistent well-quasiorder, solves the language inclusion problem \(\lang{\cN} \subseteq L_2\) for any automaton \(\cN\).
The algorithm is called ``word-based'' because the vector \(\tuple{Y_q}_{q \in Q}\) consists of finite sets of words in \(Σ^*\).
We write \(\mathord{\sqsubseteq_{\leqslant^{\ell}_{L_2}}}\hspace{-4pt} \cap \mathord{(\sqsubseteq_{\leqslant^{\ell}_{L_2}}\hspace{-2pt})^{-1}}\) as the first argument of \(\KleeneQO\) to denote the function \(f(X,Y)\) that returns \(\True\) if{}f \(X \sqsubseteq_{\leqslant^{\ell}_{L_2}} Y\) and \(Y \sqsubseteq_{\leqslant^{\ell}_{L_2}} X\).

\begin{figure}[!ht]
\RemoveAlgoNumber
\begin{algorithm}[H]
\SetAlgorithmName{\AlgRegularW}{}
\SetSideCommentRight
\caption{Word-based algorithm for \(\lang{\cN} \subseteq L_2\)}\label{alg:RegIncW}

\KwData{NFA \(\cN=\tuple{Q,Σ,\delta,I,F}\); decision procedure for \(u\mathrel{\in } L_2\); decidable left \(L_2\)-consistent wqo \(\mathord{\leqslant^{\ell}_{L_2}}\).}

\(\tuple{Y_q}_{q\in Q} := \KleeneQO (\mathord{\sqsubseteq_{\leqslant^{\ell}_{L_2}}}\hspace{-4pt} \cap \mathord{(\sqsubseteq_{\leqslant^{\ell}_{L_2}}\hspace{-2pt})^{-1}}, \lambda \vect{X}\ldotp\lfloor\vectarg{\epsilon}{F} \cup \Pre_{\cN}(\vect{X})\rfloor, \vect{\varnothing})\)\;

\ForAll{\(q \in I\)}{
	\ForAll{\(u \in Y_q\)}{
		\lIf{\(u \notin L_2\)}{\Return \textit{false}}
	}
}
\Return \textit{true}\;
\end{algorithm}
\end{figure}

\begin{theorem}\label{theorem:quasiorderAlgorithm}
Let \(\cN=\tuple{Q,Σ,\delta,I,F}\) be an NFA and let \(L_2\in \wp(\Sigma^*)\) be a language such that: 
\begin{myEnumIL}
\item membership in $L_2$ is decidable; \label{theorem:quasiorderAlgorithm:membership}
\item there exists a decidable left \(L_2\)-consistent wqo on $\Sigma^*$.\label{theorem:quasiorderAlgorithm:decidableL}%
\end{myEnumIL}%
Then, Algorithm \AlgRegularW decides the inclusion problem \(\lang{\cN} \subseteq L_2\).
\end{theorem}
\begin{proof}
Let $\leqslant_{L_2}^{\ell}$ be a decidable left $L_2$-consistent well-quasiorder on $\Sigma^*$.
Then, we check that hypothesis~\ref{theorem:FiniteWordsAlgorithmGeneral:rho}-\ref{theorem:FiniteWordsAlgorithmGeneral:INC} of Theorem~\ref{theorem:FiniteWordsAlgorithmGeneral} are satisfied.
\begin{myEnumA}
\item It follows from hypothesis~\ref{theorem:quasiorderAlgorithm:decidableL} and Lemma~\ref{lemma:properties} that \(\leqslant_{L_2}^{\ell}\) is backward complete for left concatenation and satisfies \(ρ_{\leqslant_{L_2}^{\ell}}(L_2) = L_2\).

\item Since \(\leqslant_{L_2}^{\ell}\) is a wqo, then \(\tuple{\{ρ_{\leqslant_{L_2}^{\ell}}(S) \mid S \in \wp(Σ^*)\}, \subseteq}\) is an ACC CPO.

\item Let \(\lfloor\vectarg{\epsilon}{F} \cup \Pre_{\cN}(\vect{X})\rfloor\) be the set of functions \(f_i\) each of which maps each set \(X \in \wp(Σ^*)\) into a minor of \(\vectarg{\epsilon}{F} \cup \Pre_{\cN}(\vect{X})\).
Since \(\rho_{\leqslant_{L_2}^{\ell}}(X) = ρ_{\leqslant_{L_2}^{\ell}}(\minor{X})\) for all \(X \in \wp(Σ^*)^{|Q|}\), 
we have that all functions \(f_i\) satisfy 
\[\rho_{\leqslant_{L_2}^{\ell}} (\vectarg{\epsilon}{F} \!\cup \Pre_{\cN}(\vect{X})) = ρ_{\leqslant_{L_2}^{\ell}}(f_i(\vect{X}))\enspace .\]

\item The equality \(ρ_{\leqslant_{L_2}^{\ell}}(S_1) = ρ_{\leqslant_{L_2}^{\ell}}(S_2)\) is decidable for every \(S_1, S_2 \in \wp(Σ^*)^{|Q|}\) since 
\[ρ_{\leqslant_{L_2}^{\ell}}(S_1) = ρ_{\leqslant_{L_2}^{\ell}}(S_2) \Lra S_1 \sqsubseteq_{\leqslant_{L_2}^{\ell}} S_2 \land S_2 \sqsubseteq_{\leqslant_{L_2}^{\ell}} S_1\]
and, by hypothesis~\ref{theorem:quasiorderAlgorithm:decidableL}, \(\leqslant_{L_2}^{\ell}\) is decidable.

\item Since \(\vectarg{L_2}{I} = \tuple{\nullable{q \in I}{L_2}{\Sigma^*}}_{q \in Q})\), the inclusion trivially holds for all components \(Y_q\) with \(q \notin I\).
Therefore, it suffices to check whether \(Y_q \subseteq L_2\) holds for \(q \in I\) which, since \(Y_q = \minor{S}\) with \(S \in \wp(Σ^*)\), can be decided by performing finitely many membership tests as done by lines 2-5 of Algorithm~\AlgRegularW.
By hypothesis~\ref{theorem:quasiorderAlgorithm:membership}, this check is decidable.
\end{myEnumA}
\end{proof}

\subsubsection{Right Concatenation}
Following Section~\ref{rightwqo},
a symmetric version, called \AlgRegularWr, of Algorithm \AlgRegularW and of Theorem~\ref{theorem:quasiorderAlgorithm} for \emph{right} \(L_2\)-consistent wqos can be easily derived  as follows. 

\begin{figure}[!ht]
\RemoveAlgoNumber
\begin{algorithm}[H]
\SetAlgorithmName{\AlgRegularWr}{}

\caption{Word-based algorithm for \(\lang{\cN} \subseteq L_2\)} \label{alg:RegIncWr}

\KwData{NFA \(\cN=\tuple{Q,Σ,\delta,I,F}\); decision procedure for \(u\in  L_2\); decidable right \(L_2\)-consistent wqo \(\mathord{\leqslant^r_{L_2}}\).}

\(\tuple{Y_q}_{q\in Q} := \KleeneQO (\mathord{\sqsubseteq_{\leqslant^{r}_{L_2}}}\hspace{-4pt} \cap \mathord{(\sqsubseteq_{\leqslant^{r}_{L_2}}\hspace{-2pt})^{-1}}, \lambda \vect{X}\ldotp\lfloor\vectarg{\epsilon}{I} \cup \Post_{\cN}(\vect{X})\rfloor, \vect{\varnothing})\)\;

\ForAll{\(q \in F\)}{
	\ForAll{\(u \in Y_q\)} {
		\lIf{\(u \notin L_2\)}{\Return \textit{false}}
	}
}
\Return \textit{true}\;
\end{algorithm}
\end{figure}

\begin{theorem}\label{theorem:quasiorderAlgorithmR}
Let \(\cN=\tuple{Q,Σ,\delta,I,F}\) be an NFA and let \(L_2\in \wp(\Sigma^*)\) be a language such that 
\begin{myEnumIL}
\item membership in $L_2$ is decidable;
\item there exists a decidable right \(L_2\)-consistent wqo on $\Sigma^*$.%
\end{myEnumIL}%
Then, Algorithm \AlgRegularWr decides the inclusion problem \(\lang{\cN} \subseteq L_2\).
\end{theorem}

In the following, we will consider different quasiorders on $\Sigma^*$ and we will show that they fulfill the requirements of Theorem~\ref{theorem:quasiorderAlgorithm}, so that they yield algorithms for solving a language inclusion problem.

\subsection{Nerode Quasiorders}\label{sec:nerode}
\label{sub:the_left_nerode_quasi_order_relative_to_a_language}
Recall from Chapter~\ref{chap:prel} that the \emph{left} and \emph{right} 
\emph{Nerode's quasiorders} on \(\Sigma^*\) are defined in the standard way:
\begin{align*}
	u\qlL v &\udiff\; L u^{-1} \subseteq L v^{-1} \,,&
	u\qrL v &\udiff\; u^{-1} L \subseteq v^{-1} L \enspace . 
\end{align*}
The following result shows that Nerode's quasiorders are the weakest (i.e. greatest w.r.t.\ set inclusion of binary relations) \(L_2\)-consistent quasiorders for which the algorithm \AlgRegularW can be instantiated to decide a language inclusion \(\lang{\cN}\subseteq L_2\).

\begin{lemma}\label{lemma:leftrightnerodegoodqo}
Let $L\in \wp(\Sigma^*)$.  
Then
\begin{myEnumA}
\item \(\mathord{\qlL}\) and \(\mathord{\qrL}\) are, respectively, left and right \(L\)-consistent quasiorders.
If $L$ is regular then, additionally, \(\mathord{\qlL}\) and \(\mathord{\qrL}\) are, respectively, decidable wqos. 	\label{lemma:leftrightnerodegoodqo:Consistent}
\item Let \(\ql\) and \(\qr\) be, respectively, a left and a right \(L\)-consistent quasiorder on $\Sigma^*$.
Then \( \rho_{\mathord{\qlL}} \subseteq \rho_{\ql} \) and \( \rho_{\qrL} \subseteq \rho_{\qr} \).\label{lemma:leftrightnerodegoodqo:Incl}
\end{myEnumA}
\end{lemma}
\begin{proof}\hfill

\begin{myEnumA}
\item As explained in Chapter~\ref{chap:prel}, \citet[Section 2]{deLuca1994} show that \(\qlL\) and \(\qrL\) are, respectively, left and right monotone quasiorders.

On the other hand, note that given \(u \in L\) and \(v \notin L\) we have that \(\epsilon \in Lu^{-1}\) and \(\epsilon \in u^{-1}L\) while \(\epsilon \notin Lv^{-1}\) and \(\epsilon \notin v^{-1}L\). 
Hence, \(\mathord{\qlL}\) (resp. \(\mathord{\qrL}\)) is a left (resp. right) \(L\)-consistent quasiorder.

Finally, if $L$ is regular then both relations are 
clearly decidable.

\item We consider the left case (the right case is symmetric). 

As shown by \citet[Section~2, point~4]{deLuca1994}, \(\mathord{\qlL}\) is maximum in the set of all the left \(L\)-consistent quasiorders, i.e.\ every left \(L\)-consistent quasiorder \(\ql\) is such that \(x \ql y \Ra x \qlL y \).
As a consequence, \(\rho_{\ql}(X) \subseteq \rho_{\qlL}(X)\) holds for all \(X\in \wp(\Sigma^*)\), namely, 
\(\mathord{\ql} \subseteq \mathord{\qlL}\).\qedhere
\end{myEnumA}
\end{proof}

We then derive a first instantiation of Theorem~\ref{theorem:quasiorderAlgorithm}. %
Because membership is decidable for regular languages $L_2$, Lemma~\ref{lemma:leftrightnerodegoodqo}~\ref{lemma:leftrightnerodegoodqo:Consistent} for \(\ql_{L_2}\) implies that the hypotheses \ref{theorem:quasiorderAlgorithm:membership} and \ref{theorem:quasiorderAlgorithm:decidableL} of Theorem~\ref{theorem:quasiorderAlgorithm} are satisfied, so that Algorithm \AlgRegularW instantiated to \(\ql_{L_2}\)
decides the inclusion \(\lang{\cN} \subseteq L_2\) when $L_2$ is regular. 

Furthermore, under these hypotheses,
Lemma~\ref{lemma:leftrightnerodegoodqo}~\ref{lemma:leftrightnerodegoodqo:Incl} shows that \(\qo_{L_2}^{\ell}\) is the weakest (i.e. greatest for set inclusion) 
left \(L_2\)-consistent quasiorder for which the algorithm \AlgRegularW can be instantiated 
for deciding the inclusion $\lang{\cN}\subseteq L_2$.

\begin{figure}[t]
    \centering
  \begin{minipage}{0.45\textwidth}
  \begin{tikzpicture}[->,>=stealth',shorten >=1pt,auto,node distance=5mm and 1cm,thick,initial text=]
  \tikzstyle{every state}=[scale=0.75,fill=customblue!60,draw=blue!60,text=black]

  \node[initial, state] (1) {\(1\)};
  \node[accepting, state] (2) [right=of 1] {\(2\)};
  \node (A1) [left=of 1] {$\cN_1:$};
  
  \path (1) edge node {\(a, b, c\)} (2)
        (1) edge[loop, in=60, out=120, looseness=5] node {\(a\)} (1)
            ;
  \end{tikzpicture}
  \end{minipage}
  \begin{minipage}{0.45\textwidth}
  \begin{tikzpicture}[->,>=stealth',shorten >=1pt,auto,node distance=6mm and 1cm,thick,initial text=]
  \tikzstyle{every state}=[scale=0.75,fill=customblue!60,draw=blue!60,text=black]
  
  \node[initial,state] (1) {\(1\)};
  \node[state] (2) [right=of 1] {\(2\)};
  \node[accepting, state] (5) [right=of 2] {\(5\)};
  \node[state] (3) [above=of 2] {\(3\)};
  \node[state] (4) [below=of 2] {\(4\)};
  \node (A2) [left=of 1] {$\cN_2:$};
  
  \path (1) edge[loop, in=80, out=140, looseness=5] node[above]  {\(a\)} (1)
        (1) edge node {\(a\)} (2)
        (1) edge[bend left=30] node {\(a\)} (3)
        (1) edge[bend right=30] node[below] {\(a,b\)} (4)
        (3) edge[bend left=30] node {\(a\)} (5)
        (3) edge[loop, in=60, out=120, looseness=5] node {\(a,b\)} (3)
        (2) edge node {\(c\)} (5)
        (2) edge[loop, in=60, out=120, looseness=4] node[yshift=-3pt] {\(a\)} (2)
        (4) edge[bend right=30] node[below] {\(b\)} (5)
          ;
  \end{tikzpicture}
  \end{minipage}
  \caption{Two automata \(\cN_1\) (left) and \(\cN_2\) (right) generating the regular languages \(\lang{\cN_1} = a^*(a+b+c)\) and \(\lang{\cN_2}= a^* (a(a+b)^*a+a^+c+ab+bb)\).}
    \label{fig:B}
  \end{figure}
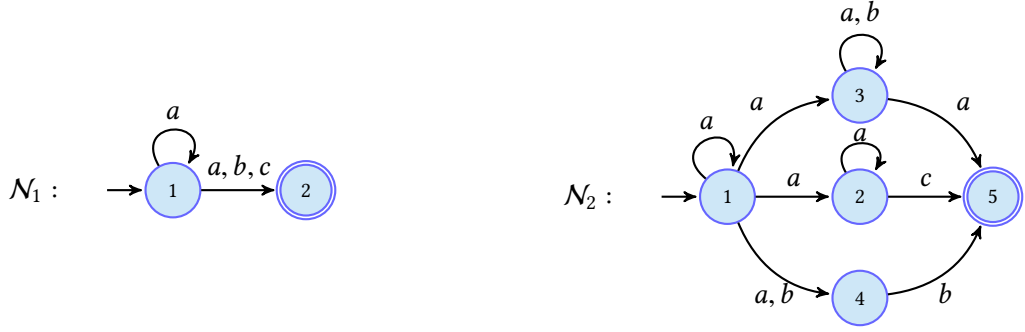

\begin{example}\label{example:Word_Regular_LInc}
We illustrate the use of the left Nerode's quasiorder in the algorithm \AlgRegularW for solving the language inclusion \(\lang{\cN_1} \subseteq \lang{\cN_2}\), where \(\cN_1\) and \(\cN_2\) are the automata shown in Figure~\ref{fig:B}.
The equations for  \(\cN_1\) are as follows:
\[
  \Eqn(\cN_1)=\begin{cases}
    X_1 = \varnothing \cup aX_1 \cup aX_2 \cup bX_2 \cup cX_2\\
    X_2 = \{\epsilon\}
  \end{cases} \enspace .
\]

\noindent
We have the following quotients (among others) for \(L = \lang{\cN_2}\). 
\begin{align*}
L \epsilon^{-1} = \; & a^* (a(a+b)^*a+a^+c+ab+bb) & L b^{-1} = \; & a^* (a + b)  \\
L a^{-1} = \; & a^* a(a+b)^*  &  L c^{-1} = \; & a^* a^+\\
L w^{-1} = \; & a^* \text{ if{}f } w \in (a(a+b)^*a+ac+ab+bb) \span
\end{align*}

\noindent
It is straightforward to check that, among others, the following relations hold between different alphabet symbols: \(b \qlL a\), 
\(c \qlL a\)
 and \(c \qlL b\). 
Then, let us show the computation of the Kleene iterates performed by Algorithm \AlgRegularW.
\begin{align*}
\vect{Y}^{(0)} &= \vect{\varnothing}\\
\vect{Y}^{(1)} &= \vectarg{\epsilon}{F} = \tuple{\varnothing, \{\epsilon\}} \\
\vect{Y}^{(2)} &= \lfloor\vectarg{\epsilon}{F}\rfloor {\sqcup} \lfloor\Pre_{\cN_1}(\vect{Y}^{(1)})\rfloor = \tuple{\varnothing, \{\epsilon\}} {\sqcup} \tuple{\minor{\varnothing \cup a\varnothing \cup a\{\epsilon\} \cup b\{\epsilon\} \cup c\{\epsilon\}}, \minor{\{\epsilon\}}}\\
&= \tuple{\minor{\{a,b,c\}}, \minor{\{\epsilon\}}} = \tuple{\{c\}, \{\epsilon\}}\\
\vect{Y}^{(3)} &= \lfloor\vectarg{\epsilon}{F}\rfloor {\sqcup} \lfloor\Pre_{\cN_1}(\vect{Y}^{(2)})\rfloor = 
\tuple{\varnothing, \{\epsilon\}} {\sqcup} \tuple{\minor{\varnothing {\cup} a\{c\} {\cup} a\{\epsilon\} {\cup} b\{\epsilon\} {\cup} c\{\epsilon\}}, \minor{\{\epsilon\}}}\\
&=
\tuple{\minor{\{ac, a, b, c\}}, \minor{\{\epsilon\}}} = \tuple{\{c\}, \{\epsilon\}} 
\end{align*}
The least fixpoint is thus \(\vect{Y} = \tuple{\{c\}, \{\epsilon\}}\).
Since $c\in \vect{Y}_1$ and \(c \notin \lang{\cN_2}\), Algorithm \AlgRegularW concludes that the language inclusion \(\lang{\cN_1} \subseteq \lang{\cN_2}\) does not hold. \eox
\end{example}

\subsubsection{On the Complexity of Nerode's quasiorders}
For the inclusion problem between languages generated by finite automata, deciding the 
(left or right) Nerode's quasiorder  can be easily shown to be as hard as the language inclusion problem, which is PSPACE-complete.
In fact, given the automata \(\cN_1=\tuple{Q_1,\delta_1,I_1,F_1,\Sigma}\) and \(\cN_2=\tuple{Q_2,\delta_2,I_2,F_2,\Sigma}\), one can define the union automaton \(\cN_3\ud \tuple{Q_1\cup Q_2\cup\{q^{\iota}\}, \delta_3, \{q^{\iota}\}, F_1\cup F_2}\) where \(\delta_3 \) maps \((q^\iota,a)\) to \(I_1\), \( (q^\iota,b) \) to \(I_2\) and behaves like \(\delta_1\) or \(\delta_2\) elsewhere. Then, it turns out that 
\[a \qr_{\lang{\cN_3}} b \Lra a^{-1}\lang{\cN_3} \subseteq b^{-1}\lang{\cN_3} \Lra \lang{\cN_1}\subseteq \lang{\cN_2}\enspace .\]
It follows that deciding the right Nerode's quasiorder \(\qr_{\lang{\cN_3}}\) is as hard as deciding \(\lang{\cN_1}\subseteq \lang{\cN_2}\).

Also, for the inclusion problem of a language generated by an 
automaton within the trace set of a one-counter net (see Section~\ref{sub:containment_in_one_counter_languages}), the right Nerode's quasiorder is a right language-consistent well-quasiorder but it turns out to be undecidable (see Lemma~\ref{lemma:RightNerodeOcnwqo}).

\subsection{State-based Quasiorders}\label{subsec:state-qos}
Consider the inclusion problem \(\lang{\cN_1} \subseteq \lang{\cN_2}\) where \(\cN_1\) and \(\cN_2\) are NFAs.
In the following, we study a class of well-quasiorders based on \(\cN_2\), called state-based quasiorders. 
These quasiorders are strictly stronger (i.e. lower w.r.t.\ set inclusion of binary relations) than the Nerode's quasiorders and sidestep the untractability or undecidability of Nerode's quasiorders yet allowing to define an algorithm solving the language inclusion problem.

\subsubsection{Inclusion in Regular Languages.}
\label{sub:automata_based}

We  define the quasiorders \(\qlN\)  and \(\qrN\) induced by an NFA \(\cN=\tuple{Q,Σ,\delta,I,F}\)
as follows:
\begin{align}\label{eqn:state-qo}
u \mindex{\qlN} v & \udiff \pre^{\cN}_{u}(F) \subseteq \pre^{\cN}_{v}(F)\,,
&
u \mindex{\qrN} v & \udiff \post^{\cN}_{u}(I) \subseteq \post^{\cN}_{v}(I) \,.
\end{align}
The superscripts in $\qlN$ and $\qrN$ stand, respectively, for left/right because they are, respectively, left and right well-quasiorders as the following result shows. 

\begin{lemma}\label{lemma:LAconsistent}
The relations \(\mathord{\qlN}\) and \(\mathord{\qrN}\) are, respectively, decidable left and right 
\(\lang{\cN}\)-consistent wqos.
\end{lemma}
\begin{proof}
Since, for every \(u \in \Sigma^*\), \(\pre^{\cN}_u(F)\) is a computable subset of a the finite set of states of \(\cN\), it turns out that \(\mathord{\qlN}\) is a decidable wqo. 
Let us check that \(\mathord{\qlN}\) is left \(\lang{\cN}\)-consistent according to Definition~\ref{def:LConsistent}~\ref{eq:LConsistentPrecise}-\ref{eq:LConsistentmonotone}. 

\begin{myEnumA}
\item Let \(u\in \lang{\cN}\) and \(v\notin \lang{\cN}\).
We have that \(\pre^{\cN}_u(F)\) contains some initial state while \(\pre^{\cN}_v(F)\) does not, hence \(u \nleq^{\ell}_{\cN} v\). 
Therefore, \(\qlN \cap (L\times L^c) = \varnothing\).

\item Let us check that $\qlN$ is left monotone.
Observe that $\pre^\cN_x$ is a monotone function and that 
\begin{eqnarray}
\pre^{\cN}_{uv} = \pre^{\cN}_{u} \comp \pre^{\cN}_v \enspace .\label{eq:prepre}
\end{eqnarray}
Therefore, for all $x_1,x_2\in \Sigma^*$ and $a\in \Sigma$, 
\begin{align*}
x_1 \qlN x_2 & \Ra  \quad\text{[By definition of \(\qlN\)]} \\
\pre^{\cN}_{x_1}(F) \subseteq \pre^{\cN}_{x_2}(F) & \Ra  \quad\text{[Since $\pre^\cN_a$ is monotone]} \\
\pre^{\cN}_{a}(\pre^{\cN}_{x_1}(F)) \subseteq \pre^{\cN}_{a}(\pre^{\cN}_{x_2}(F)) & \Lra  \quad\text{[By Equation~\eqref{eq:prepre}]} \\
\pre^{\cN}_{ax_1}(F) \subseteq \pre^{\cN}_{ax_2}(F) & \Lra  \quad\text{[By definition of \(\qlN\)]} \\
ax_1 \qlN ax_2 & \enspace . 
\end{align*}
\end{myEnumA}

The proof that \(\leq_{\cN}^r\) is a decidable right \(\lang{\cN}\)-consistent quasiorder is symmetric. 
\end{proof}

As a consequence, Theorem~\ref{theorem:quasiorderAlgorithm} applies to the wqo \(\mathord{\ql_{\cN_2}}\) (and 
\(\mathord{\qr_{\cN_2}}\)), so that one can instantiate Algorithm~\AlgRegularW with $\mathord{\ql_{\cN_2}}$ for deciding $\lang{\cN_1}\subseteq \lang{\cN_2}$. 

Turning back to the left Nerode wqo
$\ql_{\lang{\cN_2}}$, it turns out that: 
\begin{align*}
u \ql_{\lang{\cN_2}} v  \Lra \lang{\cN_2}u^{-1} \subseteq \lang{\cN_2} v^{-1} 
\Lra W_{I,\pre^{\cN_2}_u(F)} \subseteq W_{I,\pre^{\cN_2}_v(F)} \enspace .
\end{align*}
Since \(\pre^{\cN_2}_u(F) \subseteq \pre^{\cN_2}_v(F) \Ra W_{I,\pre^{\cN_2}_u(F)} \subseteq W_{I,\pre^{\cN_2}_v(F)}\), it follows that 
\[u \ql_{\cN_2} v \Ra u \ql_{\lang{\cN_2}} v\] 

\begin{example}\label{example:Word_Regular_LInc:states}
We illustrate the left state-based quasiorder by using it to solve the language inclusion \(\lang{\cN_1} \subseteq \lang{\cN_2}\) from Example~\ref{example:Word_Regular_LInc}.
We have, among others, the following set of predecessors of $F_{\cN_2}$:
\begin{align*}
\pre_{\epsilon}^{\cN_2}(F_{\cN_2}) & \,{=}\, \{5\} & \pre_{a}^{\cN_2}(F_{\cN_2}) & \,{=}\, \{3\} & \pre_{b}^{\cN_2}(F_{\cN_2}) & \,{=}\, \{4\} & \pre_{c}^{\cN_2}(F_{\cN_2}) & \,{=}\, \{2\} \\
\pre_{aa}^{\cN_2}(F_{\cN_2}) & \,{=}\, \{1, 3\} & \pre_{ab}^{\cN_2}(F_{\cN_2}) & \,{=}\, \{1\} & \pre_{ac}^{\cN_2}(F_{\cN_2}) & \,{=}\, \{1, 2\} & \pre_{aab}^{\cN_2}(F_{\cN_2}) & \,{=}\, \{1\}
\end{align*}

Recall from Example~\ref{example:Word_Regular_LInc} that, for the Nerode's quasiorder, we have that \(b \ql_{\lang{\cN_2}} a\) and \(c \ql_{\lang{\cN_2}} a\) while none of these relations hold for \(\ql_{\cN_2}\).
Let us next show the Kleene iterates computed by Algorithm \AlgRegularW when using \(\ql_{\cN_2}\).
\begin{align*}
\vect{Y}^{(0)} &=\vect{\varnothing}\\
\vect{Y}^{(1)} &= \vectarg{\epsilon}{F} = \tuple{\varnothing, \{\epsilon\}} \\
\vect{Y}^{(2)} &= \lfloor\vectarg{\epsilon}{F}\rfloor {\sqcup} \lfloor\Pre_{\cN_1}(\vect{Y}^{(1)})\rfloor = \tuple{\minor{\{a, b, c\}}, \minor{\{\epsilon\}}} = \tuple{\{a, b, c\}, \{\epsilon\}}\\
\vect{Y}^{(3)} &= \lfloor\vectarg{\epsilon}{F}\rfloor {\sqcup} \lfloor\Pre_{\cN_1}(\vect{Y}^{(2)})\rfloor = \tuple{\minor{\{aa, ab, ac, a, b, c\}}, \minor{\{\epsilon\}}} = \tuple{\{ab, a, b, c\}, \{\epsilon\}} \\
\vect{Y}^{(4)} &= \lfloor\vectarg{\epsilon}{F}\rfloor {\sqcup} \lfloor\Pre_{\cN_1}(\vect{Y}^{(3)})\rfloor = \tuple{\minor{\{aab, aa, ab, ac, a, b, c\}}, \minor{\{\epsilon\}}} = \tuple{\{ab, a, b, c\}, \{\epsilon\}} 
\end{align*}

The least fixpoint is therefore \(\vect{Y}=\tuple{\{ab, a, b, c\}, \{\epsilon\}}  \).
Since $c\in \vect{Y}_0$ and \(c \notin \lang{\cN_2}\), Algorithm \AlgRegularW concludes that the inclusion \(\lang{\cN_1} \subseteq \lang{\cN_2}\) does not hold. \eox
\end{example}

\subsubsection{Simulation-based Quasiorders.}\label{sec:simulation_basedQO}
Recall that a \demph{simulation} on an NFA $\cN= \tuple{Q,Σ,\delta,I,F}$  is a binary relation on the states of \(\cN\), i.e. \(\mathord{\preceq} \subseteq Q\times Q\), such that for all $p,q\in Q$ if \(p\preceq q\) then the following two conditions hold:
\begin{myEnumI}
\item if \(p\in F\) then \(q\in F\);
\item for every transition \(p \xrightarrow{a} p'\), there exists a transition \(q \xrightarrow{a} q'\) such that \(p'\preceq q'\).
\end{myEnumI}
\medskip 

It is well known that simulation implies language inclusion, i.e. if $\preceq$ is a simulation on $\cN$ then
\[ q \preceq q'  \Ra W^{\cN}_{q,F}\subseteq W^{\cN}_{q',F} \enspace .\]
A relation \(\mathord{\preceq}\subseteq Q\times Q\) can be lifted in the standard universal-existential way to a relation $\preceq^{\forall\exists}\subseteq \wp(Q)\times \wp(Q)$ on sets of states as follows: 
\[ X \preceq^{\forall\exists} Y \:\udiff\: \forall x\in X, \exists y\in Y,\: x\preceq y \enspace.\]
In particular, if $\preceq$ is a qo then $\preceq^{\forall\exists}$ is a qo as well. 
Also, if $\preceq$ is a simulation relation then its lifting $\preceq^{\forall\exists}$ is such 
that \(X \preceq^{\forall\exists} Y \Ra W^{\cN}_{X,F} \subseteq W^{\cN}_{Y,F}\) holds. This suggests us to
define a \emph{right simulation-based quasiorder} \(\preceq^{r}_{\cN}\) on $\Sigma^*$ induced by a simulation $\preceq$ on $\cN$ as follows:
\begin{equation}\label{eq:sim-qo}
u \mindex{\preceq^{r}_{\cN}} v \:\udiff\:  \post^{\cN}_u(I) \preceq^{\forall\exists} \post^{\cN}_v(I) \enspace .
\end{equation}

\begin{lemma}\label{lemma:simulationLConsistent}
Let \(\cN\) be an NFA and let \(\mathord{\preceq}\) be a simulation on $\cN$.
Then the right si-mulation-based quasiorder \(\mathord{\preceq^r_{\cN}}\) is a decidable right \(\lang{\cN}\)-consistent well-quasiorder.
\end{lemma}
\begin{proof}
Since, for every \(u \in \Sigma^*\), \(\post^{\cN}_u(F)\) is a computable subset of a the finite set of states of \(\cN\), it turns out that \(\mathord{\preceq^r_{\cN}}\) is a decidable wqo. 
Next, we show that \(\mathord{\preceq^r_{\cN}}\) is right \(\lang{\cN}\)-consistent according to Definition~\ref{def:LConsistent}~\ref{eq:LConsistentPrecise}-\ref{eq:LConsistentmonotone}. 
\begin{myEnumA}
\item Let \(u\in \lang{\cN}\) and \(v\notin \lang{\cN}\).
We have that \(\post^{\cN}_u(I)\) contains some final state while \(\post^{\cN}_v(I)\) does not.
Let \(q \in \post^{\cN}_u(I) \cap F\).
We have that \(q \preceq^r_{\cN} q'\) for no \(q' \in \post^{\cN}_v(I)\) since, by simulation, this would imply \(q' \in \post^{\cN}_v(I) \cap F\), which contradicts the fact that \(F \cap \post^{\cN}_v(I) = \varnothing\).
We conclude that \(u \npreceq^r_{\cN} v\), hence \(\preceq^r_{\cN} \cap (L\times L^c) = \varnothing\).

\item Next we show that \(\preceq^r_{\cN}\) is right monotone. By Equation~\eqref{def-leftmon}, we check that for all $u,v\in \Sigma^*$ and $a\in \Sigma$,  
\(u\preceq^r_{\cN} v \Ra ua \preceq^r_{\cN} va\):  %
\begin{adjustwidth}{-0.8cm}{}%
\begin{myAlign}{-10pt}{-5pt}
u \preceq^r_{\cN} v & \Lra \; \text{[By def. of \(\preceq^r_{\cN}\)]}\\
\post^{\cN}_u(I) \preceq^{\forall\exists} \post^{\cN}_v(I) & \Lra \; \text{[By def.\ of \(\preceq^{\forall\exists}\)]} \\
\forall x \in \post^{\cN}_u(I), \exists y \in \post^{\cN}_v(I), x \preceq y & \Ra \; \text{[By def.\ of \(\preceq\)]} \\
\forall x \ggoes{a} x' ,\; x \in \post^{\cN}_u(I),\; \exists y \ggoes{a} y' ,\; y\in \post^{\cN}_v(u), x' \preceq y' & \Lra \\
\span\specialcell{\hfill\text{[Since \(\post^{\cN}_{a}\circ \post^{\cN}_{u} = \post^{\cN}_{ua}(I)\)]}}\\
\forall x' \in \post^{\cN}_{ua}(I), \exists y' \in \post^{\cN}_{va}(I), \; x' \preceq y' & \Lra \;\text{[By def.\ of \(\preceq^{\forall\exists}\)]}\\
\post^{\cN}_{ua}(I) \preceq^{\forall\exists} \post^{\cN}_{va}(I) & \Lra \;\text{[By def.\ of \(\preceq_{\cN}^r\)]}\\
ua \preceq_{\cN}^r va & \enspace .
\end{myAlign}
\end{adjustwidth}
\end{myEnumA}
\end{proof}

Thus, once again, Theorem~\ref{theorem:quasiorderAlgorithmR} applies to 
\(\mathord{\preceq^r_{\cN_2}}\) and this allows us to instantiate the 
algorithm \AlgRegularWr to the quasiorder 
$\mathord{\preceq^r_{\cN_2}}$ for deciding the inclusion $\lang{\cN_1}\subseteq \lang{\cN_2}$. 

On the other hand, note that it is possible to define a left simulation \(\preceq^{\forall\exists}_{R}\) on an automaton \(\cN\) by applying \(\preceq^{\forall\exists}\) on the reverse of \(\cN\).
This left simulation induces a \emph{left simulation-based quasiorder} on \(\Sigma^*\) as follows:

\begin{equation}\label{eq:sim-qo:left}
u \mindex{\preceq^{l}_{\cN}} v \:\udiff\:  \pre^{\cN}_u(F) \preceq^{\forall\exists}_R \pre^{\cN}_v(F) \enspace .
\end{equation}

It is straightforward to check that Theorem~\ref{theorem:quasiorderAlgorithm} applies to \(\mathord{\preceq^{\ell}_{\cN_2}}\) and, therefore, we can instantiate Algorithm~\AlgRegularW for deciding \(\lang{\cN_1} \subseteq \lang{\cN_2}\). 

\begin{example}\label{example:Word_Regular_LInc:sim}
Finally, let us illustrate the use of the left simulation-based quasiorder to solve the language inclusion \(\lang{\cN_1} \subseteq \lang{\cN_2}\) of Example~\ref{example:Word_Regular_LInc}.
For the set $F_{\cN_2}$ of final states of \(\cN_2\) 
we have the same set of predecessors computed in Example~\ref{example:Word_Regular_LInc:states} and, among others, the following left simulations between these sets (For clarity, we omit the argument of the function \(\pre\), which is always \(F_{\cN_2}\)):
\begin{align*}
\pre_{c}^{\cN_2}() = \{2\}  & \preceq^{\forall\exists}_R \{3\} = \pre_{a}^{\cN_2}() & \pre_{b}^{\cN_2}() = \{4\} & \npreceq^{\forall\exists}_R \{3\} = \pre_{a}^{\cN_2}() \\
\pre_{ac}^{\cN_2}() = \{1\} & \preceq^{\forall\exists}_R \{4\} = \pre_{b}^{\cN_2}() & \pre_{ac}^{\cN_2}() = \{1\} & \npreceq^{\forall\exists}_R \{2\} = \pre_{c}^{\cN_2}()
\end{align*}

As expected, the simulation-based quasiorder lies in between the Nerode and the state-based quasiorders.
As shown in Examples~\ref{example:Word_Regular_LInc} and~\ref{example:Word_Regular_LInc:states}, we have \(b \qo_{\lang{\cN_2}}^{\ell} a\), \(c \qo_{\lang{\cN_2}}^{\ell} a\), \(b \not\qo^{\ell}_{\cN_2} a\) and \(c \not\qo^{\ell}_{\cN_2} a\) while \(c \preceq_{\cN_2}^{\ell} a\), but \(b \npreceq_{\cN_2}^{\ell} a\).

Let us show the computation of the Kleene iterates performed by Algorithm \AlgRegularW when using the quasiorder \(\mathord{\preceq_{\cN_2}^{\ell}}\).
\begin{align*}
\vect{Y}^{(0)} &= \vect{\varnothing}\\
\vect{Y}^{(1)} &= \vectarg{\epsilon}{F} = \tuple{\varnothing, \{\epsilon\}} \\
\vect{Y}^{(2)} &= \lfloor\vectarg{\epsilon}{F}\rfloor \sqcup \lfloor\Pre_{\cN_1}(\vect{Y}^{(1)})\rfloor = \tuple{\minor{\{a, b, c\}}, \minor{\{\varepsilon\}}} = \tuple{\{c\}, \{\varepsilon\}}\\
\vect{Y}^{(3)} &= \lfloor\vectarg{\epsilon}{F}\rfloor \sqcup \lfloor\Pre_{\cN_1}(\vect{Y}^{(2)})\rfloor = \tuple{\minor{\{ac, a, b, c\}}, \minor{\{\varepsilon\}}} = \tuple{\{c\}, \{\varepsilon\}} %
\end{align*}
The least fixpoint is therefore \(\vect{Y} = \tuple{\{c\}, \{\varepsilon\}}\).
Since $c\in \vect{Y}_0$ and \(c \notin \lang{\cN_2}\), Algorithm~\AlgRegularW concludes that the inclusion \(\lang{\cN_1} \subseteq \lang{\cN_2}\) does not hold. \eox
\end{example}

Let us observe that \(u \preceq^{r}_{\cN_2} v\) implies \(W_{\post^{\cN_2}_u(I),F} \subseteq W_{\post^{\cN_2}_v(I),F}\), which is equivalent to the right Nerode's quasiorder \(u\qr_{\lang{\cN_2}} v\) for $\lang{\cN_2}$. %
Furthermore, for the state-based quasiorder defined in
\eqref{eqn:state-qo}, we have that
\(u \qr_{\cN_2} v \Ra u\preceq^r_{\cN_2} v\) trivially holds.

Summing up, given an NFA \(\cN\) with \(\lang{\cN} = L\), the following containments relate the state-based, 
simulation-based and Nerode's quasiorders:
\[\mathord{\qr_{\cN}} \,\subseteq\, \mathord{\preceq^r_{\cN}} \,\subseteq\, \mathord{\qr_{L}}, \qquad \mathord{\ql_{\cN}} \,\subseteq\, \mathord{\preceq^{\ell}_{\cN}} \,\subseteq\, \mathord{\ql_{L}}\enspace .\]

Recall that these are decidable \(\lang{\cN_2}\)-consistent well-quasiorders so that Algorithm \AlgRegularW can be instantiated for each of them for deciding an inclusion $\lang{\cN_1}\subseteq \lang{\cN_2}$.
Examples~\ref{example:Word_Regular_LInc}, \ref{example:Word_Regular_LInc:states} and~\ref{example:Word_Regular_LInc:sim} show how the algorithm behaves for each of the three quasiorders considered in this section.
Despite their simplicity, the examples evidence the differences in the behavior of the algorithm when considering the different quasiorders.
In particular, we observe that the fixpoint computation for \(\mathord{\qo^r_{\lang{\cN_2}}}\) coincides with the one for \(\mathord{\preceq^r_{\cN_2}}\) which, as expected, converge faster than the one for \(\mathord{\qo^r_{\cN_2}}\).

As shown by \citet{deLuca1994}, \(\mathord{\qo^r_{\lang{\cN_2}}}\) is the coarsest well-quasiorder for which Algorithm~\AlgRegularW works (i.e. Theorem~\ref{theorem:quasiorderAlgorithm} holds), hence its corresponding fixpoint computation exhibits optimal behavior in terms of the number of closed sets considered.
However, Nerode's quasiorder is not practical since it requires checking language inclusion, which is the PSPACE-complete problem we are trying to solve, in order to decide whether two words are related.
Therefore, the coincidence of the fixpoint computations for \(\mathord{\qo^r_{\lang{\cN_2}}}\) and \(\mathord{\preceq^r_{\cN_2}}\) is of special interest since it evidences that Algorithm~\AlgRegularW might exhibit optimal behavior while using a ``simpler'' well-quasiorder such as \(\mathord{\preceq^r_{\cN_2}}\), which is a polynomial under-approximation of \(\mathord{\qo^r_{\lang{\cN_2}}}\).

\subsection{Inclusion in Traces of One-Counter Nets.}%
\label{sub:containment_in_one_counter_languages}
We show that our framework can be instantiated to systematically derive an algorithm for deciding the inclusion \(\lang{\cN} \subseteq L_2\) where \(L_2\) is the trace set of a one-counter net.
This is accomplished by defining a decidable \(L_2\)-consistent quasiorder so that Theorem~\ref{theorem:quasiorderAlgorithm} can be applied.

Intuitively, a \emph{one-counter net} is an NFA endowed with a nonnegative integer counter which 
can be incremented, decremented or
left unchanged by a transition.

\begin{definition*}[One-Counter Net]
A One-Counter Net (OCN)~\cite{hofman_trace_2018} is a tuple $\cO=\tuple{Q,\Sigma,\delta}$ where $Q$ is a finite set of states, $\Sigma$ is an alphabet and $\delta\subseteq Q\times \Sigma\times \{-1,0,1\}\times Q$ is a set of transitions.\eod
\end{definition*}

A \demph{configuration of an OCN} \(\cO = \tuple{Q,\Sigma,\delta}\) is a pair $qn$ consisting of a state \(q\in Q\) and a value \(n\in\bN\) for the counter. 
Given two configurations of an OCN, \(qn, q'n'\in Q\times \bN\), we write \(qn \xrightarrow{a} q'n'\) and call it an \(a\)-step (or simply step) if there exists a transition \( (q,a,d,q')\in\delta \) such that \(n'=n+d\).
Given \(qn\in Q\times\bN\), the \demph{trace set} \(T(qn)\subseteq \Sigma^*\) of an OCN is defined as follows:
\begin{align*}
	T(qn) & \ud \{u \in \Sigma^* \mid Z_u^{qn} \neq \varnothing\} \quad \text{ where } \\
	Z_u^{qn} & \ud \{ q_k n_k \in Q\times \bN \mid qn=q_0n_0 \xrightarrow{a_1} q_1n_1\xrightarrow{a_2}\cdots \xrightarrow{a_k} q_kn_k,\: a_1\cdots a_k=u \}\enspace .
\end{align*}
Observe that \(Z_{\epsilon}^{qn}= \{ qn \}\) and \(Z_u^{qn}\) is a finite set for every word \(u\in\Sigma^*\).

Let us consider the poset $\tuple{\bN_{\bot}\ud \bN\cup\{\bot\},\leq_{\bN_{\bot}}}$ where \(\bot\leq_{\bN_{\bot}} n\) holds for all \(n\in\bN_{\bot}\), while for all $n,n'\in
\bN$, $n\leq_{\bN_{\bot}}n'$ is the standard ordering relation between numbers.  
For a finite set of states \(S \subseteq Q\times\bN\) define the so-called macro state \(M_S \colon Q \ra \bN_{\bot}\) as follows:

\[M_S( q ) \ud \max \{ n\in \bN \mid q n \in S \}\,,\]

\noindent 
where $\max\varnothing\ud\bot$. %
Define the following quasiorder on $\Sigma^*$:

\begin{equation}\label{eq:ocnleq}
	u \leq_{{qn}}^r v \:\udiff\:\forall q\in Q,\, M_{Z_u^{qn}}(q) \leq_{\bN_{\bot}} M_{Z_v^{qn}}(q) \enspace .
\end{equation}

\begin{lemma}\label{lemma:ocnwqo}
	Let \(\cO\) be an OCN. For any configuration $q n$ of $\cO$, \(\mathord{\leq_{{qn}}^r}\) is a right \(T(qn)\)-consistent decidable well-quasiorder.
\end{lemma}
\begin{proof}
It follows from Dickson's Lemma \citep[Section~II.7.1.2]{Sakarovitch} that \(\mathord{\leq_{{qn}}^r}\) is a wqo.
Next, we show that \(\mathord{\leq_{{qn}}^r}\) is \(T(qn)\)-consistent according to Definition~\ref{def:LConsistent}~\ref{eq:LConsistentPrecise}-\ref{eq:LConsistentmonotone}.

\begin{myEnumA}
\item Since \(Z_u^{qn}\) and \(Z_v^{qn}\) are finite sets, we have that the macro state functions 
\(M_{Z_u^{qn}}\) and \(M_{Z_v^{qn}}\) are computable, hence the relation \(\mathord{\leq_{{qn}}^r}\) is decidable.
Let \(u\in T(qn)\) and \(v\notin T(qn)\).
Then \(M_{Z_u^{qn}}(q')\neq \bot\) for some \(q'\in Q\) and \(M_{Z_v^{qn}}(q') = \bot\) since \(Z_v^{qn} = \varnothing\).
It follows that  \(u  \not\leq_{{qn}}^r v \) and, therefore, \(\mathord{\leq_{{qn}}^r} \cap (T(qn) \times (T(qn))^c) = \varnothing\).

\item Next we show that 
\(u \leq_{{qn}}^r v\) implies \(ua \leq_{{qn}}^r va\)
for all \(a\in \Sigma\), since, by Equation~\eqref{def-leftmon}, this is equivalent to the fact that $\leq_{{qn}}^r$ is right monotone. 
We proceed by contradiction.

Assume that \(u \leq_{{qn}}^r v\) and \(\exists q' \in Q\),  \(M_{Z^{qn}_{ua}}(q') \not\leq_{\bN_{\bot}}  M_{Z^{qn}_{va}}(q')\).
Then we have that \(m_1\ud\max\{n \mid pn \in Z^{qn}_{ua}\} \not\leq_{\bN_{\bot}} m_2\ud\max\{n \mid pn \in Z^{qn}_{va}\}\), which implies, since
$m_1\neq \bot$, that
$m_1,m_2\in \bN$ and $m_1 > m_2$. 

On the other hand, for all \( (q,a,d,q') \in \delta\) we have \(q'(m_1-d) \in Z_u^{qn}\) and \(q'(m_2-d) \in Z_v^{qn}\).

Observe that \(\max\{n \mid pn \in Z_u^{qn}\} = m_1-d\) since otherwise we would that have \(\max\{n \mid pn \in Z_u^{qn}\} +d > m_1\) which contradicts the definition of \(m_1\).
Similarly, \(\max\{n \mid pn \in Z_v^{qn}\} = m_2-d\).

Since \(m_1 > m_2\) we have that \(m_1-d > m_2-d\) and, as a consequence, \(\max\{n \mid pn \in Z_u^{qn}\} > \max\{n \mid pn \in Z_v^{qn}\}\), which contradicts \(u \leq_{{qn}}^r v\).
\end{myEnumA}
\end{proof}

Thus, as a consequence of Theorem~\ref{theorem:quasiorderAlgorithm},
Lemma~\ref{lemma:ocnwqo} and the decidability of membership \(u\in  T(qn)\),
the following known decidability result for language inclusion of regular languages into traces of OCNs~\citep[Theorem 3.2]{JANCAR1999476} is systematically derived within our framework.

\begin{corollary}\label{theorem:ocncontainment}
Let \(\cN\) be an NFA and \(\cO\) be an OCN. For any configuration \(qn \) of $\cO$, the language inclusion 
\(\lang{\cN} \subseteq T(qn)\) is decidable.
\end{corollary}

The following result closes a conjecture made by \citet[Section 6]{deLuca1994}.

\begin{lemma}\label{lemma:RightNerodeOcnwqo}
Let \(\cO\) be an OCN.
Then the right Nerode's quasiorder \(\mathord{\qr_{T(qn)}}\) is an undecidable well-quasiorder.
\end{lemma}
\begin{proof}
Recall that \(\mathord{\qr_{T(qn)}}\) is maximum in the set of all right \(T(qn)\)-consistent quasiorders~\citep[Section~2, point~4]{deLuca1994}.
As a consequence, \(u\qr_{{qn}}v\) $\Ra$ \(u\qr_{T(qn)} v\), for all \(u,v\in\Sigma^*\).
By Lemma~\ref{lemma:ocnwqo}, \(\qr_{{qn}}\) is a wqo, so that  \(\mathord{\qr_{T(qn)}}\) is a wqo as well. 
Undecidability of \(\mathord{\qr_{T(qn)}}\) follows from the undecidability of the trace inclusion problem for nondeterministic OCNs \citep[Theorem 20]{Hofman:2013:DWS:2591370.2591405} since given the OCNs \(\cO_1=(Q_1,Σ,\delta_1)\) and \(\cO_2=(Q_2,Σ,\delta_2)\), we can define the union OCN \(\cO_3\ud (Q_1\cup Q_2\cup\{q\}, Σ, \delta_3)\) where \(\delta_3 \) maps \((q,a,0)\) to \(q_1 \in Q_1\), \( (q,b,0) \) to \(q_2 \in Q_2\) and behaves like \(\delta_1\) or \(\delta_2\) elsewhere. Then, it turns out that 
\[a \qr_{T_3(qn)} b \Lra a^{-1}T_3(q_1n) \subseteq b^{-1}T_3(q_2n) \Lra T_1(q_1n)\subseteq T_2(q_2n)\enspace .\]
Therefore, deciding the right Nerode's quasiorder \(\qr_{T_3(qn)}\) is as hard as deciding \(T_1(q_1n)\subseteq T_2(q_2n)\).
\end{proof}

It is worth to remark that, by Lemma~\ref{lemma:leftrightnerodegoodqo}~\ref{lemma:leftrightnerodegoodqo:Consistent}, the left and right Nerode's quasiorders \(\mathord{\ql_{T(qn)}}\) and \(\mathord{\qr_{T(qn)}}\) are \(T(qn)\)-consistent. 
However, the left Nerode's quasiorder does not need to be a wqo, otherwise \(T(qn)\) would be regular.

We conclude this section by conjecturing that our framework could be instantiated for extending 
Corollary~\ref{theorem:ocncontainment} to traces of Petri Nets, a result 
which is already known to be true~\cite{JANCAR1999476}.

\section{A Novel Perspective on the Antichain Algorithm}%
\label{sec:novel_perspective_AC}

Let \(\cN_1 = \tuple{Q_1,\delta_1,I_1,F_1,\Sigma}\) and \(\cN_2 = \tuple{Q_2,\delta_2,I_2,F_2,\Sigma}\) be two NFAs 
and consider
the state-based left \(\lang{\cN_2}\)-consistent wqo
 \(\mathord{\qo_{\cN_2}^{\ell}}\) defined by Equivalence~\eqref{eqn:state-qo}. 
Theorem~\ref{theorem:quasiorderAlgorithm} shows that Algorithm \AlgRegularW decides the language inclusion \(\lang{\cN_1} \subseteq \lang{\cN_2}\) by manipulating finite sets of words. 

Since \(u \qo_{\cN_2}^{\ell} v \Lra \pre^{\cN_2}_u(F_2) \subseteq \pre^{\cN_2}_v(F_2)\), we could equivalently consider 
the 
set of states \(\pre^{\cN_2}_u(F_2)\in \wp(Q_2)\) rather than
a word $u\in \Sigma^*$. 
This observation suggests the design of an algorithm analogous to \AlgRegularW but computing on the poset
\(\tuple{\AC_{\tuple{\wp(Q_2),\subseteq}},\sqsubseteq}\) of antichains 
of sets of states of the complete lattice $\tuple{\wp(Q_2),\subseteq}$. 

To that end, the poset \(\tuple{\AC_{\tuple{\wp(Q_2),\subseteq}},\sqsubseteq}\) is viewed as an abstraction of the poset \(\tuple{\wp(Σ^*), \subseteq}\) by using the abstraction and concretization functions  \(\alpha\colon \wp(\Sigma^*) \ra \AC_{\tuple{\wp(Q_2),\subseteq}}\) and \(\gamma\colon \AC_{\tuple{\wp(Q_2),\subseteq}}\ra\wp(\Sigma^*)\) and using the abstract function \({\Pre}_{\cN_1}^{\cN_2}:(\AC_{\tuple{\wp(Q_2),\subseteq}})^{|Q_1|}\ra (\AC_{\tuple{\wp(Q_2),\subseteq}})^{|Q_1|}\) defined as follows:
\begin{align}
& \alpha(X) \ud \lfloor \{ \pre_u^{\cN_2}(F_2) \in \wp(Q_2) \mid u\in X\} \rfloor \,,\nonumber\\
& \gamma(Y) \ud \{v \in \Sigma^* \mid \exists u\in \Sigma^*,\, \pre_{u}^{\cN_2}(F_2) \in Y \,\land\, \pre_{u}^{\cN_2}(F_2) \subseteq \pre_{v}^{\cN_2}(F_2)\} \,,\label{def-antichain-state-abs}\\
&\mindex{\Pre_{\cN_1}^{\cN_2}}(\tuple{X_q}_{q\in Q_1}) \ud \langle \lfloor \big\{ \pre_a^{\cN_2}(S)  \in \wp(Q_2) \mid  \exists a\in \Sigma, q'\in Q_1, q'\in\delta_1(q,a) \wedge S \in X_{q'} \big\} \rfloor \rangle_{q\in Q_1} \nonumber .
\end{align}

Observe that the functions $\alpha$ and ${\Pre}_{\cN_1}^{\cN_2}$ are well-defined because minors are antichains. 

\begin{lemma}\label{lemma:rhoisgammaalpha}
The following properties hold:
\begin{myEnumA}
\item \(\tuple{\wp(\Sigma^*),\subseteq}\galois{\alpha}{\gamma}\tuple{\AC_{\tuple{\wp(Q_2),\subseteq}},\sqsubseteq}\) is a GC.
\label{lemma:rhoisgammaalpha:GC}
\item \(\gamma \comp \alpha = \rho_{\qo^{\ell}_{\cN_2}}\).\label{lemma:rhoisgammaalpha:rho}
\item For all \(\vect{X}\in \alpha(\wp(\Sigma^*))^{|Q_1|}\), \(\Pre_{\cN_1}^{\cN_2}(\vect{X}) = {\alpha\comp \Pre_{\cN_1} \comp \gamma(\vect{X})}\). \label{lemma:rhoisgammaalpha:pre}
\end{myEnumA}
\end{lemma}

\begin{proof}
\begin{myEnumA}
\item 
Let us first observe that $\alpha$ and $\gamma$ are well-defined. 
First, $\alpha(X)$ is an antichain of  $\tuple{\wp(Q_2),\subseteq}$ since it is a minor for the well-quasiorder \(\subseteq\) and, therefore, it is finite.
On the other hand, $\gamma(Y)$ is clearly an element of $\tuple{\wp(Σ^*), \subseteq}$ by definition. 

Then, for all $X\in \wp(Σ^*)$ and 
$Y\in \AC_{\tuple{\wp(Q_2),\subseteq}}$, 
it turns out that:
\begin{adjustwidth}{-0.5cm}{}
\begin{myAlign}{0pt}{}
\alpha(X) \sqsubseteq Y & \Lra \quad\text{[By definition of \(\sqsubseteq\)]} \\
\forall z \in \alpha(X), \exists y \in Y, \; y \subseteq z &\Lra \quad\text{[By definition of  \(\alpha\) and \(\minor{\cdot}\)]} \\
\forall v \in X, \exists y \in Y, \; y \subseteq \pre^{\cN_2}_v(F_2) &\Lra \quad\text{[By definition of \(\gamma\)]} \\
\forall v \in X, \; v \in γ(Y) & \Lra \quad\text{[By definition of  \(\subseteq\)]} \\
X \subseteq \gamma(Y) &\enspace . 
\end{myAlign}
\end{adjustwidth}

\item For all \(X \in \wp(Σ^*)\) we have that
\begin{adjustwidth}{-0.95cm}{}
\begin{myAlign}{0pt}{}
	\gamma(\alpha(X)) &=\quad\text{[By definition of $\alpha,\gamma$]}\\
\{v \,{\in}\, \Sigma^* \mid \exists u\,{\in}\, \Sigma^*, \pre_{u}^{\cN_2}(F_2) \,{\in}\, \lfloor \{ \pre_w^{\cN_2}(F_2) \mid w\in X\} \rfloor  
	\span \land \pre_{u}^{\cN_2}(F_2) \subseteq \pre_{v}^{\cN_2}(F_2)\} \\
	&=\quad\text{[By definition of minor]} \\
  \{v \in \Sigma^* \mid \exists u\in X,\, \pre_{u}^{\cN_2}(F_2) \subseteq \pre_{v}^{\cN_2}(F_2)\} &=\quad\text{[By definition of \(\mathord{\qo^{\ell}_{\cN_2}}\)]}\\
	\{v \in \Sigma^* \mid \exists u \in X ,\,  u \qo^{\ell}_{\cN_2} v\}&=\quad\text{[By definition of\ \(\rho_{\qo^{\ell}_{\cN_2}}\)]}\\
	\rho_{\qo^{\ell}_{\cN_2}}(X) &\enspace .
\end{myAlign}
\end{adjustwidth}

\item For all \(\vect{X}\in \alpha(\wp(\Sigma^*))^{|Q_1|}\) we have that
\begin{adjustwidth}{-0.8cm}{}
\begin{myAlign}{0pt}{0pt}
\alpha(\Pre_{\cN_1}(\gamma(\vect{X}))) &= \quad \text{[By def. of \(\Pre_{\cN_1}\)]} \\
\tuple{\alpha({\textstyle \bigcup_{a \in \Sigma, q\ggoes{a}_{\cN_1} q'}} a\gamma(\vect{X}_{q'}))}_{q \in Q_1} &=\quad \text{[By definition of \(\alpha\)]} \\
\langle\lfloor \{ \pre^{\cN_2}_u(F_2)  \mid u \in {\textstyle \bigcup_{a \in \Sigma, q\ggoes{a}_{\cN_1} q'}} a\gamma(\vect{X}_{q'})\rfloor\rangle_{q\in Q_1} &=\\
\span\specialcell{\hfill\text{[By \(\pre^{\cN_2}_{av} = \pre^{\cN_2}_a\comp \pre^{\cN_2}_v\)]}}\\
\langle\lfloor \{ \pre^{\cN_2}_a(\{ \pre^{\cN_2}_u(F_2) \mid u \in {\textstyle\bigcup_{q\ggoes{a}_{\cN_1} q'}}\gamma(\vect{X}_{q'})\})  \mid a \in \Sigma\}\rfloor\rangle_{q\in Q_1} &= \quad \text{[By rewriting]}\\
\langle\lfloor \{ \pre^{\cN_2}_a(S)  \mid a \in \Sigma, q\ggoes{a}_{\cN_1} q', S\in \{ \pre^{\cN_2}_u(F_2) \mid u \in \gamma(\vect{X}_{q'})\}\}\rfloor\rangle_{q\in Q_1} &=\\
\span\specialcell{\hfill\text{[By \( \minor{\pre^{\cN_2}_a(X)} = \minor{\pre^{\cN_2}_a(\minor{X})}\)]}}\\
\langle\lfloor \{ \pre^{\cN_2}_a(S)  \mid a \in \Sigma, q\ggoes{a}_{\cN_1} q', S\in \minor{\{ \pre^{\cN_2}_u(F_2) \mid u \in \gamma(\vect{X}_{q'})\}}\}\rfloor\rangle_{q\in Q_1} & =\quad \text{[By definition of \(\alpha\)]} \\
\langle\lfloor \{ \pre^{\cN_2}_a(S)  \mid a \in \Sigma, q\ggoes{a}_{\cN_1} q', S\in \alpha(\gamma(\vect{X}_{q'}))\rfloor\rangle_{q\in Q_1} &=\\
\span\specialcell{\hfill\text{[Since \(\vect{X} \in \alpha\), \(\alpha(\gamma(\vect{X}_{q'})) = \vect{X}_{q'}\)]}}\\
\langle \lfloor \{ \pre^{\cN_2}_a(S)  \mid a\in\Sigma, q\ggoes{a}_{\cN_1}q', S \in \vect{X}_{q'} \} \rfloor   \rangle_{q\in Q_1} &=\quad \text{[By def. of ${\Pre}_{\cN_1}^{\cN_2}$]} \\
{\Pre}_{\cN_1}^{\cN_2}(\vect{X}) & \enspace .  
\end{myAlign}%
\end{adjustwidth}%
\end{myEnumA}%
\end{proof}

It follows from Lemma~\ref{lemma:rhoisgammaalpha} that the GC \(\tuple{\wp(\Sigma^*),\subseteq}\galois{\alpha}{\gamma}\tuple{\AC_{\tuple{\wp(Q_2),\subseteq}},\sqsubseteq}\) and the abstract 
function \(\Pre_{\cN_1}^{\cN_2}\) satisfy the hypotheses~\ref{theorem:EffectiveAlgorithm:prop:rho}-\ref{theorem:EffectiveAlgorithm:prop:abseps} of Theorem~\ref{theorem:EffectiveAlgorithm}.
Thus, in order to obtain an algorithm for deciding \(\lang{\cN_1} \subseteq \lang{\cN_2}\) it remains to show that requirement~\ref{theorem:EffectiveAlgorithm:prop:absincl} of Theorem~\ref{theorem:EffectiveAlgorithm} holds, i.e. there is an algorithm to decide whether \(\vect{Y} \sqsubseteq \alpha(\vectarg{L_2}{I_2})\) for every \(\vect{Y} \in \alpha(\wp(\Sigma^*))^{|Q_1|}\).
In order to do that, we first provide some intuitions on how the resulting algorithm works.

First, observe that the Kleene iterates of the function \(\lambda \vect{X}\ldotp\alpha(\vectarg{\epsilon}{F_1}) \sqcup \Pre_{\cN_1}^{\cN_2}(\vect{X})\) of Theorem~\ref{theorem:EffectiveAlgorithm} are vectors of antichains in \(\tuple{\AC_{\tuple{\wp(Q_2),\subseteq}},\sqsubseteq}\), where 
each component is indexed by some \(q\in Q_1\) and represents (through its minor set) a set of sets of states that are predecessors of \(F_2\) in \(\cN_2\) by a word $u$ generated by \(\cN_1\) from that state \(q\), i.e. \(\pre_u^{\cN_2}(F_2)\) with \(u \in W^{\cN_1}_{q,F_1}\).
Since \(\epsilon \in W_{q,F_1}^{\cN_1}\) for all \(q \in F_1\) and \(\pre_\epsilon^{\cN_2}(F_2) = F_2\) the 
iterations of the procedure $\Kleene$ begin with the initial vector \(\alpha(\vectarg{\epsilon}{F_1}) = \tuple{\nullable{q\in F_1}{F_2}{\varnothing}}_{q\in Q_1}\).

On the other hand, note that by taking the minor of each vector component, we are considering smaller sets which still preserve the relation \(\sqsubseteq\) since 
\begin{equation*}
A \sqsubseteq B \Lra \minor{A} \sqsubseteq B \Lra A \sqsubseteq \minor{B} \Lra \minor{A} \sqsubseteq \minor{B}\enspace .
\end{equation*}

Let \(\tuple{Y_q}_{q\in Q_1}\) be the fixpoint computed by the \(\Kleene\) procedure. 
It turns out that, for each component $q\in Q_1$, \(Y_q = \minor{\{\pre_u^{\cN_2}(F_2)\mid u \in W_{q,F_1}^{\cN_1}\}}\) holds.
Whenever the inclusion \(\lang{\cN_1} \subseteq \lang{\cN_2}\) holds, all the sets of states in \(Y_q\) for some initial state \(q \in I_1\) are predecessors of \(F_2\) in \(\cN_2\) by words in \(\lang{\cN_2}\), so that they all contain at least one initial state in \(I_2\).
As a result, we obtain Algorithm \AlgRegularA, that is, 
a ``state-based'' inclusion algorithm for deciding \(\lang{\cN_1} \subseteq \lang{\cN_2}\). 

\begin{figure}[!ht]
\RemoveAlgoNumber
\begin{algorithm}[H]
\SetAlgorithmName{\AlgRegularA}{}

\caption{State-based algorithm for {\(\lang{\cN_1} \subseteq \lang{\cN_2}\)}}\label{alg:RegIncA}

\KwData{NFAs \(\cN_1=\tuple{Q_1,\delta_1,I_1,F_1,\Sigma}\) and \(\cN_2=\tuple{Q_2,\delta_2,I_2,F_2,\Sigma}\).}
\medskip
\(\tuple{Y_q}_{q\in Q_1} := \Kleene (\lambda \vect{X}\ldotp\alpha(\vectarg{\epsilon}{F_1}) \sqcup \Pre_{\cN_1}^{\cN_2}(\vect{X}),\vect{\varnothing})\)\;

\ForAll{\(q \in I_1\)}{
	\ForAll{\(S \in Y_q\)} {
		\lIf{\(S \cap I_2 = \varnothing\)}{\Return \textit{false}}
	}
}
\Return \textit{true}\;
\end{algorithm}
\end{figure}

\begin{theorem}\label{theorem:statesQuasiorderAlgorithm}
	Let \(\cN_1,\cN_2\) be NFAs.
	The algorithm \AlgRegularA decides the inclusion \(\lang{\cN_1} \subseteq \lang{\cN_2}\).
\end{theorem}
\begin{proof}
We show that all the conditions~\ref{theorem:EffectiveAlgorithm:prop:rho}-\ref{theorem:EffectiveAlgorithm:prop:absincl} of Theorem~\ref{theorem:EffectiveAlgorithm} are satisfied for the abstract domain \(\tuple{D,\qo_D}=\tuple{\AC_{\tuple{\wp(Q_2),\subseteq}},\sqsubseteq}\) as defined by the Galois Connection of Lemma~\ref{lemma:rhoisgammaalpha}~\ref{lemma:rhoisgammaalpha:GC}. 

\begin{myEnumA}
\item Since, by Lemma~\ref{lemma:rhoisgammaalpha}~\ref{lemma:rhoisgammaalpha:rho}, \(\rho_{\qo^{\ell}_{\cN_2}}(X) = \gamma(\alpha(X))\) it follows from Lemmas~\ref{lemma:properties} and~\ref{lemma:LAconsistent} that \(\gamma(\alpha(L_2)) = L_2\).
Moreover, for all \(a\in\Sigma\), \(X\in\wp(\Sigma^*)\) we have that:
\begin{align*}
\gamma\alpha(a X) & = \quad \text{[In GCs \(\gamma = \gamma \alpha \gamma\)]} \\
\gamma\alpha\gamma\alpha(a X) & = \quad \text{[By Lemma~\ref{lemma:properties}~\ref{lemma:properties:bw} with \(\rho_{\leqslant^{\ell}_{\cN_2}} = \gamma\alpha\)]} \\
\gamma\alpha\gamma \alpha(a\gamma \alpha(X)) &= \quad \text{[In GCs 
\(\gamma = \gamma \alpha \gamma\)]} \\
\gamma\alpha(a\gamma\alpha(X)) & \enspace.
\end{align*}

\item \( (\AC_{\tuple{\wp(Q_2),\subseteq}},\sqsubseteq) \) is effective because 
$Q_2$ is finite.
\item  By Lemma~\ref{lemma:rhoisgammaalpha}~\ref{lemma:rhoisgammaalpha:pre} we have that
\(\alpha(\Pre_{\cN_1}(\gamma(\vect{X}))) = {\Pre}_{\cN_1}^{\cN_2}(\vect{X})\) for all \(\vect{X}\in \alpha(\wp(\Sigma^*))^{|Q_1|}\).
\item \(\alpha(\{\epsilon\}) = \{F_2\}\) and \(\alpha({\varnothing})=\varnothing\), hence \(\minor{\alpha(\vectarg{\epsilon}{F_1})}\) is trivial to compute. \label{prop:alphaepsilon}

\item Since \(\alpha(\vectarg{L_2}{I_1})=\tuple{\alpha(\nullable{q\in I_1}{L_2}{\Sigma^*})}_{q\in Q_1}\), for all $\vect{Y}\in\alpha(\wp(\Sigma^*))^{|Q_1|}$ the relation \(\vect{Y} \sqsubseteq \alpha(\vectarg{L_2}{I_1})\) trivially holds for all components \(q \notin I_1\).
For the components $q\in I_1$, it suffices to show that
\(Y_q \sqsubseteq \alpha(L_2) \Lra \forall S \in Y_q, \; S \cap I_2 \neq \varnothing\), which is the check performed by lines 2-5 of algorithm \AlgRegularA.
\begin{align*}
Y_q \sqsubseteq \alpha(L_2) & \Lra \quad \text{[Because \(Y_q = \alpha(U)\) for some \(U \in \wp(\Sigma^*)\)]} \\
\alpha(U) \sqsubseteq \alpha(L_2) & \Lra \quad \text{[By GC]} \\
U \subseteq \gamma(\alpha(L_2)) & \Lra \quad \text{[By L.~\ref{lemma:properties},~\ref{lemma:LAconsistent} and~\ref{lemma:rhoisgammaalpha}, $\gamma(\alpha(L_2))=L_2$]} \\
U \subseteq L_2 & \Lra \quad \text{[By definition of \(\pre_u^{\cN_2}\)]} \\
\forall u \in U, \pre_u^{\cN_2}(F_2) \cap I_2 \neq \varnothing & \Lra \quad \text{[Since \(Y_q =\alpha(U) = \lfloor \{ \pre_u^{\cN_2}(F_2) \mid u\in U\} \rfloor \)]} \\
\forall S \in Y_q, S \cap I_2 \neq \varnothing &\enspace .
\end{align*}

\end{myEnumA}
Thus, by Theorem~\ref{theorem:EffectiveAlgorithm}, Algorithm \AlgRegularA decides \(\lang{\cN_1} \subseteq \lang{\cN_2}\). %
\end{proof}

\subsection{Relationship to the Antichains Algorithm}%
\label{sub:relationship_to_the_antichain_algorithm}
\citet{DBLP:conf/cav/WulfDHR06} introduced two so-called antichains algorithms, denoted 
\emph{forward} and \emph{backward}, for deciding the universality of the language accepted by an NFA, i.e. whether the language is $\Sigma^*$ or not.
Then, they extended the backward algorithm to decide the inclusion between the languages accepted by two NFAs.

In what follows we show that Algorithm \AlgRegularA is equivalent to the corresponding extension of the forward algorithm and, therefore, dual to the backward antichains algorithm for language inclusion by \citet{DBLP:conf/cav/WulfDHR06}[Theorem 6].

To do that, we first define the poset of antichains in which the forward antichains algorithm computes its fixpoint.
Then, we give a formal definition of the forward antichains algorithm for deciding language inclusion and show that this algorithm coincides with \AlgRegularA when applied to the reverse automata.
Since language inclusion between the languages generated by two NFAs holds if{}f inclusion holds between the languages generated by their reverse NFAs, we conclude that the algorithm \AlgRegularA is equivalent to the forward antichains algorithm.

Finally, we show how the different variants of the antichains algorithm, including the original backward antichains algorithm~\cite{DBLP:conf/cav/WulfDHR06}[Theorem 6], can be derived within our framework by considering the adequate quasiorders.

\paragraph*{Forward Antichains Algorithm}

Let \(\cN_1=\tuple{Q_1,\Sigma,\delta_1,I_1,F_1}\) and \(\cN_2=\tuple{Q_2,\Sigma,\delta_2,I_2,F_2}\) be two NFAs and consider the language inclusion problem \(\lang{\cN_1} \subseteq \lang{\cN_2}\).
Let us consider the following poset of antichains 
\( \tuple{\AC_{\tuple{\wp(Q_2),\subseteq}},\wsqsubseteq} \) where
\[X \wsqsubseteq Y \udiff \forall y \in Y, \exists x \in X, \; x \subseteq y\enspace \]
and notice that \(\wsqsubseteq\) coincides with the reverse 
relation \(\sqsubseteq^{-1}\). 
As observed by \citet[Lemma 1]{DBLP:conf/cav/WulfDHR06}, it turns out that \( \tuple{\AC_{\tuple{\wp(Q_2),\subseteq}},\wsqsubseteq, \wsqcup, \wsqcap, \{\varnothing\}, \varnothing} \) is a finite lattice, where \(\wsqcup\) and \(\wsqcap\) denote, resp., lub and glb, and $\{\varnothing\}$ and $\varnothing$ are, resp., the least and greatest elements. 
This lattice \( \tuple{\AC_{\tuple{\wp(Q_2),\subseteq}},\wsqsubseteq} \) is the domain in which the forward antichains algorithm computes on for deciding universality \citep[Theorem~3]{DBLP:conf/cav/WulfDHR06}.
The following result extends this forward algorithm in order to decide language inclusion.
\begin{theorem}[\textbf{{\citep[Theorems~3 and 6]{DBLP:conf/cav/WulfDHR06}}}] \label{theorem:antichainpaper}
Let
\begin{align*}
\vect{\fp} \ud \textstyle{\wbigsqcup}\{\vect{X} \in (\AC_{\tuple{\wp(Q_2),\subseteq}})^{|Q_1|} \mid \vect{X} = \Post_{\cN_1}^{\cN_2}(\vect{X})\;\wsqcap\; \tuple{\nullable{q \in I_1}{\{I_2\}}{\varnothing}}_{q\in Q_1}\}
\end{align*}
where 
\begin{align*}
\mindex{\Post_{\cN_1}^{\cN_2}}(\tuple{X_q}_{q\in Q_1}) \ud  \langle \lfloor\{\post_a^{\cN_2}(x) {\in} \wp(Q_2) \mid \exists a {\in} \Sigma, q'{\in} Q_1, &  q{\in}\delta_1(q',a) \wedge x \in X_{q'}  \}\rfloor \rangle_{q \in Q_1}\enspace .
\end{align*}
Then, \(\lang{\cN_1} \nsubseteq \lang{\cN_2}\) if and only if there exists \(q \in F_1\) such that \(\vect{\fp}_q \,\wsqsubseteq\, \{F_2^c\} \).
\end{theorem}

\begin{proof}
Let us first introduce some notation necessary to describe the forward antichains algorithm by \citet{DBLP:conf/cav/WulfDHR06} for deciding \(\lang{\cN_1} \subseteq \lang{\cN_2}\).
In the following, we consider the poset
\(\tuple{Q_1\times \wp(Q_2),\subseteq_\times}\) where 
\[(q_1,x_1) \subseteq_\times
(q_2,x_2) \udiff q_1=q_2 \wedge x_1 \subseteq x_2 \enspace . \]
Then, let
\(\tuple{\AC_{\tuple{Q_1\times \wp(Q_2),\subseteq_\times}},\wsqsubseteq_\times,\wsqcup_\times, \wsqcap_\times}\) be the lattice of antichains over the poset \(\tuple{Q_1\times \wp(Q_2),\subseteq_\times}\) where:
\begin{align*}
X \wsqsubseteq_\times Y & \udiff \forall (q,y) \in Y, \exists (q,x) \in X , x \subseteq y \\
\textstyle{\min_{\times}}(X) &\ud \{(q,x) \in X \mid \forall (q',x') \in X, q=q' \Ra x' \nsubseteq x\} \\
X \wsqcup_\times Y & \ud \textstyle{\min_{\times}}(\{(q,x \cup y) \mid (q,x) \in X,\, (q,y) \in Y\}) \\
X \wsqcap_\times Y & \ud \textstyle{\min_{\times}}(\{(q,z) \mid (q,z) \in X \cup Y \}) \enspace .
\end{align*}
Also, let $\Post: \AC_{\tuple{Q_1\times \wp(Q_2),\subseteq_\times}} \ra \AC_{\tuple{Q_1\times \wp(Q_2),\subseteq_\times}}$ be defined as follows:
\begin{align*}  
\Post(X) \ud \textstyle{\min_{\times}}(\{ (q,\post_a^{\cN_2}(x)) \in Q_1\times \wp(Q_2) \mid \exists a \in \Sigma, q\in Q_1, \hfill (q',x) \in X , q' \ggoes{a}_{\cN_1} q\}) \enspace .
\end{align*}

Then, it turns out that the dual of the backward antichains algorithm of \citet[Theorem~6]{DBLP:conf/cav/WulfDHR06} states that \(\lang{\cN_1} \nsubseteq \lang{\cN_2}\) if{}f there exists \(q \in F_1\) such that \(\fp \mathrel{\wsqsubseteq_\times} \{(q,F_2^c)\}\) where
\[\fp = {\textstyle\wbigsqcup_\times}\{X \in \AC_{\tuple{Q_1\times \wp(Q_2),\subseteq_\times}} \mid X = \Post(X)\;\wsqcap_\times\; (I_1 \times \{I_2\})\}\enspace .\]

\noindent
We observe that for every \(X\in\AC_{\tuple{Q_1\times \wp(Q_2),\subseteq_\times}}\), a pair \((q,x) \in Q_1\times \wp(Q_2)\) such that $(q,x)\in X$ is used by 
\citet[Theorem~6]{DBLP:conf/cav/WulfDHR06} simply as
a way to associate states $q$ of \(\cN_1\) with sets $x$ of states  of \(\cN_2\).
In fact, every antichain 
\(X\in\AC_{\tuple{Q_1\times \wp(Q_2),\subseteq_\times}}\)
can be equivalently formalized 
by a vector 
\[\tuple{\{x \in \wp(Q_2) \mid (q,x)\in X\}}_{q \in Q_1}
\in (\AC_{\tuple{\wp(Q_2),\subseteq}})^{|Q_1|}\]
indexed by states \(q\in Q_1\) and whose components are antichains in $\AC_{\tuple{\wp(Q_2),\subseteq}}$. 

Correspondingly, we consider 
the lattice \(\tuple{\AC_{\tuple{\wp(Q_2),\subseteq}},\wsqsubseteq}\), where for 
every pair of elements $X,Y\in \AC_{\tuple{\wp(Q_2),\subseteq}}$ we have that
\begin{align*}
X \wsqsubseteq Y &\udiff \forall y \in Y, \exists x \in X , x \subseteq y& 
\textstyle{\min}(X) &\ud \{x \in X \mid \forall x' \in X, x' \nsubset x\} \\
X \wsqcup Y & \ud \textstyle{\min}(\{x \cup y \in \wp(Q_2) \mid x \in X, y \in Y\}) &
X \wsqcap Y & \ud \textstyle{\min}(\{z \in \wp(Q_2) \mid z \in X \cup Y\}) \enspace .
\end{align*}

Then, \(\Post\) can be replaced by \(\Post_{\cN_1}^{\cN_2}: (\AC_{\tuple{\wp(Q_2),\subseteq}})^{|Q_1|} \ra (\AC_{\tuple{\wp(Q_2),\subseteq}})^{|Q_1|}
\), its equivalent formulation on vectors defined as follows:
\begin{align*}
\Post_{\cN_1}^{\cN_2}(\tuple{X_q}_{q\in Q_1}) \ud \langle\textstyle{\min}(\{\post_a^{\cN_2}(x) \in \wp(Q_2) \mid \exists a \in \Sigma, q'\in Q_1, \hfill x \in X_{q'} , q' \ggoes{a}_{\cN_1} q \})\rangle_{q \in Q_1}\enspace .
\end{align*}

In turn, \(\fp\in \AC_{\tuple{Q_1\times \wp(Q_2),\subseteq_\times}}\) is replaced by the 
following vector:
\[\vect{\fp} \ud \textstyle{\wbigsqcup}\{\vect{X}\in (\AC_{\tuple{\wp(Q_2),\subseteq}})^{|Q_1|} \mid \vect{X}=
 \Post_{\cN_1}^{\cN_2}(\vect{X})\;\wsqcap\; \tuple{\nullable{q \in I_1}{\{I_2\}}{\varnothing}}_{q\in Q_1}\} \enspace .\]
Finally, the check \(\exists q \in F_1 , \fp \mathrel{\wsqsubseteq_\times} \{(q,F_2^c)\}\) becomes \(\exists q \in F_1 ,  \vect{\fp}_q \mathrel{\wsqsubseteq} \{F_2^c\} \).
\end{proof}

Let \(\cN^R\) denote the reverse automaton of \(\cN\), where arrows are flipped and the initial/final states become final/initial.
Note that language inclusion can be decided by considering the reverse automata since 
\[\lang{\cN_1} \subseteq \lang{\cN_2} \Lra \lang{\cN_1^R} \subseteq \lang{\cN_2^R}\enspace . \]
Furthermore, it is straightforward to check that \(\Post_{\cN_1}^{\cN_2} = \Pre_{\cN_1^R}^{\cN_2^R}\).
We therefore obtain the following result as a consequence of Theorem~\ref{theorem:antichainpaper}. 
\begin{corollary}\label{theorem:antichainpaperReverse}
Let
\begin{align*}
\vect{\fp} \ud \textstyle{\wbigsqcup}\{\vect{X} \in (\AC_{\tuple{\wp(Q_2),\subseteq}})^{|Q_1|} \mid \vect{X} = \Pre_{\cN_1}^{\cN_2}(\vect{X})\;\wsqcap\; \tuple{\nullable{q \in F_1}{\{F_2\}}{\varnothing}}_{q\in Q_1}\} \enspace .
\end{align*}
Then, \(\lang{\cN_1} \nsubseteq \lang{\cN_2}\) if{}f there exists \(q \in I_1\) such that \(\vect{\fp}_q \,\wsqsubseteq\, \{I_2^c\} \).
\end{corollary}

\paragraph*{From the Forward Antichains Algorithm to \textsc{FAIncS}}
Since \(\wsqsubseteq = \mathord{\sqsubseteq^{-1}}\), we have that \(\wsqcap = \sqcup\), \(\wsqcup = \sqcap\) and the greatest element $\varnothing$ for $\wsqsubseteq$ is the least element for $\mathord{\sqsubseteq}$.
Moreover, by~\eqref{def-antichain-state-abs}, $\alpha(\vectarg{\epsilon}{F_1}) = \tuple{\nullable{q \in F_1}{\{F_2\}}{\varnothing}}_{q\in Q_1}$.
Therefore, we can rewrite the vector 
$\vect{\fp}$ of 
Corollary~\ref{theorem:antichainpaperReverse} as 
\[
\vect{\fp} = {\textstyle\bigsqcap}\{\vect{X} \in (\AC_{\tuple{\wp(Q_2),\subseteq}})^{|Q_1|} \mid \vect{X} = \Pre_{\cN_1}^{\cN_2}(\vect{X})\;\sqcup\; \alpha(\vectarg{\epsilon}{F_1})\}
\] 
which is precisely the lfp in $\tuple{(\AC_{\tuple{\wp(Q_2),\subseteq}})^{|Q_1|}, \sqsubseteq}$ of $\Pre_{\cN_1}^{\cN_2}$ above
$\alpha(\vectarg{\epsilon}{F_1})$. 

Hence, it turns out that the Kleene iterates of the least fixpoint computation 
that converge to \(\vect{\fp}\) exactly coincide with the iterates computed by the $\Kleene$ procedure of the state-based algorithm 
\AlgRegularA.
In particular, if  \(\vect{Y}\) is the output vector of the call to $\Kleene$ at line~1 of  
\AlgRegularA then  \(\vect{Y} = \vect{\fp}\).
Furthermore, 
\[\exists q\in I_1, \vect{\fp}_q \:\wsqsubseteq\: \{I_2^c\} \Lra \exists q\in I_1, \exists S \in \vect{\fp}_q, \; S \cap I_2 = \varnothing\enspace .\]
Summing up, the \(\sqsubseteq\)-lfp algorithm \AlgRegularA coincides with the \(\wsqsubseteq\)-gfp antichains algorithm given by Corollary~\ref{theorem:antichainpaperReverse}.

\paragraph*{Backward Antichains Algorithm}

We can also derive an antichains algorithm for deciding language inclusion fully equivalent to the backward one of \citet[Theorem 6]{DBLP:conf/cav/WulfDHR06} by considering the  
lattice \(\tuple{\AC_{\tuple{\wp(Q_2),\supseteq}},\sqsubseteq}\) for the dual lattice $\tuple{\wp(Q_2),\supseteq}$ 
and by replacing the functions \(\alpha\), \(\gamma\) and \(\Pre_{\cN_1}^{\cN_2}\) of Lemma~\ref{lemma:rhoisgammaalpha}, respectively, with:
\begin{align*}
& \alpha^c(X) \ud \lfloor \{ \cpre_u^{\cN_2}(F_2^c) \in \wp(Q_2)\mid u\in X\} \rfloor \, ,\hspace{-13pt} \\
& \gamma^c(Y) \ud \{u \in \Sigma^* \mid \exists y \in Y , y \supseteq \cpre_{u}^{\cN_2}(F^c_2) \}, \\
& {\CPre}_{\cN_1}^{\cN_2}(\tuple{X_q}_{q\in Q_1}) \ud  \langle \lfloor \{ \cpre_a^{\cN_2}(S) \in \wp(Q_2)  \mid  \exists a\in \Sigma, q'\in Q_1, \hfill q'\in\delta_1(q,a) \wedge  S \in X_{q'} \} \rfloor \rangle_{q\in Q_1} \enspace . 
\end{align*}
where \(\cpre_u^{\cN_2}(S) \ud (\pre_u^{\cN_2}(S^c))^c\) for $u\in \Sigma^*$.

When instantiating Theorem~\ref{theorem:EffectiveAlgorithm} using these functions, we obtain an lfp algorithm computing on the lattice \(\tuple{\AC_{\tuple{\wp(Q_2),\supseteq}},\sqsubseteq}\).
Indeed, it turns out that 
\[\lang{\cN_1} \subseteq \lang{\cN_2} \Lra \Kleene\big(\lambda \vect{X} \ldotp \CPre_{\cN_1}^{\cN_2}(\vect{X}) \sqcup \alpha^c(\vectarg{\epsilon}{F_1}),\vect{\varnothing}\big) \sqsubseteq \alpha^c(\vectarg{L_2}{I_1})\enspace .\]

It is easily seen that this algorithm coincides with the backward antichains algorithm defined by \citet[Theorem 6]{DBLP:conf/cav/WulfDHR06} since both compute on the same lattice, \(\minor{X}\) corresponds to the maximal (w.r.t.\ set inclusion) elements of \(X\), \(\alpha^c(\{\epsilon\}) = \{F_2^c\}\) and for all \(X \in \alpha^c(\wp(\Sigma^*))\), we have that \(X \sqsubseteq \alpha^c(L_2) \Lra \forall S \in X, \; I_2 \nsubseteq S\).

\paragraph*{Variants of the Antichains Algorithm}

We have shown that the two forward/backward antichains algorithms introduced by \citet{DBLP:conf/cav/WulfDHR06} can be systematically derived by instantiating our framework and (possibly) considering the reverse automata.
Similarly, we can derive within our framework an algorithm equivalent to the backward antichains algorithm applied to the reverse automata and an algorithm equivalent to the forward antichains algorithm (without reverting the automata).
Table~\ref{table:antichainsAlgorithms} summarizes the relation between our framework and the antichains algorithms given (explicitly or implicitly) by \citet{DBLP:conf/cav/WulfDHR06}.

\begin{table}[!ht]
\centering
\setlength{\tabcolsep}{4pt}
\setlength{\extrarowheight}{1ex}
\begin{tabular}{c?c|c}
& \emph{Backward} & \emph{Forward} \\
\toprule
\(\lang{\cN_1} \subseteq \lang{\cN_2}\) & \(\cpre_u^{\cN_2}(F_2^c) \subseteq \cpre_v^{\cN_2}(F_2^c)\) & \(\post_u^{\cN_2}(I_2) \subseteq \post_v^{\cN_2}(I_2)\)\\
\(\lang{\cN_1^R} \subseteq \lang{\cN_2^R}\) & \(\cpost_u^{\cN_2}(I_2^c) \subseteq \cpost_v^{\cN_2}(I_2^c)\) & \(\pre_u^{\cN_2}(F_2) \subseteq \pre_v^{\cN_2}(F_2)\)
\end{tabular}
\caption{Summary of the quasiorders that should be used within our framework, i.e. using Theorem~\ref{theorem:EffectiveAlgorithm}, to derive the different antichains algorithms that are (explicitly or implicitly) given by \citet{DBLP:conf/cav/WulfDHR06}.
Each cell of the form \(f(u) \subseteq f(v)\) is the definition of the quasiorder \(u \qo v \ud f(u) \subseteq f(v)\) that should be used to derive the antichains algorithm given by the column for solving the language inclusion given by the row.}\label{table:antichainsAlgorithms}
\end{table}

The original antichains algorithms were later improved by \citet{Abdulla2010} and, subsequently, by \citet{DBLP:conf/popl/BonchiP13}. Among their improvements, they showed how to exploit a precomputed binary relation between pairs of states of the input automata such that language inclusion holds for all the pairs in the relation.
When that binary relation is a simulation relation, our framework allows to partially match their results by using the quasiorder \(\preceq^{r}_{\cN}\) defined in Section~\ref{subsec:state-qos}.
However, this quasiorder relation \(\preceq^{r}_{\cN}\) does not consider pairs of states \(Q_1 \times Q_1\) whereas the aforementioned algorithms do.

\section{Inclusion for Context Free Languages}%
\label{sec:context_free_languages}
In Section~\ref{sec:an_algorithmic_framework_for_language_inclusion_based_on_complete_abstractions} we used the general abstraction scheme presented in Section~\ref{sec:inclusion_checking_by_complete_abstractions} to derive two techniques (Theorems~\ref{theorem:FiniteWordsAlgorithmGeneral} and~\ref{theorem:EffectiveAlgorithm}) for defining algorithms for solving language inclusion problems.
Then, in Sections~\ref{sec:instantiating_the_framework_language_based_well_quasiorders} and~\ref{sec:novel_perspective_AC} we applied these techniques on different scenarios and derived algorithms for solving language inclusion problems \(L_1 \subseteq L_2\) where \(L_1\) and \(L_2\) are regular languages.

In this section, we show that the abstraction scheme from Section~\ref{sec:inclusion_checking_by_complete_abstractions} is general enough to cover language inclusion problems \(L_1 \subseteq L_2\) where \(L_1\) is context-free.
In particular, we replicate the developments from Sections~\ref{sec:an_algorithmic_framework_for_language_inclusion_based_on_complete_abstractions},~\ref{sec:instantiating_the_framework_language_based_well_quasiorders} and~\ref{sec:novel_perspective_AC} in order to extend our quasiorder-based framework for deciding the inclusion \(L_1 \subseteq L_2\) where \(L_1\) is a context-free language and \(L_2\) is regular.

\subsection{Extending the Framework to CFGs}
Similarly to the case of automata,  a CFG \(\cGr = (\cV,\Sigma,P)\) in CNF induces the following set of equations:
\[\Eqn(\cGr) \ud \{X_i = {\textstyle \bigcup_{X_i \to \beta_j \in P}} \beta_j \mid i \in [0,n]\} \enspace .\]

Given a subset of variables \(S \subseteq \cV\) of a grammar, the set of words generated from some variable in \(S\) is defined as
\[\mindex{W_{S}^{\cGr}} \ud \{w \in \Sigma^* \mid \exists X \in S, \; X \ra^* w\} \enspace .\]
When \(S = \{X\}\) we slightly abuse the notation and write \(W_{X}^{\cGr}\). 
Also, we drop the superscript \(\cGr\) when the grammar is clear from the context.
The language generated by \(\cGr\) is therefore \(\lang{\cGr} = W^{\cGr}_{X_0}\).

Next, we define the function \(\Fn_{\cGr}: \wp(\Sigma^*)^{|\cV|}\to \wp(\Sigma^*)^{|\cV|}\) and the vector \(\vect{b} \in \wp(\Sigma^*)^{|\cV|}\), which are used to formalize the equations in \(\Eqn(\cGr)\), as follows:
\begin{align*}
\vect{b} & \ud\tuple{b_i}_{i\in[0,n]} \in \wp(\Sigma^*)^{|\cV|} &&\text{with } b_i \ud \{ \beta \mid X_i\ra \beta\in P,\:\beta\in \Sigma\cup \{ \epsilon \}\}, \\
\Fn_{\cGr }(\vect{X}) & \ud \tuple{\beta_1^{(i)}\cup\ldots\cup\beta_{k_i}^{(i)}}_{i\in[0,n]} &&\text{with } \beta_j^{(i)}\in\cV^2 \text{ and } X_i\ra\beta_j^{(i)}\in P \enspace .
\end{align*}

Notice that function \(\lambda \vect{X}\ldotp \vect{b}\mathrel{\cup} \Fn_{\cGr}(\vect{X})\) is a well-defined monotone function in  \(\wp(\Sigma^*)^{|\cV|}\ra \wp(\Sigma^*)^{|\cV|}\), which therefore has the least fixpoint
\begin{equation}\label{eq:CFGFixpoint}
\tuple{Y_i}_{i\in[0,n]} = \lfp (\lambda \vect{X}\ldotp \vect{b}\cup \Fn_{\cGr}(\vect{X}))
\end{equation}
It is known \cite{ginsburg} that the language  accepted by \(\cGr\) is such that \(\lang{\cGr} = Y_{0}\).

\begin{example}\label{example:cfg}
Consider the following grammar in CNF:
\[\cGr = \tuple{\{X_0, X_1\}, \{a,b\}, \{X_0\ra X_0X_1 \mid X_1X_0 \mid b,\: X_1 \ra a\}}\enspace .\]  
The corresponding equation system is
\[\Eqn(\cGr) = \begin{cases}
    X_0 = X_0X_1 \cup X_1X_0 \cup \{b\}\\
    X_1 =\{a\}
  \end{cases}\]
so that
\begin{equation*}
  \left( \begin{array}{c}
     W_{X_0} \\ W_{X_1}
  \end{array} \right)=
  \lfp\biggl(\lambda \left( \begin{array}{c}
    X_0 \\ X_1
  \end{array} \right) .
  \left(\begin{array}{c}
      X_0X_1 \cup X_1X_0 \cup \{b\} \\
      \{a\}
    \end{array}\right)\biggr) = \left( \begin{array}{c}
     a^*ba^* \\ a
  \end{array} \right) \enspace .
\end{equation*}

\noindent
Moreover, we have that \(\vect{b} \in \wp(\Sigma^*)^2\) and \(\Fn_{\cGr }:\wp(\Sigma^*)^2 \ra \wp(\Sigma^*)^2\) are given by
\begin{align*}
\vect{b} & =\tuple{\{b\},\{a\}} & \Fn_{\cGr }(\tuple{X_0,X_1}) &=\tuple{X_0X_1 \cup X_1X_0, \varnothing} \tag*{\eox}
\end{align*}
\end{example}

Thus, it follows from Equation~\eqref{eq:CFGFixpoint} that
\begin{equation}\label{eq:CFGIncLfp}
\lang{\cGr} \subseteq L_2 \:\Lra\:  
\lfp (\lambda\vect{X}\ldotp \vect{b}\cup \Fn_{\cGr}(\vect{X})) \subseteq \vectarg{L_2}{X_0}
\end{equation} 
where \(\vectarg{L_2}{X_0} \ud \tuple{\nullable{i=0}{L_2}{\Sigma^*}}_{i\in[0,n]}\).

As we did for the automata case in Section~\ref{sec:an_algorithmic_framework_for_language_inclusion_based_on_complete_abstractions}, we next apply Theorem~\ref{theorem:inc-check-comp-abs} in order to derive algorithms for solving the language inclusion problem \(\lang{\cGr} \subseteq L_2\) by using backward complete abstractions of \(\wp(Σ^*)\). 

\begin{theorem}\label{theorem:rhoCFG}
Let \(\rho \!\in\! \uco(\wp(\Sigma^*))\) be backward complete for both \(\lambda X. Xa\) and \(\lambda X. aX\), for all \(a \!\in\! \Sigma\) and let \(\cGr=(\cV,\Sigma,P)\) be a CFG in CNF. 
Then \(\rho\) is backward complete for \(\Fn_{\cGr}\) and \(\lambda\vect{X}\ldotp \vect{b}\cup \Fn_{\cGr}(\vect{X})\). 
\end{theorem}
\begin{proof}
Let us first show that backward completeness for left and right concatenation can be extended from letter to words.
We give the proof for the concatenation to the left, the case of the concatenation to the right is symmetric.
We prove that \(\rho(w X) = \rho(w \rho(X))\) for every \(w\in\Sigma^*\).
We proceed by induction on the length of \(w\).

The base case is trivial because \(\rho\) is idempotent.
For the inductive case \(|w| > 0\) let \(w = a u\) for some 
\(u\in\Sigma^*\) and \(a\in \Sigma\), so that:
\begin{align*}
	\rho(a u X) &= \quad\text{[By backward completeness for \(\lambda X\ldotp a X\)]}\\
	\rho(a \rho(u X)) &= \quad\text{[By inductive hypothesis]}\\
	\rho(a \rho(u \rho(X))) &= \quad\text{[By backward completeness for \(\lambda X\ldotp a X\)]}\\
	\rho(a u \rho(X)) &\enspace .
\end{align*}

Next we turn to the binary concatenation case, i.e. we prove that \(\rho(Y Z) = \rho(\rho(Y)\rho(Z))\) for all \(Y, Z \in \wp(\Sigma^*)\):
\begin{align*}
	\rho(\rho(Y)\rho(Z)) &=\quad \text{[By definition of concatenation]}\\
	\rho(\textstyle{\bigcup_{u\in\rho(Y)}} u \rho(Z)) &=\quad \text{[By Equation~\eqref{equation:lubAndGlb}]}\\
	\rho(\textstyle{\bigcup_{u\in\rho(Y)}} \rho(u \rho(Z)) ) &=\quad \text{[By backward completeness of \(\lambda X\ldotp w X\)]}\\
	\rho(\textstyle{\bigcup_{u\in\rho(Y)}} \rho(u Z))
	&=\quad\text{[By Equation~\eqref{equation:lubAndGlb}]}\\
	\rho(\textstyle{\bigcup_{u\in\rho(Y)}} u Z) &=\quad\text{[By definition of concatenation]}\\
	\rho(\rho(Y) Z) &=\quad\text{[By definition of concatenation]}\\
	\rho(\textstyle{\bigcup_{v\in Z}} \rho(Y) v)&=\quad
	\text{[By Equation~\eqref{equation:lubAndGlb}]}\\
	\rho(\textstyle{\bigcup_{v\in Z}} \rho(\rho(Y) v))&=\quad
	\text{[By backward completeness of \(\lambda X\ldotp X w\)]}\\
	\rho(\textstyle{\bigcup_{v\in Z}} \rho(Y v))&=\quad
	\text{[By Equation~\eqref{equation:lubAndGlb}]}\\
	\rho(\textstyle{\bigcup_{v\in Z}} Y v)&=\quad\text{[By definition of concatenation]}\\
	\rho(Y Z) & \enspace .
\end{align*}
Then, the proof follows the same lines of the proof of Theorem~\ref{theorem:backComplete}.
Indeed, it follows from the definition of \(\Fn_{\cGr}(\tuple{X_i}_{i\in[0,n]})\)
that: \begin{align*}
	\rho({\textstyle\bigcup_{j=1}^{k_i}}\beta^{(i)}_j)	&=\quad \text{[By definition of \(\beta^{(i)}_j\)]}\\
\rho({\textstyle\bigcup_{j=1}^{k_i}}X^{(i)}_j  Y^{(i)}_j) &=\quad
	\text{[By Equation~\eqref{equation:lubAndGlb}]}\\
\rho({\textstyle\bigcup_{j=1}^{k_i}}\rho(X^{(i)}_j  Y^{(i)}_j)) &=\quad
  \text{[By backward comp. of \(\rho\) for concatenation]}\\
\rho({\textstyle\bigcup_{j=1}^{k_i}}\rho( \rho(X^{(i)}_j)  \rho(Y^{(i)}_j))) &=\quad
	\text{[By Equation~\eqref{equation:lubAndGlb}]}\\
\rho({\textstyle\bigcup_{j=1}^{k_i}}\rho(X^{(i)}_j)  \rho(Y^{(i)}_j)) & \enspace . 
\end{align*}

\noindent
Hence, by a straightforward
componentwise application on vectors in \(\wp(\Sigma^*)^{|\cV|}\), we obtain that \(\rho\) is backward complete for \(\Fn_\cGr\). 
Finally,  \(\rho\) is backward complete for 
\(\lambda\vect{X}\ldotp (\vect{b}\cup \Fn_{\cGr}(\vect{X}))\),
because: 
\begin{myAlignEP}
	\rho(\vect{b}\cup \Fn_{\cGr}(\rho(\vect{X}))) &= 
\quad\text{[By Equation~\eqref{equation:lubAndGlb}]}\\
	\rho(\rho(\vect{b})\cup\rho(\Fn_{\cGr}(\rho (\vect{X})))) &= 
\quad\text{[By backward comp. for \(\Fn_{\cGr}\)]}\\
\rho(\rho(\vect{b})\cup\rho(\Fn_{\cGr}(\vect{X}))) &=
\quad\text{[By Equation~\eqref{equation:lubAndGlb}]}\\
\rho(\vect{b}\cup \Fn_{\cGr}(\vect{X})) &\enspace . 
\end{myAlignEP}
\end{proof}

As a consequence, by backward completeness of $\rho$ for \(\lambda\vect{X}\ldotp (\vect{b}\cup \Fn_{\cGr}(\vect{X}))\), by \eqref{eqn:lfpcompleteness} 
it turns out that:
\[\rho (\lfp(\lambda\vect{X}\ldotp \vect{b}\cup \Fn_{\cGr}(\vect{X}))) = 
\lfp(\lambda\vect{X}\ldotp \rho(\vect{b}\cup \Fn_{\cGr}(\vect{X}))) \enspace.\]

Note that if \(\rho\) is backward complete for both left and right concatenation and \(\rho(L_2)=L_2\) then, as a straightforward consequence of Equivalence~\eqref{eq:CFGIncLfp} and Theorems~\ref{theorem:inc-check-comp-abs} and~\ref{theorem:rhoCFG}, we have that:
\begin{equation}\label{equation:CFGcheck}
	\lang{\cGr}\subseteq L_2 \Lra \lfp(\lambda\vect{X}\ldotp \rho(\vect{b}\cup \Fn_{\cGr}(\vect{X}))) \subseteq \vectarg{L_2}{X_0} \enspace .
\end{equation}

Next, we present two techniques for solving the language inclusion problem \(\lang{\cGr} \subseteq L_2\) by relying on Equivalence~\eqref{equation:CFGcheck}.
As with the two techniques presented in Section~\ref{sec:SolvingAbstractInclusionCheck}, the first of these techniques allows us to define algorithms for deciding the inclusion by working on finite languages while the second one relies on the use of Galois Connections.

\subsection{Solving the Abstract Inclusion Check using Finite Languages}

The following result, which is an adaptation of Corollary~\ref{corol:FiniteWordsAlgorithm} for grammars, shows that the fixpoint iteration for \(\lfp(\rho (\vect{b}\cup \Fn_{\cGr}(\vect{X})))\) can be replicated by iterating on a set of functions \(\mathcal{F}\), and then abstracting the result, provided that all functions in \(\mathcal{F}\) meet a set of requirements.

\begin{lemma}\label{lemma:FiniteWordsAlgorithmCFG}
Let \(\cGr=\tuple{\cV,\Sigma,P}\) be a CFG in CNF, let \(ρ \in \uco(Σ^*)\) be backward complete for \(\lambda X\in \wp(\Sigma^*)\ldotp aX\) and \(\lambda X\in \wp(\Sigma^*)\ldotp Xa\) for all \(a\in \Sigma\) and let \(\mathcal{F}\) be a set of functions such that every \(f \in \mathcal{F}\) is of the form \(f: \wp(Σ^*)^{|\cV|} \to \wp(Σ^*)^{|\cV|}\) and satisfies
\(\rho (\vect{b}\cup \Fn_{\cGr}(\vect{X})) = ρ(f(\vect{X}))\).
Then, for all \(0 \leq n\),
\[(ρ(\vect{b}\cup \Fn_{\cGr}(\vect{X}))^n = ρ(\mathcal{F}^n(\vect{X})) \enspace .\]
\end{lemma}
\begin{proof}
We proceed by induction on \(n\).
\begin{myItem}
\item \emph{Base case:} Let \(n = 0\).
Then \(\mathcal{F}^0(\vect{X}) = (ρ(\vect{b} \!\cup \Fn_{\cGr}(\vect{X}))^0 = \vect{\varnothing}\).

\item \emph{Inductive step:} Assume that \(ρ(\mathcal{F}^n(\vect{X})) = (ρ(\vect{b} \!\cup \Fn_{\cGr}(\vect{X}))^n\) holds for some value \(n \geq 0\).
To simplify the notation, let \(\cP(\vect{X}) = \vect{b} \!\cup \Fn_{\cGr}(\vect{X})\) so that \(ρ\mathcal{F}^n = (ρ\cP)^n\).
Then
\begin{align*}
ρ\mathcal{F}^{n{+}1}(\vect{X}) & = \quad \text{[Since \(\mathcal{F}^{n{+}1} = \mathcal{F}^n\mathcal{F}\)]} \\
ρ\mathcal{F}^n\mathcal{F}(\vect{X}) & = \quad \text{[By Inductive Hypothesis]} \\
(ρ\cP)^n\mathcal{F}(\vect{X}) & = \quad \text{[By Theorem~\ref{theorem:rhoCFG}, \(ρ\) is bw. complete for \(\cP\)]}\\
(ρ\cP)^nρ\mathcal{F}(\vect{X})  & = \quad \text{[By Inductive Hypothesis]} \\
(ρ\cP)^nρ\cP(\vect{X})  & = \quad \text{[By definition of \((ρ(\cP))^n\)]} \\
(ρ\cP)^{n{+}1}(\vect{X})
\end{align*}
\end{myItem}
We conclude that \((ρ(\vect{b} \!\cup \Fn_{\cGr}(\vect{X}))^n = ρ(\mathcal{F}^n(\vect{X}))\) for all \(0 \leq n\).
\end{proof}

We are now in position to show that the procedure \(\KleeneQO(\abseq,\mathcal{F},b)\) can be used to compute \(\lfp(\lambda\vect{X}\ldotp \rho(\vect{b}\cup \Fn_{\cGr}(\vect{X})))\).

\begin{lemma}\label{lemma:KleeneQOLfp:CFG}
Let \(ρ \in \uco(Σ^*)\) be backward complete for \(\lambda X\in \wp(\Sigma^*)\ldotp aX\) and \(\lambda X\in \wp(\Sigma^*)\ldotp Xa\) for all \(a\in \Sigma\) such that \(\tuple{\{ρ(S) \mid S \in \wp(Σ^*)\}, \subseteq}\) is an ACC CPO and let \(\cGr=\tuple{\cV,\Sigma,P}\) be a CFG in CNF.
Let \(\mathcal{F}\) be a set of functions such that every \(f \in \mathcal{F}\) is of the form \(f: \wp(Σ^*)^{|\cV|} \to \wp(Σ^*)^{|\cV|}\) and satisfies
\(\rho (\vect{b}\cup \Fn_{\cGr}(\vect{X})) = ρ(f(\vect{X}))\).
Then, 
\[\lfp(\lambda \vect{X}\ldotp\rho (\vect{b} \!\cup \Fn_{\cGr}(\vect{X}))) = ρ\left(\KleeneQO(\abseq,\mathcal{F},\vect{\varnothing})\right) \enspace .\]
Moreover, the iterates of \(\,\Kleene(\lambda \vect{X}\ldotp\rho (\vect{b} \!\cup \Fn_{\cGr}(\vect{X})),\vect{\varnothing})\) coincide in lockstep with the abstraction of the iterates of \(\,\KleeneQO(\abseq,\mathcal{F},\vect{\varnothing})\)
\end{lemma}
\begin{proof}
Since \(\tuple{\{ρ(S) \mid S \in \wp(Σ^*)\}, \subseteq}\) is an ACC CPO, by Theorem~\ref{theorem:Kleene}, we have that
\[\lfp(\lambda \vect{X}\ldotp\rho (\vect{b} \!\cup \Fn_{\cGr}(\vect{X}))) = \Kleene(\lambda \vect{X}\ldotp\rho (\vect{b} \!\cup \Fn_{\cGr}(\vect{X})), \vect{\varnothing})\]
On the other hand, by Lemma~\ref{lemma:FiniteWordsAlgorithmCFG}, the iterates of the above Kleene iteration coincide in lockstep with the abstraction of the iterates of \(\KleeneQO(\abseq,\mathcal{F},\vect{\varnothing})\) and, therefore,
\[\Kleene(\lambda \vect{X}\ldotp\rho (\vect{b} \!\cup \Fn_{\cGr}(\vect{X}))), \vect{\varnothing}) = ρ\left(\KleeneQO(\abseq,\mathcal{F},\vect{\varnothing})\right)\]
As a consequence,
\[\lfp(\lambda \vect{X}\ldotp\rho (\vect{b} \!\cup \Fn_{\cGr}(\vect{X}))) = ρ\left(\KleeneQO(\abseq,\mathcal{F},\vect{\varnothing})\right) \enspace .\]
\end{proof}

We are now in position to introduce the equivalent of Theorem~\ref{theorem:FiniteWordsAlgorithmGeneral} for grammars.

\begin{theorem}\label{theorem:FiniteWordsAlgorithmGeneral:CFG}
Let \(\cGr=\tuple{\cV,\Sigma,P}\) be a CFG in CNF, let \(L_2\) be a regular language, let \(ρ \in \uco(Σ^*)\) and let \(\mathcal{F}\) be a set of functions.
Assume that the following properties hold:
\begin{myEnumI}
\item The abstraction \(ρ\) satisfies \(ρ(L_2) = L_2\) and it is backward complete for both \(\lambda X\in \wp(\Sigma^*)\ldotp aX\) and \(\lambda X\in \wp(\Sigma^*)\ldotp Xa\) for all \(a\in \Sigma\).\label{theorem:FiniteWordsAlgorithmGeneral:CFG:rho}
\item The set \(\tuple{\{ρ(S) \mid S \in \wp(Σ^*)\}, \subseteq}\) is an ACC CPO.\label{theorem:FiniteWordsAlgorithmGeneral:CFG:ACC}
\item Every function \(f_i\) in the set \(\mathcal{F}\) is of the form \(f_i: \wp(Σ^*)^{|\cV|} \to \wp(Σ^*)^{|\cV|}\), it is computable and satisfies \(\rho (\vect{b} \!\cup \Fn_{\cGr}(\vect{X})) = ρ(f_i(\vect{X}))\).\label{theorem:FiniteWordsAlgorithmGeneral:CFG:F}
\item There is an algorithm, say \(\abseq^{\sharp}(\vect{X}, \vect{Y})\), which decides the abstraction equivalence \(ρ(\vect{X}) = ρ(\vect{Y})\), for all \(\vect{X}, \vect{Y} \in \wp(Σ^*)^{|\cV|}\).\label{theorem:FiniteWordsAlgorithmGeneral:CFG:EQ}
\item There is an algorithm, say \(\absincl(\vect{X})\), which decides the inclusion \(ρ(\vect{X}) \subseteq \vectarg{L_2}{X_0}\), for all \(\vect{X} \in \wp(Σ^*)^{|\cV|}\).\label{theorem:FiniteWordsAlgorithmGeneral:CFG:INC}
\end{myEnumI}
\medskip
Then, the following is an algorithm which decides whether \(\lang{\cGr} \subseteq L_2\):

\medskip

\(\tuple{Y_i}_{i \in [0,n]} := \KleeneQO (\abseq^{\sharp},\mathcal{F}, \vect{\varnothing})\)\emph{;}

\emph{\textbf{return}} \(\absincl(\tuple{Y_i}_{i \in [0,n]})\)\emph{;}
\end{theorem}

\begin{proof}
It follows from hypotheses~\ref{theorem:FiniteWordsAlgorithmGeneral:CFG:rho},~\ref{theorem:FiniteWordsAlgorithmGeneral:CFG:ACC} and~\ref{theorem:FiniteWordsAlgorithmGeneral:CFG:F}, by Lemma~\ref{lemma:KleeneQOLfp:CFG}, that 
\begin{equation}\label{eq:lfpKleeneQO:CFG}
\lfp(\lambda \vect{X}\ldotp\rho (\vect{b} \!\cup \Fn_{\cGr}(\vect{X}))) = ρ\left(\KleeneQO(\abseq,\mathcal{F},\vect{\varnothing})\right)
\end{equation}
Observe that function \(\abseq\) can be replaced by function \(\abseq^{\sharp}\) due to hypothesis~\ref{theorem:FiniteWordsAlgorithmGeneral:CFG:EQ}.
Moreover, it follows from Equivalence~\eqref{equation:CFGcheck}, which holds by hypothesis~\ref{theorem:FiniteWordsAlgorithmGeneral:CFG:rho}, and Equation~\eqref{eq:lfpKleeneQO:CFG} that
\[\lang{\cGr}\subseteq L_2 \Lra ρ\left(\KleeneQO (\abseq^{\sharp}, \mathcal{F}, \vect{\varnothing})\right) \subseteq \vectarg{L_2}{X_0}\enspace .\]

Finally, hypotheses~\ref{theorem:FiniteWordsAlgorithmGeneral:CFG:EQ} and~\ref{theorem:FiniteWordsAlgorithmGeneral:CFG:INC} guarantee, respectively, the decidability of the inclusion check \(ρ\mathcal{F}(X) \subseteq ρ(X)\) performed at each step of the \(\KleeneQO\) iteration and the decidability of the inclusion of the lfp in \(\vectarg{L_2}{X_0}\).
\end{proof}

\subsection{Solving the Abstract Inclusion Check using Galois Connections}

The following result is the equivalent of Theorem~\ref{theorem:EffectiveAlgorithm} for context-free languages.
It shows that the language inclusion problem \(\lang{\cGr} \subseteq L_2\) can be solved by working on an abstract domain.

\begin{theorem}\label{theorem:EffectiveAlgorithmCFG}
Let \(\cGr=\tuple{\cV,\Sigma,P}\) be a CFG in CNF and let \(L_2\) be a language over \(\Sigma\).
  Let \(\tuple{\wp(\Sigma^*),\subseteq} \galois{\alpha}{\gamma}\tuple{D,\sqsubseteq}\) be a GC where \( \tuple{D,\leq_D}\) is a poset.
Assume that the following properties hold:
\begin{myEnumI}
\item \(L_2\in\gamma(D)\) and for every \( a \in \Sigma\), \(X \in \wp(\Sigma^*)\), \(\gamma\alpha(a X) = \gamma\alpha(a \gamma\alpha(X))\) and \(\gamma\alpha(Xa) = \gamma\alpha\gamma(\alpha(X)a)\).\label{theorem:EffectiveAlgorithmCFG:prop:rho}
\item \((D,\leq_D,\sqcup,\bot_D)\) is an effective domain, meaning that: \((D,\leq_D,\sqcup,\bot_D)\) is an ACC join-semilattice with bottom $\bot_D$, 
every element of \(D\) has a finite representation, the binary relation 
\(\leq_D\) is decidable and the binary lub \(\sqcup\) is computable.\label{theorem:EffectiveAlgorithmCFG:prop:absdecidable}
\item There is an algorithm, say \(\Fn^{\sharp}(\vect{X}^\sharp)\), which computes \(\alpha(\Fn_{\cGr}(\gamma(\vect{X})))\),
for all \(\vect{X}^\sharp\in \alpha(\wp(\Sigma^*))^{|\mathcal{V}|}\).
\label{theorem:EffectiveAlgorithmCFG:prop:abspre}
\item There is an algorithm, say \(\base^\sharp\), which computes \(\alpha(\vect{b})\).\label{theorem:EffectiveAlgorithmCFG:prop:abseps}
\item There is an algorithm, say \(\absincl(\vect{X}^\sharp)\), which decides the abstract inclusion \(\vect{X}^\sharp \leq_D \alpha(\vectarg{L_2}{X_0})\), for all \(\vect{X}^\sharp\in \alpha(\wp(\Sigma^*)^{|\cV|})\).
\label{theorem:EffectiveAlgorithmCFG:prop:absincl}
\end{myEnumI}
\medskip
Then, the following is an algorithm which decides whether \(\lang{\cGr} \subseteq L_2\):

\medskip

\(\tuple{Y_i^\sharp}_{i \in [0,n]} := \Kleene (\lambda \vect{X}^\sharp\ldotp\base^\sharp \sqcup \Fn^{\sharp}(\vect{X}^\sharp), \vect{\bot_D})\)\emph{;}

\emph{\textbf{return}} \(\absincl(\tuple{Y_i^\sharp}_{i \in [0,n]})\)\emph{;}
\end{theorem}
\begin{proof}
Let \(\rho = \gamma\alpha\in \uco(\wp(\Sigma^*))\).
Then, it follows from property~\ref{theorem:EffectiveAlgorithmCFG:prop:rho} that \(L_2 \in \rho\), \(\rho(aX) = \rho(a\rho(X))\) and \(\rho(Xa) = \rho(\rho(X)a)\).
Therefore
\begin{align*}
	\lang{\cN}\subseteq L_2 &\Lra
	\quad\text{[By~\eqref{equation:CFGcheck}]}\\
	\lfp(\lambda \vect{X}\ldotp\rho (\vect{b} \cup \Fn_{\cGr}(\vect{X}))) \subseteq \vectarg{L_2}{X_0} &\Lra
	\quad\text{[By Lemma~\ref{lemma:alpharhoequality}]}\\
	\gamma(\lfp (\lambda \vect{X}^\sharp\ldotp\alpha(\vect{b}) \sqcup \alpha(\Fn_{\cGr}(\gamma(\vect{X}^\sharp))))) \subseteq \vectarg{L_2}{X_0} &\Lra 
	\quad\text{[By GC and since $L_2\in \rho$]}\\
	\lfp (\lambda \vect{X}^\sharp\ldotp\alpha(\vect{b}) \sqcup \alpha(\Fn_{\cGr}(\gamma(\vect{X}^\sharp)))) \leq_D \alpha(\vectarg{L_2}{X_0}) \enspace .
\end{align*}

\noindent
By hypotheses~\ref{theorem:EffectiveAlgorithmCFG:prop:absdecidable},~\ref{theorem:EffectiveAlgorithmCFG:prop:abspre} and~\ref{theorem:EffectiveAlgorithmCFG:prop:abseps} it turns out that \(\Kleene (\lambda \vect{X}^\sharp\ldotp\base^\sharp \sqcup \Fn^{\sharp}(\vect{X}^\sharp), \vect{\bot_D})\) is an algorithm computing \(\lfp (\lambda \vect{X}^\sharp\ldotp\alpha(\vect{b}) \sqcup \alpha(\Fn_{\cGr}(\gamma(\vect{X}^\sharp))))\). 
In particular, these hypotheses ensure that the Kleene iterates of \(\lfp (\lambda \vect{X}^\sharp\ldotp\alpha(\vect{b}) \sqcup \alpha(\Fn_{\cGr}(\gamma(\vect{X}^\sharp))))\) starting from \(\vect{\bot_D}\) are computable, finitely many and that it is decidable whether the iterates have reached the fixpoint.
The hypothesis~\ref{theorem:EffectiveAlgorithmCFG:prop:absincl} ensures decidability of the required \(\leq_D\)-inclusion check of this least fixpoint in \(\alpha(\wp(\Sigma^*))^{|\cV|}\).
\end{proof}

\subsection{Instantiating the Framework}
Let us instantiate the general algorithmic framework provided by Theorem~\ref{theorem:FiniteWordsAlgorithmGeneral:CFG} to the class of closure operators induced by quasiorder relations on words. 
Recall that a quasiorder \(\leqslant\) on \(\Sigma^*\) is monotone if 
\begin{equation}\label{def-mon}
\forall x_1, x_2 \in \Sigma^*, \forall a,b \in \Sigma, \; x_1 \leqslant x_2 \Ra ax_1 b \leqslant ax_2b \enspace .
\end{equation}

\noindent
It follows that \(x_1 \leqslant x_2 \Ra \forall u,v \in \Sigma^*, \; ux_1 v \leqslant ux_2 v\).
The following result is the equivalent to Lemma~\ref{lemma:properties} for \(L\)-consistent quasiorders and it allows us to characterize \(L\)-consistent quasiorders in terms of the induced closure.

\begin{lemma}\label{lemma:propertiesCFG}
Let \(L\in \wp(\Sigma^*)\) and \(\mathord{\leqslant_L}\) be a quasiorder on \(\Sigma^*\).
Then, \(\mathord{\leqslant_L}\) is an \(L\)-consistent quasiorder on \(\Sigma^*\) if and only if
\begin{myEnumA}
\item \(\rho_{\leqslant_L}(L) = L\), and \label{lemma:propertiesCFG:L}
\item \(\rho_{\leqslant_L}\) is backward complete for \(\lambda X\ldotp a X b\) for all \(a,b\in \Sigma\).\label{lemma:propertiesCFG:bw}
\end{myEnumA}
\end{lemma}
\begin{proof}
\begin{myEnumA}
\item It follows from Lemma~\ref{lemma:properties}~\ref{lemma:properties:L} since, by Definition~\ref{def:LConsistent}, a quasiorder is \(L\)-consistent if{}f it is left and right \(L\)-consistent. 

\item We first prove that if \(\mathord{\leqslant}_L\) is monotone.
Then for all $X\in \wp(\Sigma^*)$ we have that 
\(\rho_{\leqslant_L}(a X b) = \rho_{\leqslant_L}(a \rho_{\leqslant_L}(X) b)\) for all \(a, b\in\Sigma\). 

Monotonicity of concatenation together with monotonicity and extensivity of the closure $\rho_{\leqslant_L}$ imply that \(\rho_{\leqslant_L}(a X b) \subseteq \rho_{\leqslant_L}(a \rho_{\leqslant_L}(X) b)\) holds.
For the reverse inclusion, we have that:
\begin{align*}
\rho_{\leqslant_L}(a \rho_{\leqslant_L}(X)b) &= \; \text{[By definition of \(\rho_{\leqslant_L}\)]}\\
\rho_{\leqslant_L}\left( \{ a y b\mid \exists x\in X, x \leqslant_L y \} \right)
&= \; \text{[By definition of \(\rho_{\leqslant_L}\)]}\\
\{ z \mid \exists x\in X, y\in \Sigma^*,\, x\leqslant_L y \land a y b\leqslant_L z \}
&\subseteq \; \text{[By monotonicity of \(\leqslant_L\)]}\\
\{ z \mid \exists x\in X, y\in \Sigma^*,\, axb\leqslant_L ayb \land a y b\leqslant_L z \}
&= \; \text{[By transitivity of \(\leqslant_L\)]}\\
\{ z \mid \exists x\in X , a xb\leqslant_L z\}
&= \; \text{[By definition of \(\rho_{\leqslant_L}\)]}\\
\rho_{\leqslant_L}(a X b) &\enspace . 
\end{align*}

Next, we show that if \(\rho_{\leqslant_L}(a X b) = \rho_{\leqslant_L}(a \rho_{\leqslant_L}(X) b)\) for all $X\in \wp(\Sigma^*)$ and \(a,b\in\Sigma\) then \(\leqslant_L\) is monotone.
Let $x_1,x_2\in \Sigma^*$, $a,b\in \Sigma$. 
If $x_1 \leqslant_L x_2$ then
$\{x_2\} \subseteq \rho_{\leqslant_L}(\{x_1 \})$, and in turn 
$a\{x_2\}b \subseteq  a\rho_{\leqslant_L}(\{x_1 \})b$.
Since $\rho_{\leqslant_L}$ is monotone, we have that
$\rho_{\leqslant_L}(a\{x_2\}b) \subseteq  \rho_{\leqslant_L}(a\rho_{\leqslant_L}(\{x_1 \})b)$, so that, by backward completeness, 
$\rho_{\leqslant_L}(a\{x_2\}b) \subseteq  \rho_{\leqslant_L}(a\{x_1 \}b)$.
It follows that, $a\{x_2\}b \subseteq \rho_{\leqslant_L}(a\{x_1\}b)$, namely, 
$ax_1b \leqslant_L ax_2b$. By Equation~\eqref{def-mon}, this shows that $\leqslant_L$ is monotone. 
\end{myEnumA}
\end{proof}

Analogously to the case of regular languages presented in Section~\ref{sec:instantiating_the_framework_language_based_well_quasiorders}, Theorem~\ref{theorem:FiniteWordsAlgorithmGeneral:CFG} induces an algorithm for deciding the language inclusion \(\lang{\cGr} \subseteq L_2\) for any CFG \(\cGr\) and regular language \(L_2\).
Indeed, we can apply Theorem~\ref{theorem:FiniteWordsAlgorithmGeneral:CFG} with \(\minor{\vect{b} \! \cup \Fn_{\cGr}(\vect{X})}\) interpreted as the set of functions \(f_i \ud \minor{\vect{b} \! \cup \Fn_{\cGr}(\vect{X})}_i\) where, again, each \(\minor{\cdot}_i\) is a function mapping each set \(X \in \wp(Σ^*)\) into a minor \(\minor{X}_i\).

As a consequence, we obtain Algorithm~\AlgGrammarW which, given a language \(L_2\) whose membership problem is decidable and a decidable \(L_2\)-consistent well-quasiorder, determines whether \(\lang{\cGr} \subseteq L_2\) holds. 

\begin{figure}[!htp]
\RemoveAlgoNumber
\begin{algorithm}[H]
\SetAlgorithmName{\AlgGrammarW}{}
\SetSideCommentRight
\caption{Word-based algorithm for \(\lang{\cGr} \subseteq L_2\)}\label{alg:CFGIncW}

\KwData{CFG \(\cGr=\tuple{\cV,\Sigma, P}\); decision procedure for \(u\mathrel{\in } L_2\); decidable \(L_2\)-consistent wqo \(\mathord{\leqslant_{L_2}}\).}

\(\tuple{Y_i}_{i \in [0,n]} := \KleeneQO (\mathord{\sqsubseteq_{\leqslant_{L_2}}}\hspace{-4pt} \cap \mathord{(\sqsubseteq_{\leqslant_{L_2}}\hspace{-2pt})^{-1}}, \lambda \vect{X}\ldotp\lfloor\vect{b} \cup \Fn_{\cGr}(\vect{X})\rfloor, \vect{\varnothing})\)\;

\ForAll{\(u \in Y_0\)}{
  \lIf{\(u \notin L_2\)}{\Return \textit{false}}
}
\Return \textit{true}\;
\end{algorithm}
\end{figure}

\begin{theorem}\label{theorem:quasiorderAlgorithmGr}
Let \(\cGr=\tuple{\cV,\Sigma,P}\) be a CFG in CNF and let \(L_2\in \wp(\Sigma^*)\) be a language such that:%
\begin{myEnumIL}
  \item membership $u\in L_2$ is decidable; \label{theorem:quasiorderAlgorithmGr:membership}
  \item there exists a decidable \(L_2\)-consistent well-quasiorder on $\Sigma^*$.\label{theorem:quasiorderAlgorithmGr:decidableL}%
  \end{myEnumIL}%
  Then, Algorithm \AlgGrammarW decides the inclusion \(\lang{\cGr} \subseteq L_2\).
\end{theorem}
\begin{proof}
Let $\leqslant_{L_2}$ be a decidable $L_2$-consistent well-quasiorder on $\Sigma^*$.
Then, we check that hypotheses~\ref{theorem:FiniteWordsAlgorithmGeneral:CFG:rho}-\ref{theorem:FiniteWordsAlgorithmGeneral:CFG:INC} of Theorem~\ref{theorem:FiniteWordsAlgorithmGeneral:CFG} are satisfied.

\begin{myEnumA}
\item It follows from hypothesis~\ref{theorem:quasiorderAlgorithmGr:decidableL} and Lemma~\ref{lemma:propertiesCFG} that \(\leqslant_{L_2}\) is backward complete for left and right concatenation and satisfies \(ρ_{\leqslant_{L_2}}(L_2) = L_2\).

\item Since \(\leqslant_{L_2}\) is a well-quasiorder, it follows that \(\tuple{\{ρ_{\leqslant_{L_2}}(S) \mid S \in \wp(Σ^*)\}, \subseteq}\) is an ACC CPO.

\item Let \(\lfloor\vect{b} \cup \Fn_{\cGr}(\vect{X})\rfloor\) be the set of functions \(f_i\) each of which maps each set \(X \in \wp(Σ^*)\) into a minor of \(\vect{b} \cup \Fn_{\cGr}(\vect{X})\).
Since \(\rho_{\leqslant_{L_2}}(X) = ρ_{\leqslant_{L_2}}(\minor{X})\) for all \(X \in \wp(Σ^*)^{|\cV|}\) then all functions \(f_i\) satisfy 
\[\rho (\vect{b} \!\cup \Fn_{\cGr}(\vect{X})) = ρ(f_i(\vect{X}))\enspace . \]

\item The equality \(ρ_{\leqslant_{L_2}}(S_1) = ρ_{\leqslant_{L_2}}(S_2)\) is decidable for every \(S_1, S_2 \in \wp(Σ^*)^{|\cV|}\) since \(ρ_{\leqslant_{L_2}}(S_1) = ρ_{\leqslant_{L_2}}(S_2) \Lra S_1 \sqsubseteq_{\leqslant_{L_2}} S_2 \land S_2 \sqsubseteq_{\leqslant_{L_2}} S_1\) and \(\leqslant_{L_2}\) is decidable.

\item Since \(\vectarg{L_2}{X_0} = \tuple{\nullable{i = 0}{L_2}{\Sigma^*}}_{i \in [0,n]})\), the inclusion trivially holds for all components \(Y_i\) with \(i \neq 0\).
Therefore, it suffices to check whether \(Y_0 \subseteq L_2\) holds.
Since \(Y_0 = \minor{S}\) for some set \(S \in \wp(Σ^*)\), the inclusion \(Y_0 \subseteq L_2\) can be decided by performing finitely many membership tests, which is exactly the check performed by lines 2-4 of Algorithm~\AlgGrammarW.
By hypothesis~\ref{theorem:quasiorderAlgorithmGr:membership}, this check is decidable.%
\end{myEnumA}%
\end{proof}

\subsubsection{Myhill and State-based Quasiorders}
In the following, we will consider two quasiorders on $\Sigma^*$ and we will show that they fulfill the requirements of Theorem~\ref{theorem:quasiorderAlgorithmGr}, so that they yield algorithms for deciding the inclusion \(\lang{\cGr} \subseteq L_2\) for every CFG \(\cGr\) and regular language \(L_2\).

The \emph{context} of a word \(w\in \Sigma^*\) w.r.t a given language \(L \in \wp(\Sigma^*)\) is defined as:  
\begin{align*}
\mindex{\ctx_L}(w) &\ud \{(u, v) \in \Sigma^* \times \Sigma^* \mid  uwv\in L\}\enspace .
\end{align*}
Correspondingly, let us define the following quasiorder relation on \(\Sigma^*\):
\begin{align}\label{def-Myhillqo}
  u\mindex{\qo_L} v &\udiff\; \ctx_L(u) \subseteq \ctx_L(v) \enspace . 
\end{align}
\citet[Section 2]{deLuca1994} call \(\qo_L\) the \emph{Myhill quasiorder relative to \(L\)}. 
The following result is the analogue
of Lemma~\ref{lemma:leftrightnerodegoodqo} for $L$-consistent and Myhill's quasiorders: 
it shows that the Myhill's quasiorder is the weakest (i.e. greatest w.r.t.\ set inclusion between binary relations) \(L\)-consistent quasiorder for which Algorithm \AlgGrammarW can be instantiated to decide the inclusion \(\lang{\cGr}\subseteq L\).

\begin{lemma}\label{lemma:myhillgoodqo}
Let $L\in \wp(\Sigma^*)$.  
\begin{myEnumA}
\item \(\mathord{\qo_L}\) is an \(L\)-consistent quasiorder.
If $L$ is regular then, additionally, \(\mathord{\qo_L}\) is a decidable well-quasiorder. \label{lemma:myhillgoodqo:Consistent}
\item If \(\mathord{\leqslant}\) is an \(L\)-consistent quasiorder on $\Sigma^*$ then \( \rho_{\qo_L} \subseteq \rho_{\leqslant} \).\label{lemma:myhillgoodqo:Incl}
\end{myEnumA}
\end{lemma}
\begin{proof}
The proof follows the same lines of the proof of Lemma~\ref{lemma:leftrightnerodegoodqo}.

\begin{myEnumA}
\item
\citet[Section 3]{deLuca1994} observe that \(\mathord{\qo_L}\) is monotone. 
Moreover, if 
$L$ is regular then \(\mathord{\qo_L}\) is a wqo \citep[Proposition~2.3]{deLuca1994}. 
Let us observe that given \(u \in L\) and \(v \notin L\) we have that \((\epsilon, \epsilon) \in \ctx_L(u)\) while \((\epsilon, \epsilon) \notin \ctx_L(v)\). 
Hence, \(\mathord{\qo_L} \cap (L \times L^c) = \varnothing\) and, therefore, \(\mathord{\qo_L}\) is an \(L\)-consistent quasiorder.

Finally, if $L$ is regular then \(\qo_L\) is clearly decidable.

\item 
As shown by \citet{deLuca1994}, \(\mathord{\qo_L}\) is maximum in the set of all \(L\)-consistent quasiorders, i.e.\ every \(L\)-consistent quasiorder \(\leqslant\) is such that 
  \(x \leqslant y \Ra x \qo_L y \).
As a consequence, \(\rho_{\leqslant}(X) \subseteq \rho_{\qo_L}(X)\) holds for all \(X\in \wp(\Sigma^*)\), namely \(\mathord{\qo} \subseteq \mathord{\qo_L}\).
\end{myEnumA}
\end{proof}

\begin{figure}[t]
    \centering
  \begin{tikzpicture}[->,>=stealth',shorten >=1pt,auto,node distance=5mm and 1cm,thick,initial text=]
  \tikzstyle{every state}=[scale=0.75,fill=customblue!60,draw=blue!60,text=black]
  
  \node[initial,state] (1) {\(1\)};
  \node[state] (2) [right=of 1] {\(2\)};
  \node[state, accepting] (3) [right=of 2] {\(3\)};
  
  \path (1) edge node {\(a\)} (2)
        (2) edge node {\(a\)} (3)
        (2) edge[loop above] node {\(b\)} (2)
        (3) edge[loop above] node {\(a,b\)} (3)
        (1) edge[bend right=50] node {\(b\)} (3)
            ;
  \end{tikzpicture}
  \caption{A finite automaton \(\cN\) with \(\lang{\cN}= (b+ab^*a)(a+b)^*\).}\label{fig:C}
\end{figure}

\begin{example}\label{example:CFGIncL}
Let us illustrate the use of the Myhill quasiorder $\qo_{\lang{\cN}}$ in Algorithm \AlgGrammarW 
for solving the language inclusion \(\lang{\cGr} \subseteq \lang{\cN}\), where \(\cGr\) is the CFG in Example~\ref{example:cfg} and \(\cN\) is the NFA depicted in Figure~\ref{fig:C}.
Recall that the equations for \(\cGr\) are:
\[\Eqn(\cGr) = \begin{cases}
    X_0 = X_0X_1 \cup X_1X_0 \cup \{b\}\\
    X_1 =\{a\}
  \end{cases} \enspace .\]

\noindent
We write \(\{(S,T)\} \cup \{(X,Y)\}\) to denote the set \(\{(u,v) \mid (u,v) \in S\times T \cup X \times Y\}\).
Then, we have the following contexts (among others) for $L=\lang{\cN}=(b+ab^*a)(a+b)^*$:
\begin{align*}
\ctx_L(\epsilon) & = \{(\epsilon, L)\} \cup \{(ab^*, b^*a\Sigma^*)\} \cup \{(L, \Sigma^*)\}& 
\ctx_L(a) & = \{(\epsilon, b^*a\Sigma^*)\} \cup \{ab^*, \Sigma^*\} \cup \{(L,\Sigma^*)\} \\
\ctx_L(b) &= \{(\epsilon, \Sigma^*)\} \cup \{(ab^*, b^*a\Sigma^*)\} \cup \{(L, \Sigma^*)\}& 
\ctx_L(ba) & = \{(\epsilon, \Sigma^*)\} \cup \{(ab^*, \Sigma^*)\} \cup \{(L, \Sigma^*)\}
\end{align*}
Moreover, \(\ctx_L(ab) = \ctx_L(a)\) and \(\ctx_L(ba) = \ctx_L(aa) = \ctx_L(aaa) = \ctx_L(aab) = \ctx_L(aba)\) and, since \(a \qo_{L} ba\) and \(\varepsilon \qo_L b\), it follows that \(\minor{\Sigma^*} = \{\epsilon, a\}\).

Recall that, as shown in Example~\ref{example:cfg}, $\vect{b}=\tuple{\{b\}, \{a\}}$ and $\Fn_{\cGr }(\tuple{X_0,X_1}) =\tuple{X_0X_1 \cup X_1X_0, \varnothing}$.
Next, we show the computation of the Kleene iterates according to Algorithm \AlgGrammarW when using the quasiorder \(\mathord{\qo_L}\).
\begin{align*}
\vect{Y}^{(0)} &= \vect{\varnothing}\\
\vect{Y}^{(1)} &= \lfloor\vect{b}\rfloor = \tuple{\{b\}, \{a\}} \\
\vect{Y}^{(2)} &= \lfloor\vect{b}\rfloor \sqcup \lfloor\Fn_{\cGr}(\vect{Y}^{(1)})\rfloor 
= \tuple{\{b\}, \{a\}} \sqcup \tuple{\minor{\{ba,ab\}},\minor{\varnothing}} = \tuple{\minor{\{ba,ab,b\}}, \minor{\{a\}}} = \tuple{\{ab, b\}, \{a\}}\\
\vect{Y}^{(3)} &= \lfloor\vect{b}\rfloor \sqcup \lfloor\Fn_{\cGr}(\vect{Y}^{(2)})\rfloor = \tuple{\{b\}, \{a\}} \sqcup  \tuple{\minor{\{aba, ba, aab, ab\}},\minor{\varnothing}}\\
&=\tuple{\minor{\{aba, ba, aab, ab, b\}}, \minor{\{a\}}} = \tuple{\{ab, b\}, \{a\}} 
\end{align*}
The least fixpoint is therefore \(\vect{Y} = \tuple{\{ab, b\}, \{a\}}\).
Since $ab\in \vect{Y}_0$ but \(ab \notin \lang{\cN}\) then Algorithm \AlgGrammarW concludes that the inclusion \(\lang{\cGr} \subseteq \lang{\cN}\) does not hold. \eox
\end{example}

Similarly to Section~\ref{sec:instantiating_the_framework_language_based_well_quasiorders}, we also consider a state-based quasiorder that can be used with Algorithm \AlgGrammarW. 
First, given an NFA \(\cN = \tuple{Q, \delta, I, F, \Sigma}\) we define the state-based equivalent of the context of a word \(w \in \Sigma^*\) as follows:
\[\ctx_{\cN}(w) \ud \{(q,q') \in Q \times Q \mid q \stackrel{w}{\leadsto} q' \} \enspace .\]
Then, the quasiorder \(\qo_{\cN}\) on $\Sigma^*$ is defined as follows: for all $u,v\in \Sigma^*$,
\begin{equation}\label{eqn:state-qo:CFG}
u \qo_{\cN} v  \udiff \ctx_{\cN}(u) \subseteq \ctx_{\cN}(v)
\end{equation}
The following result is the analogue of Lemma~\ref{lemma:LAconsistent} and 
shows that \(\mathord{\qo_{\cN}}\) is a \(\lang{\cN}\)-consistent well-quasiorder, hence it can be used with Algorithm~\AlgGrammarW to decide the inclusion \(\lang{\cGr} \subseteq \lang{\cN}\).

\begin{lemma}\label{lemma:LAconsistent:CFG}
The relation \(\mathord{\qo_{\cN}}\) is a decidable \(\lang{\cN}\)-consistent wqo.
\end{lemma}
\begin{proof}
For every \(u \in \Sigma^*\), \(\ctx_{\cN}(u)\) is a finite and computable set, so that \(\mathord{\qo_{\cN}}\) is a decidable wqo. 
Next, we show that \(\mathord{\qo_{\cN}}\) is \(\lang{\cN}\)-consistent according to Definition~\ref{def:LConsistent}~\ref{eq:LConsistentPrecise}-\ref{eq:LConsistentmonotone}. 

\begin{myEnumA}
\item By picking \(u\!\in\! \lang{\cN}\) and \(v\!\notin\! \lang{\cN}\) we have that \(\ctx_{\cN}(u)\) contains a pair \((q_i, q_f)\) with \(q_i \!\in\! I\) and \(q_f \in F\) while \(\ctx_{\cN}(v)\) does not, hence \(u \not\qo_{\cN} v\).
Therefore, \(\qo_{\cN} \cap (\lang{\cN} \times \lang{\cN}^c) = \varnothing\).

\item Let us check that $\qo_{\cN}$ is monotone.
To that end, observe that $\ctx_{\cN}: \tuple{\Sigma^*,\qo_{\cN}} \ra \tuple{\wp(Q^2),\subseteq}$ is a monotone function.
Therefore, for all $x_1,x_2\in \Sigma^*$ and $a,b\in \Sigma$ we have that
\begin{align*}
x_1 \qo_{\cN} x_2 & \Ra  \quad\text{[By def.\ of \(\qo_{\cN}\)]} \\
\ctx_{\cN}(x_1) \subseteq \ctx_{\cN}(x_2) & \Ra  \quad\text{[Since $\ctx_\cN$ is monotone]} \\
\ctx_{\cN}(ax_1 b) \subseteq \ctx_{\cN}(a x_2 b) & \Ra \quad \text{[By def.\ of \(\qo_{\cN}\)]} \\
ax_1b \qo_{\cN} ax_2b & \enspace . \tag*{\qedhere}
\end{align*}
\end{myEnumA}
\end{proof} 

For the Myhill wqo
$\qo_{\lang{\cN}}$, it turns out that for all $u,v\in \Sigma^*$,
\begin{align*}
u \qo_{\lang{\cN}} v  \Lra \begin{array}{c}
\ctx_{\lang{\cN}}(u)\\
\subseteq \\
\ctx_{\lang{\cN}}(v) \end{array} \Lra \begin{array}{c}
\{(x, y) \mid x \in W_{I,q} \land y \in W_{q', F} \land q \stackrel{u}{\leadsto} q'\} \\
\subseteq\\
\{(x, y) \mid x \in W_{I,q} \land y \in W_{q', F} \land q \stackrel{v}{\leadsto} q'\} 
\end{array}
\end{align*}
Therefore, \(u \qo_{\cN} v \Ra u \qo_{\lang{\cN}} v\), hence \(\mathord{\qo_{\cN}} \subseteq \mathord{\qo_{\lang{\cN}}}\) holds.

\begin{example}\label{example:CFGIncA}
Let us illustrate the use of the state-based quasiorder $\qo_{\cN}$ to solve the language inclusion \(\lang{\cGr} \subseteq \lang{\cN}\) of Example~\ref{example:CFGIncL}.
Here, we have the following contexts (among others):
\begin{align*}
\ctx_{\cN}(\epsilon) & = \{(q_1, q_1), (q_2, q_2), (q_3,q_3)\} & \ctx_{\cN}(a) & = \{(q_1, q_2), (q_2, q_3), (q_3,q_3)\} \\
\ctx_{\cN}(b) &= \{(q_1,q_3), (q_2,q_2), (q_3, q_3)\} & \ctx_{\cN}(aa) & = \{(q_1,q_3), (q_2,q_3), (q_3,q_3)\}
\end{align*}
Moreover, \(\ctx_{\cN}(ab) = \ctx_{\cN}(a)\) and \(\ctx_{\cN}(ba) = \ctx_{\cN}(aa) = \ctx_{\cN}(baa) = \ctx_{\cN}(aab) = \ctx_{\cN}(aba)\).
Recall from Example~\ref{example:CFGIncL} that for the Myhill wqo we have that \(a \qo_{\lang{\cN}} ba\), while for the state-based qo \(a \not\qo_{\cN} ba\).
Next, we show the Kleene iterates computed by Algorithm \AlgGrammarW when using the wqo \(\mathord{\qo_{\cN}}\).
\begin{align*}
\vect{Y}^{(0)} &= \vect{\varnothing}\\
\vect{Y}^{(1)} &= \minor{\vect{b}} = \tuple{\{b\}, \{a\}} \\
\vect{Y}^{(2)} &= \lfloor\vect{b}\rfloor \sqcup \lfloor\Fn_{\cGr}(\vect{Y}^{(1)})\rfloor = \tuple{\minor{\{ba,ab,b\}}, \minor{\{a\}}} = \tuple{\{ba, ab, b\}, \{a\}}\\
\vect{Y}^{(3)} &= \lfloor\vect{b}\rfloor \sqcup \lfloor\Fn_{\cGr}(\vect{Y}^{(2)})\rfloor = \tuple{\minor{\{aba, aab, ab, baa, aba, ba, b\}}, \minor{\{a\}}} = \tuple{\{ba, ab, b\}, \{a\}} 
\end{align*}
The least fixpoint is therefore \(\vect{Y} = \tuple{\{ba, ab, b\}, \{a\}}\).
Since  $ab\in \vect{Y}_0$ but \(ab \notin \lang{\cN}\), Algorithm \AlgGrammarW concludes that the inclusion \(\lang{\cGr} \subseteq \lang{\cN}\) does not hold. \eox
\end{example}

\subsection{A Systematic Approach to the Antichain Algorithm}\label{sec:ACGrammar}
Consider a CFG \(\cGr=\tuple{\cV,\Sigma,P}\) and an NFA \(\cN=\tuple{Q,Σ,δ,I,F}\) and let \(\qo_{\cN}\) be the \(\lang{\cN}\)-consistent wqo defined in~\eqref{eqn:state-qo:CFG}.
Theorem~\ref{lemma:FiniteWordsAlgorithmCFG} shows that the algorithm \AlgGrammarW solves the inclusion problem \(\lang{\cGr}\subseteq \lang{\cN}\) by working with finite languages.

Similarly to the case of the quasiorder \(\qo_{\cN}^{\ell}\) (Section~\ref{sec:novel_perspective_AC}) it suffices to keep the sets \(\ctx_{\cN}(u)\) of pairs of states of $Q$ for each word \(u\) instead of the words themselves.
Therefore, we can systematically derive a ``state-based'' algorithm analogous to \AlgGrammarW but working on the antichain poset \(\tuple{\AC_{\tuple{\wp(Q\times Q),\subseteq}},\sqsubseteq}\) viewed as an abstraction of \(\tuple{\wp(Σ^*), \subseteq}\).
Let us define the abstraction and concretization maps 
\(\alpha\colon \wp(\Sigma^*) \ra \AC_{\tuple{\wp(Q\times Q),\subseteq}}\) and
\(\gamma\colon \AC_{\tuple{\wp(Q\times Q),\subseteq}}\ra\wp(\Sigma^*)\) and the abstract function
\({\Fn}_{\cGr}^{\cN}(\tuple{X_i}_{i \in [0,n]}):{\wp(Q\times Q)^{|\cV|}}\ra \wp(Q \times Q)^{|\cV|}\) as follows:
\begin{align*}
	\alpha(X)&\ud \minor{\{\ctx_{\cN}(u) \mid u \in X\}} \\
	\gamma(Y) & \ud \{u \in \Sigma^* \mid \exists y \in Y, y \subseteq \ctx_{\cN}(u)\} \\
  \Fn_{\cGr }^{\cN}(\tuple{X_i}_{i\in[0,n]}) & \ud \langle \minor{\{X_j {\comp} X_k \mid X_i {\to} X_j X_k \in P\}} \rangle_{i \in [0,n]}
\end{align*}
where \(X \comp Y \ud \{(q,q') \mid (q,q'') \in X \land (q'',q') \in Y\}\) is standard composition of  
relations  $X,Y\subseteq Q\times Q$.

\begin{lemma}\label{lemma:rhoisgammaalphaCFG}
The following hold:
\begin{myEnumA}
\item \(\tuple{\wp(\Sigma^*),\subseteq}\galois{\alpha}{\gamma}\tuple{\AC_{\tuple{\wp(Q\times Q),\subseteq}},\sqsubseteq}\) is a GC.\label{lemma:rhoisgammaalpha:GCCFG}
\item \(\gamma \comp \alpha = \rho_{\qo_{\cN}}\)\label{lemma:rhoisgammaalpha:rhoCFG}
\item \(\Fn_{\cGr}^{\cN}(\vect{X}) = \alpha \comp \Fn_{\cGr} \comp \gamma(\vect{X})\) for all \(\vect{X}\in \alpha(\wp(\Sigma^*)^{|\cV|})\)\label{lemma:rhoisgammaalpha:preCFG}
\end{myEnumA}
\end{lemma}

\begin{proof}
\begin{myEnumA}
\item Let us first observe that $\alpha$ and $\gamma$ are well-defined. 
First, $\alpha(X)$ is an antichain of  $\tuple{\wp(Q\times Q),\subseteq}$ since it is a minor for the well-quasiorder \(\subseteq\) and, therefore, it is finite.
On the other hand, $\gamma(Y)$ is clearly an element of $\tuple{\wp(Σ^*), \subseteq}$ by definition. 

Then, for all $X\in \wp(Σ^*)$ and 
$Y\in \AC_{\tuple{\wp(Q \times Q),\subseteq}}$, 
it turns out that:
\begin{align*}
\alpha(X) \sqsubseteq Y & \Lra \quad\text{[By definition of \(\sqsubseteq\)]} \\
\forall z \in \alpha(X), \exists y \in Y, \; y \subseteq z &\Lra \quad\text{[By definition of  \(\alpha\) and \(\minor{\cdot}\)]} \\
\forall v \in X, \exists y \in Y, \; y \subseteq \ctx_{\cN}(v) &\Lra \quad\text{[By definition of \(\gamma\)]} \\
\forall v \in X, x \in γ(Y) & \Lra \quad\text{[By definition of  \(\subseteq\)]} \\
X \subseteq \gamma(Y) &\enspace . 
\end{align*}

\item For all \(X \in \wp(Σ^*)\) we have that:
\begin{adjustwidth}{-0.5cm}{}
\begin{myAlign}{0pt}{}
\gamma(\alpha(X)) &= \quad\text{[By definition of $\alpha,\gamma$]}\\
\{v \in \Sigma^* \mid \exists u\in \Sigma^*, \ctx_{\cN}(u) \in \lfloor \{ \ctx_{\cN}(w) \mid w\in X\} \rfloor \land \ctx_{\cN}(u) \subseteq \ctx_{\cN}(v)\} \span  \\
&= \quad \text{[By definition of minor]} \\
\{v \in \Sigma^* \mid \exists u\in X,\, \ctx_{\cN}(u) \subseteq \ctx_{\cN}(v)\} &= \quad \text{[By definition of \(\mathord{\qo_{\cN}}\)]}\\
\{v \in \Sigma^* \mid \exists u \in X ,\,  u \qo_{\cN} v\}&= \quad\text{[By definition of\ \(\rho_{\qo_{\cN}}\)]}\\
\rho_{\qo_{\cN}}(X) &\enspace .
\end{myAlign}
\end{adjustwidth}

\item First, we show that \(\ctx_{\cN}(uv) = \ctx_{\cN}(u)\comp \ctx_{\cN}(v)\) for every pair of words \(u,v \in Σ^*\).
\begin{align*}
\ctx_{\cN}(uv)  & = \;\text{[By def. of \(\ctx_{\cN}\)]} \\
\{(q,q') \in Q^2 \mid q\goes{uv}q'\}  & = \\
\span\specialcell{\hfill\text{[Since \(q \goes{uv}q' \Lra \exists q'' \in Q, \; q \goes{u}q'' \land q''\goes{v}q'\)]}} \\
\{(q,q') \in Q^2 \mid \exists q'' \in Q, q\goes{u}q'' \land q''\goes{v}q'\} &= \\
\span\specialcell{\hfill\text{[By definition of \(\comp\) for binary relations]}}  \\
\{(q,q'') \in Q^2 \mid q\goes{u}q''\} \comp \{(q'',q') \in Q^2 \mid q''\goes{v}q'\} &=\; \text{[By definition of \(W_{q,q'}\) and \(\ctx_{\cN}\)]} \\
\ctx_{\cN}(u)\comp \ctx_{\cN}(v)
\end{align*}

Secondly, we show that \(\minor{X \comp Y} = \minor{\minor{X}\comp \minor{Y}}\) for every \(X, Y \in \wp(Q\times Q)\).
It is straightforward to check that \(\minor{X} \comp \minor{Y} \subseteq X \comp Y\) and, therefore, \(\minor{\minor{X}\comp \minor{Y}} \subseteq \minor{X \comp Y}\).
Next, we prove the reverse inclusion by contradiction.

Let \(x\comp y \in \minor{X \comp Y}\) with \(x \in X\) and \(y \in Y\).
Assume \(x \comp y \notin \minor{\minor{X}\comp \minor{Y}}\).
Then, there exists \(\tilde{x} \in \minor{X}\) and \(\tilde{y} \in \minor{Y}\) such that \(\tilde{x} \comp \tilde{y} \in \minor{\minor{X} \comp \minor{Y}}\) and \(\tilde{x} \comp \tilde{y} \subseteq x \comp y\) which contradicts the fact that \(x \comp y \in \minor{X \comp Y}\) unless \(\tilde{x} \comp \tilde{y} = x \comp y\), in which case \(x \comp y \in\minor{\minor{X} \comp \minor{Y}} \).
Therefore, \(\minor{X \comp Y} \subseteq \minor{\minor{X} \comp \minor{Y}}\).

Finally, we show that \(\alpha(\Fn_{\cGr}(\gamma(\vect{X}))) = {\Fn}_{\cGr}^{\cN}(\vect{X})\) for all \(\vect{X}\in \alpha(\wp(\Sigma^*))^{|\cV|}\).
\begin{adjustwidth}{-0.7cm}{}
\begin{myAlign}{0pt}{0pt}
\alpha(\Fn_{\cGr}(\gamma(\vect{X}))) &= \hspace{35pt}\\
\span\specialcell{\hfill\text{[By definition of \(\Fn_{\cGr}\)]}}\\
\tuple{\alpha({\textstyle \bigcup_{X_i \to X_jX_k \in P}} \gamma(\vect{X}_{j})\gamma(\vect{X}_k))}_{i\in [0,n]} &= \\
\span\specialcell{\hfill\text{[By definition of \(\alpha\)]}} \\ 
\langle\lfloor \{ \ctx_{\cN}(w) \mid w \in {\textstyle \bigcup_{X_i \to X_jX_k \in P}} \gamma(\vect{X}_{j})\gamma(\vect{X}_k)\}\rfloor\rangle_{i\in [0,n]} &= \\
\langle\lfloor \{ \ctx_{\cN}(w) \mid \exists X_i \to X_j X_k \in P, \; w\in \gamma(\vect{X}_{j})\gamma(\vect{X}_k)\}\rfloor\rangle_{i\in [0,n]} &= \\
\span\specialcell{\hfill\text{[By definition of concatenation]}}\\
\langle\lfloor \{ \ctx_{\cN}(uv) \mid \exists X_i \to X_j X_k \in P, \; u\in \gamma(\vect{X}_{j}) \land v \in \gamma(\vect{X}_k)\}\rfloor\rangle_{i\in [0,n]} &= \\
\span\specialcell{\hfill\text{[Since \(\ctx_{\cN}(uv) = \ctx_{\cN}(u)\comp \ctx_{\cN}(v)\)]}} \\
\langle\lfloor \{ \ctx_{\cN}(u)\comp \ctx_{\cN}(v) \mid \exists X_i \to X_j X_k \in P, \; u\in \gamma(\vect{X}_{j}) \land v \in \gamma(\vect{X}_k)\}\rfloor\rangle_{i\in [0,n]} & \\
\span\specialcell{\hfill\text{[By definition of \(X \comp Y\)]}}\\
\langle\lfloor \{\ctx_{\cN}(u) \mid u\in \gamma(\vect{X}_{j}),X_i\to X_jX_k\}\comp \{\ctx_{\cN}(v) \mid v \in \gamma(\vect{X}_k),X_i{\ra} X_jX_k\}\rfloor\rangle_{i\in [0,n]} & \\
\span\specialcell{\hfill\text{[Since \(\minor{X{\comp} Y} = \minor{\minor{X}{\comp}\minor{Y}}\)]}}\\
\langle\lfloor \minor{\{\ctx_{\cN}(u) \mid u\in \gamma(\vect{X}_{j}),X_i{\ra} X_jX_k\}}\comp \minor{\{\ctx_{\cN}(v) \mid v \in \gamma(\vect{X}_k),X_i{\ra} X_jX_k\}}\rfloor\rangle_{i\in [0,n]} &= \\
\span\specialcell{\hfill\text{[Since \(\alpha(\gamma(X))=\minor{X}\)]}}\\
\langle\lfloor \minor{\{\vect{X}_j \mid X_i \to X_jX_k\} }\comp \minor{\{\vect{X}_k \mid X_i \to X_jX_k\} } \rfloor\rangle_{i\in [0,n]} & \\
\span\specialcell{\hfill\text{[Since \(\minor{X{\comp} Y} = \minor{\minor{X}{\comp}\minor{Y}}\)]}} \\
\langle\lfloor \{\vect{X}_j \mid X_i \to X_jX_k\} \comp \{\vect{X}_k \mid X_i \to X_jX_k\} \rfloor\rangle_{i\in [0,n]} & \\
\span\specialcell{\hfill\text{[By definition of \(\comp\)]}}\\
\langle\lfloor\{ \vect{X}_j\comp \vect{X}_k \mid X_i \to X_jX_k\}\rfloor\rangle_{i\in [0,n]} &= \\
\span\specialcell{\hfill \text{[By definition of \(\Fn_{\cGr}^{\cN}\)]}} \\
{\Fn}_{\cGr}^{\cN}(\vect{X}) \enspace . 
\end{myAlign}
\end{adjustwidth}
\end{myEnumA}
\end{proof}

\RemoveAlgoNumber
\begin{algorithm}[!ht]
\SetAlgorithmName{\AlgGrammarA}{}

\caption{State-based algorithm for \(L(\cGr) \subseteq L(\cN)\)}\label{alg:CFGIncA}

\KwData{CFG \(\cGr = \tuple{\cV,\Sigma,P}\) and NFA \(\cN = \tuple{Q,Σ,δ,I,F}\)}

\vspace{5pt}
\(\tuple{Y_i}_{i\in[0,n]} := \Kleene (\lambda \vect{X}\ldotp\lfloor\vect{b}\rfloor \sqcup \Fn_{\cGr}^{\cN}(\vect{X}), \vect{\varnothing})\)\;

\ForAll{\(y \in Y_0\)}{
	\lIf{\(y \cap (I \times F) = \varnothing\)}{\Return \textit{false}}
}
\Return \textit{true}\;
\end{algorithm}

\begin{theorem}\label{theorem:statesQuasiorderAlgorithmCFG}
	Let \(\cGr\) be a CFG and $\cN$ be an NFA. 
	The algorithm \AlgGrammarA decides \(L(\cGr) \subseteq L(\cN)\).
\end{theorem}
\begin{proof}
We show that all the hypotheses~\ref{theorem:EffectiveAlgorithmCFG:prop:rho}-\ref{theorem:EffectiveAlgorithmCFG:prop:absincl} of Theorem~\ref{theorem:EffectiveAlgorithmCFG} are satisfied for the abstract domain \(\tuple{D,\leq_D}=\tuple{\AC_{\tuple{\wp(Q \times Q),\subseteq}},\sqsubseteq}\) as defined by the GC of Lemma~\ref{lemma:rhoisgammaalphaCFG}~\ref{lemma:rhoisgammaalpha:GCCFG}. 

\begin{myEnumI}
\item Since, by Lemma~\ref{lemma:rhoisgammaalphaCFG}~\ref{lemma:rhoisgammaalpha:rhoCFG}, we have that \(\rho_{\qo_{\cN}}(X) = \gamma(\alpha(X))\), it follows from Lemmas~\ref{lemma:propertiesCFG}~\ref{lemma:propertiesCFG:L} and~~\ref{lemma:LAconsistent:CFG} that \(\gamma(\alpha(L_2)) = L_2\).
Moreover, for every \(a\in\Sigma\) and \(X\in\wp(\Sigma^*)\) we have \(\gamma\alpha(a X) = \gamma\alpha(a\gamma\alpha(X))\):
\begin{align*}
\gamma\alpha(a X) & = \quad \text{[In GCs \(\gamma = \gamma \alpha \gamma\)]} \\
\gamma\alpha\gamma\alpha(a X) & = \quad \text{[By Lemma~\ref{lemma:propertiesCFG}~\ref{lemma:propertiesCFG:bw} with \(\rho_{\leqslant_{\cN}} = \gamma\alpha\)]} \\
\gamma\alpha\gamma \alpha(a\gamma \alpha(X)) &= \quad \text{[In GCs \(\gamma = \gamma \alpha \gamma\)]} \\
\gamma\alpha(a\gamma\alpha(X)) & \enspace.
\end{align*}

\item \( (\AC_{\tuple{\wp(Q\times Q),\subseteq}},\sqsubseteq) \) is effective because $Q$ is finite.
\item  By Lemma~\ref{lemma:rhoisgammaalphaCFG}~\ref{lemma:rhoisgammaalpha:preCFG} we have that
\(\alpha(\Fn_{\cGr}(\gamma(\vect{X}))) = {\Fn}_{\cGr}^{\cN}(\vect{X})\) for all vectors \(\vect{X}\in \alpha(\wp(\Sigma^*))^{|\cV|}\).
\item \(\alpha(\{b\}) = \{(q,q') \mid q\ggoes{b}q'\}\) and \(\alpha({\varnothing})=\varnothing\), hence \(\minor{\alpha(\vect{b})}\) is trivial to compute.

\item Since \(\alpha(\vectarg{L_2}{X_0})=\tuple{\alpha(\nullable{i = 0}{L_2}{\Sigma^*})}_{i \in [0,n]}\), for all $\vect{Y}\in\alpha(\wp(\Sigma^*))^{|\cV|}$ the relation \(\vect{Y} \sqsubseteq \alpha(\vectarg{L_2}{X_0})\) trivially holds for all components \(Y_i\) with \(i \neq 0\).
For $Y_0$, it suffices to show that
\(Y_0 \sqsubseteq \alpha(L_2) \Lra \forall S \in Y_q, \; S \cap (I \times F) \neq \varnothing\), which is the check performed by lines 2-5 of algorithm \AlgGrammarA.
\begin{adjustwidth}{-0.5cm}{}%
\begin{myAlign}{-\baselineskip}{0pt}
Y_0 \sqsubseteq \alpha(L_2) & \Lra \quad \text{[Since \(Y_0 = \alpha(U)\) for some \(U \in \wp(\Sigma^*)\)]} \\
\alpha(U) \sqsubseteq \alpha(L_2) & \Lra \quad \text{[By GC]} \\
U \subseteq \gamma(\alpha(L_2)) & \Lra \quad \text{[By Lemmas~\ref{lemma:propertiesCFG},~\ref{lemma:LAconsistent:CFG} and~\ref{lemma:rhoisgammaalphaCFG}, $\gamma(\alpha(L_2))=L_2$]} \\
U \subseteq L_2 & \Lra \quad \text{[Since \(Y_0 =\alpha(U) = \lfloor \{ \ctx_{\cN}(u) \mid u\in U\} \rfloor \)]} \\
\forall u \in U, \ctx_{\cN}(u) \cap (I \times F) \neq \varnothing & \Lra \quad \text{[By definition of \(\ctx_{\cN}(u)\)]} \\
\forall S \in Y_0, S \cap I \neq \varnothing &\enspace .
\end{myAlign}
\end{adjustwidth}%
\end{myEnumI}
Thus, by Theorem~\ref{theorem:EffectiveAlgorithmCFG}, Algorithm \AlgGrammarA decides \(\lang{\cGr} \subseteq \lang{\cN}\). 
\end{proof}

The resulting algorithm \AlgGrammarA shares some features with two previous works.
On the one hand, it is related to the work of \citet{Hofmann2014} which defines an abstract interpretation-based language inclusion decision procedure similar to ours.  

Even though Hofmann and Chen's algorithm and ours both manipulate sets of pairs of states of an automaton, 
their abstraction is based on equivalence relations and not quasiorders. 
Since quasiorders are strictly more general than equivalences our framework can be instantiated to 
a larger class of abstractions, most importantly coarser ones. 
Finally, it is worth pointing out 
that the approach of \citet{Hofmann2014} aims at including languages of finite and also infinite words.

A second related work is that of \citet{Holk2015} who define an antichain like algorithm manipulating sets of pairs of states. 
\citet{Holk2015} tackle the language inclusion problem \(\lang{\cGr} \subseteq \lang{\cN}\), where \(\cGr\) is a grammar and \(\cN\) and automaton, by rephrasing the problem as a data flow analysis problem over a relational domain.
In this scenario, the solution of the problem requires the computation of a least fixpoint on the relational domain, followed by an inclusion check between sets of relations.
Then, they use the ``antichains principle'' to improve the performance of the fixpoint computation and, finally, they move from manipulating relations to manipulating pairs of states.
As a consequence, \citet{Holk2015} obtain an antichains algorithm for deciding \(\lang{\cGr} \subseteq \lang{\cN}\).

By contrast, our approach is direct and systematic, since we derive \AlgGrammarA starting from the well-known Myhill quasiorder.
We believe our approach evidences the relationship between the original antichains algorithm of \citet{DBLP:conf/cav/WulfDHR06} for regular languages and the one of \citet{Holk2015} for context-free languages, which is the relation between Algorithms~\AlgRegularA and~\AlgGrammarA.
Specifically, we show that these two algorithms are conceptually identical and differ in the quasiorder used to define the abstraction in which the computation takes place.

\section{An Equivalent Greatest Fixpoint Algorithm}%
\label{sec:greatest_fixpoint_based_algorithm}

Let us recall from \citet[Theorem~4]{cou00} that if \(g \colon C\ra C\) is a monotone function on a complete lattice $\tuple{C,\leq,\vee,\wedge}$ which admits 
its unique \demph{right-adjoint} \(\widetilde{g} \colon C\ra C\), i.e. for every $c,c'\!\in\! C,\linebreak\, g(c)\leq c' \Lra c\leq \widetilde{g}(c')$ holds,
then the following equivalence holds
for all \(c,c'\in C\)
\begin{equation}\label{eqn:duality}
\lfp(\lambda x\ldotp c \vee g(x)) \leq c' \;\Lra\;
c\leq \gfp(\lambda y\ldotp c' \wedge \widetilde{g}(y)) \enspace .
\end{equation}
This property has been used by \citet{cou00} to derive equivalent least/greatest fixpoint-based invariance proof methods for programs. 

In the following, we use Equivalence~\eqref{eqn:duality} to derive an algorithm for deciding the language inclusion \(\lang{\cN_1}\subseteq \lang{\cN_2}\), which relies on the computation of a greatest fixpoint rather than a least fixpoint.
This can be achieved by exploiting the following simple observation, which provides an adjunction between concatenation and quotients
of sets of words.
\begin{lemma}\label{lemma:adjointbinary} 
For all \(X,Y \in \wp(\Sigma^*)\) and \(w\in \Sigma^*\),  \(wY \subseteq Z \Lra Y \subseteq w^{-1}Z\) and \(Yw \subseteq Z \Lra Y \subseteq Zw^{-1}\).
\end{lemma}
\begin{proof}
By definition, for all \(u\in \Sigma^*\), \(u \in w^{-1}Z\) if{}f \( wu \in Z\). 
Hence, 
\[Y\subseteq w^{-1}Z \Lra \forall u\in Y,\: wu \in Z \Lra wY\subseteq Z \enspace . \] 
 Symmetrically, \(Yw\subseteq Z\) \(\Lra\) \(Y\subseteq Zw^{-1}\) holds.
\end{proof}

Given an NFA \(\cN = \tuple{Q,Σ,\delta,I,F}\), we define $\widetilde{\Pre}_\cN:\wp(\Sigma^*)^{|Q|} \ra \wp(\Sigma^*)^{|Q|}$ as a function on $Q$-indexed vectors of sets of words as follows:
\[
	\widetilde{\Pre}_\cN(\tuple{X_q}_{q\in Q})	\ud \langle {\textstyle\bigcap_{a\in \Sigma, q'\in \delta(q,a)}}\; a^{-1} X_q
\rangle_{q'\in Q} \enspace ,
\]
where, as usual, \(\bigcap \varnothing = \Sigma^*\). It turns out that $\widetilde{\Pre}_\cN$ is the usual weakest liberal precondition which is 
right-adjoint
of $\Pre_\cN$.
\begin{lemma}\label{lemma:FnAdjoint}
For all \(\vect{X},\vect{Y}\in \wp(\Sigma^*)^{|Q|}\), \(\Pre_{\cN}(\vect{X})\subseteq \vect{Y}\Lra \vect{X}\subseteq \widetilde{\Pre}_{\cN}(\vect{Y})\).
\end{lemma}
\begin{proof}
For all \(\vect{X},\vect{Y}\in \wp(\Sigma^*)^{|Q|}\),
	\begin{align*}
		\Pre_{\cN}(\tuple{X_q}_{q\in Q}) \subseteq \tuple{Y_q}_{q\in Q} &\Lra 
		\quad\text{[By definition of $\Pre_{\cN}$]} \\
		\forall q\in Q, \; {\textstyle \bigcup_{q\ggoes{a}{q'}}} a X_{q'} \subseteq Y_q &\Lra
		\\
		\forall q,{q'}\in Q, \; q\ggoes{a} q' \Ra  a X_{q'} \subseteq Y_q &\Lra
		\quad\text{[By Lemma~\ref{lemma:adjointbinary}]}\\
		\forall q,{q'}\in Q, \; q\ggoes{a} q' \Ra X_{q'} \subseteq a^{-1} Y_q &\Lra \quad\text{[\((\forall i \in I, \; X \subseteq Y_i) \Lra X \subseteq {\textstyle\bigcap_{i\in I}} Y_i\)]}
		\\
		\forall {q'}\in Q, X_{q'} \subseteq {\textstyle\bigcap_{q\ggoes{a} q'}} a^{-1} Y_q&\Lra
		\quad\text{[By definition of $\widetilde{\Pre}_{\cN}$]} \\
   \tuple{X_q}_{q\in Q} \subseteq \widetilde{\Pre}_{\cN}(\tuple{Y_q}_{q\in Q}) 
	\end{align*}
\end{proof}

Hence, from Equivalences~\eqref{eq:lfp} and~\eqref{eqn:duality} we obtain: 
\begin{equation}
\lang{\cN_1} \subseteq L_2 \:\Lra\:
\vectarg{\epsilon}{F_1} \subseteq \gfp(\lambda \vect{X}\ldotp \vectarg{L_2}{I_1} \cap \widetilde{\Pre}_{\cN_1}(\vect{X})) \enspace .  \label{eq:inclgfplfp}
\end{equation}

The following algorithm \AlgRegularGfp decides the inclusion \(\lang{\cN_1} \subseteq L_2\) by implementing the greatest fixpoint
computation from Equivalence~\eqref{eq:inclgfplfp}. 

\begin{figure}[!ht]
\RemoveAlgoNumber
\begin{algorithm}[H]
\SetAlgorithmName{\AlgRegularGfp}{}
\SetSideCommentRight
\caption{Greatest fixpoint algorithm for \(\lang{\cN_1}\subseteq L_2\)}\label{alg:RegIncGfp}

\KwData{NFA \(\cN_1=\tuple{Q_1,\delta_1,I_1,F_1,\Sigma}\); regular language \(L_2\).}

\vspace{5pt}

\(\tuple{Y_q}_{q\in Q} := \Kleene (\lambda \vect{X}\ldotp\vectarg{L_2}{I_1} \cap \widetilde{\Pre}_{\cN_1}(\vect{X}), \vect{{\Sigma^*}})\)\;

\ForAll{\(q \in F_1\)}{
	\lIf{\(\epsilon \notin Y_q\)}{\Return \textit{false}}
}
\Return \textit{true}\;
\end{algorithm}
\end{figure}

\noindent
The intuition behind algorithm~\AlgRegularGfp is that
\[\lang{\cN_1} \subseteq L_2 \Lra \epsilon \in {\textstyle\bigcap_{w \in \lang{\cN_1}}} w^{-1}L_2 \enspace .\]
Therefore, \AlgRegularGfp computes the set \({\textstyle\bigcap_{w \in \lang{\cN_1}}} w^{-1}L_2\) by using the automaton \(\cN_1\) and by 
considering prefixes of \(\lang{\cN_1}\) of increasing lengths. This means that
after \(n\) iterations of the \(\Kleene\) procedure, Algorithm \AlgRegularGfp has computed, for every state \(q \in Q_1\), the set 
\[\bigcap_{wu\in \lang{\cN_1}, |w| \leq n, q_0 \in I_1, q_0 \stackrel{w}{\leadsto} q} \hspace{-30pt}w^{-1}L_2 \enspace , \]

The regularity of \(L_2\) together with the property of regular languages of being closed under intersections and quotients show that each iterate computed by $\Kleene (\lambda \vect{X}\ldotp\vectarg{L_2}{I_1} \cap \widetilde{\Pre}_{\cN_1}(\vect{X}), \vect{{\Sigma^*}})$ is a (computable) regular language. 
To the best of our knowledge, this language inclusion algorithm \AlgRegularGfp has never been described in the literature before.
\\
\indent
Next, we discharge the fundamental assumption on which the correctness of Algorithm \AlgRegularGfp depends on: the Kleene iterates computed by \AlgRegularGfp are finitely many.
In order to do that, we consider an abstract version of the greatest fixpoint computation exploiting 
a closure operator which guarantees that the abstract Kleene iterates are finitely many. 
This closure operator $\rho_{\qo_{\cN_2}}$ will be defined by using an ordering relation $\qo_{\cN_2}$ 
induced by an NFA $\cN_2$ such that 
\(L_2=\lang{\cN_2}\) and will be shown to be 
\emph{forward complete} for the function \(\lambda \vect{X}\ldotp \vectarg{L_2}{I_1} \cap \widetilde{\Pre}_{\cN_1}(\vect{X})\) 
used by \AlgRegularGfp.
\\
\indent
Forward completeness of abstract interpretations \cite{gq01}, also called
exactness \citep[Definition~2.15]{mine17},  is different 
from and orthogonal to backward completeness introduced in Section~\ref{sec:inclusion_checking_by_complete_abstractions}
and crucially used in Sections~\ref{sec:an_algorithmic_framework_for_language_inclusion_based_on_complete_abstractions}-\ref{sec:context_free_languages}.
In particular, a remarkable consequence 
of exploiting a forward complete abstraction is that the Kleene iterates of the concrete and abstract greatest fixpoint computations coincide.  

The intuition here is that this forward complete closure $\rho_{\leq_{\cN_2}}$ allows us to establish that all Kleene iterates of \(\gfp(\vect{X}\ldotp \vectarg{L_2}{I_1} \cap \widetilde{\Pre}_{\cN_1}(\vect{X}))\) belong to the image of the closure $\rho_{\qo_{\cN_2}}$.
More precisely, every Kleene iterate is a language which is upward closed for \(\qo_{\cN_2}\).
Interestingly, a similar phenomenon occurs in well-structured transition systems~\cite{ACJT96,Finkel2001}.
\\
\indent
Let us now describe in detail this abstraction. 
A closure \(\rho\in\uco(C)\) on a concrete domain $C$ is forward complete for a monotone function \(f:C\ra C\) if{}f \(\rho f \rho = f \rho\) holds. 
The intuition here is that forward completeness means that no loss of precision
is accumulated when the output of a computation of $f\rho$ is approximated by $\rho$, or, equivalently, $f$ maps abstract elements 
of $\rho$ into abstract elements of $\rho$. 
Dually to the case of backward completeness, forward completeness implies that \(\gfp(f)=\gfp(f\rho) = \gfp(\rho f \rho)\), when these greatest fixpoints exist (this is the case, e.g., when $C$ is a complete lattice).  

It turns out that forward and backward completeness are related by the following duality on function $f$.

\begin{lemma}[\textbf{\citep[Corollary~1]{gq01}}]\label{lemma:forwardbackwardtransfer}
Let $\tuple{C,\leq}$ be a complete lattice and assume that \(f\colon C\ra C\) admits the right-adjoint \(\widetilde{f}\colon C\ra C\), i.e.
$f(c) \leq c' \Lra c \leq \widetilde{f}(c')$ holds. 
Then, \(\rho\) is backward complete for \(f\) if{}f \(\rho\) is forward complete for \(\widetilde{f}\).
\end{lemma}

Thus, by Lemma~\ref{lemma:forwardbackwardtransfer}, in the following result instead of 
assuming the hypotheses implying that a closure $\rho$ is forward complete for the right-adjoint $\widetilde{\Pre}_{\cN_1}$ we
state some hypotheses which guarantee that $\rho$ is backward complete for its left-adjoint, which, by Lemma~\ref{lemma:FnAdjoint}, is ${\Pre}_{\cN_1}$. 

\begin{theorem}\label{theorem:dualalgorithm}
Let \(\cN_1 = \tuple{Q_1,\delta_1,I_1,F_1,\Sigma}\) be an NFA, let \(L_2\) be a regular language and let \(\rho \in \uco(\tuple{\wp(\Sigma^*),\subseteq})\). Let us assume that:  
\begin{myEnumA}
	\item \(\rho(L_2) = L_2\);
	\item \(\rho\) is backward complete for \(\lambda X\ldotp a X\) for all \(a\in \Sigma\).
\end{myEnumA}
\medskip
Then
\[\lang{\cN_1}\subseteq L_2 \Lra \vectarg{\epsilon}{F_1} \subseteq \gfp(\vect{X}\ldotp \rho(\vectarg{L_2}{I_1} \cap \widetilde{\Pre}_{\cN_1}(\rho(\vect{X}))))\enspace .\]
Moreover, the Kleene iterates computed by $\gfp(\vect{X}\ldotp \rho(\vectarg{L_2}{I_1} \cap \widetilde{\Pre}_{\cN_1}(\rho(\vect{X}))))$ coincide in lockstep with those of \(\gfp(\vect{X}\ldotp \vectarg{L_2}{I_1} \cap \widetilde{\Pre}_{\cN_1}(\vect{X}))\).
\end{theorem}
\begin{proof}
	Theorem~\ref{theorem:backComplete} shows that if \(\rho\) is backward complete for \(\lambda X\ldotp a X\) for every \(a\in\Sigma\) then it is backward complete for \(\Pre_{{\cN_1}}\). 
	Thus, by Lemma~\ref{lemma:forwardbackwardtransfer}, \(\rho\) is forward complete for \(\widetilde{\Pre}_{\cN_1}\), hence it is forward complete for \(\lambda \vect{X}\ldotp \vectarg{L_2}{I_1} \cap \widetilde{\Pre}_{\cN_1} (\vect{X})\) since: 
	\begin{align*}
	\rho (\vectarg{L_2}{I_1} \cap  \widetilde{\Pre}_{\cN_1} (\rho(\vect{X}))) & = 
	\quad\text{[By forward comp. for $\widetilde{\Pre}_{\cN_1}$ and \(\rho(L_2)=L_2\)]}\\
	\rho (\rho( \vectarg{L_2}{I_1}) \cap  \rho(\widetilde{\Pre}_{\cN_1} (\rho(\vect{X})))) & = 
	\quad\text{[Since \(\rho(\cap \rho(X)) = \cap\rho(X)\)]}\\
	\rho (\vectarg{L_2}{I_1}) \cap  \rho (\widetilde{\Pre}_{\cN_1} (\rho(\vect{X}))) & = 
	\quad\text{[By forward comp. for $\widetilde{\Pre}_{\cN_1}$ and \(\rho (L_2)=L_2\)]}\\
	\vectarg{L_2}{I_1} \cap  \widetilde{\Pre}_{\cN_1} (\rho (\vect{X}))&\enspace .
	\end{align*}
Since, by forward completeness,	we have that
\[\gfp(\vect{X}\ldotp \vectarg{L_2}{I_1} \cap \widetilde{\Pre}_{\cN_1}(\vect{X})) = \gfp(\vect{X}\ldotp \rho(\vectarg{L_2}{I_1} \cap \widetilde{\Pre}_{\cN_1}(\rho(\vect{X}))))\enspace, \] by Equivalence \eqref{eq:inclgfplfp}, we conclude 
that 
\[\lang{\cN_1}\subseteq L_2 \Lra \vectarg{\epsilon}{F_1} \subseteq \gfp(\vect{X}\ldotp \rho(\vectarg{L_2}{I_1} \cap \widetilde{\Pre}_{\cN_1}(\rho(\vect{X})))) \enspace . \]

\noindent
Finally, we observe that the Kleene iterates computing \(\gfp(\lambda \vect{X}\ldotp \vectarg{L_2}{I_1} \cap  \widetilde{\Pre}_{\cN_1} (\vect{X}))\) and those computing $\gfp(\vect{X}\ldotp \rho(\vectarg{L_2}{I_1} \cap \widetilde{\Pre}_{\cN_1}(\rho(\vect{X}))))$ coincide in lockstep since 
$\rho (\vectarg{L_2}{I_1} \cap  \widetilde{\Pre}_{\cN_1} (\rho(\vect{X}))) =
\vectarg{L_2}{I_1} \cap  \widetilde{\Pre}_{\cN_1} (\rho (\vect{X}))$ and 
\(\rho(\vectarg{L_2}{I_1})=\vectarg{L_2}{I_1}\). %
\end{proof}

We can now establish that the sequence of Kleene iterates computed by \(\gfp(\vect{X}\ldotp \vectarg{L_2}{I_1} \cap \widetilde{\Pre}_{\cN_1}(\vect{X}))\) is finite.
Let \(L_2=\lang{\cN_2}\), for some NFA \(\cN_2\), and consider the corresponding left
state-based quasiorder \(\mathord{\qo_{\cN_2}^{l}}\) on $\Sigma^*$ as defined by~\eqref{eqn:state-qo}.

Lemma~\ref{lemma:LAconsistent} tells us that \(\mathord{\qo_{\cN_2}^{l}}\) is a left \(L_2\)-consistent wqo.  
Furthermore, since \(Q_2\) is finite we have that both \(\mathord{\qo_{\cN_2}^{l}}\) and \((\mathord{\qo_{\cN_2}^{l}})^{-1}\) are wqos, so that, in turn, \( \tuple{\rho_{\qo_{\cN_2}^{l}},\subseteq}\) is a poset which is both ACC and DCC.  
In particular, the definition of \(\mathord{\qo_{\cN_2}^{l}}\) implies that every chain in \( \tuple{\rho_{\qo_{\cN_2}^{l}},\subseteq}\) has at most \(2^{|Q_2|}\) elements, so that
if we compute \(2^{|Q_2|}\) Kleene iterates then we have surely computed the greatest fixpoint.

Moreover, as a consequence of the DCC, the Kleene iterates of
\(\gfp(\lambda \vect{X}\ldotp\rho_{\leq_{\cN_2}}(\vectarg{L_2}{I_1} \cap \widetilde{\Pre}_{\cN_1}(\rho_{\leq_{\cN_2}}(\vect{X}))))\) are finitely many, hence so are the iterates of \(\gfp(\lambda \vect{X}\ldotp \vectarg{L_2}{I_1} \cap \widetilde{\Pre}_{\cN_1} (\vect{X}))\) because they go in lockstep as stated by Theorem~\ref{theorem:dualalgorithm}. 

\begin{corollary}
Let \(\cN_1\) be an NFA and let \(L_2\) be a regular language.
Then, Algorithm \AlgRegularGfp decides the inclusion \(\lang{\cN_1} \subseteq L_2\)
\end{corollary} 

Finally, it is worth citing that \citet{fiedor2019nested} put forward an algorithm for deciding WS1S formula which relies on the same lfp computation used in \AlgRegularA.
Then, they derive a dual gfp computation by relying on Park's duality~\cite{park1969fixpoint}: \(\lfp (\lambda X \ldotp f(X)) = (\gfp (\lambda X \ldotp (f(X^c))^c))^c\).
Their approach differs from ours since we use the Equivalence~\eqref{eqn:duality} to compute a gfp, different from the lfp, which still allows us to decide the inclusion problem.
Furthermore, their algorithm decides whether a given automaton accepts \(\epsilon\) and it is not clear how their algorithm could be extended for deciding language inclusion.
\clearpage{}%
\clearpage{}%
%
\chapter{Searching on Compressed Text}
\label{chap:zearch}

In this chapter, we show how to instantiate the quasiorder-based framework from Chapter~\ref{chap:LangInc} to search on compressed text.
Specifically, we adapt Algorithm \AlgGrammarA to report the number of lines in a grammar-compressed text containing a match for a given regular expression.

The problem of searching in compressed text is of practical interest due the growing amount of information handled by modern systems, which demands efficient techniques both for compression, to reduce the storage cost, and for regular expression searching, to speed up querying. 
As an evidence of the importance of this problem, note that state of the art tools for searching with regular expressions, such as \tool{grep} and \tool{ripgrep}, provide a method to search on compressed files by decompressing them on-the-fly. 

In the following, we focus on the problem of \emph{counting}, i.e. finding the number of lines of the input text that contain a match for the expression.   
This type of query is supported out of the box by many tools\footnote{Tools such as \tool{grep}, \tool{ripgrep}, \tool{awk} and \tool{ag}, among others, can be used to report the number of matching lines in a text.}, which evidences its practical interest.
However, when the text is given in compressed form, the fastest approach in practice is the \emph{decompress and search approach}, i.e. querying the uncompressed text as it is recovered by the decompressor.
In this chapter, we challenge this approach.

Lossless compression of textual data is achieved by finding repetitions in the input text and replacing them by references.
We focus on grammar-based compression schemes in which each tuple ``reference → repeated text'' is considered as a rule of a context-free grammar.
The resulting grammar, produced as the output of the compression, generates a language consisting of a single word: the uncompressed text.
Figure~\ref{fig:compress} depicts the output of a grammar-based compression algorithm.

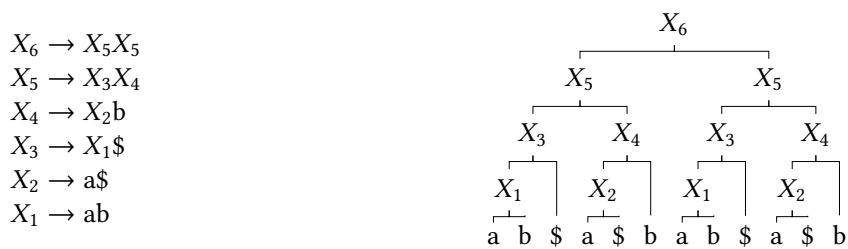
\begin{figure}[!ht]
\centering
\resizebox{!}{95pt}{
\begin{minipage}[b]{.40\textwidth}
\vspace{0pt}\hspace{0pt}$X_6→ X_5 X_5$\\
\vspace{0pt}\hspace{0pt}$X_5→ X_3 X_4$\\
\vspace{0pt}\hspace{0pt}$X_4 → X_2 \mathrm{b}$ \\
\vspace{0pt}\hspace{0pt}$X_3 → X_1 \$ $\\
\vspace{0pt}\hspace{0pt}$X_2 → \mathrm{a} \$ $\\
\vspace{0pt}\hspace{0pt}$X_1 → \mathrm{a} \mathrm{b}$ \\
\end{minipage}
\begin{minipage}[b]{.4\textwidth}
\begin{tikzpicture}
  \tikzset{%
    every node/.style={align=center},
    edge from parent/.style={%
      draw,edge from parent
      path={(\tikzparentnode.south)- +(0,-8pt)-| (\tikzchildnode)}
    },
    level distance=22pt,
    sibling distance=0pt,
    frontier/.style={distance from root=85pt} %
  }
  \Tree[.$X_6$ [.$X_5$ [.$X_3$ [.$X_1$ a b ]\$ ] [.$X_4$ [.$X_2$ a \$ ] b ]] [.$X_5$ [.$X_3$ [.$X_1$ a b ]\$ ]  [.$X_4$ [.$X_2$ a \$ ] b ] ] ]
\end{tikzpicture}
\end{minipage}}
\caption{List of grammar rules (left) generating the string
  ``\textrm{a\hspace{1pt}b\hspace{1pt}\$\hspace{1pt}a\hspace{1pt}\$\hspace{1pt}b\hspace{1pt}a\hspace{1pt}b\hspace{1pt}\$\hspace{1pt}a\hspace{1pt}\$\hspace{1pt}b}'' (and no other)
  as evidenced by the parse tree (right).}
\label{fig:compress}
\end{figure}

Intuitively, the decompress and search approach prevents the searching algorithm from taking advantage of the repetitions in the data found by the compressor.
For instance, in the grammar shown in Figure~\ref{fig:compress}, the decompress and search approach results in processing the subsequence ``$\textrm{ab\$a\$b}$'' twice.
By working on the compressed data, our algorithm would process that subsequence once and reuse the information each time it finds the variable \(X_5\).

Given a grammar-compressed text and a regular expression, deciding whether the compressed text matches the expression amounts to deciding the emptiness of the intersection of the languages generated by the grammar and an automaton built for the regular expression.

In order to solve this emptiness problem, we reduce it to an inclusion problem.
Note that this reduction is possible since the grammar generates a single word and, therefore, \(\{w\} \cap L \neq \varnothing \Lra w \in L\), where \(L\) is the language generated by the regular expression.
Then, we could instantiate the quasiorder-based framework described in Chapter~\ref{chap:LangInc} with different quasiorders to decide the inclusion.

However, in order to go beyond a yes/no answer and report or count the exact matches, we need to compute some extra information for each variable of the grammar.
This extra information is computed for the terminals of the grammar and then propagated through the variables according to the grammar rules in a bottom-up fashion.
To do that, we iterate thorough the grammar rules and compose, for each of them, the information previously computed for the variables on the right hand side.
For example, when processing rule \(X_3 {\to} X_1\$\) of Figure~\ref{fig:compress} our algorithm composes the information for \(X_1\) with the one for \(\$\).
The information computed for the string ``$\textrm{ab\$}$'', will be reused every time the variable \(X_3\) appears in the right hand side of a rule.

Following this idea, we present an algorithm for counting the lines in a grammar-compressed text containing a match for a regular expression whose runtime does not depend on the size $T$ of the uncompressed text.
Instead, it runs in time \emph{linear} in the size of its \emph{compressed version}.
Furthermore, the information computed for counting can be used to perform an \emph{on-the-fly}, \emph{lazy} decompression to recover the matching lines from the compressed text.
Note that, for reporting the matching lines, the dependency on $T$ in unavoidable.

The salient features of our approach are:
\paragraph{Generality} Our algorithm is not tied to any particular grammar-based compressor.
Instead, we consider the compressed text is given by a straight line program (SLP for short), i.e. a grammar generating the uncompressed text and nothing else.

Finding the smallest SLP $g$ generating a text of length $T$ is an NP-hard problem, as shown by \citet{Charikar2005Smallest}, for which grammar-based compressors such as LZ78~\cite{ziv1978compression}, LZW~\cite{welch1984technique}, RePair~\cite{larsson2000off} and Sequitur~\cite{nevill1997compression} produce different approximations.
For instance, \citet{Hucke2016Smallest} showed that the LZ78 algorithm produces a representation of size $Ω\bigl(\len{g}{\cdot} (T/\log{T})^{2/3}\bigr)$ and the representation produced by the RePair algorithm has size $Ω\bigl(\len{g}{\cdot} \log{T}/\log(\log{T})\bigr)$.

Since it is defined over SLPs, our algorithm applies to all such approximations, including $g$ itself.

\paragraph{Nearly optimal data structures} We define data structures enabling the algorithm to run in time linear in the size of the compressed text. 
With these data structures our algorithm runs in $\mathcal{O}(t {\cdot} s^3)$ time using $\mathcal{O}(t {\cdot} s^2)$ space where $t$ is the size of the compressed text, i.e. the grammar, and $s$ is the size of the automaton built from the regular expression. 
When the automaton is deterministic, the complexity drops to $\mathcal{O}(t{\cdot}s)$ time and $\mathcal{O}(t{\cdot}s)$ space.

As shown by \citet{Amir2018FineGrained}, there is no combinatorial\footnote{Interpreted as
any \emph{practically efficient} algorithm that does not suffer from the issues of Fast Matrix Multiplication such as large constants and inefficient memory usage.} algorithm improving these time complexity bounds beyond \emph{polylog} factors, hence our algorithm is \emph{nearly optimal}.

\paragraph{Efficient implementation} We present \tool{zearch}, a purely \emph{sequential} implementation of our algorithm which uses the above mentioned data structures.\footnote{\tool{zearch} can optionally report the matching lines.}

The experiments show that \tool{zearch} requires up to $25\pct$ less time than the state of the art: running \tool{hyperscan} on the uncompressed text as it is recovered by \tool{lz4} (in \emph{parallel}).
Furthermore, when the grammar-based compressor achieves high compression ratio (above 13:1), running \tool{zearch} on the compressed text is as fast as running \tool{hyperscan} directly on the uncompressed text.
Such compression ratios are achieved, for instance, when working with automatically generated log files.

\section{Finding the Matches}\label{sec:findingMatches}
Recall that the problem of deciding whether a grammar-compressed text contains a match for a regular expression can be reduced to an emptiness problem for the intersection of the languages generated by a grammar and an automaton.
Indeed, given an SLP \(\cP\) generating a text \(T\) over an alphabet \(Σ\), i.e. \(\lang{\cP}=\{T\}\) where \(T \in Σ^*\), and an automaton \(\cN = \tuple{Q, Σ, δ, I, F}\) representing a regular expression, we find that:
\[\text{There exists a substring of } T \text{ in }\lang{\cN}  \Lra \lang{\cP} \cap \bigl(Σ^* \cdot \lang{\cN}\cdot Σ^*\bigr) \neq \varnothing \enspace .\]

\noindent On the other hand, since \(\lang{\cP}\) contains exactly one word we have that
\[\lang{\cP} \cap \bigl(Σ^* \cdot \lang{\cN}\cdot Σ^*\bigr) \neq \varnothing \Lra \lang{\cP} \subseteq \bigl(Σ^* \cdot \lang{\cN}\cdot Σ^*\bigr) \enspace .\]

As a consequence, the problem of deciding whether a grammar-compressed text contains a match for a regular expression can be solved by using Algorithm \AlgGrammarA with the quasiorder \(\qo_{\cN}\) as described in Chapter~\ref{chap:LangInc}.

Observe that, as the following example evidences, when restricting to SLPs the iteration of the \(\Kleene\) procedure updates the abstraction for each variable of the grammar \emph{exactly once} since there are no loops in SLPs.
As a consequence, it is enough to process the rules in an orderly manner and compute the abstraction for each variable, i.e. \(α(X)\), exactly once.

\begin{example}\label{example:search}
Let \(\cP\) be the SLP from Figure~\ref{fig:compress} and let \(\cN\) and \(\cN'\) be the automata from Figure~\ref{fig:NFAsearch}.
Next, we show the Kleene iterates computed by Algorithm \AlgGrammarA which, as shown in Chapter~\ref{chap:LangInc}, works on the abstract domain \(\tuple{\AC_{\tuple{\wp(Q\times Q),\subseteq}},\sqsubseteq}\) with the abstraction function defined as \(α(X) = \minor{\{\ctx_{\cN}(u) \mid u \in X\}}\).

To simplify the notation, we denote the pair \((q_i,q_j)\) by \(ij\).
{\small
\begin{align*}
\left(\hspace{-5pt}\begin{array}{c}
α(W_{X_6}^{\cP}) \\[2pt]
α(W_{X_5}^{\cP}) \\[2pt]
α(W_{X_4}^{\cP}) \\[2pt]
α(W_{X_3}^{\cP})\\[2pt]
α(W_{X_2}^{\cP})\\[2pt]
α(W_{X_1}^{\cP})\end{array}\hspace{-5pt}\right) \,{=}\, \left(\hspace{-5pt}\begin{array}{c}
\varnothing \\[2pt]
\varnothing \\[2pt]
\varnothing \\[2pt]
\varnothing \\[2pt]
\minor{\{{11},{33}\}} \\[2pt]
\minor{\{{11},{33},{13}\}}\end{array}\hspace{-5pt}\right) \,{\Ra}\, 
\left(\hspace{-5pt}\begin{array}{c}
\varnothing \\[2pt]
\varnothing \\[2pt]
\minor{\{{11},{33}\}} \\[2pt]
\minor{\{{11},{33},{13}\}}\\[2pt]
\minor{\{{11},{33}\}} \\[2pt]
\minor{\{{11},{33},{13}\}}\end{array}\hspace{-5pt}\right) \,{\Ra}\, 
\left(\hspace{-5pt}\begin{array}{c}
\varnothing \\[2pt]
\minor{\{{11},{33},{13}\}} \\[2pt]
\minor{\{{11},{33}\}} \\[2pt]
\minor{\{{11},{33},{13}\}}\\[2pt]
\minor{\{{11},{33}\}} \\[2pt]
\minor{\{{11},{33},{13}\}}\end{array}\hspace{-5pt}\right) \,{\Ra}\,
\left(\hspace{-5pt}\begin{array}{c}
\minor{\{{11},{33},{13}\}} \\[2pt]
\minor{\{{11},{33},{13}\}} \\[2pt]
\minor{\{{11},{33}\}} \\[2pt]
\minor{\{{11},{33},{13}\}}\\[2pt]
\minor{\{{11},{33}\}} \\[2pt]
\minor{\{{11},{33},{13}\}}\end{array}\hspace{-5pt}\right)
\end{align*}}
Since for every variable \(X_n\) the value of \(α(X_n)\) is computed by combining the values of \(α(X_i)\) and \(α(X_j)\) for some \(i,j < n\), the \(\Kleene\) procedure is equivalent to computing, sequentially, the values of \(α(X_1), α(X_2), \ldots, \linebreak α(X_5)\).
In this case, since \((q_1,q_3) \in α(X_5)\) then Algorithm~\AlgGrammarA concludes that the language inclusion \(\lang{\cP} \subseteq \lang{\cN}\) holds, i.e. there exists a substring \(w\) of the uncompressed such that \(w \in \lang{\cN'}\). \eox
\end{example}

\begin{figure}[t]
    \centering
\begin{minipage}[l]{0.49\textwidth}
  \begin{tikzpicture}[->,>=stealth',shorten >=1pt,auto,node distance=5mm and 1cm,thick,initial text=]
  \tikzstyle{every state}=[scale=0.75,fill=customblue!60,draw=blue!60,text=black]
  
  \node[initial,state] (1) {\(1\)};
  \node[state] (2) [right=of 1] {\(2\)};
  \node[state, accepting] (3) [right=of 2] {\(3\)};
  
  \path (1) edge[loop above, color=white] node [color=white] {\(a,b,\$\)} (1)
        (1) edge node {\(a,b\)} (2)
        (2) edge node {\(b\)} (3);
  \end{tikzpicture}
\end{minipage}
\begin{minipage}[r]{0.49\textwidth}
  \begin{tikzpicture}[->,>=stealth',shorten >=1pt,auto,node distance=5mm and 1cm,thick,initial text=]
  \tikzstyle{every state}=[scale=0.75,fill=customblue!60,draw=blue!60,text=black]
  
  \node[initial,state] (1) {\(1\)};
  \node[state] (2) [right=of 1] {\(2\)};
  \node[state, accepting] (3) [right=of 2] {\(3\)};
  
  \path (1) edge[loop above] node {\(a,b,\$\)} (1)
        (1) edge node {\(a,b\)} (2)
        (2) edge node {\(b\)} (3)
        (3) edge[loop above] node {\(a,b,\$\)} (3);
  \end{tikzpicture}
\end{minipage}
  \caption{NFAs \(\cN'\) (left) and \(\cN\) (right) on \(Σ = \{a,b,\$\}\) with \(\lang{\cN'} = \{ab,bb\}\) and \(\lang{\cN}= Σ^* \cdot \lang{\cN'} \cdot Σ^*\).}\label{fig:NFAsearch}
\end{figure}

Furthermore, in an SLP each variable generates exactly one word and, therefore, the abstraction of a variable consists of a single set, i.e. \(α(X) \in \AC_{\tuple{\wp(Q\times Q),\subseteq}}\) is a singleton as shown in Example~\ref{example:search}.
As a consequence, we can drop the \(\minor{\cdot}\) from function \(\Fn^{\cN}_{\cP}\) defined in Section~\ref{sec:ACGrammar}, since \(\minor{\{\ctx_{\cN}(w)\}} = \{\ctx_{\cN}(w)\}\) for any word, and write:
\[\Fn_{\cP }^{\cN}(\tuple{X_i}_{i\in[0,n]}) \ud \langle \{X_j \comp X_k \mid X_i {\to} X_j X_k \in P\} \rangle_{i \in [0,n]}\]
Recall that, by definition, for all \(X_j, X_k \in \wp(Q \times Q)^{|\cV|}\),
\[X_j \comp X_k = \{(q_1,q_2) \mid \exists q' \in Q, \; (q_1,q') \in X_j \land (q',q_2) \in X_k\} \enspace .\]

Finally, given an NFA \(\cN'\) it is straightforward to build an automaton \(\cN\) generating the language \(Σ^* \cdot \lang{\cN} \cdot Σ^*\) by adding self-loops reading each letter of the alphabet to every initial and every final state of \(\cN'\) as shown in Figure~\ref{fig:NFAsearch}.
Instead of adding these transitions to \(\cN\), which, as shown in Example~\ref{example:search}, results in adding the pairs \(\{(q,q) \mid q \in I \cup F\}\) to \(\ctx_{\cN}(w)\) for every word \(w \in Σ^*\), we consider them as implicit.

As a consequence, when the input grammar is an SLP and we are interested in deciding whether \(\lang{\cP} \subseteq Σ^* \cdot \lang{\cN} \cdot Σ^*\), Algorithm~\AlgGrammarA can be written as Algorithm~\AlgSLPIncS.
Observe that Algorithm~\AlgSLPIncS uses the transition function \(δ\) to store and manipulate the sets \(\ctx_{\cN}(X_i)\) for each variable \(X_i\) of the grammar, i.e.
\[(q_1,X_i,q_2) \in δ \Lra (q_1,q_2) \in \ctx_{\cN}(X_i)\enspace .\]

\begin{figure}[!ht]
\RemoveAlgoNumber
\begin{algorithm}[H]
\SetAlgorithmName{\AlgSLPIncS}{}

\caption{Algorithm for deciding $\lang{\cP} \subseteq Σ^*\cdot \lang{\cN}\cdot Σ^*$.}\label{alg:SLPIncS}

\KwData{An SLP $\cP=\tuple{V,Σ,\pr}$ and an NFA $\cN=\tuple{Q,Σ,δ,I,F}$.}

\SetKwProg{myproc}{Procedure}{}{}

\myproc{\textsc{main}}{
  \ForAll{$\ell = 1,2,\ldots,\len{V}{-}1$}{\label{alg:algorithmTheoryDecide:step}
    \textbf{let} $(X_\ell → α_\ell β_\ell) \in \pr$; \label{alg:algorithmTheoryDecide:loopl}
    
    \ForAll{$q_1,q'\in Q$ s.t. $(q_1,α_\ell,q')\in δ$ or $q_1=q'\in I$}{\label{alg:algorithmTheoryDecide:loopa}
      \ForAll{$q_2 \in Q$ s.t. $(q',β_\ell,q_2)∈ δ$ or $q'=q_2 \in F$}{\label{alg:algorithmTheoryDecide:loopb}
        $δ:= δ \cup \{(q_1,X_\ell,q_2)\}$; \label{alg:algorithmTheoryDecide:add}
      }
    }
  }\label{alg:algorithmTheoryDecide:looplend}
  \Return \(((q_1,X_{\len{\cV}},q_2)\in δ \land q_1 \in I \land q_2 \in F\ \mathbin{?}\ \True \colon \False)\);\label{alg:algorithmTheoryDecide:return}
}
\end{algorithm}
\end{figure}

\section{Counting Algorithm}\label{sec:provingCorrectness}
State of the art tools for regular expression search are equipped with a number of features\footnote{\url{https://beyondgrep.com/feature-comparison/}} to perform different operations beyond deciding the existence of a match in the text.
Among the most relevant of these features we find \emph{counting}.
Tools like \tool{grep}\footnote{\url{https://www.gnu.org/software/grep}}, \tool{rg}\footnote{\url{https://github.com/BurntSushi/ripgrep}}, \tool{ack}\footnote{\url{https://github.com/beyondgrep/ack2}} or \tool{ag}\footnote{\url{https://geoff.greer.fm/ag/}} report the number of lines containing a match, ignoring matches across lines.
Next we extend Algorithm~\AlgSLPIncS to perform this sort of counting.

Let $\NL$ denote the new-line delimiter and let $\widehat{Σ} = Σ {\setminus} \{\NL\}$.
Given a string $w \in Σ^+$ compressed as an SLP $\cP=\tuple{V,Σ,\pr}$ and an automaton $\cN=\tuple{Q,\widehat{Σ},δ,I,F}$ built from a regular expression, Algorithm \AlgCountLines reports the number of lines in $w$ containing a match for the expression.
Note that, as the tools mentioned in the previous paragraph, we deliberately ignore matches across lines.

As an overview, our algorithm computes some \emph{counting information} for each alphabet symbol of the grammar (procedure \textsc{init\_automaton}) which is then propagated, in a bottom-up manner, to the axiom rule.
Such propagation is achieved by iterating through the grammar rules (loop in line~\ref{alg:algorithmTheoryCount:step}) and combining, for each rule, the information for the symbols on the right hand side to obtain the information for the variable on the left (procedure \textsc{count}).
Finally, the output of the algorithm is computed from the information propagated to the axiom symbol (line~\ref{alg:algorithmTheoryCount:return}).

\begin{figure}[!ht]
\RemoveAlgoNumber
\begin{algorithm}[H]
\SetAlgorithmName{\AlgCountLines}{}

\caption{Algorithm for counting the lines in $\lang{\cP}$ that contain a word in \(\lang{\cN}\).}\label{alg:CountLines}

\KwData{An SLP $\cP=\tuple{V,Σ,\pr}$ and an NFA $\cN=\tuple{Q,\widehat{Σ},δ,I,F}$.}

\SetKwProg{myproc}{Procedure}{}{}

\vspace{3pt}

\myproc{\textsc{count}($X$, $α$, $β$, $m$)}{
$\algvarnl{X}:=\algvarnl{α} \lor \algvarnl{β}$;

\(\algvarleft{X} := ( \neg \algvarnl{α} \mathbin{?} \algvarleft{α} \lor \algvarleft{β} \lor m \colon \algvarleft{α} )\);

\(\algvarright{X} := ( \neg \algvarnl{β} \mathbin{?} \algvarright{α}\lor \algvarright{β}\lor m \colon \algvarright{β} )\);

\(\algvarcount{X} := \algvarcount{α} + \algvarcount{β} +  \bigl( \algvarnl{α} \land \algvarnl{β} \land (\algvarright{α} \lor \algvarleft{β} \lor m) \;{\mathbin{?}}\; {1}{\colon}{0}\bigr)\);
}

\vspace{3pt}

\myproc{\textsc{init\_automaton}()}{
  \ForAll{$a \in Σ$}{
    \(\algvarnl{a} := (a = \NL) \)\;

    \(\algvarcount{a} := 0 \)\;

    \(\algvarleft{a} := \bigl((q_0,a,q_f) \in δ, \; q_0 \in I,\; q_f \in F\bigr)\)\;

    \(\algvarright{a} := \algvarleft{a}\)\;
  }
}

\vspace{3pt}

\myproc{\textsc{main}}{
  \textsc{init\_automaton()}\;

  \ForAll{$\ell = 1,2,\ldots,\len{V}$}{\label{alg:algorithmTheoryCount:step}
    \textbf{let} $(X_\ell → α_\ell β_\ell) \in \pr$; \label{alg:algorithmTheoryCount:loopl}
    
    new\_match := $\False$;

    \ForAll{$q_1,q'\in Q$ s.t. $(q_1,α_\ell,q')\in δ$ or $q_1{=}q'\in I$}{\label{alg:algorithmTheoryCount:loopa}
      \ForAll{$q_2 \in Q$ s.t. $(q',β_\ell,q_2)∈ δ$ or $q'{=}q_2 \in F$}{\label{alg:algorithmTheoryCount:loopb}
        $δ:= δ \cup \{(q_1,X_\ell,q_2)\}$; \label{alg:algorithmTheoryCount:add}

        \parbox{\dimexpr\textwidth-2\algomargin\relax}{new\_match := new\_match $\lor \bigl( q_1\in I\land q'\notin \bigl(I\cup F\bigr) \land q_2\in F\bigr)$;}\label{alg:algorithmTheoryCount:nm}
      }
    }
    \textsc{count}{($X_\ell,α_\ell,β_\ell,$new\_match)};
  }\label{alg:algorithmTheoryCount:looplend}
  \Return \(\algvarcount{X_{\len{V}}} + (\algvarnl{X_{\len{V}}} \ \mathbin{?}\ \algvarleft{X_{\len{V}}}+\algvarright{X_{\len{V}}} \colon \algvarleft{X_{\len{V}}})\);\label{alg:algorithmTheoryCount:return}
}
\end{algorithm}
\end{figure}

Define a \demph{line} as a maximal factor of $w$ each symbol of which belongs to $\widehat{Σ}$, a \demph{closed line} as a line which is not a prefix nor a suffix of $w$ and a \demph{matching line} as a line in $\widehat{\lang{\cN}}$, where $\widehat{\lang{\cN}} = \widehat{Σ}^*\cdot\lang{\cN}\cdot\widehat{Σ}^*$.

\begin{example}
Consider the word \(w = ``\textrm{a\hspace{1pt}b\hspace{1pt}\NL\hspace{1pt}a\hspace{1pt}\NL\hspace{1pt}b\hspace{1pt}a\hspace{1pt}b\hspace{1pt}\NL}"\) and an NFA \(\cN\) with \(\lang{\cN} = \{ba\}\).
Then the strings ``\(ab\)'', ``\(a\)'' and ``\(bab\)'' are \emph{lines} of which only the strings ``\(ab\)'' and ``\(a\)'' are \emph{closed lines} and ``\(bab\)'' is the only \emph{matching line}.\eox
\end{example}

\begin{definition}[Counting Information]
Let \(\cN\) be an NFA and let \(\cP=\tuple{\cV,Σ,P}\) be an SLP.
The \emph{counting information of} $τ \in (V \cup Σ)$, with $τ\produces^* u$ and $u \in Σ^+$, is the tuple $\counting{τ}\ud\tuple{\varnl{τ},\varleft{τ},\varright{τ},\varcount{τ}}$ where
\begin{align*}
\varnl{τ} &\ud \exists k\,; (u)_k = \NL & 
\varleft{τ} & \ud \exists i, \; (u)_{1,i} \in \widehat{Σ}^*\cdot\lang{\cN} \\
\varright{τ} &\ud \exists j, \; (u)_{j,\dag} \in \lang{\cN}\cdot\widehat{Σ}^*& 
\varcount{τ} & \ud \len{\{(i{+}1,j{-}1) \mid (u)_{i,j} \in \NL\cdot\widehat{\lang{\cN}}\cdot\NL\}}\tag*{\rule{0.5em}{0.5em}}
\end{align*}
\end{definition}

Note that $\varnl{τ}$, $\varleft{τ}$ and $\varright{τ}$ are boolean values while $\varcount{τ}$ is an integer.
It follows from the definition that the number of \emph{matching lines} in $u$, with $τ \produces^*u$, is given by the number of \emph{closed matching lines} ($\varcount{τ}$) plus the prefix of $u$ if{}f it is a \emph{matching line} ($\varleft{τ}$) and the suffix of $u$ if{}f it is a \emph{matching line} ($\varright{τ}$) different from the prefix ($\varnl{τ})$.
Since whenever $\varnl{τ} = \False$ we have $\varleft{τ} = \varright{τ}$ , it follows that
\[\sharp\text{\emph{matching lines} in } u = \varcount{τ} + \left\{ 
\begin{array}{ll}
  1 & \text{if } \varleft{τ}\\
  0 & otherwise \end{array} \\
   \right. + \left\{ 
\begin{array}{ll}
  1 & \text{if } \varnl{τ} \land \varright{τ}\\
  0 & otherwise \end{array} \\
   \right.
\]

Computing the counting information of $τ$ requires deciding membership of certain factors of $u$ in $\widehat{\lang{A}}$.
As explained before, we reduce these membership queries to language inclusion checks which are solved by Algorithm~\AlgSLPIncS. 
This operation corresponds to lines~\ref{alg:algorithmTheoryCount:loopa} to \ref{alg:algorithmTheoryCount:add} of Algorithm \AlgCountLines.
As a result, after processing the rule for $τ$, we have $(q_1,τ,q_2) \in δ$ if{}f the automaton moves from $q$ to $q'$ reading\begin{myEnumAL}
\item $u$,
\item a suffix of $u$ and $q_1 \in I$, or
\item a prefix of $u$ and $q_2 \in F$.
\end{myEnumAL}

Procedures \textsc{count} and \textsc{init\_automaton} are quite straightforward, the main difficulty being the computation of $\algvarcount{X}$ which we explain next.
Let $x,y \in Σ^+$ be the strings generated by $α$ and $β$, respectively.
Given rule $X \to α β$, $X$ generates all the matching lines generated by $α$ and $β$ plus, possibly, a ``new'' matching line of the form $z = (x)_{i,\dag}(y)_{1,j}$ with $1< i \leq \len{x}$ and $1 \leq j < \len{y}$.
Such an extra matching line appears if{}f both $α$ and $β$ generate a $\NL$\, symbol and one of the following holds:

\begin{myEnumA}
\item The suffix of $x$ matches the expression.
\item The prefix of $y$ matches the expression.
\item There is a new match $m \in z$ with $m \notin x$, $m \notin y$ (line~\ref{alg:algorithmTheoryCount:nm}).
\end{myEnumA}

\begin{example}
Let $\cN$ be an automaton with $\lang{\cN}=\{ab,ba\}$ and let $X\to α β$ be a grammar rule with $α\produces^* ba\NL a$ and $β\produces^*b \NL aba$. 
Then $X\produces^* ba\NL ab \NL aba$.

The matching lines generated by $α$, $β$ and $X$ are, respectively, $\{ba\}$, $\{aba\}$ and $\{ba,ab,aba\}$.
Moreover 
\[\counting{α}  = \tuple{\True,\True,\False,0} \quad \text{ and } \quad \counting{β}  = \tuple{\True,\False,\True,0}\enspace .\]
Therefore, applying function \textsc{count} we find that $\counting{X}=\tuple{\True,\True,\True,1}$ so the number of matching lines is $1{+}1{+}1{=}3$, as expected. \eox
\end{example}

Note that the counting information computed by Algorithm \AlgCountLines can be used to uncompress \emph{only} the matching lines by performing a top-down processing of the SLP. 
For instance, given $X \to α β$ with $\counting{X}=\tuple{\True,\True,\False,0}$ and $\counting{α}=\tuple{\True,\True,\False,0}$, there is no need to decompress the string generated by $β$ since we are certain it is not part of any matching line (otherwise we should have $\varcount{X}>0$ or $\varright{X}=\True$).

Next, we describe the data structures that we use to implement Algorithm \AlgCountLines with \emph{nearly optimal} complexity.

\subsection{Data Structures}\label{sec:DataStructures}
We assume the alphabet symbols, variables and states are indexed and use the following data structures, illustrated in Figure~\ref{fig:datastructure}: an array $\mathcal{A}$ with $t{+}\len{Σ}$ elements, where $t$ is the number of rules of the SLP, and two $s \times s$ matrices $\mathcal{M}$ and $\mathcal{K}$ where $s$ is the number of states of the automaton.

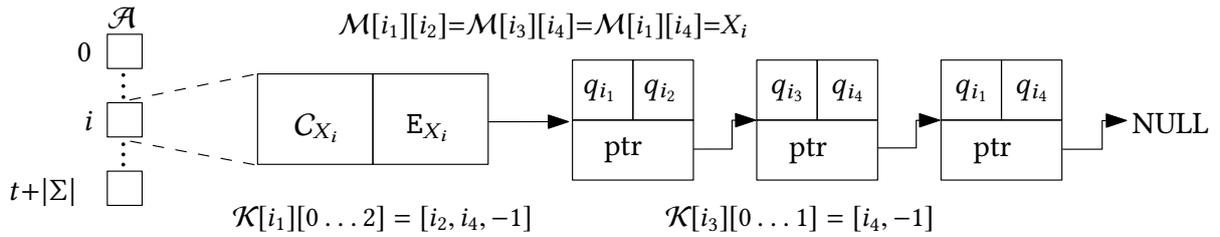
\begin{figure}[!ht]
\centering
\resizebox{\textwidth}{!}{
\begin{tikzpicture}[ipe stylesheet]
%
\draw[shift={(285.368, 716.58)}, scale=0.6563]
    (0, 0) rectangle (32, -32);
  \draw[shift={(306.37, 716.58)}, scale=0.6563]
    (0, 0) rectangle (32, -32);
  \draw[shift={(285.368, 695.577)}, scale=0.6563]
    (0, 0) rectangle (64, -32);
  \node[ipe node]
     at (188.695, 690.754) {$\counting{X_i}$};
  \draw[shift={(176, 711.737)}, xscale=0.8341, yscale=0.9892]
    (0, 0) rectangle (48, -32);
  \draw[shift={(123.816, 725.609)}, xscale=0.7558, yscale=0.7487]
    (0, 0) rectangle (16, -16);
  \draw[shift={(123.816, 701.651)}, xscale=0.7558, yscale=0.7487]
    (0, 0) rectangle (16, -16);
  \draw[shift={(123.816, 677.693)}, xscale=0.7558, yscale=0.7487]
    (0, 0) rectangle (16, -16);
  \node[ipe node]
     at (113.401, 716.068) {$0$};
  \node[ipe node]
     at (115.793, 692.139) {$i$};
  \node[ipe node]
     at (89.623, 667.626) {$t{+}\len{Σ}$};
  \draw[ipe dash dashed]
    (131.5102, 702.6101)
     -- (175.901, 711.756);
  \draw[ipe dash dashed]
    (131.7285, 688.1998)
     -- (175.901, 679.767);
  \node[ipe node]
     at (227.757, 690.888) {$\edges{X_i}$};
  \node[ipe node]
     at (296.463, 683.341) {ptr};
  \node[ipe node]
     at (289.156, 705.124) {$q_{i_1}$};
  \node[ipe node]
     at (310.158, 705.124) {$q_{i_2}$};
  \node[ipe node]
     at (479.878, 690.956) {NULL};
  \pic[ipe mark tiny]
     at (129.3755, 711.3358) {ipe disk};
  \pic[ipe mark tiny]
     at (129.3876, 707.4932) {ipe disk};
  \pic[ipe mark tiny]
     at (129.4054, 703.6406) {ipe disk};
  \pic[ipe mark tiny]
     at (129.3667, 687.1024) {ipe disk};
  \pic[ipe mark tiny]
     at (129.3876, 683.5351) {ipe disk};
  \pic[ipe mark tiny]
     at (129.3684, 679.9975) {ipe disk};
  \node[ipe node, font=\small]
     at (167.443, 658.533) {$\mathcal{K}[i_1][0…2] = [i_2,i_4,{-}1]$};
  \draw[shift={(349.367, 716.58)}, scale=0.6563]
    (0, 0) rectangle (32, -32);
  \draw[shift={(370.37, 716.58)}, scale=0.6563]
    (0, 0) rectangle (32, -32);
  \draw[shift={(349.367, 695.577)}, scale=0.6563]
    (0, 0) rectangle (64, -32);
  \draw[shift={(327.381, 685.411)}, xscale=0.4477, yscale=0.4841, ->]
    (0, 0)
     -- (29.958, 0.037)
     -- (29.958, 20.86)
     -- (48.343, 20.954);
  \node[ipe node]
     at (360.735, 683.116) {ptr};
  \node[ipe node]
     at (354.359, 704.974) {$q_{i_3}$};
  \node[ipe node]
     at (375.362, 704.974) {$q_{i_4}$};
  \draw[shift={(413.367, 716.58)}, scale=0.6563]
    (0, 0) rectangle (32, -32);
  \draw[shift={(434.37, 716.58)}, scale=0.6563]
    (0, 0) rectangle (32, -32);
  \draw[shift={(413.367, 695.577)}, scale=0.6563]
    (0, 0) rectangle (64, -32);
  \node[ipe node]
     at (424.735, 683.116) {ptr};
  \node[ipe node]
     at (418.359, 704.974) {$q_{i_1}$};
  \node[ipe node]
     at (439.362, 704.974) {$q_{i_4}$};
  \draw[shift={(391.504, 685.815)}, xscale=0.4429, yscale=0.4629, ->]
    (0, 0)
     -- (29.958, 0.037)
     -- (29.958, 20.86)
     -- (48.343, 20.954);
  \draw[shift={(455.566, 685.358)}, xscale=0.46, yscale=0.4804, ->]
    (0, 0)
     -- (29.958, 0.037)
     -- (29.958, 20.86)
     -- (48.343, 20.954);
  \node[ipe node, font=\small]
     at (317.726, 658.465) {$\mathcal{K}[i_3][0…1] = [i_4,{-}1]$};
  \node[ipe node, font=\small]
     at (204.039, 725.877) {$\mathcal{M}[i_1][i_2] {=}\mathcal{M}[i_3][i_4] {=}
\mathcal{M}[i_1][i_4] {=} X_i$};
  \node[ipe node]
     at (124.061, 727.992) {$\mathcal{A}$};
  \draw[->]
    (256.2657, 694.9724)
     -- (283.9817, 695.0154);
  \draw[shift={(216.062, 711.739)}, xscale=0.8341, yscale=0.9892]
    (0, 0) rectangle (48, -32); 
 \end{tikzpicture}
}
\caption{Data structures enabling nearly optimal running time for Algorithm \AlgCountLines. The image shows the contents of $\mathcal{M}$ after processing rule $X_i → α_i β_i$ and the contents of $\mathcal{K}$ after processing $X_\ell → α_\ellβ_\ell$ with $β_\ell = X_i$.}\label{fig:datastructure}
\end{figure}

Each element $\mathcal{A}[i]$ contains the information related to variable $X_i$, i.e. $\counting{X_i}$ and the list of transitions labeled with $X_i$, denoted $\edges{X_i}$.
We store $\counting{X}$ using one bit for each $\algvarnl{X}$, $\algvarleft{X}$ and $\algvarright{X}$ and an integer for $\algvarcount{X}$.

For each rule $X_\ell → α_\ell β_\ell$ the matrix $\mathcal{K}$ is set so that row $i$ contains the set of states reachable from the state $q_i$ by reading the string generated by $β_\ell$, i.e. \(\mathcal{K}[i] = \{q_j \mid (q_i,β_\ell,q_j) \in δ\}\).
If there are less than $s$ such states we use a sentinel value (${-}1$ in Figure~\ref{fig:datastructure}).

Finally, each element $\mathcal{M}[i][j]$ stores the index $\ell$ of the last variable for which $(q_i,X_\ell,q_j)$ was added to $δ$.
Note that since rules are processed one at a time, matrices $\mathcal{K}$ and $\mathcal{M}$ can be reused for all rules.

Observe that it is straightforward to update the matrices \(\cM\) and \(\cK\) in \(\mathcal{O}(s^2)\) time for each rule $X_\ell → α_\ell β_\ell$ since there are up to \(s^2\) transitions \((q_i,β_\ell,q_j) \in δ\).
These data structures provide $\mathcal{O}(1)$ runtime for the following operations: 
\begin{myEnum}
\item[-] Accessing the information corresponding to $α_\ell$ and $β_\ell$ at line~\ref{alg:algorithmTheoryCount:loopl} (using $\mathcal{A}$).
\item[-] Accessing the list of pairs $(q,q')$ with $(q,α_\ell,q') \in δ$ at line~\ref{alg:algorithmTheoryCount:loopa} (using $\edges{X_i}$).
\item[-] Accessing the list of states $q_2$ with $(q',β_\ell,q_2) \in δ$ at line~\ref{alg:algorithmTheoryCount:loopb} (using $\mathcal{K}$).
\item[-] Inserting a pair $(q,q')$ in $\edges{X_i}$ (avoiding duplicates) at line~\ref{alg:algorithmTheoryCount:add} (using $\mathcal{M}$).
\end{myEnum}

As a result, Algorithm \AlgCountLines runs in $\mathcal{O}(t{\cdot}s^3)$\footnote{The algorithm performs \(t\) iterations of loop in line~\ref{alg:algorithmTheoryCount:loopl}, up to \(s^2\) iterations of loop in line~\ref{alg:algorithmTheoryCount:loopa} and up to \(s\) iterations for loop in line~\ref{alg:algorithmTheoryCount:loopb}.} time using $\mathcal{O}(t{\cdot}s^2)$ space when the automaton built from the regular expression is an NFA and it runs in $\mathcal{O}(t{\cdot}s)$ time and $\mathcal{O}(t{\cdot}s)$ space when the automaton is a DFA (each row of $\mathcal{K}$ stores up to one state, hence the loop in line~\ref{alg:algorithmTheoryCount:loopb} results in, at most, one iteration).

\citet[Thm.~3.2]{Amir2018FineGrained} proved that, under the Strong Exponential Time Hypothesis, there is no combinatorial algorithm for deciding whether a grammar-compressed text contains a match for a DFA running in $\mathcal{O}((t {\cdot} s)^{1-ε})$ time  with $ε{>}0$.
For NFAs, they proved~\citep[Thm.~4.2]{Amir2018FineGrained} that, under the $k$-Clique Conjecture, there is no combinatorial algorithm running in $\mathcal{O}((t {\cdot} s^3)^{1-ε})$ time.
Therefore, our algorithm is \emph{nearly optimal} in both scenarios.

\section{Implementation}\label{sec:implementation}
We implemented Algorithm \AlgCountLines, using the data structures described in the previous section, in a tool named \tool{zearch}\footnote{\url{https://github.com/pevalme/zearch}}.
Our tool works on \tool{repair}\footnote{\url{https://storage.googleapis.com/google-code-archive-downloads/v2/code.google.com/re-pair/repair110811.tar.gz}}-compressed text and, beyond counting the matching lines, it can also report them by partially decompressing the input file.
The implementation consists of less than 2000 lines of C code.

The choice of this particular compressor, which implements the RePair algorithm of \citet{larsson2000off}, is due to the little effort required to adapt Algorithm \AlgCountLines to the specifics of the grammar built by \tool{repair} and the compression it achieves (see Table~\ref{table:compression}).
However \tool{zearch} can handle any grammar-based compression scheme by providing a way to recover the SLP from the input file.

Recall that we assume the alphabet symbols, variables and states are indexed. 
For text compressed with \tool{repair}, the indexes of the alphabet symbols are $0…255$ ($Σ$ is fixed\footnote{Our algorithm also applies to larger alphabets, such as UTF8, without altering its complexity.}) and the indexes of the variables are $256…t{+}256$.
Typically, grammar-based compressors such as \tool{repair} encode the grammar so that rule $X \to α β$ appears always after the rules with $α$ and $β$ on the left hand side.
Thus, each iteration of the loop in line~\ref{alg:algorithmTheoryCount:loopl} reads a subsequent rule from the compressed input.

We translate the input regular expression into an $ε$-free NFA using the automata library \tool{libfa}\footnote{\url{http://augeas.net/libfa/index.html}} which applies Thompson's algorithm~\cite{thompson1968programming} with on-the-fly $ε$-removal.

\section{Empirical Evaluation}\label{sec:experimental}
Next we present a \emph{summary} of the experiments carried out to assess the performance of \tool{zearch}.
The details of the experiments, including the runtime and number of matching lines reported for each expression on each file and considering more tools, file sizes and regular expressions are available on-line\footnote{\url{https://pevalme.github.io/zearch/graphs/index.html}}, where we report graphs as the ones shown in Figure~\ref{fig:webGraphs}.
The following explanations about how the experiments reported in this thesis were carried out also apply to the larger set of experiments available on-line.

\begin{figure}[!htp]
\centering
\includegraphics[width=0.95\textwidth]{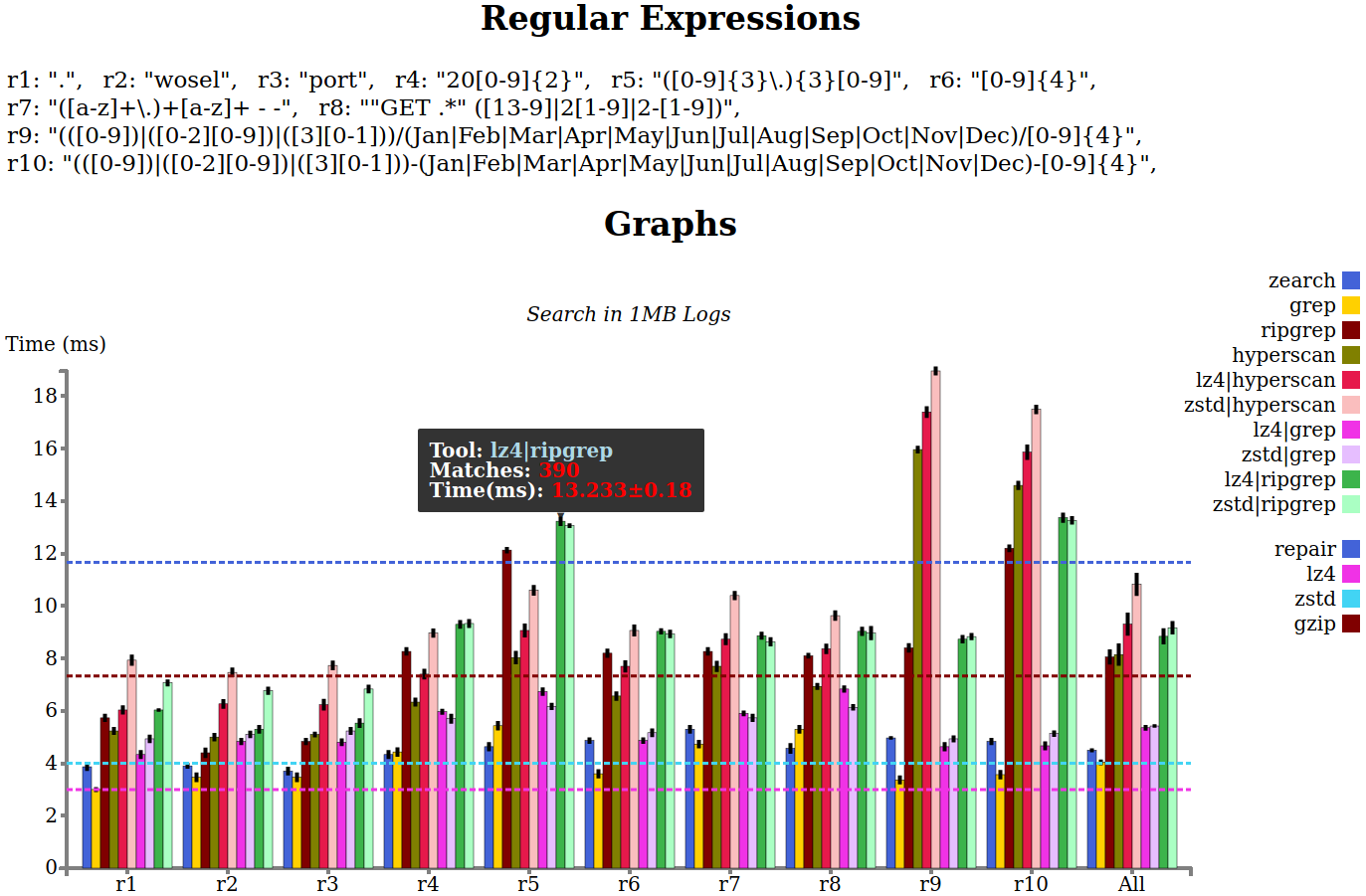}
\includegraphics[width=0.9\textwidth]{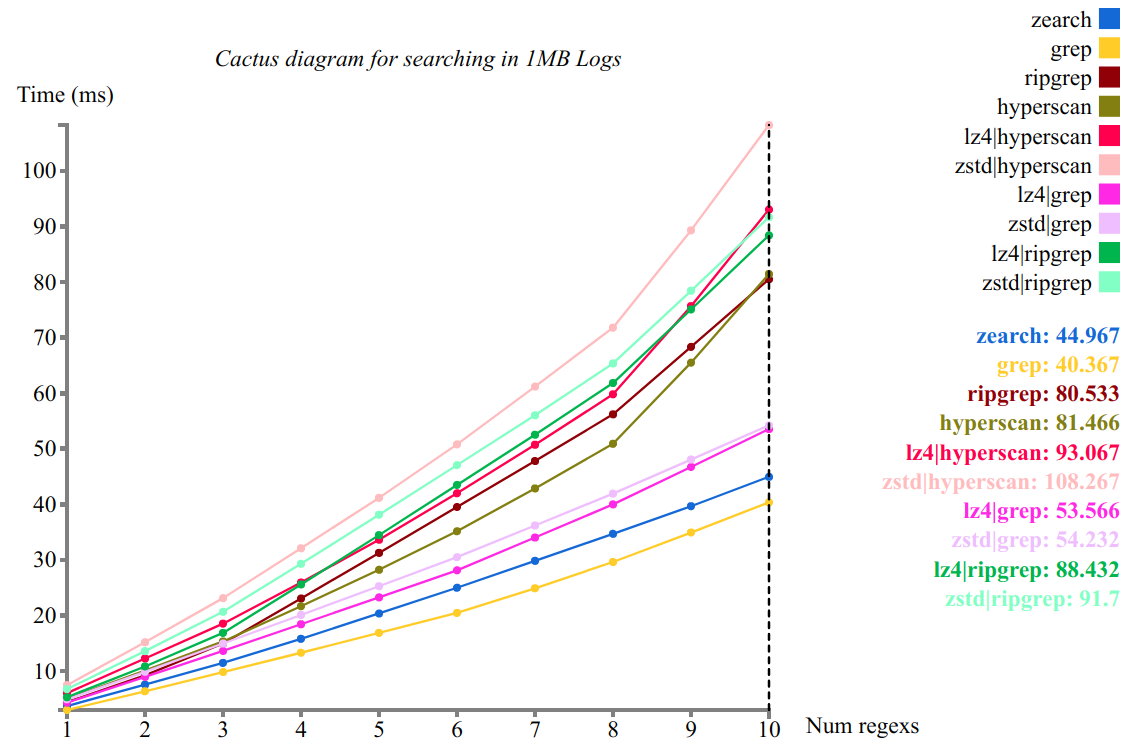}
\caption{The \emph{first graph} shows the time required to report the number of lines in a log file matching a regular expression.
All tools are fed with the same regular expression.
The decompress and search approach is implemented in parallel i.e. searching on the output uncompressed text as it is recovered by the decompressor.
As a reference, we show the time required for decompressing the file with different tools (horizontal lines).
The \emph{second graph} is the \emph{cactus plot} corresponding the data from the first graph. 
In this case, we observe that \tool{zearch} is faster than any other tool, except \tool{grep}.}\label{fig:webGraphs}
\end{figure}

All tools for regular expression searching considered in this benchmark are used to count the matching lines without reporting them.
As expected, all tools report the exact same result for all benchmarks.
To simplify the terminology, we refer to counting the matching lines as \emph{searching}, unless otherwise stated.

\subsection{Tools}
Our benchmark compares the performance of \tool{zearch} against the fastest implementations we found for the following operations:

\begin{myEnumI}
\item Searching the compressed text without decompression.
\item Searching the uncompressed text.
\item Decompressing the text without searching. 
\item Searching the uncompressed text as it is recovered by the decompressor.
\end{myEnumI}
\medskip

For searching the compressed text we consider \tool{GNgrep}, the tool developed by \citet{navarro2003regular} for searching on text compressed with the grammar-based compressor \tool{LZW} defined by \citet{welch1984technique}.
To the best of our knowledge, this is the only existing tool departing from the \emph{decompress and search} approach.

For searching uncompressed text we consider \tool{grep} and \tool{hyperscan}.
We improve the performance of \tool{grep} by compiling it without \emph{perl regular expression} compatibility, which is not supported by \tool{zearch}.
We used the library \tool{hyperscan} by means of the tool (provided with the library) \tool{simplegrep}, which we modified\footnote{\url{https://gist.github.com/pevalme/f94bedc9ff08373a0301b8c795063093}} to \emph{efficiently} read data either from stdin or an input file.
These tools are top of the class\footnote{\url{https://rust-leipzig.github.io/regex/2017/03/28/comparison-of-regex-engines/}} for regular expression searching.

For (de)compressing the files we use \tool{zstd} and \tool{lz4} which are among the best lossless compressors\footnote{\url{https://quixdb.github.io/squash-benchmark/}}, being \tool{lz4} considerably faster while \tool{zstd} achieves better compression.
We use both tools with the highest compression level, which has little impact on the time required for decompression.

We use versions \tool{grep v3.3}, \tool{hyperscan v5.0.0}, \tool{lz4 v1.8.3} and \tool{zstd v1.3.6} running in an Intel Xeon E5640 CPU 2.67 GHz with 20 GB RAM which supports SIMD instructions up to SSE4-2.
We restrict to ASCII inputs and set \verb!LC_ALL=C! for all experiments, which significantly improves the performance of \tool{grep}.
Since both \tool{hyperscan} and \tool{GNgrep} count positions of the text where a match ends, we extend each regular expression (when used with these tools) to match the whole line.
We made this decision to ensure all tools solve the same counting problem and produce the \emph{same output}.

\subsection{Files and Regular Expressions}
Our benchmark consists of an automatically generated \emph{Log}\footnote{\url{http://ita.ee.lbl.gov/html/contrib/NASA-HTTP.html}} of HTTP requests, English \emph{Subtitles}~\cite{openSubtitles}, and a concatenation of English \emph{Books}\footnote{\url{https://web.eecs.umich.edu/~lahiri/gutenberg_dataset.html}}.
Table~\ref{table:compression} shows how each compressor behaves on these files.
\begin{table}[!ht]
\centering
\renewcommand{\arraystretch}{0.8}%
\setlength{\tabcolsep}{5pt}
\setlength{\extrarowheight}{.2ex}
\resizebox{\textwidth}{!}{
\begin{tabular}{rr|r?r|r|r|r?r|r|r|r?r|r|r|r}
\toprule
& & & \multicolumn{4}{c?}{\textbf{Compressed size}} & \multicolumn{4}{c?}{\textbf{Compression time}} & \multicolumn{4}{c}{\textbf{Decompression time}} \\
& & \multicolumn{1}{c?}{\textbf{File}} & \multicolumn{1}{c|}{\tool{LZW}} & \multicolumn{1}{c|}{\tool{repair}} & \multicolumn{1}{c|}{\tool{zstd}} & \multicolumn{1}{c?}{\tool{lz4}} & \multicolumn{1}{c|}{\tool{LZW}} & \multicolumn{1}{c|}{\tool{repair}} & \multicolumn{1}{c|}{\tool{zstd}} & \multicolumn{1}{c?}{\tool{lz4}} & \multicolumn{1}{c|}{\tool{LZW}} & \multicolumn{1}{c|}{\tool{repair}} & \multicolumn{1}{c|}{\tool{zstd}} & \multicolumn{1}{c}{\tool{lz4}} \\
\midrule

\parbox[t]{2mm}{\multirow{6}{*}{\rotatebox[origin=c]{90}{\resizebox{65pt}{!}{\textbf{Uncompressed}}}}} & \parbox[t]{2mm}{\multirow{3}{*}{\rotatebox[origin=c]{90}{\small\textbf{1 MB}}}} & \textit{Logs} & \loser{0.19} & \second{0.08} & \winner{0.07} & 0.12 & \second{0.04} & 0.19 & \loser{0.51} & \winner{0.03} & \loser{0.02} & 0.01 & \second{0.01} & \winner{0.004}  \\
& & \textit{Subtitles} & \loser{0.36} & \second{0.13} & \winner{0.11} & 0.15 & \second{0.04} & 0.25 & \loser{0.3} & \winner{0.03} & \loser{0.02} & 0.01 & \second{0.01} & \winner{0.004}  \\
& & \textit{Books} & 0.42 & \second{0.34} & \winner{0.27} & \loser{0.43} & \winner{0.04} & 0.29 & \loser{0.42} & \second{0.08} & \loser{0.02} & 0.02 & \second{0.01} & \winner{0.004}  \\
\cmidrule{2-15}
 & \parbox[t]{2mm}{\multirow{3}{*}{\rotatebox[origin=c]{90}{\small\textbf{500 MB}}}} & \textit{Logs} & \loser{96} & \second{38} & \winner{33} & 65 & \second{16.9} & 123.2 & \loser{819.1} & \winner{13.3} & \loser{7.8} & 5.5 & \second{1.1} & \winner{0.64}  \\
& & \textit{Subtitles} & \loser{191} & \second{66} & \winner{55} & 114 & \winner{19.9} & 169.3 & \loser{415.2} & \second{22.8} & \loser{8.6} & 8.2 & \second{1.2} & \winner{0.81}  \\
& & \textit{Books} &206 & \second{153} & \winner{129} & \loser{216} & \winner{20.2} & 198.6 & \loser{646.3} & \second{40.6} & 8.6 & \loser{9.7} & \second{2.0} & \winner{0.8}  \\
\bottomrule
\end{tabular}}
\caption{Sizes (in MB) of the compressed files and (de)compression times (in seconds). Maximum compression levels enabled.
(Blue = best; bold black = second best; red = worst).}
\label{table:compression}
\end{table}

We first run each experiment 3 times as warm up so that the files are loaded in memory.
Then we measure the running time 30 times and compute the \emph{confidence interval} (with 95\% confidence) for the running time required to count the number of matching lines for a regular expression in a certain file using a certain tool.

We consider the \emph{point estimate} of the confidence interval and omit the \emph{margin of error} which never exceeds the $9\pct$ of the point estimate for the reported experiments.
The on-line version of these experiment \emph{does report} the margin of error as a black mark on the top of each bar.
The height of the bar is the point estimate computed for the given experiment while the black mark denotes the confidence interval (see Figure~\ref{fig:webGraphs}).
Figure~\ref{fig:comparison} summarizes the obtained results when considering, for all files, the regular expressions: ``\verb!what!'', ``\verb!HTTP!'', ``\verb!.!'', ``\verb!I .* you !'', ``\verb! [a-z]{4} !'', ``\verb! [a-z]*[a-z]{3} !'', ``\verb![0-9]{4}!'', ``\verb![0-9]{2}/(Jun|Jul|Aug)/[0-9]{4}!''.

For clarity, we report only on the most relevant tools among the ones considered.
For \tool{lz4} and \tool{zstd}, we report the time required to decompress the file and send the output to \tool{/dev/null}.

\begin{figure}[!ht]
\resizebox{\textwidth}{!}{
\begin{tikzpicture}
  \begin{groupplot}[
      group style={
          group name=my plots,
          group size=3 by 1,
          xlabels at=edge bottom,
          ylabels at=edge left,
          horizontal sep=0.6cm,
          vertical sep=1cm,
          },
      xtick={1,5,10,25,50,100,250,500},
      ymajorgrids=true,
      xlabel near ticks,
      ylabel near ticks,
      xlabel={Uncompressed size (MB)},
      ylabel={Time (ms)},
      grid style=dashed,
      width=0.5\textwidth,
      xmode=log,
      ymode=log,
      log ticks with fixed point,
      tick label style={font=\tiny},
      legend columns=-1,
      legend style={column sep=0.7ex},
  ]
\nextgroupplot[xmin=20, xmax=600, ymax=1000, ymin=30, yticklabels={15,25,50,100,250,500,1K,2K,4K,8K}, ytick={15,25,50,100,250,500,1000,2000,4000,8000}, legend to name=testLegend, title={\small \emph{Logs}}, title style={yshift=-7pt}]
\addplot[color=black, mark=*, mark size=1.3pt] coordinates{(1,1)};\addlegendentry{{\small\tool{zearch}}};
\addplot[color=blue, mark=o] coordinates{(1,1)};\addlegendentry{{\small\tool{grep}}};
\addplot[color=blue, mark=square] coordinates{(1,1)};\addlegendentry{{\small\tool{hyperscan}}};
\addplot[color=green, mark=square] coordinates{(1,1)};\addlegendentry{{\small\tool{zstd|hyperscan}}};
\addplot[color=green, mark=o] coordinates{(1,1)};\addlegendentry{{\small\tool{lz4|hyperscan}}};
\addplot[color=red, mark=o] coordinates{(1,1)};\addlegendentry{{\small\tool{lz4}}};
\addplot[color=red, mark=square] coordinates{(1,1)};\addlegendentry{{\small\tool{zstd}}};
\addplot[color=black, mark=square] coordinates{(1,1)};\addlegendentry{{\small\tool{GNgrep}}};
\addplot[color=black, mark=*, mark size=1.3pt] coordinates{(1,4.162)(5,9.442)(10,16.238)(25,33.638)(50,65.592)(100,131.692)(250,336.371)(500,718.913)};
\addplot[color=blue, mark=o] coordinates{(1,3.621)(5,8.654)(10,15.333)(25,35.058)(50,67.537)(100,131.483)(250,326.021)(500,669.133)};
\addplot[color=blue, mark=square] coordinates{(1,6.05)(5,11.262)(10,18.108)(25,37.467)(50,69.746)(100,132.85)(250,324.925)(500,633.954)};
\addplot[color=green, mark=square] coordinates{(1,8.696)(5,20.567)(10,35.438)(25,81.492)(50,159.208)(100,314.675)(250,735.458)(500,1430.696)};
\addplot[color=green, mark=o] coordinates{(1,7.083)(5,15.754)(10,25.133)(25,53.775)(50,100.746)(100,195.75)(250,479.75)(500,953.683)};
\addplot[color=red, mark=o] coordinates{(1,3.667)(5,7.667)(10,11.667)(25,23.333)(50,43.0)(100,81.333)(250,196.0)(500,400.0)};
\addplot[color=red, mark=square] coordinates{(1,4.667)(5,10.0)(10,18.333)(25,43.667)(50,87.667)(100,173.0)(250,381.333)(500,734.667)};
\addplot[color=black, mark=square] coordinates{(1,27.908)(5,57.292)(10,93.279)(25,202.146)(50,382.688)(100,734.075)(250,1810.258)(500,3802.442)};

\nextgroupplot[xmin=20, xmax=600, ymax=2000, ymin=40, yticklabels={15,25,50,100,250,500,1K,2K,4K,8K}, ytick={15,25,50,100,250,500,1000,2000,4000,8000}, title={\small \emph{Subtitles}}, title style={yshift=-5pt}]
\addplot[color=black, mark=*, mark size=1.3pt] coordinates{(1,5.608)(5,17.029)(10,33.975)(25,80.425)(50,154.429)(100,307.608)(250,802.433)(500,1629.858)};
\addplot[color=blue, mark=o] coordinates{(1,4.917)(5,13.783)(10,24.942)(25,59.725)(50,116.021)(100,227.662)(250,563.188)(500,1119.308)};
\addplot[color=blue, mark=square] coordinates{(1,6.592)(5,13.571)(10,22.408)(25,48.421)(50,90.763)(100,174.846)(250,427.65)(500,843.35)};
\addplot[color=green, mark=square] coordinates{(1,9.213)(5,23.221)(10,41.746)(25,97.817)(50,190.213)(100,375.288)(250,897.696)(500,1749.292)};
\addplot[color=green, mark=o] coordinates{(1,7.754)(5,17.833)(10,29.475)(25,63.15)(50,118.667)(100,230.312)(250,568.638)(500,1124.304)};
\addplot[color=red, mark=o] coordinates{(1,3.667)(5,8.0)(10,12.667)(25,24.0)(50,44.667)(100,84.0)(250,203.0)(500,395.333)};
\addplot[color=red, mark=square] coordinates{(1,4.333)(5,10.667)(10,21.0)(25,50.333)(50,97.667)(100,194.333)(250,446.333)(500,847.333)};
\addplot[color=black, mark=square] coordinates{(1,35.446)(5,96.308)(10,175.008)(25,409.642)(50,797.646)(100,1578.1)(250,3895.021)(500,7771.454)};

\nextgroupplot[xmin=20, xmax=600, ymax=4000, ymin=40, yticklabels={15,25,50,100,250,500,1K,2K,4K,8K}, ytick={15,25,50,100,250,500,1000,2000,4000,8000}, title={\small \emph{Books}}, title style={yshift=-5pt}]
\addplot[color=black, mark=*, mark size=1.3pt] coordinates{(1,9.688)(5,34.879)(10,67.888)(25,173.675)(50,342.996)(100,708.104)(250,1864.617)(500,3887.012)};
\addplot[color=blue, mark=o] coordinates{(1,3.967)(5,10.675)(10,19.717)(25,45.85)(50,88.692)(100,172.05)(250,427.342)(500,845.658)};
\addplot[color=blue, mark=square] coordinates{(1,5.796)(5,10.838)(10,17.504)(25,35.638)(50,65.487)(100,124.412)(250,302.933)(500,593.975)};
\addplot[color=green, mark=square] coordinates{(1,9.8)(5,26.608)(10,49.242)(25,121.112)(50,237.808)(100,477.337)(250,1178.146)(500,2351.9)};
\addplot[color=green, mark=o] coordinates{(1,6.875)(5,17.425)(10,28.337)(25,60.808)(50,113.312)(100,218.917)(250,536.704)(500,1061.717)};
\addplot[color=red, mark=o] coordinates{(1,4.0)(5,10.333)(10,15.667)(25,32.667)(50,58.0)(100,112.667)(250,267.667)(500,530.667)};
\addplot[color=red, mark=square] coordinates{(1,6.0)(5,17.0)(10,33.667)(25,87.333)(50,172.667)(100,349.667)(250,847.333)(500,1697.0)};
\addplot[color=black, mark=square] coordinates{(1,37.821)(5,103.808)(10,189.812)(25,454.75)(50,866.688)(100,1690.279)(250,4179.625)(500,8372.242)};
\end{groupplot}
\end{tikzpicture}
}
\begin{nscenter}
\resizebox{\textwidth}{!}{
\pgfplotslegendfromname{testLegend} 
}
\end{nscenter}

\caption{Average running time required to count the lines matching a regular expression in a file and time required for decompression.
Colors indicate whether the tool performs the search on the uncompressed text (blue); the compressed text (black); the output of the decompressor (green); or decompresses the file without searching (red).
}\label{fig:comparison}
\vspace{-10pt}
\end{figure}
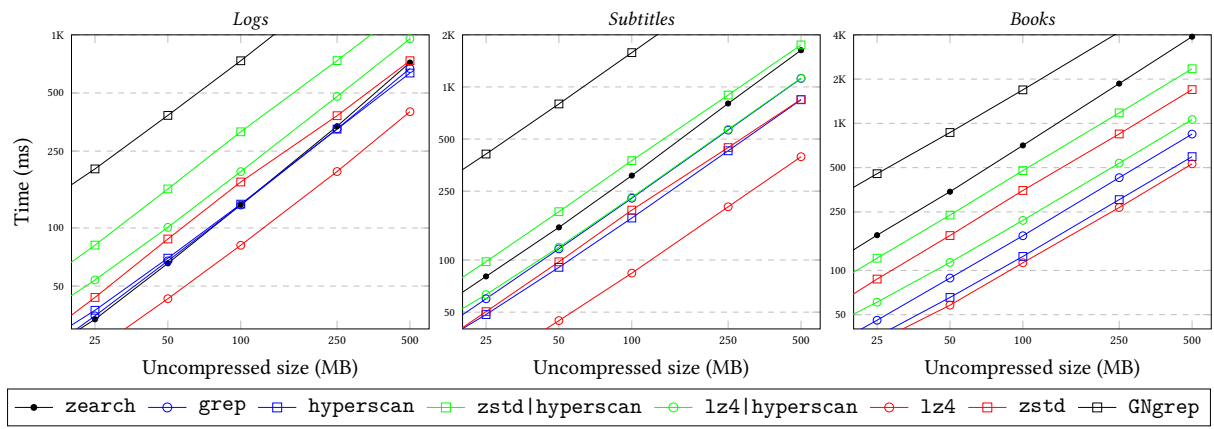

\subsection{Analysis of the Results.}
Figure~\ref{fig:comparison} and Table~\ref{table:compression} show that the performance of \tool{zearch} improves with the compression ratio.
This is to be expected since \tool{zearch} processes each grammar rule exactly once and better compression results in less rules to be processed.
In consequence, \tool{zearch} is the fastest tool for counting matching lines in compressed \emph{Log} files while it is the second slowest one for the \emph{Books}.

In particular, \tool{zearch} is more than $25\pct$ faster than any other tool working on compressed \emph{Log} files.
Actually \tool{zearch} is competitive with \tool{grep} and \tool{hyperscan}, even though these tools operate on the uncompressed text.
These results are remarkable since \tool{hyperscan}, unlike \tool{zearch}, uses algorithms specifically designed to take advantage of SIMD parallelization.\footnote{According to the documentation, \tool{hyperscan} \emph{requires}, at least, support for SSSE3.}

Finally, the fastest tool for counting matching lines in compressed \emph{Subtitles} and \emph{Books}, i.e. \linebreak\tool{lz4|hyperscan}, applies to files larger than the ones obtained when compressing the data with \tool{repair} (see Table~\ref{table:compression}).
However, when considering a better compressor such as \tool{zstd}, which achieves slightly more compression than \tool{repair}, the decompression becomes slower.
As a result, \tool{zearch} outperforms \tool{zstd|hyperscan} by more than $7\pct$ for \emph{Subtitles} files and $50\pct$ for \emph{Logs}.

\paragraph{Contrived Example}
Next, we discharge the full potential of our approach by considering a contrived experiment in which the data is highly repetitive.
In particular, we consider a file where all lines are identical and consist of the sentence ``\textrm{This is a contrived experiment.\NL}''.
Table~\ref{table:compressionContrived} shows the compression achieved on this data for each of the compressors.

\begin{table}[!ht]
\centering
\renewcommand{\arraystretch}{0.8}%
\setlength{\tabcolsep}{4pt}
\setlength{\extrarowheight}{.2ex}
\resizebox{\textwidth}{!}{
\begin{tabular}{r?r|r|r|r?r|r|r|r?r|r|r|r}
\toprule
 & \multicolumn{4}{c?}{\textbf{Compressed size}} & \multicolumn{4}{c?}{\textbf{Compression time}} & \multicolumn{4}{c}{\textbf{Decompression time}} \\
\multicolumn{1}{c?}{\textbf{File size}} & \multicolumn{1}{c|}{\tool{LZW}} & \multicolumn{1}{c|}{\tool{repair}} & \multicolumn{1}{c|}{\tool{zstd}} & \multicolumn{1}{c?}{\tool{lz4}} & \multicolumn{1}{c|}{\tool{LZW}} & \multicolumn{1}{c|}{\tool{repair}} & \multicolumn{1}{c|}{\tool{zstd}} & \multicolumn{1}{c?}{\tool{lz4}} & \multicolumn{1}{c|}{\tool{LZW}} & \multicolumn{1}{c|}{\tool{repair}} & \multicolumn{1}{c|}{\tool{zstd}} & \multicolumn{1}{c}{\tool{lz4}} \\
\midrule

1MB & \loser{13} & \winner{0.072} & \second{0.135} & 4.1 & 0.01 & \loser{0.08} & \second{0.01} & \winner{0.004} & \loser{0.01} & 0.01 & \second{0.003} & \winner{0.003}   \\
\midrule
500MB & 950 & \second{0.09} & \winner{44} & \loser{2000} & 14.5 & \loser{53.1} & \second{0.99} & \winner{0.28} & \loser{3.9} & 3.3 & \second{0.24} & \winner{0.2}  \\
\bottomrule
\end{tabular}}
\caption{Sizes (in KB) of the compressed files and (de)compression times (in seconds). Maximum compression levels enabled.
(Blue = best; bold black = second best; red = worst).}
\label{table:compressionContrived}
\end{table}

As expected this contrived file results in really high compression ratios.
As we show next, this scenario evidences the virtues of \tool{zearch} which is capable of searching in 500MB of data by processing a grammar consisting of 57 rules.

Table~\ref{table:comparisonContrived} summarizes the results obtained when searching the 500 MB contrived file for different regular expressions.\footnote{We run each experiment 30 times and report the point estimate of the confidence interval with 95\pct\, confidence.}
For each expression, we report the time required to\begin{myEnumIL}
\item search on the compressed data without decompression,
\item search on the uncompressed data and 
\item search with the best implementation of the parallel decompress and search approach.\footnote{The best implementation might vary depending on the expression.}
\end{myEnumIL}

\begin{table}[!ht]
\centering
\renewcommand{\arraystretch}{0.8}%
\setlength{\tabcolsep}{4pt}
\setlength{\extrarowheight}{.2ex}
\resizebox{!}{50pt}{
\begin{tabular}{l?r|r?r|r?r}
\toprule
\multicolumn{1}{c?}{\textbf{Expression}} & \multicolumn{1}{c|}{\tool{zearch}} & \multicolumn{1}{c?}{\tool{GNgrep}} & \multicolumn{1}{c|}{\tool{grep}} & \multicolumn{1}{c?}{\tool{hyperscan}} & \multicolumn{1}{c}{\tool{decompress and search}} \\
\midrule
``\tool{experiment}'' & \winner{2.267} & \loser{14K} & \second{1352} & \loser{1784} & 1652 \\
\cmidrule{1-6}
``\tool{This}'' & \winner{2.533} & \loser{14K} & \second{764} & 2166 & 959 \\
\cmidrule{1-6}
``\tool{.}'' & \winner{2.467} & \loser{14K} & \second{703} & \loser{1276} & 886  \\
\cmidrule{1-6}
``\tool{[a-z]\{4\}}'' & \winner{2.667} & \loser{14K} & \second{1138} & 1270 & 1360  \\
\cmidrule{1-6}
``\tool{[a-z]\{11\}}'' & \winner{2.233} & \second{37} & \loser{1690} & 1312 & 1397  \\
\cmidrule{1-6}
``\tool{That}'' & \winner{2.433} & \second{37.2} & \loser{607} & 239 & 444  \\
\bottomrule
\end{tabular}}
\caption{Time (ms) required to report the number of lines matching a regular expression in the 500 MB large contrived file. 
(Blue = fastest; bold black = second fastest; red = slowest).}
\label{table:comparisonContrived}
\end{table}

As shown by Tables~\ref{table:compressionContrived} and~\ref{table:comparisonContrived}, our tool is about \emph{10 times faster} at searching than \tool{lz4} at decompression.
Therefore, \tool{zearch} clearly outperforms any decompress and search approach, even if decompression and search are done in parallel.
This is to be expected since \tool{zearch} only needs to process 90 Bytes of data (the size of the grammar) while the rest of the tools need to process 500 MB.

Similarly, \tool{GNgrep} processes 950 KB of data (the size of the \tool{LZW}-compressed data).
As a consequence, when there are no matches of the expression, \tool{GNgrep} is faster than decompression as evidenced by the last two rows of Table~\ref{table:comparisonContrived}.
However, \tool{GNgrep} reports the number of matching lines by explicitly finding the positions in the data where the match begins, which results rather inefficient when all lines of the file contain a match, as evidenced by the first 4 rows of Table~\ref{table:comparisonContrived}.

\section{Fine-Grained Analysis of the Implementation}\label{sec:complexity}

The grammars produced by \tool{repair} break the definition of SLP in behalf of compression by allowing the axiom rule to have more than two symbols on the right hand side.
This is due to the fact that the axiom rule is built with the remains of the input text after creating all grammar rules.

Typically, the length of the axiom is larger or equal than the number of rules in the SLP so the way in which the axiom is processed heavily influences the performance of \tool{zearch}.

On the other hand, our experiments show that the performance of \tool{zearch} is typically far from its worst case complexity.
This is because the worst case scenario assumes each string generated by a grammar variable labels a path between each pair of states of the automaton.
However, we only observed such behavior in contrived examples.

\subsection{Processing the Axiom Rule.}%

Algorithm \AlgCountLines could process the axiom rule $X_{\len{V}} \to σ$ by building an SLP with the set of rules
\[\{S_{1} \to (σ)_1(σ)_2\} \cup \{S_i \to S_{i{-}1}(σ)_{i{+}1} \mid i = 2\ldots \len{σ}{-}2\} \cup \{X_{\len{V}} \to S_{\len{σ}{-}2}(σ)_\dag\} \enspace .\]
However it is more efficient to compute the set of states reachable from the initial ones when reading the string generated with $S_1$ and update this set for each symbol $(σ)_i$.
To perform the counting note that $\counting{S_i}$ is only used to compute $\counting{S_{i+1}}$ and can be discarded afterwards.
This yields an algorithm running in $\mathcal{O}\left(\len{V}\cdot s^3{+}\len{σ}\cdot s^2\right)$ time using $\mathcal{O}\left(\len{V}\cdot s^2\right)$ space where $\len{V}$ is the number of rules of the input grammar and $X_{\len{V}} \to σ$ its axiom.

\subsection{Number of Operations Performed by the Algorithm}
Define $s_{τ,q}=\len{\{q' \mid (q,τ,q') \in δ\}}$ and $s_τ = \sum_{q \in Q}s_{τ,q}$ and let us recall the complexity of Algorithm \AlgCountLines according to the data structures described in Section~\ref{sec:DataStructures}.
The algorithm iterates over the $\len{V}$ rules of the grammar and, for each of them: 
\begin{myEnumI}
\item Initializes matrix $\mathcal{K}$ with $s_{β_\ell}$ elements\footnote{We need to set up to $s$ sentinel values for the rows in $\mathcal{K}$ not used for storing $s_{β_\ell}$}
\item Iterates through $\mathcal{K}[q'][0…s_{β_\ell,q'}]$ for each pair $(q_1,q') \in \edges{α_\ell}$.  
\end{myEnumI}
\smallskip
Then it processes the axiom rule iterating, for each symbol $(σ)_i$, through $s_{(σ)_i}$ transitions.

These are all the operations performed by the algorithm with running time dependent on the size of the input.
Hence, Algorithm \AlgCountLines runs in 
\[\mathcal{O}\left(\sum_{\ell=1}^{\len{V}}\tilde{s}_\ell+\sum_{i=1}^\len{σ}s_{(σ)_i} \right) \text{ time, where } \tilde{s}_\ell=s_{β_\ell}+s+\hspace{-10pt}\sum_{(q_1,q') \in \edges{α_\ell}}\hspace{-10pt}\left(1{+}s_{β_\ell,q'}\right)\enspace p.\]

Note that $\tilde{s}_\ell \leq s^3$, $s_{(σ)_i} \leq s^2$ so the worst case time complexity of the algorithm is $\mathcal{O}(\len{V}\cdot s^3+\len{σ}\cdot s^2)$.
However, in the experiments we observed that $\tilde{s}_\ell$ and $s_{(σ_i)}$ are usually much smaller than $s^3$ and $s^2$, respectively, as reported in Table~\ref{table:Behavior}.

\begin{table}[!ht]
\centering
\resizebox{\textwidth}{!}{
\renewcommand{\arraystretch}{0.60}%
\setlength{\tabcolsep}{4pt}
\setlength{\extrarowheight}{1ex}
\small
\begin{tabular}{lr?r|rrrrr?r|rrrrr}
\toprule
\multicolumn{1}{c}{\multirow{2}{*}{\textbf{Expression}}} & \multicolumn{1}{c?}{\multirow{2}{*}{$s$}}  & \multicolumn{1}{c|}{\multirow{2}{*}{$s^3$}} & \multicolumn{5}{c?}{percentiles for $\tilde{s}_{\ell}$} & \multicolumn{1}{c|}{\multirow{2}{*}{$s^2$}} & \multicolumn{5}{c}{percentiles for $s_{(σ)_i}$} \\
\multicolumn{1}{c}{} & \multicolumn{1}{c?}{} & \multicolumn{1}{c|}{} & {\footnotesize $50$\pct} & {\footnotesize $75$\pct} & {\footnotesize $95$\pct} & {\footnotesize $98$\pct} & {\footnotesize$100$\pct} & \multicolumn{1}{c|}{} & {\footnotesize $50$\pct} & {\footnotesize $75$\pct} & {\footnotesize $95$\pct} & {\footnotesize $98$\pct} & {\footnotesize$100$\pct}\\
\midrule
{\small ``\tool{what}''} & 5 & 125 & 0 & 0 & 1 & 1 & 9 & 25 & 0 & 0 & 1 & 1 & 2 \\
{\small ``\tool{HTTP}''} & 5 & 125 & 0 & 0 & 0 & 0 & 10 & 25 & 0 & 0 & 0 & 0 & 2 \\
{\small ``\tool{.}''} & 2 & 8 & 0 & 0 & 0 & 0 & 4 & 4 & 0 & 0 & 0 & 0 & 1 \\
{\small ``\tool{I .* you} ''} & 9 & 729 & 3 & 13 & 16 & 18 & 29 & 81 & 3 & 3 & 5 & 5 & 9 \\
{\small ``\tool{ [a-z]{4} }''} & 7 & 343 & 2 & 10 & 11 & 12 & 18 & 49 & 1 & 1 & 2 & 2 & 4 \\
{\small ``\tool{ [a-z]*[a-z]{3} }''} & 7 & 343 & 3 & 11 & 14 & 18 & 31 & 49 & 1 & 3 & 3 & 4 & 8 \\
{\small ``\tool{ [0-9]\{4\}}''} & 6 & 216 & 8 & 8 & 8 & 8 & 18 & 36 & 1 & 1 & 1 & 1 & 5 \\
{\small ``\tool{.*[A-Za-z ]\{5\}}''} & 7 & 343 & 14 & 25 & 48 & 48 & 48 & 49 & 11 & 14 & 14 & 14 & 14 \\
{\small ``\tool{.*[A-Za-z ]\{10\}}''} & 12 & 1728 & 29 & 51 & 86 & 95 & 98 & 144 & 16 & 26 & 29 & 29 & 29 \\
{\small ``\tool{.*[A-Za-z ]\{20\}}''} & 22 & 10648 & 57 & 87 & 132 & 153 & 198 & 484 & 23 & 38 & 52 & 58 & 59 \\
{\small ``\tool{((((.)*.)*.)*.)*}''} & 6 & 216 & 12 & 29 & 209 & 209 & 209 & 36 & 29 & 29 & 29 & 29 & 29 \\
{\small ``\tool{(((((.)*.)*.)*.)*.)*}''} & 7 & 343 & 14 & 34 & 249 & 249 & 249 & 49 & 34 & 34 & 34 & 34 & 34 \\
\bottomrule
\end{tabular}
}
\caption{Analysis of the values $\tilde{s}_\ell$ and $s_{(σ)_i}$ obtained when considering different regular expressions to search \emph{Subtitles} (100 MB uncompressed long).
The fifth column of the fourth row indicates that when considering the expression ``\tool{I .* you}'', for 75\% of the grammar rules we have $\tilde{s}_\ell \leq 13$ while $s^3=729$.}
\label{table:Behavior}
\end{table}

As the experiments show, \tool{zearch} exhibits almost linear behavior with respect to the size of the automaton built from the expression.
Nevertheless, there are regular expressions that trigger the worst case behavior (last two rows in Table~\ref{table:Behavior}), which cannot be avoided due to the result of \citet{Amir2018FineGrained} described before.

\section{Fine-Grained Complexity}
In Section~\ref{sec:provingCorrectness} we obtained upper bounds for the worst-case time complexity of our algorithm depending on whether the automata built from the expression is an NFA or a DFA.

However, we observed in Section~\ref{sec:complexity} that the actual behavior of our implementation is, in general, far from its worst-case scenario (see Table~\ref{table:Behavior}).
This is due to the fact that the worst-case scenario assumes an NFA where each pair of states are connected by a transition for each symbol in the alphabet but this is rarely the type of automata obtained from non-contrived regular expressions.

This difference between the worst-case time complexity of the algorithm and its behavior in practice also appears when considering the problem of searching with regular expressions on plain text.
Indeed, this problem led \citet{Backurs2016Hard} to analyze the complexity of searching on plain text for different classes of regular expressions.

In their work, \citet{Backurs2016Hard} restrict themselves to \demph{homogeneous regular expressions}, i.e. regular expressions in which operators at the same level of the formula are equal\footnote{Write the regular expression as a tree. The expression is homogeneous if all non-leaf nodes at the same depth have the same label.}, which are grouped in classes depending on the sequence of the operators involved.
Then, they obtain a lower bound for the search complexity for each class of expressions by building reductions from the \demph{Orthogonal Vector Problem}\index{OVP} (OVP for short) which, given two sets of vectors \(A,B \subseteq \{0,1\}^d\) in \(d\) dimensions, with \(N\) and \(M\) elements respectively, asks whether there exists \(a\in A\) and \(b \in B\) such that \(a\cdot b =0\).

\begin{conjecture*}[OV Conjecture~\cite{bringmann2015quadratic}]
There are no reals \(\varepsilon, d > 0\) such that the OVP in \(d < N^{o(1)}\) dimensions with \(M = \Theta(N^α)\) for \(α \in (0,1]\) can be solved in \(\mathcal{O}((N\cdot M)^{1-\varepsilon})\) time.
\end{conjecture*}

The idea behind the conjecture is that any algorithm defying it would yield an algorithm for SAT violating the Strong Exponential Time Hypothesis.

\citet{Backurs2016Hard} relied on the OV conjecture to determine whether a search problem is \emph{easy}\index{easy searching problem}, i.e. there is an algorithm running in $\mathcal{O}(T+s)$ time where $T$ is the size of the input text and $s$ is the number of states of the automaton, or \emph{hard}\index{hard searching problem}, i.e. assuming the Strong Exponential Time Hypothesis (SETH) any algorithm has $Ω((T\cdot s)^{1-ε})$ time complexity with $ε > 0$.

This analysis can be extended to consider searching on compressed text and decide whether our implementation is optimal on different classes of homogeneous regular expressions.
To do that, we apply the following remark, inherited from~\citet{Amir2018FineGrained}, who used the OVP to analyze whether the decompress and solve approach can be outperformed by manipulating the compressed text for different problems.

\begin{remark}\label{remark:OVP}
Let $A=\{a_1,…,a_{N}\}\subseteq\{0,1\}^d$ and $B=\{b_1,…,b_{M}\} \subseteq\{0,1\}^d$ be an instance of the OVP in $d \leq N^{o(1)}$ dimensions with \(M = \mathcal{O}(N)\).
We define a string \(T\) with a representation as an SLP of size $t = \mathcal{O}(N\cdot d)$ and a regular expression \(π\) of size \(s = \mathcal{O}(M\cdot d)\) such that the string contains a match for the expression if{}f we have a solution for the OVP.

If there is an algorithm for regular expression searching on compressed text that operates, for a class of regular expressions that includes \(π\), in $\mathcal{O}((t\cdot s)^{1{-}ε})$ with \(\epsilon > 0\) then it would solve the OVP in $\mathcal{O}((N\cdot M)^{1{-}ε})$ (since the dimension is fixed) which contradicts the OV conjecture.
\end{remark}

\subsection{Complexity of Searching on Compressed Text}

Given a regular expression, we say it is \emph{homogeneous of type}\index{homogeneous expression} ``\verb!|+!'' if{}f the regular expression is a disjunction of \verb!+! operators and terminals.
We extend this notation to any combination of operators.
For instance, the expressions ``\tool{a+b+}'' and ``\tool{a+b}'' are homogeneous of type ``\tool{·+}'' while ``\tool{a+b*}'' is not homogeneous.
Recall that the \emph{size} of a regular expression is the number of operators and terminals used to define the expression.
For instance, ``\tool{a+b+}'' and ``\tool{a+b}'' have size 4 and 3, respectively.

The following three results use Remark~\ref{remark:OVP} to show that the time complexity of regular expression searching on compressed text is \(Ω(t\cdot s)\), where \(t\) is the size of the SLP and \(s\) is the size of the expression, when the regular expression is homogeneous of type ``\verb!·+!'', ``\verb!·*!'' or ``\verb!·|!''.

\begin{theorem}\label{theorem:HomogeneousDotPlus}
There is no algorithm for searching with a regular expression on grammar-compressed text that operates in \(\mathcal{O}((t\cdot s)^{1-\epsilon})\) time with \(\varepsilon > 0\), where \(t\) is the size of the compressed text and \(s\) is the size of the regular expression, when the expression is homogeneous of type ``\verb!·+!''.
\end{theorem}
\begin{proof}
Let $A=\{a_1,…,a_{N}\}\subseteq\{0,1\}^d$ and $B=\{b_1,…,b_{M}\} \subseteq\{0,1\}^d$ be an instance of the OVP in $d \leq N^{o(1)}$ dimensions with \(M = \mathcal{O}(N)\).
Without loss of generality, assume the dimension is even.
Consider the regular expression 
\[π \ud \text{``} F(b_1)zF(b_2)z…zF(b_M)z\text{''}\]
on the alphabet \(Σ = \{x,y,z\}\) with
\(\def\arraystretch{0.3}
F(b_i) \ud f(b_i,1)f(b_i,2),…,f(b_i,d)\) and
\[f(b,j) \ud \left\{ \begin{array}{lcc}
     xx^+ & \text{if}  & (b)_j = 1 \text{ and \(j\) is even}  \\
     x^+ & \text{if}  & (b)_j = 0 \text{ and \(j\) is even} \\
     yy^+ & \text{if}  & (b)_j = 1 \text{ and \(j\) is odd} \\
     y^+ & \text{if}  & (b)_j = 0 \text{ and \(j\) is odd}
     \end{array}
\right. \enspace ,\]
where \((b)_j\) is the \(j\)-th component of the vector \(b\).
Clearly, \(π\) is homogeneous of type ``\verb!·+!'' and has size \(s = \mathcal{O}(M \cdot d)\).

Now, we define an SLP $\cP$ on \(Σ = \{x,y,z\}\) such that $\lang{\cP}=\{w\}$ with
\[w \ud \left((xxyy)^{\frac{d}{2}}z\right)^{M{-}1}\tilde{F}(a_1)z…\left((xxyy)^{\frac{d}{2}}z\right)^{M{-}1}\tilde{F}(a_N)z\left((xxyy)^{\frac{d}{2}}z\right)^{M{-}1} \enspace\]
where \(\tilde{F}(a_i) \ud \tilde{f}(a_i,1)\tilde{f}(a_i,2),…,\tilde{f}(a_i,d)\) and 
\[\tilde{f}(a,j) \ud \left\{ \begin{array}{lcc}
     x & if  & (a)_j = 1 \text{ and \(j\) is even} \\
     xx & if  & (a)_j = 0 \text{ and \(j\) is even} \\
     y & if  & (a)_j = 1 \text{ and \(j\) is odd} \\
     yy & if  & (a)_j = 0 \text{ and \(j\) is odd}
     \end{array}
\right. \enspace .\]
The substring $\left((xxyy)^{\frac{d}{2}}z\right)^{M{-}1}$ can be generated with an SLP of size $\mathcal{O}(d+\log M)$, hence $w$ can be compressed as an SLP $\cP$ of size $t =\mathcal{O}(N\cdot d + d + \log M)$ and, since \(d \leq N^{o(1)}\) is a constant and \(M = \mathcal{O}(N)\), we find that $t = \mathcal{O}(N)$.

Clearly, \(π\) and \(\cP\) can be built in \(\mathcal{O}(M \cdot d)\) and \(\mathcal{O}(N \cdot d)\) time, respectively.

Finally, we show that there exists $a \in A$, $b \in B$ such that $a⋅b = 0$ if{}f there is a factor of $w$ that matches $π$.
Let $a_{i_1} \in A$ and $b_{i_2} \in B$.
Then $a_{i_1}⋅b_{i_2} = 0$ if{}f 
\begin{myEnumI}
\item The factor \(\left((xxyy)^{\frac{d}{2}}z\right)^{i_2-1}\) of \(w\) that \emph{precedes} the factor \(\tilde{F}(a_{i_1})z\) matches the subexpression \linebreak``\(F(b_1)z\ldots F(b_{i_2{-}1})z\)''.
\item The factor \(\tilde{F}(a_{i_1})z\) of \(w\) matches the subexpression ``\(F(b_{i_2})z\)''.
\item The factor \(\left((xxyy)^{\frac{d}{2}}z\right)^{M{-}i_2}\) of \(w\) that \emph{succeeds} the factor \(\tilde{F}(a_{i_1})z\) matches the subexpression ``\(F(b_{i_2{+}1})z\ldots F(b_M)z\)''.
\end{myEnumI}
It follows from Remark~\ref{remark:OVP} that there is no algorithm for searching with an homogeneous regular expression of type ``\verb!·+!'' working on $\mathcal{O}((t\cdot s)^{1{-}ε})$ time.

Finally, note that if the dimension of the OVP is odd then it suffices to replace the \(\left((xxyy)^{\frac{d}{2}}z\right)^{M{-}1}\) factors from \(w\) by \(\left((xxyy)^{\frac{d{-}1}{2}}xxz\right)^{M{-}1}\).
\end{proof}

Note that for any homogeneous regular expression of type ``\verb!·+!'' of size \(s\), we can build in \(\mathcal{O}(s)\) time an equivalent homogeneous regular expression of type ``\verb!·*!'' and size \(\mathcal{O}(s)\).
Therefore, we obtain the following corollary from Theorem~\ref{theorem:HomogeneousDotPlus}.

\begin{corollary}\label{corol:HomogeneousDotStar}
There is no algorithm for searching with a regular expression on grammar-compressed text that operates in \(\mathcal{O}((t\cdot s)^{1-\epsilon})\) with \(\varepsilon > 0\), where \(t\) is the size of the compressed text and \(s\) is the size of the regular expression, when the expression is homogeneous of type ``\verb!·*!''.
\end{corollary}

\begin{theorem}\label{theorem:HomogeneousDotOr}
There is no algorithm for searching with regular expressions on grammar-compressed text that operates in \(\mathcal{O}((t\cdot s)^{1-\epsilon})\) with \(\varepsilon > 0\), where \(t\) is the size of the compressed text and \(s\) is the size of the regular expression, when the expression is homogeneous of type ``\verb!·|!''.
\end{theorem}
\begin{proof}
The proof is identical to that of Theorem~\ref{theorem:HomogeneousDotPlus} but considering the expression
\[π \ud \text{``}F(b_1)zF(b_2)z…zF(b_M)z\text{''}\]
on the alphabet \(Σ=\{0,1,z\}\) with
\[F(b_i)=f(b_i,1)f(b_i,2),…,f(b_i,d) \;\text{ and }\; f(b, j) = \left\{ \begin{array}{lcc}
     0 & if  & (b)_j = 1 \\
     0|1 & if  & (b)_j = 0
     \end{array}
\right.\]
and the word
\[w \ud \left(0^dz\right)^{M{-}1}a_1z\left(0^dz\right)^{M{-}1}a_2z…\left(0^dz\right)^{M{-}1}a_Nz\left(0^dz\right)^{M{-}1} \enspace .\]
Note that, unlike the proof of Theorem~\ref{theorem:HomogeneousDotPlus}, this proof does not depend on the parity of the dimension of the OVP.
\end{proof}

\subsection{Complexity of Our Implementation}
In the following, we analyze the complexity of the implementation of Algorithm \AlgCountLines described in Section~\ref{sec:implementation} when the input regular expression is homogeneous of type ``\verb!·+!'', ``\verb!·*!'' or ``\verb!·|!''

As explained in Section~\ref{sec:implementation}, \tool{zearch} uses \tool{libfa}, which applies Thompson's algorithm~\cite{thompson1968programming} with on-the-fly $ε$-removal, to build an NFA for the input regular expression.

However, given a regular expression of size \(s\) we can decide in \(\mathcal{O}(s)\) time whether a expression is homogeneous of type ``\verb!·+!'', ``\verb!·*!'' or ``\verb!·|!'' and, as we show next, use a specialized algorithm for building a DFA with \(\mathcal{O}(s)\) states in \(\mathcal{O}(s)\) time for the given expression.
Therefore, \tool{zearch} admits a straightforward modification that allows it to search on grammar-compressed text with homogeneous regular expressions of type ``\verb!·+!'', ``\verb!·*!'' or ``\verb!·|!'' in \(\mathcal{O}(t \cdot s)\) time, where \(s\) is the size of the expression.
Next, we show how to build such DFAs from the given regular expressions.

First, observe that every homogeneous regular expression of type ``\verb!·+!'' of size \(s\) such that it contains no concatenation of the form ``\verb!a+a+!'' can be captured by a DFA with \(s{+}1\) states as we show next.
Let \(a_1,\ldots,a_n\) be the sequence of letters that appear in an homogeneous expression of type ``\verb!·+!''.
Then, 
\[\cD = \tuple{\{q_i \mid 0 \leq i \leq n\}, Σ, \{(q_i,a_i,q_i), (q_{i{-}1},a_{i},q_{i}) \mid 1 \leq i \leq n\}, \{q_1\}, \{q_{n}\}}\]
is a DFA for the given expression.

If the expression contains a concatenation of the form ``\tool{a+a+}'' then \(\cD\) is no longer deterministic.
In that case, we can replace ``\tool{a+a+}'' by ``\tool{aa+}'' and, therefore, remove from \(\cD\) the self-loop corresponding to the first \(a\).
It is straightforward to check that this change results in a deterministic automaton and it does not alter the generated language, hence it does not alter the result of the search. 
Figure~\ref{fig:DFAc+} shows the DFA for an homogeneous regular expression of type ``\verb!·+!''.

\begin{figure}[!ht]
\centering
\begin{tikzpicture}[->,>=stealth',shorten >=1pt,auto,node distance=5mm and 1cm,thick,initial text=]
\tikzstyle{every state}=[scale=0.75,fill=customblue!60,draw=blue!60,text=black]
\node[initial,state] (0) {\(q_0\)};
\node[state] (1) [right=of 0] {\(q_1\)};
\node[state] (2) [right=of 1] {\(q_2\)};
\node[state] (3) [right=of 2] {\(q_3\)};
\node[state] (4) [right=of 3] {\(q_4\)};
\node[state, accepting] (5) [right=of 4] {\(q_5\)};

\path 
(0) edge node {\(a\)} (1)
(1) edge node {\(b\)} (2)
(2) edge node {\(b\)} (3)
(3) edge node {\(a\)} (4)
(4) edge node {\(c\)} (5)
(1) edge [loop above] node {\(a\)} (1)
(3) edge [loop above] node {\(b\)} (3)
(4) edge [loop above] node {\(a\)} (4)
(5) edge [loop above] node {\(c\)} (5);
\end{tikzpicture}
\caption{DFA for the regular expression ``\tool{a+b+b+a+c+}''=``\tool{a+bb+a+c+}''.}\label{fig:DFAc+}
\end{figure}
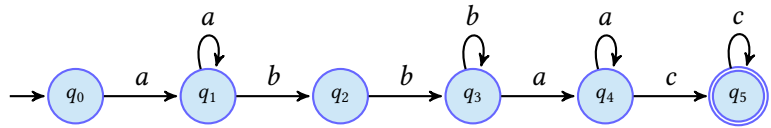

On the other hand, let \(a_1,\ldots,a_n\) be the sequence of letters that appear in an homogeneous expression of type ``\verb!·*!''.
For every \(a\), let \(j_a^k\) be the smallest index such that \(k \leq j_a^k \leq n \) and \(a = a_{j_a^k}\).
Then, the DFA obtained by making every state of \(\cD\) final and adding the transitions \(\{(q_{i-1},a_k,q_{j_{a_k}^i}) \mid 1 \leq i \leq k \leq n\}\), is an automaton for the given expression.
Note that, if the expression contains a concatenation of the form ``\tool{a*a*}'', which will break the determinism of our automata, then we can safely replace it by ``\tool{a*}''.
Figure~\ref{fig:DFAc*} shows the DFA for an homogeneous regular expression of type ``\verb!·*!''.

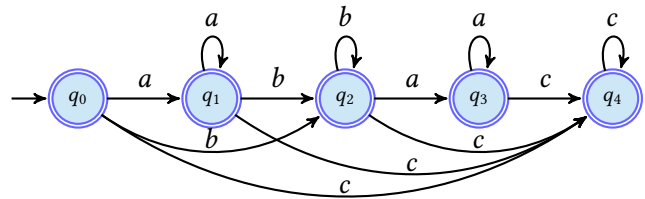
\begin{figure}[!ht]
\centering
\begin{tikzpicture}[->,>=stealth',shorten >=1pt,auto,node distance=5mm and 1cm,thick,initial text=]
\tikzstyle{every state}=[scale=0.75,fill=customblue!60,draw=blue!60,text=black]
      
\node[initial, state, accepting] (0) {\(q_0\)};
\node[state, accepting] (1) [right=of 0] {\(q_1\)};
\node[state, accepting] (2) [right=of 1] {\(q_2\)};
\node[state, accepting] (3) [right=of 2] {\(q_3\)};
\node[state, accepting] (4) [right=of 3] {\(q_4\)};

\path 
(0) edge node {\(a\)} (1)
(1) edge node {\(b\)} (2)
(2) edge node {\(a\)} (3)
(3) edge node {\(c\)} (4)
(1) edge [loop above] node {\(a\)} (1)
(2) edge [loop above] node {\(b\)} (2)
(3) edge [loop above] node {\(a\)} (3)
(4) edge [loop above] node {\(c\)} (4)
(0) edge [bend right=35] node [yshift=-3pt] {\(b\)} (2)
(0) edge [bend right=35] node [yshift=-3pt] {\(c\)} (4)
(1) edge [bend right=35] node [yshift=-3pt] {\(c\)} (4)
(2) edge [bend right=35] node [yshift=-3pt] {\(c\)} (4);
\end{tikzpicture}
\caption{DFA for the regular expression ``\tool{a*b*b*a*c*}''=``\tool{a*b*a*c*}''.}\label{fig:DFAc*}
\end{figure}

Finally, it is straightforward to build a DFA for an homogeneous regular expression of the type ``\verb!·|!'' with \(n{+}1\) states where \(n\) is the number of concatenations.
Figure~\ref{fig:DFAco} shows the DFA for an homogeneous regular expression of type ``\verb!·|!''.

\begin{figure}[!ht]
\centering
\begin{tikzpicture}[->,>=stealth',shorten >=1pt,auto,node distance=5mm and 1cm,thick,initial text=]
\tikzstyle{every state}=[scale=0.75,fill=customblue!60,draw=blue!60,text=black]
      
\node[initial,state] (0) {\(q_0\)};
\node[state] (1) [right=of 0] {\(q_1\)};
\node[state] (2) [right=of 1] {\(q_2\)};
\node[state] (3) [right=of 2] {\(q_3\)};
\node[state, accepting] (4) [right=of 3] {\(q_4\)};

\path 
(0) edge node {\(a,b\)} (1)
(1) edge node {\(a,c\)} (2)
(2) edge node {\(b,c\)} (3)
(3) edge node {\(a,c\)} (4);
\end{tikzpicture}
\caption{DFA for the regular expression ``\tool{(a|b)(a|c)(b|c)(a|c)}''.}\label{fig:DFAco}
\end{figure}
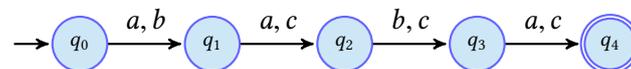

It is worth to remark that the DFA of Figure~\ref{fig:DFAco} is the result of applying Thompson's construction on the input expression.
As a consequence, \tool{zearch} already builds a DFA when the input expression is homogeneous of type ``\verb!·|!'' and, therefore, it performs the search in \(\mathcal{O}(t \cdot s)\).

We conclude that \tool{zearch} admits a straightforward modification to exhibit \(\mathcal{O}(t\cdot s)\) time complexity when working on homogeneous regular expressions of types ``\verb!·+!'', ``\verb!·*!'' and ``\verb!·|!'' and, therefore, be nearly optimal for these classes of regular expressions.

%
%
%
%
%
%
%
%
%
%
%
%
%
%
%
%
%
%
%
%
%
%
%

%
%
\clearpage{}%
\clearpage{}%
%

\chapter{Building Residual Automata}
\label{chap:RFA}
As shown in Chapter~\ref{chap:prel}, residual automata (RFAs for short) are a class of automata that lies between deterministic (DFAs) and nondeterministic automata (NFAs).
They share with DFAs a significant property: the existence of a canonical minimal form for any regular language.
On the other hand, they share with NFAs the existence of automata that are exponentially smaller (in the number of states) than the corresponding minimal DFA for the language.
These properties make RFAs specially appealing in certain areas of computer science such as Grammatical Inference~\cite{denis2004learning,Kasprzik2011Inference}.

RFAs were first introduced by \citet{denis2000residual,denis2002residual} who defined an algorithm for \emph{residualizing} an automaton (see Section~\ref{sec:FSA}), showed that there exists a \emph{unique} \emph{canonical} RFA for every regular language and proved that the residual-equivalent of double-reversal method for DFAs~\cite{brzozowski1962canonical} holds for RFAs, i.e.\ residualizing an automaton \(\cN\) whose reverse is residual yields the canonical RFA for \(\lang{\cN}\).
Later, \citet{tamm2015generalization} generalized the double-reversal method for RFAs in the same lines as that of \citet{Brzozowski2014} for the double-reversal method for DFAs.

The similarities between the determinization and residualization (see Section~\ref{sec:FSA}) operations and between the double-reversal methods for DFAs and RFAs evidence the existence of a relationship between these two classes of automata.
However, the connection between them is not clear and, as a consequence, the relation between the generalization by \citet{Brzozowski2014} of the double-reversal method for DFAs and the one by \citet{tamm2015generalization} for RFAs is not immediate.

In this chapter, we show that \emph{quasiorders} are fundamental to RFAs as \emph{congruences} are for DFAs, which evidences the relation between these two classes of automata.
To do that, we define a framework of finite-state automata constructions based on \emph{quasiorders} over words.

As explained in Chapter~\ref{chap:related}, \citet{ganty2019congruence} studied the problem of building DFAs using congruences, i.e., equivalence relations over words with good properties w.r.t. concatenation, and derived several well-known results about minimization of DFAs, including the double-reversal method and its generalization by \citet{Brzozowski2014}.
While the use of congruences over words suited for the construction of a subclass of residual automata, namely, \emph{deterministic} automata, these are no longer useful to describe the more general class of \emph{nondeterministic} residual automata.
By moving from \emph{congruences} to \emph{quasiorders}, we are able to introduce nondeterminism in our automata constructions.

We consider quasiorders with good properties w.r.t. \emph{right} and \emph{left} concatenation.
In particular, we define the so-called right \emph{language-based} quasiorder, whose definition relies on a given regular language; and the right \emph{automata-based} quasiorder, whose definition relies on a finite representation of the language, i.e., an automaton.
We also give counterpart definitions for quasiorders that behave well with respect to \emph{left} concatenation.

When instantiating our automata constructions using the right language-based quasiorder, we obtain the canonical RFA for the given language; while using the right automata-based quasiorder yields an RFA for the language generated by the automaton that has, at most, as many states as the RFA obtained by the residualization operation defined by \citet{denis2002residual}.
Similarly, \emph{left} automata-based and language-based quasiorders yield co-residual automata, i.e., automata whose reverse is residual.

Our quasiorder-based framework allows us to give a simple correctness proof of the double-reversal method for building the canonical RFA.
Moreover, it allows us to generalize this method in the same fashion as \citet{Brzozowski2014} generalized the double-reversal method for building the minimal DFA.
Specifically, we give a characterization of the class of automata for which our automata-based quasiorder construction yields the canonical RFA.

We compare our characterization with the class of automata, defined by \citet{tamm2015generalization}, for which the residualization operation of \citet{denis2002residual} yields the canonical RFA and show that her class of automata is strictly contained in the class we define.
Furthermore, we highlight the connection between the generalization of \citet{Brzozowski2014} and the one of \citet{tamm2015generalization} for the double-reversal methods for DFAs and RFAs, respectively.

Finally, we revisit the problem of learning RFAs from a quasiorder-based perspective.
Specifically, we observe that the NL\(^*\) algorithm defined by \citet{bollig2009angluin}, inspired by the popular Angluin's L\(^*\) algorithm for learning DFAs~\cite{angluin1987learning}, can be seen as an algorithm that starts from a quasiorder and refines it at each iteration.
At the end of each iteration, the automaton built by NL\(^*\) coincides with our quasiorder-based automata construction applied to the refined quasiorder.

\section{Automata Constructions from Quasiorders}%
\label{sec:automataConstructions}
In this chapter, we consider monotone quasiorders on \(\Sigma^*\) (and their corresponding closures) and we use them to define RFAs constructions for regular languages.
The following lemma gives a characterization of right and left quasiorders.

\begin{lemma}
\label{lemma:QObwComplete}
The following properties hold:
\begin{myEnumA}
\item \(\qr\) is a right quasiorder if{}f \(ρ_{\qr}(u)\, v \subseteq ρ_{\qr}(uv)\), for all \(u,v \in \Sigma^*\).
\item\(\ql\) is a left quasiorder if{}f \(v\, ρ_{\ql}(u) \subseteq ρ_{\ql}(vu)\), for all \(u,v \in \Sigma^*\).
\end{myEnumA}
\end{lemma}

\begin{proof}\hfill
\begin{myEnumA}
\item To simplify the notation, we denote \(ρ_{\qr}\), the closure induced by \(\qr\), by \(ρ\).%
\smallskip
\begin{myEnumA}
\item[(\(\Ra\))]
Let \(x \in ρ(v)u\), i.e. \(x = \tilde{v}u\) with \(v \qr \tilde{v}\). 
Since \(\qr\) is a right quasiorder and \(v \qr \tilde{v}\) then \(vu \qr \tilde{v}u\).
Therefore \(x \in ρ(vu)\).

\item[(\(\La\))]
Assume that for each \(u,v \in \Sigma^*\) and \(\tilde{v} \in ρ(v)\) we have that \(\tilde{v}u \in ρ(vu)\).
Then, \(v \qr \tilde{v} \Ra vu \qr \tilde{v}u\).
\end{myEnumA}
\medskip

\item To simplify the notation we denote \(ρ_{\ql}\), the closure induced by \(\ql\), by \(ρ\).
\smallskip
\begin{myEnumA}
\item[(\(\Ra\))]
Let \(x \in uρ(v)\), i.e. \(x = u\tilde{v}\) with \(v \ql \tilde{v}\). 
Since \(\ql\) is a left quasiorder and \(v \ql \tilde{v}\) then \(uv \ql u\tilde{v}\).
Therefore \(x \in ρ(uv)\).

\item[(\(\La\))]
Assume that for each \(u,v \in \Sigma^*\) and \(\tilde{v} \in ρ(v)\) we have that \(u\tilde{v} \in ρ(uv)\).
Then \(v \ql \tilde{v} \Ra uv \ql u\tilde{v}\).
\end{myEnumA}
\end{myEnumA}
\end{proof}

Given a regular language \(L\), we are interested in left and right \(L\)-consistent quasiorders.
We use the principals of these quasiorders as states of automata constructions that yield RFAs and co-RFAs generating the language \(L\).
Therefore, in the sequel, we only consider quasiorders that induce a finite number of principals, i.e., quasiorders \(\qo\) such that the equivalence \(\mathord{\sim} \ud \mathord{\qo} \cap (\mathord{\qo})^{-1}\) has finite index.

Next, we introduce the notion of \emph{\(L\)-composite principals} which, intuitively, correspond to states of our automata constructions that can be removed without altering the generated language.

\begin{definition}[\(L\)-Composite Principal]%
\label{def:CompositeClosed}
Let \(L\) be a regular language and let \(\qr\) (resp. \(\ql\)) be a right (resp.\ left) quasiorder on \(Σ^*\).
Given \(u \in \Sigma^*\), the principal \(ρ_{\qr}(u)\) (resp. \(ρ_{\ql}(u)\)) is \(L\)-\emph{composite} if{}f %
\begin{align*}
u^{-1}L & = \hspace{-10pt}\bigcup_{x\in\Sigma^*,\; x \qrn u }\hspace{-10pt} x^{-1}L &
\text{\emph{(}resp. }Lu^{-1} & = \hspace{-10pt}\bigcup_{x\in\Sigma^*,\; x \qln u }\hspace{-10pt} Lx^{-1}\text{\emph{)}}
\end{align*}
If \(ρ_{\qr}(u)\) (resp. \(ρ_{\ql}(u)\)) is not \(L\)-composite then it is \emph{\(L\)-prime}.\eod
\end{definition}

We sometimes use the terms \emph{composite} and \emph{prime principal} when the language \(L\) is clear from the context.
Observe that, if \(ρ_{\qr}(u)\) is \(L\)-composite, for some \(u \in \Sigma^*\), then so is \(ρ_{\qr}(v)\), for every \(v \in \Sigma^*\) such that \(u \rr v\).
The same holds for a left quasiorder \(\ql\).

Given a regular language \(L\) and a right \(L\)-consistent quasiorder \(\qr\), the following automata construction yields an RFA that generates exactly \(L\).

\begin{definition}[Automata construction \(\cH^{r}(\qr, L)\)]%
\label{def:right-const:qo}
Let \(\qr\) be a right quasiorder and let \(L \subseteq \Sigma^*\) be a language.
Define the automaton \(\cH^{r}(\qr, L) \ud \tuple{Q, \Sigma, \delta, I, F}\) where \(Q = \{ρ_{\qr}(u) \mid u \in Σ^*, \; ρ_{\qr}(u) \text{ is \(L\)-prime}\}\), \(I = \{ρ_{\qr}(u) \in Q \mid \varepsilon \in ρ_{\qr}(u)\}\), \(F = \{ρ_{\qr}(u) \in Q \mid u \in L\}\) and \( \delta(ρ_{\qr}(u), a) = \{ ρ_{\qr}(v) \in Q \mid ρ_{\qr}(u) \cdot a \subseteq ρ_{\qr}(v)\}\) for all \(ρ_{\qr}(u) \in Q, a \in Σ\). \hfill\rule{0.5em}{0.5em}
\end{definition}

\begin{lemma}\label{lemma: HrGeneratesL}
Let \(L\subseteq \Sigma^*\) be a regular language and let \(\qr\) be a right \(L\)-consistent quasiorder.
Then, \(\cH^r(\qr,L)\) is an RFA such that \(\lang{\cH^r(\qr,L)} = L\).
\end{lemma}
\begin{proof}
To simplify the notation, we denote \(ρ_{\qr}\), the closure induced by the quasiorder \(\qr\), simply by \(ρ\).
Let \(\mathcal{H} = \cH^{r}(\qr, L) = \tuple{Q, \Sigma, \delta, I, F}\).
We first show that \(\mathcal{H}\) is an RFA, i.e.
\begin{equation}\label{eq:right-langsUco}
W_{ρ(u), F}^{\mathcal{H}} = u^{-1}L, \quad \text{for each } ρ(u)\in Q \enspace .
\end{equation}

Let us prove that \(w \in u^{-1}L \Ra w \in W_{ρ(u), F}^{\mathcal{H}}\).
We proceed by induction on the length of \(w\).
\begin{myItem}
\item \emph{Base case:}
Assume \(w = \varepsilon\).
Then, 
\[\varepsilon \in u^{-1}L \Ra u \in L \Ra ρ(u) \in F \Ra \varepsilon \in W_{ρ(u), F}^{\mathcal{H}}\enspace .\]

\item \emph{Inductive step:}
Assume that the hypothesis holds for each word \(x \in \Sigma^*\) with \(\len{x} \leq n\), where \(n \geq 1\), and let \(w \in \Sigma^*\) be such that \(\len{w} = n{+}1\).
Then \(w = a  x\) with \(\len{x} = n\) and \(a \in Σ\).
\begin{adjustwidth}{-0.8cm}{}
\begin{myAlign}{0pt}{}
ax \in u^{-1} L & \Ra \quad \text{[By definition of quotient]} \\
x \in (ua)^{-1}L & \Ra \; \\
\hspace{-5pt}\span \text{[By Def.~\ref{def:CompositeClosed}, \(ρ(ua)\) is \(L\)-prime (so \(z \ud ua\)) or \((ua)^{-1}L = \hspace{-5pt}\bigcup_{x_i \qrn ua}\hspace{-5pt}x_i^{-1}L\) (so \(z \ud x_i\))]}\\
\exists ρ(z) \in Q, \; x \in z^{-1}L \land ρ(ua) \subseteq ρ(z) & \Ra \quad \text{[By I.H., Lemma~\ref{lemma:QObwComplete} and Def.~\ref{def:right-const:qo}]} \\
x \in W_{ρ(z), F}^{\mathcal{H}} \land ρ(z) \in δ(ρ(u), a) & \Ra \quad \text{[By definition of \(W_{S,T}\)]}\\
ax \in W_{ρ(u), F}^{\mathcal{H}} \enspace .
\end{myAlign}
\end{adjustwidth}
\end{myItem}

We now prove the other side of the implication, \(w \in W_{ρ(u), F}^{\mathcal{H}} \Ra w \in u^{-1}L\).

\begin{myItem}
\item \emph{Base case:}
Let \(w = \varepsilon\).
By Definition~\ref{def:right-const:qo}, 
\[\varepsilon \in W_{ρ(u), F}^{\mathcal{H}} \Ra \exists ρ(x) \in Q, \; x \in L \land ρ(u)  \varepsilon \subseteq ρ(x) \enspace .\]
Since \(ρ(L) = L\), we have that \(u\,\varepsilon \in L \), hence \(\varepsilon \in u^{-1}L\).

\item \emph{Inductive step:}
Assume the hypothesis holds for each \(x \in \Sigma^*\) with \(\len{x} \leq n\), where \(n \geq 1\), and let \(w \in \Sigma^*\) be such that \(\len{w} = n{+}1\).
Then \(w = a  x\) with \(\len{x} = n\) and \(a \in Σ\).
\begin{align*}
ax \in W_{ρ(u), F}^{\mathcal{H}} & \Ra \quad \text{[By Definition~\ref{def:right-const:qo}]} \\
x \in W_{ρ(y), F}^{\mathcal{H}} \land ρ(u) a \subseteq ρ(y)  & \Ra \quad \text{[By I.H. and since \(ρ\) is induced by \(\qr\)]} \\
x \in y^{-1}L \land y \qr ua & \Ra \quad \text{[By \citet{deLuca1994}]} \\
x \in y^{-1}L \land y^{-1}L \subseteq (ua)^{-1}L & \Ra \quad \text{[Since \(x \in (ua)^{-1} L \Ra ax \in u^{-1}L\)]} \\
ax \in u^{-1}L \enspace .
\end{align*}
\end{myItem}

We have shown that \(\mathcal{H}\) is an RFA. 
Finally, we show that \(\lang{\mathcal{H}} = L\). 
First note that
\[\lang{\mathcal{H}} = \bigcup_{ρ(u) \in I} W_{ρ(u), F}^{\mathcal{H}} = \bigcup_{ρ(u) \in I} u^{-1}L \enspace ,\]
where the first equality holds by definition of \(\lang{\mathcal{H}}\) and the second by Equation~\eqref{eq:right-langsUco}.
On one hand, we have that \(\bigcup_{ρ(u) \in I} u^{-1}L \subseteq L\) since, by Definition~\ref{def:right-const:qo}, \(\varepsilon \in ρ(u)\) for each \(ρ(u) \in I\), hence \(u \qr \varepsilon\) which, as shown by \citet{deLuca1994}, implies that \(u^{-1}L \subseteq \varepsilon^{-1}L = L\).

Let us now show that \(L \subseteq \bigcup_{ρ(u) \in I} u^{-1}L\).
First, let us assume that \(ρ(\varepsilon) \in I\).
Then, 
\[L = \varepsilon^{-1}L \subseteq \bigcup_{ρ(u) \in I} u^{-1}L \enspace .\]
Now suppose that \(ρ(\varepsilon)\notin I\), i.e. \(ρ(\varepsilon)\) is \(L\)-composite.
Then,
\[ L = \varepsilon^{-1}L = \bigcup_{u \qrn \varepsilon} u^{-1}L = \bigcup_{ρ(u) \in I} u^{-1}L \enspace .\]
where the last equality follows from \(ρ(u) \in I \Lra \varepsilon \in ρ(u)\).
\end{proof}

Given a regular language \(L\) and a left \(L\)-consistent quasiorder \(\ql\), we can give a similar automata construction of a co-RFA that recognizes exactly \(L\)

\begin{definition}[Automata construction \(\cH^{\ell}(\ql, L)\)]%
\label{def:left-const:qo}
Let \(\ql\) be a left quasiorder and let \(L \subseteq \Sigma^*\) be a language.
Define the automaton \(\cH^{\ell}(\ql, L)\ud \tuple{Q, \Sigma, \delta, I, F}\) where \(Q = \{ρ_{\ql}(u) \mid u \in Σ^*, \; ρ_{\ql}(u) \text{ is \(L\)-prime}\}\), \(I = \{ ρ_{\ql}(u) \in Q \mid u \in L \}\), \(F = \{ρ_{\ql}(u) \in Q \mid \varepsilon \in ρ_{\ql}(u)\}\), and \(\delta(ρ_{\ql}(u), a) = \{ρ_{\ql}(v)\in Q \mid a\cdot ρ_{\ql}(v) \subseteq ρ_{\ql}(u)\}\) for all \(ρ_{\ql}(u) \in Q, a \in \Sigma\). \hfill\rule{0.5em}{0.5em}
\end{definition}

\begin{lemma}%
\label{lemma:HlgeneratesL}
Let \(L\subseteq \Sigma^*\) be a language and let \(\ql\) be a left \(L\)-consistent quasiorder.
Then \(\cH^{\ell}(\ql,L)\) is a co-RFA such that \(\lang{\cH^{\ell}(\ql,L)} = L\).
\end{lemma}
\begin{proof}
To simplify the notation we denote \(ρ_{\ql}\), the closure induced by the quasiorder \(\ql\), simply by \(ρ\).
Let \(\mathcal{H} = \cH^{\ell}(\ql, L) = \tuple{Q, \Sigma, \delta, I, F}\).
We first show that \(\mathcal{H}\) is a co-RFA.
\begin{equation}\label{eq:left-langsUco}
W_{I, ρ(u)}^{\mathcal{H}} = Lu^{-1}, \quad \text{for each } ρ(u)\in Q \enspace .
\end{equation}

Let us prove that \(w \in Lu^{-1} \Ra w \in W_{I, ρ(u)}^{\mathcal{H}}\).
We proceed by induction.
\begin{myItem}
\item \emph{Base case:}
Let \(w = \varepsilon\).
Then 
\[\varepsilon \in Lu^{-1} \Ra u \in L \Ra ρ(u) \in I \Ra \varepsilon \in W_{I, ρ(u)}^{\mathcal{H}}\enspace .\]
\item \emph{Inductive step:}
Assume the hypothesis holds for all \(x \in \Sigma^*\) with \(\len{x} \leq n\), where \(n \geq 1\), and let \(w \in \Sigma^*\) be such that \(\len{w} = n{+}1\).
Then \(w = x a\) with \(\len{x} = n\) and \(a \in Σ\).
\begin{adjustwidth}{-0.8cm}{}
\begin{myAlign}{0pt}{}
xa \in Lu^{-1} & \Ra \quad \text{[By definition of quotient]} \\
x \in L(au)^{-1} & \Ra \;  \\
\hspace{-5pt}\span \text{[By Def.~\ref{def:CompositeClosed}, \(ρ(ua)\) is \(L\)-prime (so \(z \ud au\)) or \(L(au)^{-1} = \hspace{-5pt}\bigcup_{x_i \qln au}\hspace{-5pt}Lx_i^{-1}\) (so \(z \ud x_i\))]}\\
\exists ρ(z) \in Q, \; x \in L z^{-1} \land ρ(au) \subseteq ρ(z) & \Ra \quad \text{[By I.H., Lemma~\ref{lemma:QObwComplete} and Def.~\ref{def:left-const:qo}]} \\
x \in W_{I, ρ(z)}^{\mathcal{H}} \land ρ(u) \in δ(ρ(z), a) & \Ra \quad \text{[By definition of \(W_{S,T}\)]}\\
xa \in W_{I, ρ(u)}^{\mathcal{H}} \enspace .
\end{myAlign}
\end{adjustwidth}
\end{myItem}
We now prove the other side of the implication, \(w \in W_{I, ρ(u)}^{\mathcal{H}} \Ra w \in Lu^{-1}\).

\begin{myItem}
\item \emph{Base case:}
Let \(w = \varepsilon\).
Then 
\[\varepsilon \in W_{I, ρ(u)}^{\mathcal{H}} \Ra \exists ρ(x) \in Q,\; x \in L \land \varepsilon  ρ(u) \subseteq ρ(x)\enspace . \]
Since \(ρ(L) = L\), we have that \( \varepsilon  u \in L\), hence \(\varepsilon \in Lu^{-1}\).

\item \emph{Inductive step:}
Assume the hypothesis holds for each \(x \in \Sigma^*\) with \(\len{x} \leq n\), where \(n \geq 1\), and let \(w \in \Sigma^*\) be such that \(\len{w} = n{+}1\).
Then \(w = x\cdot a\) with \(\len{x} = n\) and \(a \in Σ\).
\begin{align*}
xa \in W_{I, ρ(u)}^{\mathcal{H}} & \Ra \quad \text{[By Definition~\ref{def:left-const:qo}]} \\
a\cdot ρ(u) \subseteq ρ(y) \land x \in W_{I, ρ(y)}^{\mathcal{H}} & \Ra \quad \text{[By I.H. and since \(ρ\) is induced by \(\ql\)]} \\
y \ql au \land x \in Ly^{-1} & \Ra \quad \text{[By \citet{deLuca1994}]} \\ 
Ly^{-1} \subseteq L(au)^{-1} \land x \in Ly^{-1} & \Ra \quad \text{[Since \(x \in L(au)^{-1} \Ra xa \in Lu^{-1}\)]} \\
xa \in u^{-1}L \enspace .
\end{align*}
\end{myItem}

We have shown that \(\mathcal{H}\) is a co-RFA.
Finally, we show that \(\lang{\mathcal{H}} = L\). 
First note that
\[\lang{\mathcal{H}} = \bigcup_{ρ(u) \in F} W_{I, ρ(u)}^{\mathcal{H}} = \bigcup_{ρ(u) \in F} Lu^{-1} \enspace ,\]
where the first equality holds by definition of \(\lang{\mathcal{H}}\) and the second by Equation~\eqref{eq:left-langsUco}.
On one hand, we have that \(\bigcup_{ρ(u) \in F} Lu^{-1} \subseteq L\) since, by Definition~\ref{def:left-const:qo}, \(\varepsilon \in ρ(u)\) for each \(ρ(u) \in F\), hence \(u \ql \varepsilon\) which, as shown by \citet{deLuca1994}, implies that \(Lu^{-1} \subseteq L\varepsilon^{-1} = L\).

Let us now show that \(L \subseteq \bigcup_{ρ(u) \in F} Lu^{-1}\).
First, let us assume that \(ρ(\varepsilon) \in F\).
Then, 
\[L = L\varepsilon^{-1} \subseteq \bigcup_{ρ(u) \in F} Lu^{-1} \enspace .\]
Now suppose that \(ρ(\varepsilon)\notin F\), i.e. \(ρ(\varepsilon)\) is \(L\)-composite.
Then, 
\[ L = L\varepsilon^{-1} = \bigcup_{u \qln \varepsilon} Lu^{-1} = \bigcup_{ρ(u) \in F} u^{-1}L\enspace .\]
where the last equality follows from \(ρ(u) \in F \Lra \varepsilon \in ρ(u)\).
\end{proof}

Observe that the automaton \(\mathcal{H}^r = \cH^{r}(\qr, L)\) (resp. \(\mathcal{H}^{\ell} = \cH^{\ell}(\ql, L)\)) is \emph{finite}, since we assume \(\qr\) (resp. \(\ql\)) induces a finite number of principals.
Note also that \(\mathcal{H}^{r}\) (resp. \(\mathcal{H}^{\ell}\)) possibly contains empty (resp.\ unreachable) states but no state is unreachable (resp.\ empty).

Moreover, notice that by keeping all principals of \(\qr\) (resp. \(\ql\)) as states, instead of only the \(L\)-prime ones as in Definition~\ref{def:right-const:qo} (resp. Definition~\ref{def:left-const:qo}), we would obtain an RFA (resp.\ a co-RFA) with (possibly) more states that also recognizes \(L\).

Finally, Lemma~\ref{lemma:leftRightReverse} shows that \(\mathcal{H}^{\ell}\) and \(\mathcal{H}^r\) inherit the left-right duality between \(\ql\) and \(\qr\) through the reverse operation.

\begin{lemma}\label{lemma:leftRightReverse}
Let \(\qr\) and \(\ql\) be a right and a left quasiorder, respectively, and let \(L \subseteq \Sigma^*\) be a language.
If the following property holds
\begin{equation}%
\label{eq:leftRightReverse}
u \qr v \Lra u^R \ql v^R
\end{equation}
then \(\cH^{r}(\qr, L) \) is isomorphic to \( \left(\cH^{\ell}(\ql, L^R)\right)^R\).
\end{lemma}
\begin{proof}
Let \(\cH^{r}(\qr, L) = \tuple{Q, \Sigma, \delta, I, F}\) and \((\cH^{\ell}(\ql, L^R))^R = \tuple{\wt{Q}, \Sigma, \wt{\delta}, \wt{I}, \wt{F}}\).
We will show that \(\cH^{r}(\qr, L)\) is isomorphic to \((\cH^{\ell}(\ql, L^R))^R\).

Let \(\varphi: Q \rightarrow \wt{Q}\) be a mapping assigning to each state \(ρ_{\qr}(u) \in Q\) with \(u \in \Sigma^*\), the state \(ρ_{\ql}(u^R) \in \wt{Q}\).
Next, we show that \(\varphi\) is an NFA isomorphism between \(\cH^{r}(\qr, L)\) and \((\cH^{\ell}(\ql, L^R))^R\).

Observe that:
\begin{align*}
u^{-1}L = \bigcup_{x \qrn u}x^{-1}L & \Lra \quad\text{[Since \(\left(\bigcup S_i\right)^R = \bigcup S_i^R\)]} \\
(u^{-1}L)^R = \bigcup_{x \qrn u}(x^{-1}L)^R & \Lra \quad\text{[Since \((u^{-1}L)^R = L^R(u^R)^{-1} \)]}\\
L^R(u^R)^{-1} = \bigcup_{x \qrn u} L^R(x^R)^{-1} & \Lra \quad\text{[By Equation~\eqref{eq:leftRightReverse}]}\\
L^R(u^R)^{-1} = \bigcup_{x^R \qln u^R} L^R(x^R)^{-1} & \enspace . 
\end{align*}
Therefore \(ρ_{\qr}(u)\) is \(L\)-composite if{}f \(ρ_{\ql}(u^R)\) is \(L^R\)-composite, hence \(\varphi(Q) = \wt{Q}\).

Since 
\[\varepsilon \in ρ_{\qr}(u) \Lra u \qr \varepsilon \Lra u^r \ql \varepsilon \Lra \varepsilon \in ρ_{\ql}(u^R) \enspace ,\]
we have that \(ρ_{\qr}(u)\) is an initial
state of \(\cH^{r}(\qr, L)\) if{}f \(ρ_{\ql}(u^R)\) is a
final state of \(\cH^{\ell}(\ql, L^R)\), i.e.\ an initial state of \((\cH^{\ell}(\ql, L^R))^R\).
Therefore, \(\varphi(I) = \wt{I}\).

Since 
\[ρ_{\qr}(u) \subseteq L \Lra u \in L \Lra u^r \in L^R \enspace ,\]
we have that \(ρ_{\qr}(u)\) is a final state of \(\cH^{r}(\qr, L)\) if{}f \(ρ_{\ql}(u^R)\) is an initial state of \(\cH^{\ell}(\ql, L^R)\), i.e. a final state of \((\cH^{\ell}(\ql, L^R))^R\).
Therefore, \(\varphi(F) = \wt{F}\).

It remains to show that \(q' \in \delta(q, a)\Lra \varphi(q') \in \wt{\delta}(\varphi(q),a)\), for all \(q, q' \in Q\) and \(a \in \Sigma\).
Assume that \(q = ρ_{\qr}(u)\) for some \(u \in \Sigma^*\), \(q' = ρ_{\qr}(v)\) for some \(v \in Σ^*\) and \(q' \in \delta(q, a)\) with \(a \in \Sigma\).
Then,
\begin{myAlignEP}
ρ_{\qr}(v) \in δ(ρ_{\qr}(u), a) & \Lra \quad \text{[By Definition~\ref{def:right-const:qo}]} \\
ρ_{\qr}(u)a \subseteq ρ_{\qr}(v) & \Lra \quad \text{[By definition of \(ρ_{\qr}\) and Lemma~\ref{lemma:QObwComplete}]} \\
v \qr ua & \Lra \quad \text{[By Equation~\eqref{eq:leftRightReverse} and \((ua)^R = au^R\)]} \\
v^r \ql au^R & \Lra \quad \text{[By definition of \(ρ_{\ql}\) and Lemma~\ref{lemma:QObwComplete}]} \\
aρ_{\ql}(u^R) \subseteq ρ_{\ql}(v^R) & \Lra \quad \text{[By Definition~\ref{def:left-const:qo}]} \\
ρ_{\ql}(v^R) \in \wt{δ}(ρ_{\ql}(u^R), a) & \Lra \quad \text{[Definition of \(q, q'\) and \(\varphi\)]} \\
\varphi(q') \in \wt{δ}(\varphi(q), a) \enspace .\tag*{\qedhere}
\end{myAlignEP}
\end{proof}

\subsection{On the Size of \texorpdfstring{\(\cH^{r}(\qr,L)\)}{Hr} and \texorpdfstring{\(\cH^{\ell}(\ql,L)\)}{Hl}}

We conclude this section with a note on the sizes of the automata constructions \(\cH^{r}(\qr,L)\) and \(\cH^{\ell}(\ql,L)\) when applied to quasiorders satisfying \(\mathord{\qr_1} \subseteq \mathord{\qr_2}\) and \(\mathord{\ql_1} \subseteq \mathord{\ql_2}\), respectively.

The following result establishes a relationship between the \(L\)-composite principals for two comparable right quasiorders \(\mathord{\qr_1} \subseteq \mathord{\qr_2}\).
This result is used in Theorem~\ref{theorem:numLPrimePrincipals} to show that the number of \(L\)-prime principals induced by \(\qr_1\) is greater or equal than the number of \(L\)-prime principals induced by \(\qr_2\).

As a consequence, if \(\mathord{\qr_1} \subseteq \mathord{\qr_2}\) then the automaton \(\cH^{r}(\qr_1,L)\) has, at least, as many states as \(\cH^{r}(\qr_2,L)\).
The same holds for left quasiorders and \(\cH^{\ell}\).

\begin{lemma}\label{lemma:numprincipals}
Let \(L \subseteq Σ^*\) be a regular language and let \(u \in Σ^*\). 
Let \(\qr_1\) and \(\qr_2\) be two \(L\)-consistent right quasiorders such that \(\mathord{\qr_1} \subseteq \mathord{\qr_2}\).
Then
\[ρ_{\qr_1}(u) \text{ is \(L\)-composite} \Ra \left( ρ_{\qr_2}(u) \text{ is \(L\)-composite} \lor  \exists x \qrn_1 u, \; ρ_{\qr_2}(u) = ρ_{\qr_2}(x)\right)\enspace .\]
Similarly holds for left quasiorders.
\end{lemma}

\begin{proof}
Let \(u \in Σ^*\) be such that \(ρ_{\qr_1}(u)\) is \(L\)-composite.
Then we have that \(u^{-1}L = \bigcup_{x \in Σ^*, x \qrn_1 u} x^{-1}L\).
On the other hand, since \(\qr_2\) is a right \(L\)-consistent quasiorder, we have that \(\mathord{\qr_2} \subseteq \mathord{\qrL}\), as shown by \citet{deLuca1994}.
Therefore \(u^{-1}L \supseteq \bigcup_{x \in Σ^*, x \qrn_2 u} x^{-1}L\). 
There are now two possibilities:
\begin{myItem}
\item For all \(x \in Σ^*\) such that \(x \qrn_1 u\) we have that \(x \qrn_2 u\).
In that case we have that \(u^{-1}L = \bigcup_{x\inΣ^*, \; x \qrn_2 u} x^{-1}L\), hence \(ρ_{\qr_2}(u)\) is \(L\)-composite.
\item There exists \(x \in Σ^*\) such that \(x \qrn_1 u\), hence \(x \qr_2 u\), but \(x \not\qrn_2u\).
In that case, it follows that \(ρ_{\qr_2}(x) = ρ_{\qr_2}(u)\).
\end{myItem}
The proof for left quasiorders is symmetric.
\end{proof}

\begin{theorem}\label{theorem:numLPrimePrincipals}
Let \(L\subseteq Σ^*\) and let \(\qo_1\) and \(\qo_2\) be two right or two left \(L\)-consistent quasiorders such that \(\mathord{\qo_1} \subseteq \mathord{\qo_2}\).
Then
\[\len{\{ρ_{\qo_1}(u) \mid u \in Σ^* \land ρ_{\qo_1}(u) \text{ is \(L\)-prime}\}} \geq \len{\{ρ_{\qo_2}(u) \mid u \in Σ^* \land ρ_{\qo_2}(u) \text{ is \(L\)-prime}\}}\]
\end{theorem}
\begin{proof}
We proceed by showing that for every \(L\)-prime \(ρ_{\qo_2}(u)\) there exists an \(L\)-prime \(ρ_{\qo_1}(x)\) such that \(ρ_{\qo_2}(x) = ρ_{\qo_2}(u)\).
Clearly, this entails that there are, at least, as many \(L\)-prime principals for \(\qo_1\) as there are for \(\qo_2\).

Let $ρ_{\qo_2}(u)$ be $L$-prime.

If \(ρ_{\qo_1}(u)\) is \(L\)-prime, we are done.
Otherwise, by Lemma~\ref{lemma:numprincipals}, we have that there exists $x \qon_1 u$ such that 
$ρ_{\qo_2}(u) = ρ_{\qo_2}(x)$.

We repeat the reasoning with $x$. If $ρ_{\qo_1}(x)$ is $L$-prime, we are done. Otherwise, there exists $x_1 \qon_1 x$ such that 
$ρ_{\qo_2}(u) = ρ_{\qo_2}(x) = ρ_{\qo_2}(x_1)$.

Since \(\qo_1\) induces finitely many principals, there are no infinite strictly descending chains and, therefore, there exists \(x_n\) such that \(ρ_{\qo_2}(u)=ρ_{\qo_2}(x)=ρ_{\qo_2}(x_1)=\ldots = ρ_{\qo_2}(x_n)\) and \(ρ_{\qo_1}(x_n)\) is \(L\)-prime.%
\end{proof}

\section{Language-based Quasiorders and their Approximation using NFAs}%
\label{sec:Instantiation}
In this section we instantiate our automata constructions using two classes of quasiorders, namely, the so-called \emph{Nerode's} quasiorders~\cite{deLuca1994}, whose definition is based on a given regular language; and the \emph{automata-based} quasiorders, whose definition is based on a finite representation of the language, i.e, an automaton.

Both quasiorders have been used previously in Chapter~\ref{chap:LangInc} in order to derive algorithms for solving the language inclusion problem between regular languages.
We recall their definitions next:
\begin{align}
u \qrL v  & \udiff u^{-1}L \subseteq v^{-1}L & \quad \text{\emph{Right-}language-based Quasiorder}\label{eq:Rlanguage} \\
u \qlL v  & \udiff Lu^{-1} \subseteq Lv^{-1}  & \quad \text{\emph{Left-}language-based Quasiorder}%
\label{eq:Llanguage} \\
u \qrN v & \udiff \post^{\cN}_u(I) \subseteq \post^{\cN}_v(I) & \quad \text{\emph{Right-}Automata-based Quasiorder}\label{eq:RState} \\
u \qlN v & \udiff \pre^{\cN}_u(F) \subseteq \pre^{\cN}_v(F) & \quad \text{\emph{Left-}Automata-based Quasiorder} \label{eq:LState} 
\end{align}

As explained in Chapter~\ref{chap:LangInc}, \citet{deluca2011} showed that for every regular language \(L\) there exists a finite number of quotients \(u^{-1}L\) and, therefore, \(\qrL\) and \(\qlL\) are well-quasiorders.
On the other hand, the automata-based quasiorders are also well-quasiorders.
Therefore, all the quasiorders defined above induce a finite number of principals.

\begin{remark} \label{remark:LRDual}
The pairs of quasiorders \(\qrL\)~-~\(\qlL\) and \(\qrN\)~-~\(\qlN\) are dual, i.e.  
\[u \qrL v \Lra u^R \qlL v^R \quad \text{and} \quad u \qrN v \Lra u^R \qlN v^R\enspace .\]
\end{remark}

The following result shows that the principals of \(\qrN\) and \(\qlN\) can be described, respectively, as intersections of left and right languages of the states of \(\cN\) while the principals of \(\qrL\) and \(\qlL\) can be described as intersections of left and right quotients of \(L\).

\begin{lemma}\label{lemma:positive_atoms}
Let \(\cN = \tuple{Q,Σ,δ,I,F}\) be an NFA with \(\lang{\cN}=L\).  
Then, for every \(u \in \Sigma^*\),
\begin{align*}
ρ_{\qrN}(u) &= \bigcap\textstyle{_{q \in \post_u^{\cN}(I)}} W_{I,q}^{\cN} & ρ_{\qrL}(u)  &= \bigcap\textstyle{_{w \in \Sigma^*, \; w \in u^{-1}L}} Lw^{-1} \\
ρ_{\qlN}(u)  & = \bigcap\textstyle{_{q \in \pre_u^{\cN}(I) }} W_{q,F}^{\cN} & ρ_{\qlL}(u)  &= \bigcap\textstyle{_{w \in \Sigma^*, \; w \in Lu^{-1}L }} w^{-1}L\enspace .
\end{align*}
\end{lemma}
\begin{proof}
We prove the lemma for the principals induced by \(\qrL\) and \(\qrN\).
The proofs for the left quasiorders are symmetric. 

For each \(u \in \Sigma^*\) we have that 
\begin{align*}
ρ_{\qrN}(u) &= \quad \text{[By definition of \(ρ_{\qrN}\)]}\\
\{v \in Σ^* \mid \post_u^{\cN}(I) \subseteq \post_v^{\cN}(I)\} & = \quad \text{[By definition of set inclusion]} \\
\{v \in \Sigma^* \mid  \forall q \in \post^{\cN}_u(I), \; q\in \post^{\cN}_v(I)\} & = \quad \text{[Since \(q \in \post^{\cN}_v(I) \Lra v \in W^{\cN}_{I,q}\)]}\\
\{v \in \Sigma^* \mid  \forall q \in \post^{\cN}_u(I), \; v \in W^{\cN}_{I,q}\} & = \quad \text{[By definition of intersection]}\\
\bigcap\textstyle{_{q \in \post^{\cN}_u(I)}} W_{I,q}^{\cN} & \enspace .
\end{align*}

On the other hand, 
\begin{align*}
v \in  \bigcap\textstyle{_{w \in \Sigma^*, \; w \in u^{-1}L}} Lw^{-1} & \Lra \quad \text{[By definition of intersection]} \\
\forall w \in Σ^*, \; w \in u^{-1}L \Ra v \in Lw^{-1} & \Lra \quad \text{[Since \(\forall x,y \in Σ^*, \; x \in Ly^{-1} \Lra y \in x^{-1}L\)]} \\
\forall w \in Σ^*, \; w \in u^{-1}L \Ra w \in v^{-1}L & \Lra \quad \text{[By definition of set inclusion]} \\
u^{-1}L \subseteq v^{-1}L & \Lra \quad \text{[By definition of \(ρ_{\qlL}(u)\)]} \\
v \in ρ_{\qrL}(u) 
\end{align*}
\end{proof}

As shown by Lemma~\ref{lemma:LAconsistent}, given an NFA \(\cN\) with \(L = \lang{\cN}\), the quasiorders \(\qrL\) and \(\qrN\) are right \(L\)-consistent, while the quasiorders \(\qlL\) and \(\qlN\) are left \(L\)-consistent.
Therefore, by Lemma~\ref{lemma: HrGeneratesL} and~\ref{lemma:HlgeneratesL}, our automata constructions applied to these quasiorders yield automata for \(L\).

Finally, recall that, as shown by \citet{deLuca1994}, \(\qrN\) is finer than \(\qrL\), i.e., \(\mathord{\qrN} \subseteq \mathord{\qrL}\).
In that sense we say \(\qrN\) \emph{approximates}  \(\qrL\).
As the following lemma shows, the approximation is precise, i.e., \(\mathord{\qrN} = \mathord{\qrL}\), whenever \(\cN\) is a co-RFA with no empty states\@.

\begin{lemma}\label{lemma:coResidual_qrL=qrN}
Let \(\cN = \tuple{Q, Σ, δ, I, F}\) be a co-RFA with no empty states such that \(L = \lang{\cN}\).
Then \(\mathord{\qrL} = \mathord{\qrN}\).
Similarly, if \(\cN\) is an RFA with no unreachable states and \(L = \lang{\cN}\) then \(\mathord{\qlL} = \mathord{\qlN}\).
\end{lemma}
\begin{proof}
It is straightforward to check that the following holds for every NFA \(\cN\) and \(u, v \in Σ^*\).
\[\post_u^{\cN}(I) \subseteq \post_v^{\cN}(I) \Ra W_{\post_{u}^{\cN}(I),F}^{\cN} \subseteq W_{\post_{v}^{\cN}(I),F}^{\cN}\]
Next we show that the reverse implication also holds when \(\cN\) is a co-RFA with no empty states.
Let \(u, v\in Σ^*\) be such that \(W_{\post_{u}^{\cN}(I),F}^{\cN} \subseteq W_{\post_{v}^{\cN}(I),F}^{\cN}\).
Then,
\begin{align*}
q \in \post_{u}^{\cN}(I) & \Ra \quad \text{[Since \(\cN\) is co-RFA with no empty states]}\\
\exists x \in Σ^*, \; u \in W_{I,q} = Lx^{-1} & \Ra \quad \text{[Since \(u \in Lx^{-1} \Ra x \in u^{-1}L\)]} \\
x \in W_{\post_{u}^{\cN}(I), F} & \Ra \quad \text{[Since \(W_{\post_{u}^{\cN}(I),F}^{\cN} \subseteq W_{\post_{v}^{\cN}(I),F}^{\cN}\)]} \\
x \in W_{\post_{v}^{\cN}(I), F} & \Ra \quad \text{[By definition of \(W_{S,T}^{\cN}\)]} \\
\exists q' \in Q, \; x \in W_{q',F} \land v \in W_{I, q'} & \Ra \quad \text{[Since \(x \in W_{q',F} \Ra W_{I,q'} \subseteq Lx^{-1}\)]}\\
v \in  Lx^{-1} & \Ra \quad \text{[Since \(Lx^{-1} = W_{I,q}\)]} \\
v \in W_{I, q} & \Ra \quad \text{[By definition of \(\post_v^{\cN}(I)\)]} \\
q \in \post_v^{\cN}(I) \enspace .
\end{align*}

Therefore, \(W_{\post_{u}^{\cN}(I),F}^{\cN} \subseteq W_{\post_{v}^{\cN}(I),F}^{\cN} \Ra \post_u^{\cN}(I) \subseteq \post_v^{\cN}(I)\).

The proof for RFAs with no unreachable states and left quasiorders is symmetric.
\end{proof}

Finally, the following lemma shows that, for the Nerode's quasiorders, the \(L\)-composite principals can be described as intersections of \(L\)-prime principals.

\begin{lemma}\label{lemma:CompositeIntersection}
Let \(\cN = \tuple{Q, Σ, δ, I, F}\) be an NFA with \(\lang{\cN} = L\).
Then,
\begin{equation}\label{eq:rhoCompRaIntersection}
u^{-1}L = \hspace{-10pt}\bigcup_{x\in\Sigma^*,\; x \qrn_L u}\hspace{-10pt} x^{-1}L \implies ρ_{\qrL}(u) = \hspace{-10pt}\bigcap_{x\in\Sigma^*,\; x \qrn_L u}\hspace{-10pt} ρ_{\qrL}(x) \enspace .
\end{equation}
Similarly holds for the left Nerode's quasiorder \(\qlL\).
\end{lemma}
\begin{proof}
Observe that the inclusion \(ρ_{\qrL}(u) \subseteq \bigcap_{x\in\Sigma^*, x \qrn_L u} ρ_{\qrL}(x)\) always holds since \(x \qrn_L u \Ra ρ_{\qrL}(u) \subseteq ρ_{\qrL}(x)\).
Next, we prove the reverse inclusion.

Let \(w \in \bigcap_{x\in\Sigma^*, x \qrn_L u} ρ_{\qrL}(x)\) and assume that the left hand side of Equation~\eqref{eq:rhoCompRaIntersection} holds.
Then, by definition of intersection and \(ρ_{\qrL}\), we have that \(x \qrL w\) for every \(x \in Σ^*\) such that \(x \qrn_L u\), i.e., \(x^{-1}L \subseteq w^{-1}L\) for every \(x \in Σ^*\) such that \(x^{-1}L \subsetneq u^{-1}L\).
Since, by hypothesis, \(u^{-1}L = \bigcup_{x\in\Sigma^*,\; x \qrn_L u} x^{-1}L\), it follows that \(u^{-1}L \subseteq w^{-1}L\) and, therefore, \(w \in ρ_{\qr}(u)\).

We conclude that \(\bigcap_{x\in\Sigma^*, x \qrn_L u} ρ_{\qrL}(x) \subseteq ρ_{\qrL}(u)\).
\end{proof}

\subsection{Automata Constructions}

In what follows, we will use \(\cF{}\) and \(\cG{}\) to denote the construction \(\cH\) when applied, respectively, to the language-based quasiorders induced by a regular language and the automata-based quasiorders induced by an NFA\@.

\begin{definitionNI}[\(\cG{}\) and \(\cF{}\)]%
\label{def:FG}
Let \(\cN\) be an NFA with \(L = \lang{\cN}\).
Define:
\begin{align*}
\mindex{\cF{r}}(L) & \ud  \cH^{r}(\qrL, L) & \mindex{\cG{r}}(\cN) & \ud \cH^{r}(\qrN, L) \\
\mindex{\cF{\ell}}(L) & \ud  \cH^{\ell}(\qlL, L) & \mindex{\cG{\ell}}(\cN) & \ud  \cH^{\ell}(\qlN, L) \enspace . \tag*{\rule{0.5em}{0.5em}}
\end{align*}
\end{definitionNI}

Given an NFA \(\cN\) generating the language \(L=\lang{\cN}\), all constructions in the above definition yield automata generating \(L\).
However, while the constructions using the right quasiorders result in RFAs, those using left quasiorders result in co-RFAs.
Furthermore, it follows from Remark~\ref{remark:LRDual} and Lemma~\ref{lemma:leftRightReverse} that \(\cF{\ell}(L)\) is isomorphic to \((\cF{r}(L^R))^R\) and \(\cG{\ell}(\cN)\) is isomorphic to \((\cG{r}(\cN^R))^R\).

It follows from Theorem~\ref{theorem:numLPrimePrincipals} that the automata \(\cG{r}(\cN)\) and \(\cG{\ell}(\cN)\) have, at least, as many states as \(\cF{r}(L)\) and \(\cF{\ell}(L)\), respectively.
Intuitively, \(\cF{r}(L)\) is the minimal RFA for \(L\), i.e. it is isomorphic to the canonical RFA for \(L\), since, as shown by~\citet{deLuca1994}, \(\qrL\) is the coarsest right \(L\)-consistent quasiorder.
On the other hand, as we shall see in Example~\ref{example:residualization}, \(\cG{r}(\cN)\) is a sub-automaton of \(\cN^{\text{res}}\)~\cite{denis2002residual} for every NFA \(\cN\).

Finally, it follows from Lemma~\ref{lemma:coResidual_qrL=qrN} that residualizing (\(\cG{r}\)) a co-RFA with no empty states (for instance, \(\cG{\ell}(\cN)\)) results in the canonical RFA for \(\lang{\cN}\) (\(\cF{r}(\lang{\cN})\)).

We formalize all these notions in Theorem~\ref{theoremF}.
\pagebreak
\begin{theorem}\label{theoremF}
Let \(\cN\) be an NFA with \(L = \lang{\cN}\).
Then the following hold:
\smallskip
\begin{myEnumA}
\item \(\lang{\cF{r}(L)} =\lang{\cF{\ell}(L)} = L = \lang{\cG{r}(\cN)} = \lang{\cG{\ell}(\cN)}\).%
\label{lemma:language-F}
\item \(\cF{\ell}(L)\) is isomorphic to \((\cF{r}(L^R))^R\).%
\label{lemma:FlisomorphicRfrR}
\item \(\cG{\ell}(\cN)\) is isomorphic to \((\cG{r}(\cN^R))^R\).%
\label{lemma:AlRequalArNR}
\item \(\cF{r}(L)\) is isomorphic to the canonical RFA for \(L\).%
\label{theorem:CanonicalRFAlanguage}
\item \(\cG{r}(\cN)\) is isomorphic to a sub-automaton of \(\cN^{\text{res}}\).%
\label{lemma:rightNRes}
\item \(\cG{r}(\cG{\ell}(\cN))\) is isomorphic to \(\cF{r}(L)\).\label{lemma:LS+RS=RN}
\end{myEnumA}
\end{theorem}

\begin{proof} 
In the following, let \(\cN = \tuple{Q, \Sigma, \delta, I, F}\).

\begin{myEnumA}
\item 

By Definition~\ref{def:FG}, \(\cF{r}(L) = \cH^{r}(\qrL, L)\) and \(\cG{r}(\cN) = \cH^{r}(\qrN, L)\).
On the other hand, by Lemma \ref{lemma: HrGeneratesL}, \(\lang{\cH^{r}(\qrL, L)} = \lang{\cH^{r}(\qrN, L)} = L\).
Therefore, \(\lang{\cF{r}(L)} = \lang{\cG{r}(L)} = L\).

Similarly, it follows from Lemma~\ref{lemma:HlgeneratesL} that \(\lang{\cF{\ell}(L)} = \lang{\cG{\ell}(L)} =L\).

\item For each \(u, v \in Σ^*\):
\begin{align*}
u \qlL v & \Lra \quad \text{[By Definition~\eqref{eq:Llanguage}]} \\
u^{-1}L \subseteq v^{-1}L & \Lra \quad \text{[\(A \subseteq B\Lra A^R \subseteq B^R\)]}\\
(u^{-1}L)^R \subseteq (v^{-1}L)^R & \Lra \quad\text{[Since \((u^{-1}L)^R = L^R(u^R)^{-1}\)]} \\
L^R(u^R)^{-1} \subseteq L^R(v^R)^{-1} & \Lra\quad\text{[By Definition~\eqref{eq:Rlanguage}]} \\
u^R \qr_{L^R} v^R \enspace . 
\end{align*}
Therefore, by Lemma~\ref{lemma:leftRightReverse}, \(\cF{\ell}(L)\) is isomorphic to \((\cF{r}(L^R))^R\).

\item %

For each \(u,v \in \Sigma^*\):
\begin{align*}
u \qlN v & \Lra \quad\text{[By Definition~\eqref{eq:LState}]}\\
\pre_u^{\cN^R}(F) \subseteq \pre_v^{\cN^R}(F) & \Lra\quad \text{[Since \(q \in \pre^{\cN^R}_{x}(F)\) if{}f \(q \in \post^{\cN}_{x^R}(I) \)]}\\
\post_{u^R}^{\cN}(I) \subseteq \post_{v^R}^{\cN}(I) & \Lra \quad\text{[By Definition~\eqref{eq:RState}]}\\
u^R \qlN v^R \enspace .
\end{align*}
It follows from Lemma~\ref{lemma:leftRightReverse} that \(\cG{\ell}(\cN)\) is isomorphic to \(\cG{r}(\cN^R)^R\).

\item %

Let \(ρ\) be the closure induced by \(\qrL\).
Let \(\cC = \tuple{\widetilde{Q}, \Sigma, \eta, \widetilde{I}, \widetilde{F}}\) be the canonical RFA for \(L\) and let \(\cF{r}(L) = \tuple{Q, \Sigma, \delta, I, F}\). 
Let \(\varphi: \widetilde{Q} \rightarrow Q\) be the mapping assigning to each state \(\widetilde{q}_i \in \widetilde{Q}\) of the form \(u^{-1}L\), the state \(ρ(u) \in Q\), with \(u \in \Sigma^*\).

We show that \(\varphi\) is an NFA isomorphism between \(\cC\) and \(\cF{r}(L)\).

Since 
\[u^{-1}L \subseteq L \Lra u \qrL \varepsilon \Lra \varepsilon \in ρ(u)\enspace ,\]
we have that \(u^{-1}L\) is an initial state of \(\mathcal{C}\) if{}f \(ρ(u)\) is an initial state of \(\cF{r}(L)\), hence \(\varphi(\widetilde{I}) = I\).

On the other hand, %
since 
\[\varepsilon \in u^{-1}L \Lra u \in L \enspace ,\]
we have that \(u^{-1}L\) is a final state of \(\mathcal{C}\) if{}f \(ρ(u)\) is a final state of \(\cF{r}(L)\), hence \(\varphi(\widetilde{F}) = F\).

Moreover, since 
\[ρ(u)\cdot a \subseteq ρ(v) \Lra v \qrL ua \Lra v^{-1}L \subseteq (ua)^{-1}L\enspace ,\]
we have that \(v^{-1}L = \eta(u^{-1}L, a)\) if and only if \(ρ(v) \in δ(ρ(u),a)\), for all \(u^{-1}L, v^{-1}L \in \wt{q}\) and \(a \in \Sigma\).

Finally, we need to show that \(\forall u \in Σ^*, \;  ρ(u) \in Q \Lra \exists q_i \in \widetilde{Q}, \; q_i = u^{-1}L\).
Observe that:
\begin{align*}
u^{-1}L = \bigcup_{x \qrn_L u} u^{-1}L & \Lra \quad \text{[By Definition~\eqref{eq:Rlanguage}]} \\
u^{-1}L = \bigcup_{x^{-1}L \subset u^{-1}L} x^{-1}L \enspace .
\end{align*}
Therefore, \(\forall u \in Σ^*, ρ(u) \text{ is \(L\)-prime} \Lra u^{-1}L \text{ is prime}\), hence \(\varphi(\wt{Q}) = Q\).

\item %

Recall that \(\cN^{\text{res}} = \tuple{Q_r, Σ, δ_r, I_r, F_r}\) is the RFA built by the residualization operation defined by \citet{denis2002residual} (see Chapter~\ref{chap:prel}).
Let \(\cG{r}(\cN) = \tuple{\widetilde{Q}, \Sigma, \widetilde{\delta}, \widetilde{I}, \widetilde{F}}\).

Next, we show that there is a surjective mapping \(\varphi\) that associates states and transitions of \(\cG{r}(\cN)\) with states and transitions of \(\cN^{\text{res}}\).
Moreover, if \(q \in \widetilde{Q}\) is initial (resp.\ final) then \(\varphi(q) \in Q_r\) is initial (resp.\ final) and \(q' \in \wt{δ}(q,a) \Lra \varphi(q') \in δ_r(\varphi(q),a)\).
In this way, we conclude that \(\cG{r}(\cN)\) is isomorphic to a sub-automaton of \(\cN^{\text{res}}\).

Finally, since \(\lang{\cN^{\text{res}}} = \lang{\cN}\) then it follows from Lemma~\ref{lemma: HrGeneratesL} that \(\lang{\cN^{\text{res}}} = \lang{\cN} = \lang{\cG{r}(\cN)}\).

Let \(ρ\) be the closure induced by \(\qrN\) and let \(\varphi: \widetilde{Q} \rightarrow Q_{r}\) be the mapping assigning to each state \(ρ(u) \in \widetilde{Q}\), the set \(\post_u^{\cN}(I) \in Q_{r}\) with \(u \in \Sigma^*\).

Since 
\[\varepsilon \in ρ(u) \Lra u \qrN \varepsilon \Lra \post_u^{\cN}(I) \subseteq \post_{\varepsilon}^{\cN}(I)\]
The initial states of \(\cG{r}(\cN)\) are mapped into the set the initial states of \(\cN^{\text{res}}\), hence \(\varphi(\wt{I}) = I_r\).

On the other hand, since
\[ρ(u) \subseteq L \Lra u \in L \Lra (\post_u^{\cN}(I) \cap F) \neq \varnothing \enspace ,\]
we have that the final states of \(\cG{r}(\cN)\), are mapped to the final states of \(\cN^{\text{res}}\), hence \(\varphi(\wt{F}) = F_r\).

Moreover, since
\[ρ(u)\cdot a \subseteq ρ(v) \Lra v \qrN ua \Lra \post_{v}^{\cN}(I) \subseteq \post_{ua}^{\cN}(I)\enspace ,\]
it follows that \(\forall u, v \in Σ^*\) such that \(\post_u^{\cN}(I), \post_v^{\cN}(I) \in Q_r\), we have 
\[\post_v^{\cN}(I) \in δ_r(\post_u^{\cN}(I),a) \Lra ρ(v) \in \wt{δ}(ρ(u), a)\enspace . \]

Finally, we show that \(\forall u \in Σ^*, \;  ρ(u) \in \widetilde{Q} \Ra \post_u^{\cN}(I) \in Q_r\).
By definition of \(\widetilde{Q}\) and \(Q_r\), this is equivalent to showing that for every word \(u \in Σ^*\), if \(\post_u^{\cN}(I)\) is coverable then \(ρ(u)\) is \(L\)-composite.
Observe that:
\begin{align*}
\post_u^{\cN}(I) = \hspace{-10pt}\bigcup_{\post_x^{\cN}(I) \subset \post_u^{\cN}(I)}\hspace{-10pt} \post_x^{\cN}(I) & \Lra \quad \text{[\(x \qrn_{\mathcal{N}} u \Lra \post_x^{\cN}(I) \subset \post_u^{\cN}(I)\)]} \\
\post_u^{\cN}(I) = \bigcup_{x \qrn_{\mathcal{N}} u} \post_x^{\cN}(I) & \Ra \quad \text{[Since \(W_{\post_u^{\cN}(I),T}^{\cN} = u^{-1}L\)]} \\
u^{-1}L = \bigcup_{x \qrn_{\mathcal{N}} u} x^{-1}L \enspace .
\end{align*}
It follows that if \(\post_u^{\cN}(I)\) is coverable then \(ρ(u)\) is \(L\)-composite, hence \(\varphi(\wt{Q}) \subseteq Q_r\).

\item %

As shown by Lemma~\ref{lemma:HlgeneratesL}, \(\cG{\ell}(\cN)\) is a co-RFA with no empty states and \(\lang{\cG{\ell}(\cN)} = \lang{\cN}\).
Therefore, it follows from Lemma~\ref{lemma:coResidual_qrL=qrN} that \(\cG{r}(\cG{\ell}(\cN))\) is isomorphic to \linebreak\(\cF{r}(\lang{\cG{\ell}(\cN)}) = \cF{r}(\lang{\cN})\).\qedhere
\end{myEnumA}
\end{proof}

Figure~\ref{Figure:diagramAutomata} summarizes all the connections between the automata constructions from Definition~\ref{def:FG}.

\begin{figure}[!ht]
\begin{minipage}[l]{0.45\textwidth}
\begin{tikzcd}[column sep=scriptsize, row sep=large]
\cN\ar[d, description, "R"',leftrightarrow] \ar[r, "\cG{\ell}"] \ar[rr, start anchor=70, bend left=25, "\cF{r}"] & \cG{\ell}(\cN) \ar[d, description, "R"',leftrightarrow] \ar[r, "\cG{r}"] & \cF{r}(\lang{\cN}) \ar[d, description, "R"',leftrightarrow]\\
\cN^R \ar[r, "\cG{r}"] \ar[rr, start anchor=290, bend right=25, "\cF{\ell}"] & \cG{r}(\cN^R) \ar[r, "\cG{\ell}"] & \cG{\ell}(\cG{r}(\cN^R))
\end{tikzcd}
\end{minipage}\hfill
\begin{minipage}[r]{0.54\textwidth}
The upper part of the diagram follows from Theorem~\ref{theoremF}~\ref{lemma:LS+RS=RN}, the squares follow from Theorem~\ref{theoremF}~\ref{lemma:AlRequalArNR} and the bottom curved arc follows from 
Theorem~\ref{theoremF}~\ref{lemma:FlisomorphicRfrR}.
Incidentally, the diagram shows a new relation which is a consequence of the left-right dualities between \(\qlL\) and \(\qrL\), and \(\qlN\) and \(\qrN\): \(\cF{\ell}(\lang{\cN^R})\) is isomorphic to \(\cG{\ell}(\cG{r}(\cN^R))\).
\end{minipage}

\caption{Relations between the constructions \(\cG{\ell},\cG{r},\cF{\ell}\) and \(\cF{r}\).
Note that constructions \(\cF{r}\) and \(\cF{\ell}\) are applied  to the language generated by the automaton in the origin of the labeled arrow while constructions \(\cG{r}\) and \(\cG{\ell}\) are applied directly to the automaton.}%
\label{Figure:diagramAutomata}
\end{figure}

It is well-known that determinizing a deterministic automata yields the same automaton, i.e. \(\cD^D = \cD\) for every DFA \(\cD\).
As a consequence, determinizing twice and automaton is the same as doing it once, i.e. \((\cN^{D})^{D} = \cN^D\).
However, it is not clear that the same holds for our residualization operation, i.e. it is not clear whether \(\cG{r}(\cG{r}(\cN)) = \cG{r}(\cN)\).

The following lemma gives a sufficient condition on an RFA \(\mathcal{H}\) built with our right automata construction so that applying our residualization operation yields the same automaton, i.e. \(\cG{r}(\mathcal{H}) \!=\! \mathcal{H}\).
In particular, we find that \(\cF{r}(L)\) is invariant to our residualization operation \(\cG{r}\).

\begin{lemma}
\label{lemma:qrHEqualqrifHsc}
Let \(L\) be a regular language and let \(\qr\) be a right \(L\)-consistent quasiorder.
Let \(\mathcal{H}\!=\!\cH^{r}(\qr,L)\).
If \(\mathcal{H}\) is a strongly consistent RFA then \(\mathord{\qr_{\mathcal{H}}} = \mathord{\qr}\).
\end{lemma}

\begin{proof}
Let \(\cN \!=\! \tuple{Q, Σ, δ, I, F}\) and \(\mathcal{H} \!=\! \tuple{\wt{Q}, Σ, \wt{δ}, \wt{I}, \wt{F}}\).
As shown by Lemma~\ref{lemma: HrGeneratesL}, \(\mathcal{H} \!=\! \cH^{r}(\qr, L)\) is an RFA generating \(L\), hence each state of \(\mathcal{H}\) is an \(L\)-prime principal \(ρ_{\qr}(u)\) whose right language is the quotient \(u^{-1}L\) for some \(u \in Σ^*\).

Observe that, by definition, 
\(\mathord{\qr_{\mathcal{H}}} = \mathord{\qr} \Lra \left( \forall u,v \in Σ^*, \; \post_u^{\mathcal{H}}(\wt{I}) \subseteq \post_v^{\mathcal{H}}(\wt{I}) \Lra u \qr v\right)\).
Next we prove that:
\begin{equation}%
\label{eq:PostQuotient}
\post_u^{\mathcal{H}}(\wt{I}) = \{ρ_{\qr}(x) \in \wt{Q} \mid x \qr u\} \enspace .
\end{equation}
First, we show that \(\post_u^{\mathcal{H}}(\wt{I}) \subseteq \{ρ_{\qr}(x) \in \wt{Q} \mid x \qr u\}\). 
To simplify the notation, let \(ρ\) denote \(ρ_{\qr}\).
\begin{align*}
ρ(x) \in \post_u^{\mathcal{H}}(\wt{I}) & \Lra \quad \text{[By definition of \(\post_u^{\mathcal{H}}(\wt{I})\)]}\\
\exists ρ(x_0) \in \wt{I}, \; u \in W^{\mathcal{H}}_{ρ(x_0),ρ(x)} & \Ra \quad \text{[By Definition~\ref{def:right-const:qo}]} \\
\exists ρ(x_0) \in \wt{Q}, \;\varepsilon \in ρ(x_0) \land ρ(x_0) \cdot u \subseteq ρ(x) & \Lra \quad \text{[By definition of \(ρ\)]} \\
\exists ρ(x_0) \in \wt{Q}, \; x_0 \qr \varepsilon \land x \qr u\cdot x_0 & \Ra \quad \text{[By mon. and trans. of \(\qr\)]} \\
x \qr u \enspace .
\end{align*}
We now prove the reverse inclusion. 
Let \(ρ(u), ρ(x) \in \wt{Q}\) be such that \(x \qr u\).
Then,
\begin{align*}
ρ(u) \in \wt{Q} & \Ra \quad \text{[By Lemma~\ref{lemma: HrGeneratesL}]}\\
W^{\mathcal{H}}_{ρ(u),F} = u^{-1}L & \Ra \quad \text{[Since \(\mathcal{H}\) is str. cons.]} \\
u \in W_{I,ρ(u)}^{\mathcal{H}} & \Ra \quad \text{[By def. \(W_{S,T}^{\mathcal{H}}\), \(u = za\)]} \\
\exists ρ(y) \in \wt{Q}, ρ(u_0) \in \wt{I},\; z \in W_{ρ(u_0), ρ(y)}  \land a \in W_{ρ(y), ρ(u)}& \Ra \quad \text{[By Definition~\ref{def:right-const:qo}]} \\
z \in W_{ρ(u_0), ρ(y)} \land ρ(y) \cdot a \subseteq ρ(u)& \Ra \quad \text{[By def. \(ρ = ρ_{\qr}\)]} \\
z \in W_{ρ(u_0), ρ(y)} \land u\qr y\cdot a & \Ra \quad \text{[Since \(x \qr u\)]} \\
z \in W_{ρ(u_0), ρ(y)} \land x \qr ya & \Ra \quad \text{[By Def.~\ref{def:right-const:qo}]} \\
z \in W_{ρ(u_0), ρ(y)} \land ρ(x) \in δ(ρ(y), a) & \Ra \quad \text{[By def. \(\post_u(I)\)]} \\
ρ(x) \in \post_u(I) \enspace .
\end{align*}
It follows from Equation~\eqref{eq:PostQuotient} that \(\post_u^{\mathcal{H}}(I) \subseteq \post^{\mathcal{H}}_v(I) \Lra u \qr v\), i.e. \(\mathord{\qr_{\mathcal{H}}} = \mathord{\qr}\).
\end{proof}

Finally, note that if \(\mathord{\qrL}=\mathord{\qrN}\) then clearly the automata \(\cF{r}(L)\) and \(\cG{r}(\cN)\) coincide for any NFA \(\cN\) with \(L = \lang{\cN}\).
The following result shows that the reverse implication also holds.

\begin{lemma}\label{lemma:qrlEqualqrNResEqualCan}
Let \(\cN\) be an NFA with \(L = \lang{\cN}\).
Then \(\mathord{\qrL} = \mathord{\qrN}\) if{}f \(\cG{r}(\cN)\) is isomorphic to \(\cF{r}(L)\).
\end{lemma}
\begin{proof}
As shown by Theorem~\ref{theoremF}~\ref{theorem:CanonicalRFAlanguage}, \(\cF{r}(L)\) is the canonical RFA for \(L\), hence it is strongly consistent and, by Lemma~\ref{lemma:qrHEqualqrifHsc}, we have that \(\mathord{\qr_{\cF{r}(L)}} = \mathord{\qrL}\).
On the other hand, if \(\cG{r}(\cN)\) is isomorphic to \(\cF{r}(L)\) we have that \(\mathord{\qr_{\cG{r}(\cN)}} = \mathord{\qr_{\cF{r}(L)}}\), and by Lemma~\ref{lemma:qrHEqualqrifHsc}, \(\mathord{\qr_{\cG{r}(\cN)}} =~ \qrN\).
It follows that if \(\cG{r}(\cN)\) is isomorphic to \(\cF{r}(L)\) then \(\mathord{\qrL} = \mathord{\qrN}\).

Finally, if \(\mathord{\qrL} = \mathord{\qrN}\) then \(\cH^{r}(\qrL, L) = \cH^{r}(\qrN, \lang{\cN})\), in other words,  \( \cF{r}(L) = \cG{r}(\cN)\).
\end{proof}

The following example illustrates the differences between our residualization operation, \linebreak\(\cG{r}(\cN)\), and the one defined by~\citet{denis2001residual}, \(\cN^{\text{res}}\), on a given NFA \(\cN\): the automaton \(\cG{r}(\cN)\) has, at most, as many states as \(\cN^{\text{res}}\).
This follows from the fact that for every \(u \in Σ^*\), if \(\post_u^{\cN}(I)\) is coverable then \(ρ_{\qrN}(u)\) is composite but \emph{not} vice-versa.

\begin{example}%
\label{example:residualization}
Let \(\cN = \tuple{Q, Σ, δ, I, F}\) be the automata on the left of Figure~\ref{fig:Residuals} and let \(L = \lang{\cN}\).
In order to build \(\cN^{\text{res}}\) we compute \(\post_u^{\cN}(I)\), for all \(u \in Σ^*\).
Let \(C \ud L^c \setminus \{\varepsilon, a, b, c\}\).
\begin{align*}
\post_{\varepsilon}^{\cN}(I)  & = \{0\} & \post_a^{\cN}(I)  & = \{1,2\}& \forall w \in L, \;\post_{w}^{\cN}(I)  & = \{5\} \\
\post_c^{\cN}(I)  & = \{1, 2, 3, 4\} & \post_b^{\cN}(I)  & = \{1,3\} &\forall w \in C,\; \post_{w}^{\cN}(I) & = \varnothing
\end{align*}
Since none of these sets is coverable by the others, they are all states of \(\cN^{\text{res}}\).
The resulting RFA \(\cN^{\text{res}}\) is shown in the center of Figure~\ref{fig:Residuals}.

On the other hand, let us denote \(ρ_{\qrN}\) simply by \(ρ\).
In order to build \(\cG{r}(\cN)\) we need to compute the principals \(ρ(u)\), for all \(u \!\in\! Σ^*\).
By definition of \(\qrN\), we have that \(w \!\in\! ρ(u) \Lra \post_u^{\cN}(I) \!\subseteq\! \post_w^{\cN}(I)\).
Therefore, we obtain:
\begin{align*}
ρ(\varepsilon)  & = \{\varepsilon\}  & ρ(a)  &= \{a,c\} & ρ(c)  & = \{c\} & ρ(b)  & = \{b,c\} & \forall w \in L,\;  ρ(w)  & = L & \forall w \in C,\;ρ(w) & = \Sigma^* 
\end{align*}

Since \(a \qrn_{\cN} c\), \(b \qrn_{\cN} c\) and \(\forall w\in Σ^*, \; cw \subseteq L \Lra \big(aw \subseteq L \lor bw \subseteq L\big)\), it follows that \(ρ(c)\) is \(L\)-composite.
The resulting RFA \(\cG{r}(\cN)\) is shown on the right of Figure~\ref{fig:Residuals}.
\eox%
\end{example}

\begin{figure}[!ht]%
    \centering
\begin{minipage}{0.3\textwidth}
  \begin{tikzpicture}[->,>=stealth',shorten >=1pt,auto,node distance=3mm and 5mm,thick,initial text=]
   \tikzstyle{every state}=[scale=0.75,fill=customblue!60,draw=blue!60,text=black,style={draw,ellipse,inner sep=0pt}]

  \node[initial, state] (0) {\(0\)};
  \node[state] (2) [right=of 0, yshift=0.6cm] {\(2\)};
  \node[state] (1) [above=of 2, yshift=-0.2cm] {\(1\)};
  \node[state] (3) [right=of 0, yshift=-0.6cm]{\(3\)};
  \node[state] (4) [below=of 3, yshift=0.2cm] {\(4\)};
  \node[accepting, state] (5) [right=of 0, xshift=1.5cm] {\(5\)};
  
  \path (0) edge [bend left] node {\small\(a,b,c\)} (1)
  		(0) edge node [xshift=10pt, yshift=0pt] {\small\(a,c\)} (2)
  		(0) edge node [pos=0.25, yshift=-5pt] {\small\(b,c\)} (3)
  		(0) edge [bend right] node {\small\(c\)} (4);

  \path (1) edge [bend left] node {\small\(a\)} (5)
  		(2) edge node {\small\(b\)} (5)
  		(3) edge node {\small\(c\)} (5)
  		(4) edge [bend right] node [yshift=-5pt, xshift=4pt] {\small\(a,b,c\)} (5);
  \end{tikzpicture}
  \end{minipage}%
\begin{minipage}{0.39\textwidth}
  \begin{tikzpicture}[->,>=stealth',shorten >=1pt,auto,node distance=3mm and 5mm,thick,initial text=]
 \tikzstyle{every state}=[scale=0.75,fill=customblue!60,draw=blue!60,text=black,style={draw,ellipse,inner sep=0pt}]

  \node[initial, state] (0) {\(\left\{\hspace{-5pt}\begin{array}{l}0\end{array}\hspace{-5pt}\right\}\)};
  \node[state] (2) [right=of 0] {\(\left\{1,2,3,4\right\}\)};
  \node[state] (1) [above=of 2] {\(\left\{1, 2\right\}\)};
  \node[state] (3) [below=of 2] {\(\left\{1, 3\right\}\)};
  \node[accepting, state] (4) [right=of 2, xshift=0.2cm] {\(\left\{5\right\}\)};
  
  \path (0) edge [bend left] node {\small\(a,c\)} (1)
  		(0) edge node {\small\(c\)} (2)
  		(0) edge [bend right] node [yshift=-5pt] {\small\(b,c\)} (3);

  \path (1) edge [bend left] node {\small\(a,b\)} (4)
  		(2) edge node {\small\(a,b,c\)} (4)
  		(3) edge [bend right] node [xshift=4pt, yshift=-4pt] {\small\(a,c\)} (4);
  \end{tikzpicture}
  \end{minipage}%
\begin{minipage}{0.3\textwidth}
  \begin{tikzpicture}[->,>=stealth',shorten >=1pt,auto,node distance=3mm and 4.2mm,thick,initial text=]
   \tikzstyle{every state}=[scale=0.75,fill=customblue!60,draw=blue!60,text=black,style={draw,ellipse,inner sep=0pt}]
  
  \node[initial, state] (0) {\(ρ(\varepsilon)\)};
  \node[state] (1) [right=of 0, yshift=1cm] {\(ρ(a)\)};
  \node[state] (2) [right=of 0, yshift=-1cm] {\(ρ(b)\)};
  \node[accepting, state] (3) [right=of 0, xshift=1.3cm] {\(ρ(aa)\)};
  
  \path (0) edge [bend left] node {\small\(a,c\)} (1)
  		(0) edge [bend right] node [yshift=-2pt, xshift=-4pt] {\small\(b,c\)} (2);

  \path (1) edge [bend left] node [xshift=-5pt]{\small\(a,b\)} (3)
  		(2) edge [bend right] node [xshift=5pt, yshift=-3pt] {\small\(a,c\)} (3);
  \end{tikzpicture}
  \end{minipage}
\caption{Left to right: an NFA \(\cN\) and the RFAs \(\cN^{\text{res}}\) and \(\cG{r}(\cN)\). We omit the empty states for clarity.}
\label{fig:Residuals}
\end{figure}
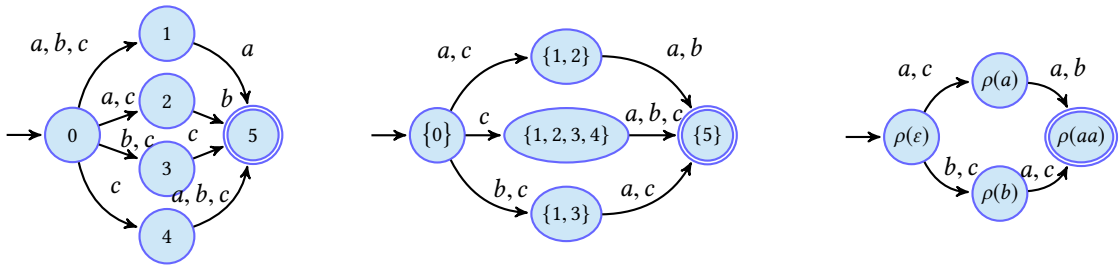

\section{Double-Reversal Method for Building the Canonical RFA}%
\label{sec:Novel}
\citet{denis2002residual} show that their residualization operation satisfies the residual-equivalent of the double-reversal method for building the minimal DFA\@.
More specifically, they prove that if an NFA \(\cN\) is a co-RFA with no empty states then their residualization operation applied to \(\cN\) results in the canonical RFA for \(\lang{\cN}\).
As a consequence, 
\[(((\cN^R)^{\text{res}})^R)^{\text{res}} \text{ is the canonical RFA for } \lang{\cN}\enspace .\]

In this section we show that the residual-equivalent of the double-reversal method works when using our automata constructions based on quasiorders, i.e.
\[\cG{r}((\cG{r}(\cN^R))^R) \text{ is isomorphic to } \cF{r}(\cN)\enspace .\]
Then, we generalize this method along the lines of the generalization of the double-reversal method for building the minimal DFA given by \citet{Brzozowski2014}.

To this end, we extend the work of \citet{ganty2019congruence} where they use congruences to offer a new perspective on the generalization of \citet{Brzozowski2014}.
By switching from congruences to monotone quasiorders, we are able to give a \emph{necessary} and \emph{sufficient} condition on an NFA \(\cN\) that guarantees that our residualization operation yields the canonical RFA for \(\lang{\cN}\).
Finally, we compare our generalization with the one given by \citet{tamm2015generalization}.

\subsection{Double-reversal Method}
We give a simple proof of the double-reversal method for building the canonical RFA for the language generated by a given NFA \(\cN\)\@.

\begin{theorem}[Double-Reversal]%
\label{theorem:DoubleReversal}
Let \(\cN\) be an NFA\@.
Then \(\cG{r}((\cG{r}(\cN^R))^R)\) is isomorphic to the canonical RFA for \(\lang{\cN}\).
\end{theorem}
\begin{proof}
It follows from Theorem~\ref{theoremF}~\ref{lemma:AlRequalArNR}, \ref{theorem:CanonicalRFAlanguage} and \ref{lemma:LS+RS=RN}.
\end{proof} 

Note that Theorem~\ref{theorem:DoubleReversal} can be inferred from Figure~\ref{Figure:diagramAutomata} by following the path starting at \(\cN\), labeled with \(R-\cG{r}-R-\cG{r}\) and ending in \(\cF{r}(\lang{\cN})\).

\subsection{Generalization of the Double-reversal Method}
Next we show that residualizing an automaton yields the canonical RFA if{}f the left language of every state is closed w.r.t. the right Nerode quasiorder.

\begin{theorem}%
\label{theorem:canonicalreverserestic}
Let \(\cN = \tuple{Q,\Sigma,\delta,I,F}\) be an NFA with \(L=\lang{\cN}\).
Then \(\cG{r}(\cN)\) is the canonical RFA for \(L\) if{}f \(\forall q \in Q,\;  ρ_{\qrL}(W_{I,q}^{\cN}) = W_{I,q}^{\cN}\).
\end{theorem}
\begin{proof}
We first show that \(\forall q \in Q,\ ρ_{\qrL}(W_{I,q}^{\cN}) = W_{I,q}^{\cN}\) is a \emph{necessary} condition, i.e.\ if \(\cG{r}(\cN)\) is the canonical RFA for \(L\) then \(\forall q \in Q,\ ρ_{\qrL}(W_{I,q}^{\cN}) = W_{I,q}^{\cN}\) holds.

By Lemma~\ref{lemma:qrlEqualqrNResEqualCan} we have that if \(\cG{r}(\cN)\) is the canonical RFA for \(L\) then \(\mathord{\qrL} = \mathord{\qrN}\).
Moreover:
\begin{align*}
ρ_{\qrL}(W_{I,q}^{\cN}) & = \; \text{[By definition of \(ρ_{\qrL}\)]}\\
\{w \in \Sigma^* \mid \exists u \in W_{I,q}^{\cN}, \; u^{-1}L \subseteq w^{-1}L\} & = \; \text{[Since \(\mathord{\qrL} = \mathord{\qrN}\)]} \\
\{w \in \Sigma^* \mid \exists u \in W_{I,q}^{\cN}, \; \post_u^{\cN}(I) \subseteq \post_w^{\cN}(I)\} & \subseteq  \; \text{[Since \( u \in W_{I,q}^{\cN} \Lra q \in \post_u^{\cN}(I)\)]}\\
\{w \in \Sigma^* \mid q \in \post_w^{\cN}(I)\} & = \; \text{[By definition of \(W_{I,q}^{\cN}\)]}\\
W_{I,q}^{\cN} \enspace .
\end{align*}
By reflexivity of \(\qrL, \) we conclude that \(ρ_{\qrL}(W_{I,q}^{\cN}) = W_{I,q}^{\cN} \).

Next, we show that \(\forall q \in Q,\;  ρ_{\qrL}(W_{I,q}^{\cN}) = W_{I,q}^{\cN}\) is also a \emph{sufficient} condition.
By Lemma~\ref{lemma:positive_atoms} and condition \(\forall q \in Q,\;  ρ_{\qrL}(W_{I,q}^{\cN}) = W_{I,q}^{\cN}\), we have that
\begin{equation}\label{eq:rhoNIntersectOfrhoL}
ρ_{\qrN}(u) = \bigcap{\textstyle_{q \in \post^{\cN}_u(I)}} W_{I,q}^{\cN} = \bigcap{\textstyle_{q \in \post^{\cN}_u(I)}} ρ_{\qrL}(W_{I,q}^{\cN})\enspace .
\end{equation}

Since \(u \in ρ_{\qrL}(W_{I,q}^{\cN})\) for all \(q \in \post_u^{\cN}(I)\), it follows that \(ρ_{\qrL}(u) \subseteq ρ_{\qrL}(W_{I,q}^{\cN})\) for all \(q \in \post_u^{\cN}(I)\) and, since \(ρ_{\qrN}(u) = \bigcap\textstyle{_{q \in \post^{\cN}_u(I)}} ρ_{\qrL}(W_{I,q}^{\cN})\), we have that \(ρ_{\qrL}(u) \subseteq ρ_{\qrN}(u)\) for every \(u \in Σ^*\), i.e., \(\mathord{\qrL} \subseteq \mathord{\qrN}\).

On the other hand, as shown by \citet{deLuca1994}, we have that \(\mathord{\qrN}\subseteq \mathord{\qrL}\).
We conclude that \(\mathord{\qrN} = \mathord{\qrL}\), hence \(\cG{r}(\cN) = \cF{r}(L)\).
\end{proof}

It is worth to remark that Theorem~\ref{theorem:canonicalreverserestic} does not hold when considering the residualization operation \(\cN^{\text{res}}\)~\cite{denis2002residual}.
As a counterexample we have the automaton \(\cN\) in Figure~\ref{fig:Residuals} where \(\cG{r}(\cN)\) is the canonical RFA for \(\lang{\cN}\), hence \(\cN\) satisfies the condition of Theorem~\ref{theorem:canonicalreverserestic}, while \(\cN^{\text{res}}\) is not canonical.

\subsubsection{Co-atoms and co-rests}
The condition of Theorem~\ref{theorem:canonicalreverserestic} is analogue to the one \citet[Theorem 16]{ganty2019congruence} give for building the minimal DFA, except that the later is formulated in terms of congruences instead of quasiorders.
In that case they prove that, given an NFA \(\cN=\tuple{Q,Σ,δ,I,F}\) with \(L = \lang{\cN}\), 
\[\cN^D \text{ is the minimal DFA for \(L\) if{}f} \forall q \in Q, \; ρ_{\rrL}(W^{\cN}_{I,q}) = W^{\cN}_{I,q}\enspace ,\]
where \(\mathord{\rrL} \ud \mathord{\qrL} \cap \mathord{(\qrL)^{-1}}\) is the right Nerode's congruence.

Moreover, \citet{ganty2019congruence} show that the principals of \(\rrL\) coincide with the so-called \emph{co-atoms}, which are non-empty intersections of complemented and uncomplemented right quotients of the language. 
This allowed them to connect their result with the generalization of the double-reversal method for DFAs of \citet{Brzozowski2014}, who establish that determinizing an NFA \(\cN\) yields the minimal DFA for \(\lang{\cN}\) if{}f the left languages of the states of \(\cN\) are unions of co-atoms of \(\lang{\cN}\).

Next, we give a formulation of the condition from Theorem~\ref{theorem:canonicalreverserestic} along the lines of the one given by \citet{Brzozowski2014} for their generalization of the double-reversal method for building the minimal DFA.

To do that, let us call the intersections used in Lemma~\ref{lemma:positive_atoms} to describe the principals of \(\qlL\) and \(\qrL\) as \emph{rests} and \emph{co-rests} of \(L\), respectively.

\begin{definitionNI}[Rest and Co-rest]\index{Rest}\index{co-Res}%
\label{def:rest}
Let \(L\) be a regular language.
A \emph{rest} (resp. \emph{co-rest}) is any non-empty intersection of left (resp.\ right) quotients of \(L\). \hfill\rule{0.5em}{0.5em}
\end{definitionNI}

As shown by Theorem~\ref{theorem:canonicalreverserestic}, residualizing an NFA \(\cN\) yields the canonical RFA for \(\lang{\cN}\) if{}f the left language of every state of \(\cN\) satisfies \(ρ_{\qrL}(W_{I,q}^{\cN}) = W_{I,q}^{\cN}\).
By definition, \(ρ_{\qrL}(S) = S\) if{}f \(S\) is a union of principals of \(\qrL\).

Therefore we derive the following statement, equivalent to Theorem~\ref{theorem:canonicalreverserestic}, that we consider as the residual-equivalent of the generalization of the double-reversal for building the minimal DFA~\cite{Brzozowski2014}.

\begin{corollary}%
Let \(\cN\) be an NFA with \(L=\lang{\cN}\).
Then \(\cG{r}(\cN)\) is the canonical RFA for \(L\) if{}f the left languages of \(\cN\) are union of co-rests.
\end{corollary}

\subsubsection{Tamm's Generalization of the Double-reversal Method for RFAs}\label{sec:TammGen}

\citet{tamm2015generalization} generalized the double-reversal method of \citet{denis2002residual} by showing that \(\cN^{\text{res}}\) is the canonical RFA for \(\lang{\cN}\) if{}f the left languages of \(\cN\) are union of the left languages of the canonical RFA for \(\lang{\cN}\).

In this section, we compare the generalization of \citet{tamm2015generalization} with ours.
The two approaches differ in the definition of the residualization operation they consider and, as we show next, the sufficient and necessary condition from Theorem~\ref{theorem:canonicalreverserestic} is more general than that of \citet[Theorem 4]{tamm2015generalization}.

\begin{lemma}\label{lemma:WeRaTamm}
Let \(\cN = \tuple{Q, Σ, δ, I, F}\) be an NFA with \(L = \lang{\cN}\) and let \(\cC = \cF{r}(\qrL, L) = \tuple{\wt{Q}, Σ, \wt{δ}, \wt{I}, \wt{F}}\) be the canonical RFA for \(L\).
Then
\[W_{I,q}^{\cN} = \bigcup_{q \in \wt{Q}}W_{\wt{I},q}^{\cC} \implies ρ_{\qrL}(W_{I,q}^{\cN}) = W_{I,q}^{\cN}\enspace .\]
\end{lemma}
\begin{proof}
Since the canonical RFA, \(\cC\), is strongly consistent then it follows from Lemma~\ref{lemma:qrHEqualqrifHsc} that \(\mathord{\qr_{\mathcal{C}}} = \mathord{\qrL}\) and, consequently, \(\cG{r}(\cC)\) is isomorphic to \(\cF{r}(L)\).
It follows from Theorem~\ref{theorem:canonicalreverserestic} that \(ρ_{\qrL}(W_{\wt{I},q}^{\cC}) = W_{\wt{I},q}^{\cC}\) for every \(q \in \wt{Q}\).
Therefore,
\begin{align*}
ρ_{\qrL}(W_{I,q}^{\cN}) & = \quad \text{[Since \(W_{I,q}^{\cN} = {\textstyle\bigcup_{q \in \wt{Q}}W_{\wt{I},q}^{\cC}}\) and \(ρ_{\qrL}(\cup S_i) = \cup ρ_{\qrL}(S_i)\)]} \\
{\textstyle\bigcup_{q \in \wt{Q}}ρ_{\qrL}(W_{\wt{I},q}^{\cC})} & = \quad \text{[Since \(ρ_{\qrL}(W_{\wt{I},q}^{\cC}) = W_{\wt{I},q}^{\cC}\) for every \(q \in \wt{Q}\)]}\\
{\textstyle\bigcup_{q \in \wt{Q}}W_{\wt{I},q}^{\cC}}\enspace .\tag*{\qedhere} 
\end{align*}
\end{proof}

Observe that, since the canonical RFA \(\cC = \tuple{\wt{Q}, Σ, \wt{δ}, \wt{I}, \wt{F}}\) for a language \(L\) is strongly consistent, the left language of each state is a principal of \(\qrL\).
In particular, if the right language of a state is \(u^{-1}L\) then its left language is the principal \(ρ_{\qrL}(u)\).
Therefore, if \(W_{I,q}^{\cN} = \bigcup_{q \in \wt{Q}}W_{\wt{I},q}^{\cC}\) then \(W_{I,q}^{\cN}\) is a closed set in \(ρ_{\qrL}\).
However, the reverse implication does not hold since \emph{only the \(L\)-prime principals are left languages of states of \(\cC\)}.

On the other hand, Lemma~\ref{lemma:CompositeIntersection} shows that \(L\)-composite principals can be described as intersections of \(L\)-prime principals when we consider the Nerode's quasiorder \(\qrL\).
As a consequence, our residualization operation applied on an NFA \(\cN\) yields the canonical RFA for \(\lang{\cN}\) if{}f the left languages of states of \(\cN\) are \emph{union of non-empty intersections of left languages of the canonical RFA} while \citet{tamm2015generalization} proves that \(\cN^{\text{res}}\) yields to the canonical RFA if{}f the left languages of states of \(\cN\) are \emph{union of left languages of the canonical RFA}.

\section{Learning Residual Automata}\label{sec:LearningNL:qo}
\citet{bollig2009angluin} devised the NL\(^*\) algorithm for learning the canonical RFA for a given regular language.
The algorithm describes the behavior of a \demph{Learner} that infers a language \(L\) by performing membership queries on \(L\) (which are answered by a \demph{Teacher}) and equivalence queries between the language generated by a candidate automaton and \(L\) (which are answered by an \demph{Oracle}).
The algorithm terminates when the \emph{Learner} builds an RFA generating the language \(L\).

For the shake of completeness, we offer an overview of the NL\(^*\) algorithm as presented by~\citet{bollig2009angluin}.

\subsection{The \texorpdfstring{NL\(^*\)}{NL*} Algorithm~\texorpdfstring{\cite{bollig2009angluin}}{}}
The \emph{Learner} maintains a prefix-closed finite set \(\Pref \subseteq Σ^*\) and a suffix-closed finite set \(\Suf \subseteq Σ^*\).
The \emph{Learner} groups the words in \(\Pref\) by building a \demph{table} \(T = (\cT, \Pref, \Suf)\) where \(T: (\Pref \cup \Pref \cdot Σ) \times \Suf \to \{{+}, {-}\}\) is a function such that for every \(u \in \Pref \cup \Pref \cdot Σ\) and \(v \in \Suf\) we have that \(T(u,v) = {+} \Lra uv \in L\).
Otherwise \(T(u,v) = {-}\).

For every word \(u \in \Pref \cup \Pref \cdot Σ\), define the function \(\row(u): \Suf \to \{{+}, {-}\}\) as \(\row(u)(v) \ud T(u,v)\).
The set of all rows of a table \(\cT\) is denoted by \(\Rows(\cT)\).

The algorithm uses the table \(\cT = (T, \Pref, \Suf)\) to build an automaton whose states are some of the rows of \(\cT\). 
In order to do that, it is necessary to define the notions of \emph{union} of rows, \emph{prime} row and \emph{composite} row.

\begin{definition}[Join Operator]%
\label{def:join}
Let \(\cT = (T, \Pref, \Suf)\) be a table.
For every pair of rows \(r_1, r_2 \in \Rows(\cT)\), define the \emph{join} \(r_1 \sqcup r_2: \Suf \to \{{+},{-}\}\) as:
\[\forall x \in \Suf, \; (r_1 \sqcup r_2)(x) \ud \left\{\begin{array}{ll}
{+} & \text{if } r_1(x) = + \lor r_2(x) = + \\
{-} & \text{otherwise}\end{array}\right. \tag*{\rule{0.5em}{0.5em}}\]
\end{definition}

Note that the join operator is associative, commutative and idempotent.
However, the join of two rows is not necessarily a row of \(\cT\).

\begin{definition}[Covering Relation]%
\label{def:coverRow}
Let \(\cT = (T, \Pref, \Suf)\) be a table.
Then, for every pair of rows \(r_1, r_2 \in \Rows(\cT)\) we have that
\[r_1 \sqsubseteq r_2 \udiff \forall x \in \Suf, \; r_1(x) = {+} \Ra r_2(x) = {+}\enspace .\]
We write \(r_1 \sqsubset r_2\) to denote \(r_1 \sqsubseteq r_2 \) and \(r_1 \neq r_2\). \hfill\rule{0.5em}{0.5em}
\end{definition}

\begin{definition}[Composite and Prime Rows]%
\label{def:PrimeRow}
Let \(\cT = (T, \Pref, \Suf)\) be a table.
We say a row \(r \in \Rows(\cT)\) is \(\cT\)-\emph{composite} if it is the join of all the rows that it strictly covers, i.e. \(r = \bigsqcup_{r' \in \Rows(\cT), \; r' \sqsubset r} r'\).
Otherwise, we say \(r\) is \(\cT\)-\emph{prime}. \hfill\rule{0.5em}{0.5em}
\end{definition}

\begin{definition}[Closed and Consistent Table]%
\label{def:Table}
Let \(\cT = (T, \Pref, \Suf)\) be a table.
Then
\begin{myEnumA}
\item \(\cT\) is \emph{closed} if{}f \label{def:Table:closed}
\(\forall u \in \Pref, a \in Σ, \; \row(ua) = \bigsqcup \{\row(v) \mid v \in \Pref, \; \row(v) \sqsubseteq \row(ua) \land \row(v) \text{ is \(\cT\)-prime}\}\).
\item \(\cT\) is \emph{consistent} if{}f \label{def:Table:Consistent} \(\row(u) \sqsubseteq \row(v) \Ra \row(ua) \sqsubseteq \row(va) \text{ for every \(u,v \in \Pref\) and \(a \in Σ\)}\)\eod
\end{myEnumA}
\end{definition}

At each iteration of the algorithm, the \emph{Learner} checks whether the current table \(\cT = (T, \Pref, \Suf)\) is closed and consistent.

If \(\cT\) is not closed, then it finds \(\row(ua)\) with \(u \in \Pref, a \in Σ\) such that \(\row(ua)\) is \(\cT\)-prime and it is not equal to some \(\row(v)\) with \(v \in \Pref\).
Then the \emph{Learner} adds \(ua\) to \(\Pref\) and updates the table \(\cT\).

Similarly, if \(\cT\) is not consistent, the \emph{Learner} finds \(u, v \in \Pref, a \in Σ, x \in \Suf\) such that \(\row(u) \subseteq \row(v)\) but \(\row(ua)(x) = {+} \land \row(va)(x) = {-}\).
Then the \emph{Learner} adds \(ax\) to \(\Suf\) and updates \(\cT\).

When the table \(\cT\) is closed and consistent, the \emph{Learner} builds the RFA \(\cR(\cT)\).

\begin{definition}[Automata Construction \(\cR(\cT)\)]
Let \(\cT = (T, \Pref, \Suf)\) be a closed and consistent table.
Define the automaton \(\cR(\cT) \ud \tuple{Q, Σ, δ, I, F}\) with \(Q = \{\row(u) \mid u \in \Pref \land \row(u) \text{ is \(\cT\)-prime}\}\), \(I {=} \{\row(u) \in Q \mid \row(u) \sqsubseteq \row(\varepsilon)\}\), \(F {=} \{\row(u) \in Q\mid \row(u)(\varepsilon) = +\}\) and \(\row(v) \in δ(\row(u),a) = \{\row(v) \in Q \mid \row(v) \sqsubseteq \row(ua)\}\) for all \(\row(u) \in Q, a \in Σ\). \hfill\rule{0.5em}{0.5em}
\end{definition}

The \emph{Learner} asks the \emph{Oracle} whether \(\lang{\cR(\cT)} = L\).
If the \emph{Oracle} answers \emph{yes} then the algorithm terminates.
Otherwise, the \emph{Oracle} returns a counterexample \(w\) for the language equivalence.
Then the \emph{Learner} adds every suffix of \(w\) to \(\Suf\), updates the table \(\cT\) and repeats the process.

\subsection{The \texorpdfstring{NL\(^{\qo}\)}{NLqo} Algorithm}
In this section we present a quasiorder-based perspective on the NL\(^*\) algorithm in which the \emph{Learner} iteratively refines a quasiorder on \(Σ^*\) by querying the \emph{Teacher} and uses and adaption of the automata construction from Definition~\ref{def:right-const:qo} to build an RFA that is used to query the \emph{Oracle}.
We capture this approach in the so-called \emph{NL\(^{\qo}\) algorithm}.

\RemoveAlgoNumber%
\begin{algorithm}[!ht]
\caption{NL\(^\qo\): A quasiorder-based version of NL\(^*\)}\label{alg:NL-star}
\SetAlgorithmName{Algorithm NL\(^{\qo}\)}{}
\SetSideCommentRight%
\KwData{A \emph{Teacher} that answers membership queries in \(L\)}
\KwData{An \emph{Oracle} that answers equivalence queries between the language generated by an RFA and \(L\)}
\KwResult{The canonical RFA for the language \(L\).}
\(\cP, \cS := \{\varepsilon\}\)\;
\While{True\label{step:teacher-yes}}{%
\label{step:loop}
	\While{\(\qA\) not closed or consistent:}{
		\If{\(\qA\) is not closed}{
			Find \(u \in \Pref, a \in \Sigma\) with \(ρ_{\qA}(u)\) \(L_{\Suf}\)-prime for \(\Pref\) and \(\forall v \in \Pref, \; ρ_{\qA}(u) \neq ρ_{\qA}(v)\)\;
			Let \(\cP := \cP \cup \{ua\}\)\;
		}
		\If{\(\qA\) is not consistent}{
			Find \(u, v \in \Pref, a \in \Sigma\) with \(u \qA v\) s.t. \(u a \not\qA v a  \)\;
			Find \(x \in (ua)^{-1}L \cap ((va)^{-1}L)^c \cap \Suf\) \;
			Let \(\cS := \cS \cup \{ax\}\)\;
		}
	}
	\label{step:DFA-const}Build \(\cR(\qA, \Pref)\)\;
	Ask the \emph{Oracle} whether \(L = \lang{\cR(\qA,\Pref)}\)\;
	\If{the \emph{Oracle} replies with a counterexample \(w\)}{
	Let \(\Suf := \Suf \cup \{x \in \Sigma^* \mid w = w'x \text{ with }w \in \Suf, w' \in \Sigma^*\}\)\;
	}\Else{\Return{\(\cR(\qA, \Pref)\)}\;}
} 
\end{algorithm}

Next we explain the behavior of algorithm NL\(^{\qo}\) and give the necessary definitions in order to understand it and its relation with the algorithm  NL\(^*\).

The \emph{Learner} maintains a prefix-closed finite set \(\Pref \subseteq Σ^*\) and a suffix-closed finite set \(\Suf \subseteq Σ^*\).
The set \(\Suf\) is used to \emph{approximate} the principals in \(\qrL\) for the words in \(\Pref\).
In order to manipulate these approximations, we define the following two operators.

\begin{definitionNI}\label{def:subsetS}
Let \(L\) be a language, \(\Suf \subseteq \Sigma^*\) and \(u \in Σ^*\).
Define: 
\begin{align*}
u^{-1}L =_{\Suf} v^{-1}L  & \udiff  \left(u^{-1}L \cap \Suf\right) = \left(v^{-1}L \cap \Suf\right) &
u^{-1}L \subseteq_{\Suf} v^{-1}L & \udiff  \left(u^{-1}L \cap \Suf\right) \subseteq \left(v^{-1}L \cap \Suf\right) .\tag*{\eod}
\end{align*}
\end{definitionNI}

These operators allow us to define an over-approximation of Nerode's quasiorder that can be decided with finitely many membership tests. 

\vspace{-4pt}

\begin{definition}[Right-language-based quasiorder w.r.t. \(\Suf\)]%
\label{def:finiteNerode}
Let \(L\) be a language, \(\Suf \subseteq \Sigma^*\) and \(u,v \in Σ^*\).
Define \(u \qA v \udiff u^{-1}L \subseteq_{\Suf} v^{-1}L\).\eod
\end{definition}

\vspace{-4pt}

Recall that the \emph{Learner} only manipulates the principals for the words in \(\Pref\).
Therefore, we need to adapt the notion of composite principal for \(\qA\).

\vspace{-4pt}

\begin{definition}[\(L_{\Suf}\)-Composite Principal w.r.t. \(\Pref\)]
Let \(\Pref, \Suf \subseteq \Sigma^*\) with \(u \in \Pref\) and let \(L \subseteq Σ^*\) be a language. 
We say \(ρ_{\qA}(u)\) is \emph{\(L_{\Suf}\)-composite w.r.t. \(\Pref\)} if{}f
\[u^{-1}L =_{\Suf} \bigcup_{x \in \Pref, \; x \qAn u} x^{-1}L\enspace .\]
Otherwise, we say it is \(L_{\Suf}\)-\emph{prime} w.r.t. \(\Pref\).\eod
\end{definition}

The \emph{Learner} uses the quasiorder \(\qA\) to build an automaton by adapting the construction from Definition~\ref{def:right-const:qo} in order to use only the information that is available by means of the sets \(\Suf\) and \(\Pref\).
Building such an automaton requires the quasiorder to satisfy two conditions: it must be \emph{closed} and \emph{consistent} w.r.t. \(\Pref\).

\begin{definitionNI}[Closedness and Consistency of \(\qA\) w.r.t. \(\Pref\)]\index{Closedness w.r.t. \(\Pref\)}\index{Consistency w.r.t. \(\Pref\)}\label{def:ClosedCons}\hfill
\begin{myEnumA}
\item \(\qA\) is \emph{closed w.r.t. \(\Pref\)} if{}f \label{def:ClosedCons:Closed}
\begin{adjustwidth}{-0.5cm}{}
\begin{myAlign}{0pt}{}
\forall u \in \cP, a \in \Sigma,\;  ρ_{\qA}(ua) \text{ is \(L_{\Suf}\)-prime w.r.t. \(\Pref\)}\Ra \exists v \in \cP,  ρ_{\qA}(ua) = ρ_{\qA}(v) .
\end{myAlign}
\end{adjustwidth}
\item \(\qA\) is \emph{consistent w.r.t. \(\Pref\)} if{}f \label{def:ClosedCons:Cons} 
\(\forall u, v \in \Pref, a \in Σ: \; u \qA v \Ra ua \qA va\). \eod
\end{myEnumA}
\end{definitionNI}

At each iteration, the \emph{Learner} checks whether the quasiorder \(\qA\) is closed and consistent w.r.t. \(\Pref\).
If \(\qA\) is not closed w.r.t. \(\Pref\), then it finds \(ρ_{\qA}(ua)\) with \(u \in \Pref, a \in Σ\) such that \(ρ_{\qA}(ua)\) is \(L_{\Suf}\)-prime w.r.t. \(\Pref\) and it is not equal to some \(ρ_{\qA}(v)\) with \(v \in \Pref\).
Then it adds \(ua\) to \(\Pref\).

Similarly, if \(\qA\) is not consistent w.r.t. \(\Pref\) then the \emph{Learner} finds \(u, v \in \Pref\), \(a \in Σ, x \in \Suf\) such that \(u \qA v\) but \(uax \in L \land vax \notin L\).
Then the \emph{Learner} adds \(ax\) to \(\Suf\).
When the quasiorder \(\qA\) is closed and consistent w.r.t. \(\Pref\), the \emph{Learner} builds the RFA \(\cR(\qA, \Pref)\).

Definition~\ref{def:right-const:qo:S} is an adaptation of the automata construction \(\cH^{r}\) from Definition~\ref{def:right-const:qo}.
Instead of considering all principals, it considers only those that correspond to words in \(\Pref\).
Moreover, the notion of \(L\)-primality is replaced by \(L_{\Suf}\)-primality w.r.t. \(\Pref\) since the algorithm does not manipulate quotients of \(L\) by words in \(Σ^*\) but the approximation through \(\Suf\) of the quotients of \(L\) by words in \(\Pref\) (see Definition~\ref{def:subsetS}).

\begin{definition}[Automata construction \(\cL(\qA, \Pref)\)]%
\label{def:right-const:qo:S}
Let \(L\!\subseteq\! Σ^*\) be a regular language and let \(\Pref,\Suf\!\subseteq\! Σ^*\).
Define the automaton \(\cL(\qA, \Pref)\!=\! \tuple{Q, \Sigma, \delta, I, F}\) with \(Q \!=\! \{ρ_{\qA}(u) \mid u\!\in\! \Pref, ρ_{\qA}(u) \text{ is \(L_{\Suf}\)-prime w.r.t. \(\Pref\)}\}\), \(I = \{ρ_{\qA}(u) \in Q \mid \varepsilon \in ρ_{\qA}(u)\}\), \(F = \{ρ_{\qA}(u) \in Q \mid u \in L\}\) and \( \delta(ρ_{\qA}(u), a) = \{ ρ_{\qA}(v) \in Q \mid ρ_{\qA}(u)   a \subseteq ρ_{\qA}(v)\}\) for all \(ρ_{\qA}(u) \in Q\) and \(a \in Σ\).\eod
\end{definition}

Finally, the \emph{Learner} asks the \emph{Oracle} whether \(\lang{\cR(\qA, \Pref)} = L\).
If the \emph{Oracle} answers \emph{yes} then the algorithm terminates.
Otherwise, the \emph{Oracle} returns a counterexample \(w\) for the language equivalence.
Then, the \emph{Learner} adds every suffix of \(w\) to \(\Suf\) and repeats the process.

Theorem~\ref{theorem:NLqo} shows that the NL\(^{\qo}\) algorithm exactly coincides with NL\(^*\).

\begin{theorem}\label{theorem:NLqo}
NL\(^{\qo}\) builds the same sets \(\Pref\) and \(\Suf\), performs the same queries to the \emph{Oracle} and the \emph{Teacher} and returns the same RFA as NL\(^*\), provided that both algorithms perform the same non-deterministic choices. 
\end{theorem}
\begin{proof}
Let \(\Pref, \Suf \subseteq Σ^*\) be a prefix-closed and a suffix-closed finite set, respectively, and let \(\cT = (T, \Pref, \Suf)\) be the table built by algorithm NL\(^*\).
Observe that for every \(u,v \in \Pref\):
\begin{align}
u \qA v & \Lra \quad \text{[By Definition~\ref{def:finiteNerode}]} \nonumber\\
{u}^{-1}L \subseteq_{\Suf} {v}^{-1}L & \Lra \quad \text{[By definition of quotient w.r.t \(S\)]} \nonumber\\
\forall x \in S, \; ux \in L \Ra vx \in L & \Lra \quad \text{[By definition of \(\cT\)]} \nonumber\\
\forall x \in S, \; (\row(u)(x) = {+}) \Ra (\row(v)(x) = {+}) & \Lra \quad \text{[By Definition~\ref{def:coverRow}]} \nonumber\\
\row(u) \sqsubseteq \row(v) \enspace .
\label{eq:QOIffRowsSubset}
\end{align}
Moreover, for every \(u,v \in \Pref\) we have that \({u}^{-1}L =_{\Suf} {v}^{-1}L\) if{}f \(\row(u) = \row(v)\).

Next, we show that the join operator applied to rows corresponds to the set union applied to quotients w.r.t \(S\).
Let \(u,v \in \Pref\) and let \(x \in \Suf\).
Then,
\begin{align}
(\row(u) \sqcup \row(v))(x) = {+} & \Lra \quad \text{[By Definition~\ref{def:join}]} \nonumber\\
(\row(u)(x) = {+}) \lor (\row(v)(x) = {+}) & \Lra \quad \text{[By definition of row]} \nonumber\\
(ux \in L )\lor (vx \in L) & \Lra \quad \text{[By definition of quotient w.r.t \(\Suf\)]} \nonumber\\
(x \in {u}^{-1}L) \lor (x \in {v}^{-1}L )& \Lra \quad \text{[By definition of \(\cup\)]}\nonumber \\ 
x \in {u}^{-1}L \cup {v}^{-1}L \enspace .
\label{eq:joinUnion}
\end{align}
Therefore, we can prove that \(\row(u)\) is \(\cT\)-\emph{prime} if{}f \(ρ_{\qA}(u)\) is \(L_{\Suf}\)-prime w.r.t. \(\Pref\).
\begin{align*}
\row(u) = {\textstyle\bigsqcup_{v \in \Pref, \; \row(v) \sqsubset \row(u)}} \row(v) & \Lra \quad \text{[By Equation~\eqref{eq:QOIffRowsSubset}]} \\
\row(u) = {\textstyle\bigsqcup_{v \in \Pref, \; {v}^{-1}L \subsetneq_{\Suf} {u}^{-1}L}} \row(v) & \Lra \quad \text{[By Equation~\eqref{eq:joinUnion}]} \\
{u}^{-1}L = {\textstyle\bigcup_{v \in \Pref, \; {v}^{-1}L \subsetneq_{\Suf} {u}^{-1}L}} {v}^{-1}L & \Lra \quad \text{[\({v}^{-1}L \subsetneq_{\Suf} {u}^{-1}L \Lra u \qAn v\)]} \\
{u}^{-1}L = {\textstyle\bigcup_{v \in \Pref, \; u \qAn v}} {v}^{-1}L  \enspace .
\end{align*}
It follows from Definitions~\ref{def:ClosedCons} \ref{def:ClosedCons:Closed} and~\ref{def:Table} \ref{def:Table:closed} and Equation~\eqref{eq:joinUnion} that \(\cT\) is closed if{}f \(\qA\) is closed.

Moreover, it follows from Definitions~\ref{def:ClosedCons} \ref{def:ClosedCons:Cons} and~\ref{def:Table} \ref{def:Table:Consistent} that \(\cT\) is consistent if{}f \(\qA\) is consistent.

On the other hand, for every \(u,v \in \Pref, a \in Σ\) and \(x \in \Suf\) we have that:
\begin{align*}
(\row(u) \subseteq \row(v)) \land (\row(ua)(x) = {+}) \land (\row(va)(x) = {-} )& \Lra \quad \text{[By Equation~\eqref{eq:QOIffRowsSubset}]} \\
(u \qA v )\land (uax \in L) \land (vax \notin L) & %
\end{align*}
It follows that if \(\cT\) and \(\qA\) are not consistent then both NL\(^*\) and NL\(^{\qo}\) can find the same word \(ax \in Σ   \Suf\) and add it to \(\Suf\).
Similarly, it is straightforward to check that if \(\row(ua)\) with \(u \in \Pref\) and \(a \in Σ\) break consistency, i.e.\ it is \(\cT\)-prime and it is not equal to any \(\row(v)\) with \(v \in \Pref\), then \(ρ_{\qA}(ua)\) is \(L_{\Suf}\)-prime for \(\Pref\) and not equal to any \(ρ_{\qA}(v)\) with \(v \in \Pref\).
Thus, if \(\cT\) and \(\qA\) are not closed then both NL\(^*\) and NL\(^{\qo}\) can find the same word \(ua\) and add it to \(\Pref\).

It remains to show that both algorithms build the same automaton modulo isomorphism, i.e., \(\cR(\cT) = \tuple{\widetilde{Q}, Σ, \wt(δ), \widetilde{I}, \wt{F}}\) is isomorphic to \(\cR(\qA, \Pref) = \tuple{Q, Σ, δ, I, F}\).
Define the mapping \(\varphi: Q \to \wt{Q}\) as \(\varphi(ρ_{\qA}(u)) = \row(u)\).
Then:
\begin{align*}
\varphi(Q) & = \{\varphi(ρ_{\qA}(u)) \mid u \in \cP \land ρ_{\qA}(u) \text{ is \(L_{\Suf}\)-prime w.r.t. \(\Pref\)}\} \\
& = \{\row(u) \mid u \in \cP \land \row(u) \text{ is \(\cT\)-prime}\} = \wt{Q} \enspace .\\
\varphi(I) & = \{\varphi(ρ_{\qA}(u)) \mid \varepsilon \in ρ_{\qA}(u)\} = \{\row(u) \mid u \qA \varepsilon\} = \{\row(u) \mid \row(u) \sqsubseteq \row(\varepsilon)\} = \wt{I} \enspace .\\
\varphi(F) & = \{\varphi(ρ_{\qA}(u)) \mid u \in L \cap \cP\} = \{\row(u) \mid u \in L \cap \cP\} = \{\row(u) \mid \row(u)(\varepsilon) = {+}\} = \wt{F}\enspace .\\
\varphi(\delta(ρ_{\qA}(u),a)) &= \varphi(ρ_{\qA}(ua)) = \{\row(v) \mid ρ_{\qA}(u) \in Q \land ρ_{\qA}(u)a \subseteq ρ_{\qA}(v)\} \\
& = \{\row(v) \mid \row(v) \in \wt{Q} \land v \qA ua\}  = \{\row(v) \mid \row(v) \in \wt{Q} \land \row(v) \sqsubseteq \row(ua)\} \\
& = \wt{\delta}(\row(u),a) = \wt{\delta}(\varphi(ρ_{\qA}(u)),a) \enspace .
\end{align*}

Finally, we show that \(\varphi\) is an isomorphism.
Clearly, the function \(\varphi\) is surjective since, for every \(u \in \Pref\), we have that \(\row(u) = \varphi(ρ_{\qA}(u))\).
Moreover \(\varphi\) is injective since for every \(u,v \in \Pref\), \(\row(u) = \row(v) \Lra {u}^{-1}L =_{\Suf} {v}^{-1}L\), hence \(\row(u) = \row(v) \Lra ρ_{\qA}(u) = ρ_{\qA}(v)\).

We conclude that \(\varphi\) is an NFA isomorphism between \(\cR(\qA,\Pref))\) and \(\cR(\cT)\).
Therefore NL\(^*\) and NL\(^{\qo}\) exhibit the same behavior, provided that both algorithms perform the same non-deterministic choices, as they both maintain the same sets \(\Pref\) and \(\Suf\) and build the same automata at each step.
\end{proof}

\paragraph*{Termination of NL\(^*\) and NL\(^{\qo}\)}

At each iteration of the NL\(^{\qo}\) algorithm, it either terminates or the counterexample \(w\) given by the \emph{Oracle} refines the quasiorder \(\qA\) which results in having, at least, one new principal \(ρ_{\qA}(w)\).

Since 
\[ρ_{\qA}(u) \neq ρ_{\qA}(v) \Ra \exists s \in \Suf, \; us \in L \land vs \notin L \Ra ρ_{\qrL}(u) \neq ρ_{\qrL}(v)\enspace ,\]
we conclude that the number of principals for \(\qA\) is smaller o equal than the number of principals for \(\qrL\).
Given that \(\qrL\) induces finitely many principals, algorithm NL\(^{\qo}\) can only add finitely many principals to \(\qA\) and, therefore, the algorithm terminates.

It is worth to remark that, in order to prove the termination of the NL\(^*\) algorithm, \citet{bollig2009angluin} first had to show that the number of rows built during the computation of the NL\(^*\) algorithm is a lower bound for the number of rows computed during an execution of the L\(^*\) algorithm of \citet{angluin1987learning}.
Then, the termination of the NL\(^*\) algorithms follows from the termination of L\(^*\).

Finally, observe that, by replacing the right quasiorder \(\qA\) by its corresponding right congruence \(\mathord{\sim_{L_{\Suf}}} \ud \mathord{\mathord{\qA}} \cap \mathord{\mathord{(\qA)^{-1}}}\) in the above algorithm (precisely, in Definitions~\ref{def:ClosedCons} and \ref{def:right-const:qo:S}), the resulting algorithm corresponds to the L\(^*\) algorithm of \citet{angluin1987learning}.
Note that, in that case, all principals \(ρ_{\sim_{L_{\Suf}}}(u)\), with \(u\in\Sigma^*\), are \(L_{\Suf}\)-prime w.r.t. \(\Pref\). 
\clearpage{}%
\clearpage{}%
%

\chapter{Future Work}
\label{chap:future}

We believe that we have only scratched the surface on the use of \emph{well-quasiorders} on words for solving problems from \emph{Formal Language Theory}.

In this section, we present some directions for further developments that show how our work can be extended to\begin{myEnumIL}\item take full advantage of simulation relations, \item better understand, and possibly improve, the performance of \tool{zearch} and \item develop new algorithms for building smaller residual automata.\end{myEnumIL}

\section{The Language Inclusion Problem}
Consider the inclusion problem \(\lang{\cN_1} \subseteq \lang{\cN_2}\), where \(\cN_1\) and \(\cN_2\) are NFAs.
Even though we have shown in Chapter~\ref{chap:LangInc} that simulations can be used to derive an algorithm for solving this language inclusion problem, we are not on par with the thoughtful use of simulation relations made by \citet{Abdulla2010} and \citet{DBLP:conf/popl/BonchiP13}. 
The main reason for which we are not able to accommodate within our framework their use of simulations is that our abstraction only manipulates sets of states of \(\cN_2\).
As a consequence, any use of simulations that involves states of \(\cN_1\) is out of reach.

However, it is possible to overcome this limitation by using \emph{alternating automata} as we show next.
Intuitively, since alternating automata can be complemented without altering their number of states, we can reduce any language inclusion problem \(\lang{\cA_1} \subseteq \lang{\cA_2}\), where \(\cA_1\) and \(\cA_2\) are alternating automata, into a universality problem \(Σ^* \subseteq \lang{\cA_3}\), where \(\cA_3 = \cA_1^c \cup \cA_2\).
Since \(\cA_3\) is built by combining the two input automata its states are the union of the states of \(\cA_1\) and \(\cA_2\).
Therefore, simulations applicable within our framework to decide \(Σ^* \subseteq \lang{\cA_3}\), which only involve states of \(\cA_3\), now involve states of \(\cA_1\) and \(\cA_2\). 

\subsection{Language Inclusion Through Alternating Automata}
Let \(S\) be a set.
We denote by \(\cBp(S)\) the set of \demph{positive Boolean formulas} over \(S\) which are of the form \(\Phi \ud s \;|\; \Phi_1 \lor \Phi_2 \;|\; \Phi_1 \land \Phi_2 \;|\; \False\), where \(s \in S\) and \(\Phi_1,\Phi_2 \in \cBp(S)\).
We say \(S' \subseteq S\) \emph{satisfies} a formula \(\Phi \in \cB^{+}(S)\) if{}f \(\Phi\) is \(\True\) when assigning the value \(\True\) to all elements in \(S'\) and \(\False\) to the elements in \(S \setminus S'\).
Given \(\Phi \in \cBp(S)\), we denote \(\eval{\Phi}\) the set of all subsets of \(S\) that satisfy \(\Phi\).
Clearly, if \(S'\) satisfies a formula \(\Phi\), any set \(S'' \subseteq S\) such that \(S' \subseteq S''\) also satisfies \(\Phi\).
Therefore, the set \(\eval{\Phi}\) is an \(\subseteq\)-upward closed set, i.e. \(ρ_{\subseteq}(\eval{\Phi}) = \eval{\Phi}\).
Finally, if a formula \(\Phi\) is not satisfiable, i.e. no set \(S' \subseteq S\) satisfies \(\Phi\), then \(\eval{\Phi} = \varnothing\). 

\begin{definition*}[AFA]
An \emph{alternating finite-state automata}\index{alternating automata} (AFA for short) is a tuple \(\cA \ud \tuple{Q,\Sigma,δ,I,F}\) where \(Q\) is the finite set of \emph{states}, \(\Sigma\) is the finite alphabet, \(\delta \colon Q \times \Sigma \to \cBp(Q)\) is the transition function, \(I \subseteq Q\) are the initial states and \(F \subseteq Q\) are the final states.\eod
\end{definition*}

Intuitively, given an active state \(q \in Q\) and an alphabet symbol \(a \in Σ\) an AFA can activate any set of states in \(\eval{δ(q,a)}\). 
Figure~\ref{fig:AFA} shows an example of an AFA. 

\begin{figure}[!ht]
  \centering
\begin{minipage}[l]{0.34\textwidth}
\begin{tikzpicture}[->,>=stealth',shorten >=1pt,auto,node distance=3mm and 8mm,thick,initial text=]
\tikzstyle{every state}=[scale=0.855,fill=customblue!60,draw=blue!60,text=black,style={draw,ellipse,inner sep=0pt}]

\node[initial, state] (1) {\(0\)};
\node[state] (2) [right=of 1, yshift=1.8cm] {\(1\)};
\node[state] (3) [right=of 1] {\(2\)};
\node[accepting, state] (4) [right=of 3] {\(3\)};

\path (1) edge node[inner sep=0mm,pos=0.2] (a1) {} node {\(a\)} (2);
\path (1) edge node[inner sep=0mm,pos=0.2] (b1) {} node {\(a\)} (3);
\path (1) edge[bend right=45] node {\(a\)} (4);

\path (2) edge node[inner sep=0mm,pos=0.5] (a2) {} node[left] {\(b\)} (3);
\path (2) edge node[inner sep=0mm,pos=0.2] (b2) {} node {\(b\)} (4);

\path (3) edge node {\(a\)} (4);
\path (4) edge[loop above] node {\(a,b\)} (4);

\path[-,shorten <=0mm,shorten >=-1.5pt] (a1) edge [bend left]  (b1) ;
\path[-,shorten <=-1.5pt,shorten >=0mm] (a2) edge [bend right]  (b2) ;
\end{tikzpicture}
\end{minipage}\hfill %
\begin{minipage}[r]{0.65\textwidth}
\(\cA \!=\! \tuple{Q,\Sigma,δ,I,F}\) with \(Q\!=\!\{q_0, q_1, q_2, q_3\}\), \(\Sigma  \!=\! \{a,b\}\), \(I  \!=\! \{q_0\}\), \(F  = \{q_3\}\) and 
\begin{myAlign}{2pt}{0pt}
δ(q_0,a) & = (q_1 \land q_2) \lor q_3 & 
δ(q_2,a) & = δ(q_3,a) = δ(q_3,b) = q_3 \\
δ(q_1,b) & = q_2 \land q_3 &
δ(q_0,b) & = δ(q_1,a) = δ(q_2,b) = \False
\end{myAlign}
\end{minipage}

\caption{Alternating automaton \(\cA\) generating the language \(\lang{\cA}= a(a+b)^*\).}
  \label{fig:AFA}
\end{figure}

Given an AFA \(\cA = \tuple{Q,\Sigma,δ,I,F}\), we extend the transition function \(\delta\) to sets of states obtaining \(\Delta: \wp(Q)\times \Sigma \to \cBp\) defined as \(Δ(S,a) \ud \bigwedge_{s \in S} δ(s,a)\).
Intuitively, \(\Delta(S,a)\) indicates the states that will be activated after reading \(a\) when all states in \(S\) are active.
Let \(X \uplus Y \ud \{x \cup y \mid x \in X, y \in Y\}\).
Then
\begin{equation}\label{eq:Delta}
\eval{\Delta(S,a)} = \left\{\begin{array}{ll}
\varnothing & \text{ if } \exists s \in S \text{ s.t. } \delta(s,a) = \False \\
\biguplus_{s \in S} \eval{\delta(s,a)} & \text{ otherwise }
\end{array}\right.
\end{equation}

We say a word \(w\) is accepted by an AFA \(\cA=\tuple{Q,Σ,δ,I,F}\) if{}f there exists a sequence of sets of active states \(S_0,\ldots,S_{\len{w}}\) such that \(S_0 = \{q_i\}\) with \(q_i \in I\), \(S_n \subseteq F\), \(S_n \neq \varnothing\) and \(S_{i} \in \eval{Δ(S_{i{-}1},(w)_i)}\) for \(1 \leq i \leq \len{w}\).

\begin{example}
Let us consider the alternating automaton \(\cA\) in Figure~\ref{fig:AFA}. 	
Then, we have that
\begin{align*}
\Delta(\{q_0\},a) & = δ(q_0,a) = (q_1 \land q_2) \lor q_3 \enspace .\\
\eval{\Delta(\{q_0\},a)} & = \eval{δ(q_0,a)} = ρ_{\subseteq}\left(\{\{q_1,q_2\},\{q_3\}\}\right) \enspace .\\[8pt]
\Delta(\{q_1,q_2\},b) & =  δ(q_1,b) \land δ(q_2,b) = q_2\land q_3 \land \False = \False\\
\eval{\Delta(\{q_1,q_2\},b)} & =  \eval{δ(q_1,b)} \biguplus \eval{δ(q_2,b)} = ρ_{\subseteq}\left(\{\{q_2,q_3\}\}\right) \uplus \varnothing = \varnothing \enspace . \\[8pt]
\Delta(\{q_1,q_2\},a) & =  δ(q_1,a) \land δ(q_2,a)  = \False \land q_3 = \False\\
\eval{\Delta(\{q_1,q_2\},a)} & =  \eval{δ(q_1,a)} \biguplus \eval{δ(q_2,a)} = \varnothing \uplus ρ_{\subseteq}\left(\{\{q_3\}\}\right) = \varnothing \enspace . \\[8pt]
Δ(\{q_3\},a) & = Δ(\{q_3\},b) = q_3 \\
\eval{Δ(\{q_3\},a)} & = \eval{Δ(\{q_3\},b)} = ρ_{\subseteq}\left(\{\{q_3\}\}\right)\enspace .
\end{align*}
Since \(q_0\) is the only initial state and \(F = \{q_3\}\), it follows that the language generated by the automaton is \(\lang{\cA} = a(a+b)^*\).\eox
\end{example}

We denote the reflexo-transitive closure of \(\eval{\Delta}\) as \(\goes{}\).
Thus, the \emph{language} of an AFA, \(\cA\), is  \(\lang{\cA} = \{w \in \Sigma^* \mid \exists q_i \in I, S \subseteq F, \; S \neq \varnothing \land \{q_i\} \goes{w} S\}\).

One of the most interesting properties of AFAs is that their complement, i.e. an AFA generating the complement language, can be built in polynomial time.

\begin{definition}[Complement of an AFA]
Let \(\cA = \tuple{Q,\Sigma,δ,I,F}\) be an AFA with \(L = \lang{\cA}\).
Its \emph{complement AFA}, denoted \(\cA^c\) is the AFA \(\cA^c \ud \tuple{Q,\Sigma,δ^c,I,Q \setminus F}\) where
\(δ^c(q,a)\) is the result of switching \(\land\) and \(\lor\) operators in \(δ(q,a)\).\eod
\end{definition}

The simplicity of the computation of the complement for AFAs, \emph{which does not alter the number of states of the automaton}, allows us to use them in order to solve the language inclusion problem \(\lang{\cN_1} \subseteq \lang{\cN_2}\), where \(\cN_1\) and \(\cN_2\) are NFAs, by reducing it to universality of alternating automata as follows:
\begin{align}
\lang{\cN_1} \subseteq \lang{\cN_2} & \Lra \quad \text{[Since NFAs \(\subseteq\) AFAs]}\nonumber \\
\lang{\cA_1} \subseteq \lang{\cA_2} & \Lra \quad \text{[\(A \subseteq B \Lra A \cap B^c = \varnothing\)]}\nonumber\\
\lang{\cA_1} \cap (\lang{\cA_2})^c = \varnothing & \Lra \quad \text{[\((A\cap B)^c = A^c \cup B^c\) and \(\varnothing^c = Σ^*\)]}\nonumber\\
(\lang{\cA_1})^c \cup \lang{\cA_2} = \Sigma^* & \Lra \quad\text{[AFAs are closed under complement]}\nonumber\\
\lang{\cA_1^c} \cup \lang{\cA_2} = \Sigma^* & \Lra \quad \text{[\(A = \Sigma^* \Lra \Sigma^* \subseteq A\)]} \nonumber\\
\Sigma^* \subseteq \lang{\cA_1^c} \cup \lang{\cA_2} & \Lra \quad \text{[With \(\cA_3 = \cA_1^c \cup \cA_2\)]} \nonumber\\
Σ^* \subseteq \lang{\cA_3}\label{eq:universality} 
\end{align}

On the other hand, \(\Sigma^*\) is the \(\lfp\) of the equation \(\lambda X. \{\varepsilon\} \cup \bigcup_{a \in \Sigma}aX\).
Therefore 
\[\Sigma^* \subseteq \lang{\cA_3} \Lra \lfp (\lambda X. \{\varepsilon\} \cup {\textstyle\bigcup_{a \in \Sigma} aX}) \subseteq \lang{\cA_3}\enspace .\]

We are now in position to leverage our quasiorder-based framework from Chapter~\ref{chap:LangInc} to derive an algorithm for deciding the universality of a regular language given by an AFA \(\cA\).
To do that, we adapt our right state-based quasiorder from Equation~\ref{eqn:state-qo}, which requires defining the successor operator for AFAs \(\post_w^{\cA}: \wp(\wp(Q)) \to \wp(\wp(Q))\), where \(w \in \Sigma^*\), as follows:

\begin{equation}\label{eq:postAFAwords}
\post_{w}^{\cA}(X) \ud \{S' \in \wp(Q) \mid \exists S \in X, S \goes{w} S'\} \enspace .
\end{equation}

It is straightforward to check that \(\post_{wa}^{\cA}(X) = \post_a^{\cA}(\post_w^{\cA}(X))\).
The following example illustrates the behavior of the function \(\post_w^{\cA}\) on the AFA from Figure~\ref{fig:AFA}.

\begin{example}
Consider again the AFA \(\cA\) from Figure~\ref{fig:AFA}.
We have that
\begin{align*}
\post_{a}^{\cA}(\{\{q_0\}\}) & = ρ_{\subseteq}\left(\{\{q_1,q_2\},\{q_3\}\}\right) \\
\post_{aa}^{\cA}(\{\{q_0\}\}) & = \post_{a}^{\cA}(\{\{q_1,q_2\},\{q_3\}\}) = ρ_{\subseteq}\left(\{\{q_3\}\}\right) \\
\post_{ab}^{\cA}(\{\{q_0\}\}) & = \post_b^{\cA}(\{\{q_1,q_2\},\{q_3\}\}) = ρ_{\subseteq}\left(\{\{q_3\}\}\right)\enspace . \tag*{\eox}
\end{align*}
\end{example}

Similarly to what we did in Section~\ref{subsec:state-qos} for NFAs, we next define a sate-based quasiorder for AFAs, \(\qo_{\cA}\).
To do that, let \(I_{\{\}}\) be the set of singleton subsets of \(I\), i.e. \(I_{\{\}} \ud \{\{q\} \mid q \in I\}\).
Then
\begin{equation}\label{eq:qoAFAState}
u \qo_{\cA} v \Lra \post_u^{\cA}(I_{\{\}}) \subseteq \post_v^{\cA}(I_{\{\}})
\end{equation}

\begin{lemma}\label{lemma:leqStateLcWQP}
Let \(\cA = \tuple{Q,\Sigma,\delta, I,F}\) be an AFA with \(L=\lang{\cA}\).
Then \(\qo_{\cA}\) is a right \(L\)-consistent well-quasiorder.
\end{lemma}
\begin{proof}
First, we show that \(\qo_{\cA}\) is right monotone.
Let \(u,v \in \Sigma^*\) and \(a \in \Sigma\).
Recall that $\post^\cA_a$ is a monotonic function and that 
\begin{equation}
\post^{\cA}_{uv} = \post^{\cA}_{v} \comp \post^{\cA}_u \enspace .\label{eq:postpost}
\end{equation}
Then
\begin{myAlign}{0pt}{0pt}
u \qo_{\cA} v & \Ra  \quad\text{[By definition of \(\qo_{\cA}\)]} \\
\post^{\cA}_{u}(I_{\{\}}) \subseteq \post^{\cA}_{v}(I_{\{\}}) & \Ra  \quad\text{[Since $\post^\cA_a$ is monotonic]} \\
\post^{\cA}_{a}(\post^{\cA}_{u}(I_{\{\}})) \subseteq \post^{\cA}_{a}(\post^{\cA}_{v}(I_{\{\}})) & \Lra  \quad\text{[By Equation~\eqref{eq:postpost}]} \\
\post^{\cA}_{ua}(I_{\{\}}) \subseteq \post^{\cA}_{va}(I_{\{\}}) & \Lra  \quad\text{[By definition of \(\qo_{\cA}\)]} \\
ua \qo_{\cA} va \enspace . 
\end{myAlign}

On the other hand, \(\qo_{\cA}\) is \(L\)-consistent since, by definition
\[\forall w \in Σ^*, \; w \in L \Lra \exists S \in \post_w^{\cA}(I_{\{\}}), \; S \neq \varnothing \land S \subseteq F\enspace .\]
Therefore, if \(u \in L\) and \(u \qo_{\cA} v\) then it follows that \(v \in L\).

Finally, it is straightforward to check that \(\qo_{\cA}\) is a well-quasiorder since \(\wp(\wp(Q))\) is finite.
\end{proof}

Since membership in AFAs is decidable, it follows from Lemma~\ref{lemma:leqStateLcWQP} and Theorem~\ref{theorem:quasiorderAlgorithmR} that Algorithm \AlgRegularWr instantiated with the wqo \(\qo_{\cA}\) decides the inclusion \(Σ^* \!\subseteq\! \lang{\cA}\), where \(\cA\) is an AFA.

Following the developments of Chapter~\ref{chap:LangInc}, given an AFA \(\cA = \tuple{Q,Σ,δ,I,F}\), we could define a Galois Connection \(\tuple{\wp(\Sigma^*),\subseteq}\galois{\alpha}{\gamma}\tuple{\AC_{\tuple{\wp(\wp(Q)),\subseteq}},\sqsubseteq}\) that yields an antichains algorithm for deciding the universality of AFAs by manipulating sets of sets of states.
By doing so, we would obtain an algorithm that computes the set \(Y = \minor{\{\post_w^{\cA}(I_{\{\}}) \mid w \in Σ^*\}}\) and checks whether all elements \(y \in Y\) satisfy \(\exists s \in y, \; s \neq \varnothing \land s \subseteq F\).

Moreover, we could enhance the state-based quasiorder for AFAs by using simulations between the states of \(\cA\) which, recall, are the union of the states of the input automata \(\cN_1\) and \(\cN_2\).
This would allow us to use the simulations that relate states of both automata, similarly to \citet{Abdulla2010} and \citet{DBLP:conf/popl/BonchiP13}.

Therefore, we believe that the full development of an antichains algorithm for AFAs is an interesting line for future work since it will allow us to understand how close our framework can get to the results of \citet{Abdulla2010} and \citet{DBLP:conf/popl/BonchiP13}.

\section{The Complexity of Searching on Compressed Text}
We believe the good results obtained during the evaluation of \tool{zearch} (see Figure~\ref{fig:comparison}) invite for a deeper study of our algorithm in order to better understand its behavior and improve its performance.

For instance, it is yet to be considered how the performance of \tool{zearch} is affected by the choice of the grammar-based compression algorithm.
By using different heuristics to build the grammar, the resulting SLP will have different properties, such as depth, width or length of the rules, which would definitely affect \tool{zearch}'s performance.
Figure~\ref{fig:grammars} shows the grammars built by different compression algorithms for the same string.

\begin{figure}[!ht]
\begin{minipage}[t]{0.33\textwidth}
	\centering
\caption*{\tool{Sequitur}}
\vspace{-10pt}
\begin{tikzpicture}[every node/.style={align=center}]
  \tikzset{
  	level distance=20pt,
    edge from parent/.style={
      draw,edge from parent
      path={(\tikzparentnode.south)+(0,3pt)-- +(0,-1pt)-| (\tikzchildnode)}
    },
    sibling distance=-5pt,
    frontier/.style={distance from root=60pt} %
  }

   \Tree[.$S$ x [.$S_2$ a [.$S_1$ b c ] d [.$S_1$ b c ] ] [.$S_2$ a [.$S_1$ b c ] d [.$S_1$ b c ] ] y ]
\end{tikzpicture}
\end{minipage}%
\begin{minipage}[t]{0.33\textwidth}
	\centering
\caption*{\tool{Re-Pair}}
\vspace{-10pt}
\begin{tikzpicture}[every node/.style={align=center}]
  \tikzset{
  	level distance=20pt,
    edge from parent/.style={
      draw,edge from parent
      path={(\tikzparentnode.south)+(0,3pt)-- +(0,-1pt)-| (\tikzchildnode)}
    },
    sibling distance=-5pt,
    frontier/.style={distance from root=100pt} %
  }

   \Tree[.$S$ x [.$S_4$ [.$S_3$ [.$S_2$ a [.$S_1$ b c ] ] d ] [.$S_1$ b c ] ] [.$S_4$ [.$S_3$ [.$S_2$ a [.$S_1$ b c ] ] d ] [.$S_1$ b c ] ] y ]
\end{tikzpicture}
\end{minipage}%
\begin{minipage}[t]{0.30\textwidth}
\centering
\caption*{\tool{LZW}}
\vspace{-10pt}
\begin{tikzpicture}[every node/.style={align=center}]
 \tikzset{
  	level distance=20pt,
    edge from parent/.style={
      draw,edge from parent
      path={(\tikzparentnode.south)+(0,3pt)-- +(0,-1pt)-| (\tikzchildnode)}
    },
    sibling distance=-5pt,
    frontier/.style={distance from root=40pt} %
  }
   \Tree[.$S$ x a b c d [.$S_1$ b c ] [.$S_2$ a b ] [.$S_3$ c d ] [.$S_1$ b c ] y ]
\end{tikzpicture}
\end{minipage}
\caption{From left to right, grammars built by the compression algorithms \tool{sequitur}~\cite{nevill1997compression}, \tool{repair}~\cite{larsson2000off} and \tool{LZW}~\cite{welch1984technique} for “{\rm xabcdbcabcdbcy}”.}
\label{fig:grammars}
\end{figure}

In particular, there are grammar-based compression algorithms such as \tool{Sequitur}~\cite{nevill1997compression} that produce SLPs which are not in CNF, i.e. in which rules might have more than two symbols on the right hand side.
Processing such a grammar, instead of the one built by \tool{repair} reduces the number of rules to be processed at the expense of a greater cost for processing each rule.
It is worth considering whether adapting \tool{zearch} to work on such SLPs will have a positive impact on its performance.

On the other hand, Algorithm \AlgCountLines allows for a conceptually simple parallelization since any set of rules such that no variable appearing on the left hand side of a rule appears on the right hand side of another, can be processed simultaneously.
Indeed, a theoretical result by \citet{ullman1988parallel} on the parallelization of Datalog queries can be used to show that counting the number of lines in a grammar-compressed text containing a match for a regular expression is in $\mathcal{NC}^2$, i.e. it is solvable in \emph{polylogarithmic time} on \emph{parallel} computer with a polynomial number of processors, when the automaton built from the expression is acyclic.
Therefore, we are optimistic about the possibilities of a parallel version of \tool{zearch}.

Finally, \emph{patterns} are a commonly used subclass of regular expressions for which specific searching algorithms have been developed~\cite{kida1998multipattern,navarro2005lzgrep,gawrychowski2013optimal,gawrychowski2014simple}.
Since the standard automata construction from regular expressions yields a DFA when the expression is a pattern, our algorithm allows us to search for patterns in \(\mathcal{O}(t\cdot s)\) time, where \(t\) is the size of the grammar and \(s\) is the length of the pattern.
However, as shown by \citet{gawrychowski2013optimal}, it is possible to decide the existence of a pattern in an \tool{LZW}-compressed text in \(\mathcal{O}(t + s)\) time.
It is yet to be considered whether the algorithm of \citet{gawrychowski2013optimal} can be adapted to the more general scenario of searching on grammar-compressed text and whether it can be extended to report the number of matching lines without altering its complexity as we did with Algorithm \AlgCountLines.

\section{The Performance of Residualization}

In Chapter~\ref{chap:RFA} we presented the automata construction \(\cG{r}(\cN)\) as an alternative to \(\cN^{\text{res}}\), the residualization operation defined by \citet{denis2002residual}.
We have shown in Theorem~\ref{theoremF}~\ref{lemma:rightNRes} that given an NFA \(\cN\), the automaton \(\cG{r}(\cN)\) is a sub-automaton of \(\cN^{\text{res}}\), meaning that our construction yields smaller automata.

On the other than, it is clear that, given an NFA \(\cN=\tuple{Q,Σ,δ,I,F}\), finding the coverable sets in \(\{\post_u^{\cN}(I) \mid u \in Σ^*\}\) is easier than finding the \(L\)-composite principals in \(\{ρ_{\qrN}(u) \mid u \in Σ^*\}\).
However, it is yet to be considered the performance of both algorithms and the actual difference in size between the RFAs \(\cG{r}(\cN)\) and \(\cN^{\text{res}}\).

\subsection{Reducing RFAs with Simulations}
Let \(\cN\) be an NFA with \(L = \lang{\cN}\).
As shown by Lemma~\ref{lemma:simulationLConsistent}, the simulation-based quasiorder \(\preceq^r_{\cN}\) is an \(L\)-consistent right well-quasiorder.
Therefore, it follows from Lemma~\ref{lemma: HrGeneratesL} that \(\cH^r(\preceq^r_{\cN},L)\) is an RFA generating the language \(L\).
Moreover, as shown in Section~\ref{sec:simulation_basedQO}, we have the following relation between the state-based, the simulation-based and the Nerode's right quasiorders:
\[
\mathord{\qr_{\cN}} \,\subseteq\, \mathord{\preceq^r_{\cN}} \,\subseteq\, \mathord{\qr_{\lang{\cN}}}\enspace .
\]
Therefore, by Theorem~\ref{theorem:numLPrimePrincipals}, we have that
\[\begin{array}{c}
 \len{\{ρ_{\qrN}(u) \mid u \in Σ^* \text{ and } ρ_{\qrN}(u) \text{ is \(L\)-prime}\}} \\
\rotatebox[origin=c]{270}{\(\geq\)} \\
 \len{\{ρ_{\preceq^r_{\cN}}(u) \mid u \in Σ^* \text{ and } ρ_{\preceq^r_{\cN}}\text{ is \(L\)-prime}\}}  \\
\rotatebox[origin=c]{270}{\(\geq\)} \\
\len{\{ρ_{\qrL}(u) \mid u \in Σ^* \text{ and } ρ_{\qrL}(u)\text{ is \(L\)-prime}\}} \enspace .
\end{array}\]

One promising direction for future work is to fully develop this idea of using simulation-based quasiorders to build even smaller RFAs.
Such technique should be implemented and evaluated in practice in comparison with the residualization operations \(\cG{r}(\cN)\) and \(\cN^{\text{res}}\). 
\clearpage{}%
\clearpage{}%
%

\chapter{Conclusions}
\label{chap:conclusions}

In this thesis, we have shown that well-quasiorders are the right tool for addressing different problems from \emph{Formal Language Theory}.
Indeed, we presented two quasiorder-based frameworks in Chapters~\ref{chap:LangInc} and~\ref{chap:RFA} that allowed us to offer a new perspective on \emph{The Language Inclusion Problem} and \emph{Residual Automata}, respectively.
In both cases, our frameworks allowed us to\begin{myEnumIL}
\item offer a \emph{new perspective on known algorithms} that facilitates their understanding and evidences the relationships between them and
\item \emph{systematically derive new algorithms}, some of which proved to be of practical interest due to their performance. 
\end{myEnumIL}

\paragraph{The Language Inclusion Problem}
We have been able to systematically derive well-known algorithms such as the antichains algorithms for regular languages of \citet{DBLP:conf/cav/WulfDHR06}, with its multiple variants (see Section~\ref{sec:novel_perspective_AC}), and the antichains algorithm for grammars of \citet{Holk2015}.
These systematic derivations result in a simpler presentation of the antichains algorithm for grammars of \citet{Holk2015} as a straightforward extension of the antichains algorithm for regular languages.
Indeed, we have shown that the antichains algorithm for regular languages and for grammars are conceptually identical and correspond to two instantiations of our framework with different quasiorders.
Recall that, previously, the use of antichains for grammars was justified through a reduction to data flow analysis.

Our framework has also allowed us to derive algorithms for deciding the inclusion of a regular language in the trace set of a one-counter net.
In doing so, we have shown that the right Nerode quasiorder for the trace set of a one-counter net is an undecidable well-quasiorder, thereby closing a conjecture made by \citet[Section 6]{deLuca1994}.

Finally, our quasiorder-based framework also allowed us to derive novel algorithms, such as gfp-based Algorithm \AlgRegularGfp, for deciding the inclusion between regular languages.
It is yet to be considered the performance of this algorithm in order to decide whether it is of practical interest.

\paragraph*{Searching on Compressed Text}
We then adapted the antichains algorithm for grammars to the problem of \emph{searching with regular expressions in grammar compressed text}.
As a result, we have presented the first algorithm for \emph{counting} the number of lines in a grammar-compressed text containing a match for a regular expression.
It is worth to remark that our algorithm applies to any grammar-based compression scheme while being nearly optimal.

Together with the presentation of our algorithm, we described in Chapter~\ref{chap:zearch} the data structures required to achieve nearly optimal complexity for searching in compressed text and used them to implement a \emph{sequential} tool --\tool{zearch}-- that significantly outperforms the \emph{parallel} state of the art to solve this problem.
Indeed, when the grammar-based compressor achieves high compression ratio, which is the case, for instance, for automatically generated \emph{Log} files, \tool{zearch} uses up to $25\pct$ less time than \tool{lz4|hyperscan}, even outperforming \tool{grep} and being competitive with \tool{hyperscan}.

Our results evidence that compression of textual data and regular expression searching, two problems considered independent in practice, are connected.
Intuitively, the search can take advantage of the information about repetitions in the text, highlighted by the compressor, to skip parts of the uncompressed text.

\paragraph{Residual Automata}
\citet{denis2002residual} introduced the notion of RFA and canonical RFA for a regular language and devised a procedure, similar to the subset construction for DFAs, to build the RFA \(\cN^{\text{res}}\) from a given automaton \(\cN\).
Furthermore, they showed that the \emph{double-reversal method} holds for RFAs with their \emph{residualization} operation, i.e. \(\cN^{\text{res}}\) is isomorphic to the canonical RFA \(\cC\) for \(\lang{\cN}\) for every co-RFA \(\cN\).

Later, \citet{tamm2015generalization} proved the following result:
\begin{lemmaC}[\citet{tamm2015generalization}]\label{lemma:tammGen}
Let \(\cN\) be an NFA and let \(\cC\) be the \emph{canonical RFA} for \(\lang{\cN}\).
Then, \(\cN^{\text{res}}\) is\linebreak isomorphic to \(\cC\) if{}f the \emph{left language} of every state of \(\cN\) is a \emph{union of left languages} of states of \(\cC\).
\end{lemmaC}

The result of \citet{tamm2015generalization} generalizes the double-reversal method for RFAs along the lines of the generalization by \citet{Brzozowski2014} of the double-reversal method for DFAs which we estate next.
\begin{lemmaC}[\citet{Brzozowski2014}]
Let \(\cN\) be an NFA and let \(\cM\) be the \emph{minimal DFA} for \(\lang{\cN}\).
Then \(\cN^D\) is isomorphic to \(\cM\) if{}f the \emph{left language} of each state of \(\cN\) is a \emph{union of co-atoms} of \(\lang{\cN}\).
\end{lemmaC}

Although the two generalizations have a common foundation, the connection between the two resulting characterizations is not immediate.
Our work, together with the work of \citet{ganty2019congruence} allows us to clarify the relation between these two results and our Theorem~\ref{theorem:canonicalreverserestic}.
Indeed, \citet{ganty2019congruence} offered a congruence-based perspective of the generalized double-reversal me-thod for building the minimal DFA which lead to the following result.

\begin{lemmaC}[\citet{ganty2019congruence}]\label{lemma:congruenceDFA}
Let \(\cN\) be an NFA and let \(\cM\) be the minimal DFA for \(\lang{\cN}\).
Then \(\cN^D\) is isomorphic to \(\cM\) if{}f 
\[ρ_{\rrL}(W_{I,q}^{\cN}) = W_{I,q}^{\cN}\enspace ,\]
where \(\mathord{\rrL} \ud \mathord{\qrL} \cap \mathord{(\qrL)^{-1}}\) is the right Nerode's congruence.
\end{lemmaC}

We believe that the similarity between the generalizations of the double-reversal methods for the minimal DFA (Lemma~\ref{lemma:congruenceDFA}) and for the canonical RFA (Theorem~\ref{theorem:canonicalreverserestic}), which says that
\[\cG{r}(\cN) \text{ is isomorphic to \(\cC\)} \Lra ρ_{\qrL}(W_{I,q}^{\cN}) = W_{I,q}^{\cN}\enspace ,\]
evidences that \emph{quasiorders are for RFAs as congruences are for DFAs}.
Figure~\ref{fig:conclusions} summarizes the existing results about these double-reversal methods.

\begin{figure}[!ht]
\begin{minipage}[l]{.5\textwidth}
\small
\begin{tabu}{@{}c|c@{}}
   \textbf{{\citet{Brzozowski2014}}}& \textbf{{ Theorem~\ref{theorem:canonicalreverserestic}}}\\[3pt]  
   \begin{tabular}{c} \(\cN^{D} \equiv \cM\) \\ if{}f \\ \(\forall q, W_{I,q}^{\cN} \text{ is a union of co-atoms} \)\end{tabular} & \begin{tabular}{c}\(\cN^{D} \equiv \cM\) \\ if{}f \\ \(ρ_{\rrL}(W_{I,q}^{\cN}) = W_{I,q}^{\cN}\)\end{tabular}  \\
   \vspace{-5pt}\\
   \tabucline[1pt off 2pt]
   \\
   \vspace{-5pt}\\
   \textbf{{\citet{ganty2019congruence}}} & \textbf{{\citet{tamm2015generalization}}}\\[3pt] 
   \begin{tabular}{c}\(\cN^{\text{res}} \equiv \cC\) \\ if{}f \\ \(\forall q, W_{I,q}^{\cN} \text{ is a union of }  W_{I,q'}^{\cC}\)\end{tabular} & \begin{tabular}{c} \(\cG{r}(\cN) \equiv \cC\) \\ if{}f \\ \( ρ_{\qrL}(W_{I,q}^{\cN}) = W_{I,q}^{\cN}\)\end{tabular} \\
\end{tabu}
\end{minipage}\hfill
\begin{minipage}[r]{.4\textwidth}
In the diagram: \(\cN\) is an NFA with \(L = \lang{\cN}\); \(\cN^{D}\) is the result of determinizing \(\cN\) with the standard subset construction; \(\cM\) is the minimal DFA for \(L\); \(\cC = \cF{r}(L)\) is the canonical RFA for \(L\) and \(\cN_1 \equiv \cN_2\) denotes that automaton \(\cN_1\) is isomorphic to \(\cN_2\).
\end{minipage}
\caption{Summary of the existing results about the generalized double-reversal method for building the minimal DFA (first row) and the canonical RFA (second row) for a given language. The results on the first column are based on the notion of \emph{atoms} of a language while the results on the second column are based on \emph{quasiorders}.}%
\label{fig:conclusions}  
\end{figure}

Moreover, as shown by Lemma~\ref{lemma:qrlEqualqrNResEqualCan}, our residualization operation \(\cG{r}(\cN)\) offers a desirable property that \(\cN^{\text{res}}\) lacks: residualizing \(\cN\) yields the canonical RFA for \(\lang{\cN}\) if{}f \(\mathord{\qrL} = \mathord{\qrN}\).
Again, this property is equivalent to the one presented by \citet{ganty2019congruence} for DFAs.
\begin{lemmaC}[\citet{ganty2019congruence}]
Let \(\cN\) be an NFA and let \(\cM\) be the minimal DFA for \(\lang{\cN}\).
Then \(\cN^D\) is isomorphic to \(\cM\) if{}f \(\mathord{\rrN} = \mathord{\rrL}\), where \(\mathord{\rrN} \ud \mathord{\qrN} \cap (\qrN)^{-1}\) is the right state-based congruence.
\end{lemmaC} 

On the other hand, since \citet{ganty2019congruence} showed that the left languages of the minimal DFA for a regular language are the blocks of the partition \(ρ_{\rrL}\), Lemma~\ref{lemma:congruenceDFA} can be equivalently stated as follows.

\begin{lemmaC}[\citet{ganty2019congruence}]\label{lemma:congruenceTamm}
Let \(\cN\) be an NFA and let \(\cM\) be the minimal DFA for \(\lang{\cN}\).
Then \(\cN^D\) is isomorphic to \(\cM\) if{}f the left language of each state of \(\cN\) is a union of left languages of states of the minimal DFA.
\end{lemmaC} 

Therefore, Lemma~\ref{lemma:congruenceTamm} can be seen as the DFA-equivalent of Tamm's condition for RFAs (Lemma~\ref{lemma:tammGen}).
Therefore, Lemma~\ref{lemma:congruenceTamm} together with Lemma~\ref{lemma:congruenceDFA}, evidence the connection between the generalization of the double reversal for RFAs of \citet{tamm2015generalization} and the one for DFAs of \citet{Brzozowski2014}.

Finally, we further support the idea that quasiorders are natural to residual automata by observing that the NL\(^*\) algorithm proposed by \citet{bollig2009angluin} for learning RFAs can be interpreted within our framework as an algorithm that, at each step, refines an approximation of the Nerode's quasiorder and builds an RFA using our automata construction. 
\clearpage{}%

\newpage
\pagestyle{plain}
\chapter*{Funding Acknowledgments}
This research was partially supported by:
\begin{myItem}
\item The Spanish Ministry of Economy and Competitiveness project No.\ PGC2018-102210-B-I00.
\item The Spanish Ministry of Science and Innovation project No. TIN2015-71819-P.
\item The Madrid Regional Government project No. S2018/TCS-4339 .
\item The Madrid Regional Government project No. S2013/ICE-2731
\item The Ramón y Cajal fellowship RYC-2016-20281.
\item German Academic Exchange Service (DAAD) program ``Research Grants - Short-Term Grants 2018 (57378443)''.
\end{myItem}

\backmatter
\newpage
\pagestyle{plain}
\addcontentsline{toc}{chapter}{Bibliography}
\printindex


\begin{thebibliography}{}

\bibitem[\protect\astroncite{Abboud et~al.}{2017}]{Amir2018FineGrained}
Abboud, A., A.~Backurs, K.~Bringmann, and
  M.~K{\"{u}}nnemann\leavevmode\nopagebreak\newline 2017.
\newblock Fine-grained complexity of analyzing compressed data: Quantifying
  improvements over decompress-and-solve.
\newblock In {\em 58th {IEEE} Annual Symposium on Foundations of Computer
  Science, {FOCS} 2017, Berkeley, CA, USA, October 15-17, 2017}, C.~Umans, ed.,
  Pp.~ 192--203. {IEEE} Computer Society.

\bibitem[\protect\astroncite{Abdulla}{2012}]{abdulla2012regular}
Abdulla, P.~A.\leavevmode\nopagebreak\newline 2012.
\newblock Regular model checking.
\newblock {\em Int. J. Softw. Tools Technol. Transf.}, 14(2):109--118.

\bibitem[\protect\astroncite{Abdulla et~al.}{1996}]{ACJT96}
Abdulla, P.~A., K.~Cerans, B.~Jonsson, and
  Y.~Tsay\leavevmode\nopagebreak\newline 1996.
\newblock General decidability theorems for infinite-state systems.
\newblock In {\em Proceedings, 11th Annual {IEEE} Symposium on Logic in
  Computer Science, New Brunswick, New Jersey, USA, July 27-30, 1996}, Pp.~
  313--321. {IEEE} Computer Society.

\bibitem[\protect\astroncite{Abdulla et~al.}{2010}]{Abdulla2010}
Abdulla, P.~A., Y.~Chen, L.~Hol{\'{\i}}k, R.~Mayr, and
  T.~Vojnar\leavevmode\nopagebreak\newline 2010.
\newblock When simulation meets antichains.
\newblock In {\em Tools and Algorithms for the Construction and Analysis of
  Systems, 16th International Conference, {TACAS} 2010, Held as Part of the
  Joint European Conferences on Theory and Practice of Software, {ETAPS} 2010,
  Paphos, Cyprus, March 20-28, 2010. Proceedings}, J.~Esparza and R.~Majumdar,
  eds., volume 6015 of {\em Lecture Notes in Computer Science}, Pp.~ 158--174.
  Springer.

\bibitem[\protect\astroncite{Ad{\'{a}}mek et~al.}{2012}]{Adamek2012}
Ad{\'{a}}mek, J., F.~Bonchi, M.~H{\"{u}}lsbusch, B.~K{\"{o}}nig, S.~Milius, and
  A.~Silva\leavevmode\nopagebreak\newline 2012.
\newblock A coalgebraic perspective on minimization and determinization.
\newblock In {\em Foundations of Software Science and Computational Structures
  - 15th International Conference, {FOSSACS} 2012, Held as Part of the European
  Joint Conferences on Theory and Practice of Software, {ETAPS} 2012, Tallinn,
  Estonia, March 24 - April 1, 2012. Proceedings}, L.~Birkedal, ed., volume
  7213 of {\em Lecture Notes in Computer Science}, Pp.~ 58--73. Springer.

\bibitem[\protect\astroncite{Allouche et~al.}{2003}]{allouche2003automatic}
Allouche, J.-P., J.~Shallit, et~al.\leavevmode\nopagebreak\newline 2003.
\newblock {\em Automatic sequences: theory, applications, generalizations}.
\newblock Cambridge university press.

\bibitem[\protect\astroncite{Angluin}{1987}]{angluin1987learning}
Angluin, D.\leavevmode\nopagebreak\newline 1987.
\newblock Learning regular sets from queries and counterexamples.
\newblock {\em Inf. Comput.}, 75(2):87--106.

\bibitem[\protect\astroncite{Antimirov}{1995}]{Antimirov1995}
Antimirov, V.~M.\leavevmode\nopagebreak\newline 1995.
\newblock Rewriting regular inequalities (extended abstract).
\newblock In {\em Fundamentals of Computation Theory, 10th International
  Symposium, {FCT} '95, Dresden, Germany, August 22-25, 1995, Proceedings},
  H.~Reichel, ed., volume 965 of {\em Lecture Notes in Computer Science}, Pp.~
  116--125. Springer.

\bibitem[\protect\astroncite{Backurs and Indyk}{2016}]{Backurs2016Hard}
Backurs, A. and P.~Indyk\leavevmode\nopagebreak\newline 2016.
\newblock Which regular expression patterns are hard to match?
\newblock In {\em {IEEE} 57th Annual Symposium on Foundations of Computer
  Science, {FOCS} 2016, 9-11 October 2016, Hyatt Regency, New Brunswick, New
  Jersey, {USA}}, I.~Dinur, ed., Pp.~ 457--466. {IEEE} Computer Society.

\bibitem[\protect\astroncite{Bille et~al.}{2014}]{bille2017compressed}
Bille, P., P.~H. Cording, and I.~L. G{\o}rtz\leavevmode\nopagebreak\newline
  2014.
\newblock Compressed subsequence matching and packed tree coloring.
\newblock In {\em Combinatorial Pattern Matching - 25th Annual Symposium, {CPM}
  2014, Moscow, Russia, June 16-18, 2014. Proceedings}, A.~S. Kulikov, S.~O.
  Kuznetsov, and P.~A. Pevzner, eds., volume 8486 of {\em Lecture Notes in
  Computer Science}, Pp.~ 40--49. Springer.

\bibitem[\protect\astroncite{Bille et~al.}{2009}]{bille2009improved}
Bille, P., R.~Fagerberg, and I.~L. G{\o}rtz\leavevmode\nopagebreak\newline
  2009.
\newblock Improved approximate string matching and regular expression matching
  on ziv-lempel compressed texts.
\newblock {\em {ACM} Trans. Algorithms}, 6(1):3:1--3:14.

\bibitem[\protect\astroncite{Bollig et~al.}{2009}]{bollig2009angluin}
Bollig, B., P.~Habermehl, C.~Kern, and
  M.~Leucker\leavevmode\nopagebreak\newline 2009.
\newblock Angluin-style learning of {NFA}.
\newblock In {\em {IJCAI} 2009, Proceedings of the 21st International Joint
  Conference on Artificial Intelligence, Pasadena, California, USA, July 11-17,
  2009}, C.~Boutilier, ed., Pp.~ 1004--1009.

\bibitem[\protect\astroncite{Bonchi and Pous}{2013}]{DBLP:conf/popl/BonchiP13}
Bonchi, F. and D.~Pous\leavevmode\nopagebreak\newline 2013.
\newblock Checking {NFA} equivalence with bisimulations up to congruence.
\newblock In {\em The 40th Annual {ACM} {SIGPLAN-SIGACT} Symposium on
  Principles of Programming Languages, {POPL} '13, Rome, Italy - January 23 -
  25, 2013}, R.~Giacobazzi and R.~Cousot, eds., Pp.~ 457--468. {ACM}.

\bibitem[\protect\astroncite{Bringmann and
  K{\"{u}}nnemann}{2015}]{bringmann2015quadratic}
Bringmann, K. and M.~K{\"{u}}nnemann\leavevmode\nopagebreak\newline 2015.
\newblock Quadratic conditional lower bounds for string problems and dynamic
  time warping.
\newblock In {\em {IEEE} 56th Annual Symposium on Foundations of Computer
  Science, {FOCS} 2015, Berkeley, CA, USA, 17-20 October, 2015}, V.~Guruswami,
  ed., Pp.~ 79--97. {IEEE} Computer Society.

\bibitem[\protect\astroncite{Brzozowski}{1962}]{brzozowski1962canonical}
Brzozowski, J.~A.\leavevmode\nopagebreak\newline 1962.
\newblock Canonical regular expressions and minimal state graphs for definite
  events.
\newblock {\em Mathematical Theory of Automata}, 12(6):529--561.

\bibitem[\protect\astroncite{Brzozowski and Tamm}{2014}]{Brzozowski2014}
Brzozowski, J.~A. and H.~Tamm\leavevmode\nopagebreak\newline 2014.
\newblock Theory of {\'{a}}tomata.
\newblock {\em Theor. Comput. Sci.}, 539:13--27.

\bibitem[\protect\astroncite{Büchi}{1989}]{Buchi89}
Büchi, J.~R.\leavevmode\nopagebreak\newline 1989.
\newblock {\em Finite Automata, their Algebras and Grammars - Towards a Theory
  of Formal Expressions}.
\newblock Springer.

\bibitem[\protect\astroncite{Charikar et~al.}{2005}]{Charikar2005Smallest}
Charikar, M., E.~Lehman, D.~Liu, R.~Panigrahy, M.~Prabhakaran, A.~Sahai, and
  A.~Shelat\leavevmode\nopagebreak\newline 2005.
\newblock The smallest grammar problem.
\newblock {\em {IEEE} Trans. Inf. Theory}, 51(7):2554--2576.

\bibitem[\protect\astroncite{Chomsky}{1959}]{DBLP:journals/iandc/Chomsky59a}
Chomsky, N.\leavevmode\nopagebreak\newline 1959.
\newblock On certain formal properties of grammars.
\newblock {\em Information and Control}, 2(2):137--167.

\bibitem[\protect\astroncite{Clarke et~al.}{2018}]{clarke18}
Clarke, E.~M., T.~A. Henzinger, H.~Veith, and
  R.~Bloem\leavevmode\nopagebreak\newline 2018.
\newblock {\em Handbook of Model Checking}, 1st edition.
\newblock Springer Publishing Company, Incorporated.

\bibitem[\protect\astroncite{Cousot}{1978}]{Cousot78-1-TheseEtat}
Cousot, P.\leavevmode\nopagebreak\newline 1978.
\newblock {\em Méthodes itératives de construction et d'approximation de
  points fixes d'opérateurs monotones sur un treillis, analyse sémantique de
  programmes (in French)}.
\newblock Thèse d'État ès sciences mathématiques, Université Joseph
  Fourier, Grenoble, France.

\bibitem[\protect\astroncite{Cousot}{2000}]{cou00}
Cousot, P.\leavevmode\nopagebreak\newline 2000.
\newblock Partial completeness of abstract fixpoint checking.
\newblock In {\em Abstraction, Reformulation, and Approximation, 4th
  International Symposium, {SARA} 2000, Horseshoe Bay, Texas, USA, July 26-29,
  2000, Proceedings}, B.~Y. Choueiry and T.~Walsh, eds., volume 1864 of {\em
  Lecture Notes in Computer Science}, Pp.~ 1--25. Springer.

\bibitem[\protect\astroncite{Cousot and Cousot}{1977}]{CC77}
Cousot, P. and R.~Cousot\leavevmode\nopagebreak\newline 1977.
\newblock Abstract interpretation: {A} unified lattice model for static
  analysis of programs by construction or approximation of fixpoints.
\newblock In {\em Conference Record of the Fourth {ACM} Symposium on Principles
  of Programming Languages, Los Angeles, California, USA, January 1977}, R.~M.
  Graham, M.~A. Harrison, and R.~Sethi, eds., Pp.~ 238--252. {ACM}.

\bibitem[\protect\astroncite{Cousot and Cousot}{1979}]{CC79}
Cousot, P. and R.~Cousot\leavevmode\nopagebreak\newline 1979.
\newblock Systematic design of program analysis frameworks.
\newblock In {\em Conference Record of the Sixth Annual {ACM} Symposium on
  Principles of Programming Languages, San Antonio, Texas, USA, January 1979},
  A.~V. Aho, S.~N. Zilles, and B.~K. Rosen, eds., Pp.~ 269--282. {ACM} Press.

\bibitem[\protect\astroncite{D'Alessandro and Varricchio}{2008}]{d2008well}
D'Alessandro, F. and S.~Varricchio\leavevmode\nopagebreak\newline 2008.
\newblock Well quasi-orders in formal language theory.
\newblock In {\em Developments in Language Theory, 12th International
  Conference, {DLT} 2008, Kyoto, Japan, September 16-19, 2008. Proceedings},
  M.~Ito and M.~Toyama, eds., volume 5257 of {\em Lecture Notes in Computer
  Science}, Pp.~ 84--95. Springer.

\bibitem[\protect\astroncite{de~Luca and Varricchio}{1994}]{deLuca1994}
de~Luca, A. and S.~Varricchio\leavevmode\nopagebreak\newline 1994.
\newblock Well quasi-orders and regular languages.
\newblock {\em Acta Inf.}, 31(6):539--557.

\bibitem[\protect\astroncite{de~Luca and Varricchio}{2011}]{deluca2011}
de~Luca, A. and S.~Varricchio\leavevmode\nopagebreak\newline 2011.
\newblock {\em Finiteness and Regularity in Semigroups and Formal Languages},
  1st edition.
\newblock Springer.

\bibitem[\protect\astroncite{de~Moura et~al.}{1998}]{de1998direct}
de~Moura, E.~S., G.~Navarro, N.~Ziviani, and R.~A.
  Baeza{-}Yates\leavevmode\nopagebreak\newline 1998.
\newblock Direct pattern matching on compressed text.
\newblock In {\em String Processing and Information Retrieval: {A} South
  American Symposium, {SPIRE} 1998, Santa Cruz de la Sierra Bolivia, September
  9-11, 1998}, Pp.~ 90--95. {IEEE} Computer Society.

\bibitem[\protect\astroncite{Denis et~al.}{2000}]{denis2000residual}
Denis, F., A.~Lemay, and A.~Terlutte\leavevmode\nopagebreak\newline 2000.
\newblock Learning regular languages using non deterministic finite automata.
\newblock In {\em Grammatical Inference: Algorithms and Applications, 5th
  International Colloquium, {ICGI} 2000, Lisbon, Portugal, September 11-13,
  2000, Proceedings}, A.~L. Oliveira, ed., volume 1891 of {\em Lecture Notes in
  Computer Science}, Pp.~ 39--50. Springer.

\bibitem[\protect\astroncite{Denis et~al.}{2001}]{denis2001residual}
Denis, F., A.~Lemay, and A.~Terlutte\leavevmode\nopagebreak\newline 2001.
\newblock Residual finite state automata.
\newblock 2010:144--157.

\bibitem[\protect\astroncite{Denis et~al.}{2002}]{denis2002residual}
Denis, F., A.~Lemay, and A.~Terlutte\leavevmode\nopagebreak\newline 2002.
\newblock Residual finite state automata.
\newblock {\em Fundam. Inform.}, 51(4):339--368.

\bibitem[\protect\astroncite{Denis et~al.}{2004}]{denis2004learning}
Denis, F., A.~Lemay, and A.~Terlutte\leavevmode\nopagebreak\newline 2004.
\newblock Learning regular languages using rfsas.
\newblock {\em Theor. Comput. Sci.}, 313(2):267--294.

\bibitem[\protect\astroncite{Ehrenfeucht
  et~al.}{1983}]{ehrenfeucht_regularity_1983}
Ehrenfeucht, A., D.~Haussler, and G.~Rozenberg\leavevmode\nopagebreak\newline
  1983.
\newblock On regularity of context-free languages.
\newblock {\em Theor. Comput. Sci.}, 27:311--332.

\bibitem[\protect\astroncite{Esparza et~al.}{2000}]{esparza00}
Esparza, J., P.~Rossmanith, and S.~Schwoon\leavevmode\nopagebreak\newline 2000.
\newblock A uniform framework for problems on context-free grammars.
\newblock {\em Bulletin of the {EATCS}}, 72:169--177.

\bibitem[\protect\astroncite{Fiedor et~al.}{2015}]{fiedor2019nested}
Fiedor, T., L.~Hol{\'{\i}}k, O.~Leng{\'{a}}l, and
  T.~Vojnar\leavevmode\nopagebreak\newline 2015.
\newblock Nested antichains for {WS1S}.
\newblock 9035:658--674.

\bibitem[\protect\astroncite{Finkel and Schnoebelen}{2001}]{Finkel2001}
Finkel, A. and P.~Schnoebelen\leavevmode\nopagebreak\newline 2001.
\newblock Well-structured transition systems everywhere!
\newblock {\em Theor. Comput. Sci.}, 256(1-2):63--92.

\bibitem[\protect\astroncite{Ganty et~al.}{2019}]{ganty2019congruence}
Ganty, P., E.~Guti{\'{e}}rrez, and P.~Valero\leavevmode\nopagebreak\newline
  2019.
\newblock A congruence-based perspective on automata minimization algorithms.
\newblock In {\em 44th International Symposium on Mathematical Foundations of
  Computer Science, {MFCS} 2019, August 26-30, 2019, Aachen, Germany},
  P.~Rossmanith, P.~Heggernes, and J.~Katoen, eds., volume 138 of {\em LIPIcs},
  Pp.~ 77:1--77:14. Schloss Dagstuhl - Leibniz-Zentrum f{\"{u}}r Informatik.

\bibitem[\protect\astroncite{Gawrychowski}{2011}]{gawrychowski2013optimal}
Gawrychowski, P.\leavevmode\nopagebreak\newline 2011.
\newblock Optimal pattern matching in {LZW} compressed strings.
\newblock In {\em Proceedings of the Twenty-Second Annual {ACM-SIAM} Symposium
  on Discrete Algorithms, {SODA} 2011, San Francisco, California, USA, January
  23-25, 2011}, D.~Randall, ed., Pp.~ 362--372. {SIAM}.

\bibitem[\protect\astroncite{Gawrychowski}{2014}]{gawrychowski2014simple}
Gawrychowski, P.\leavevmode\nopagebreak\newline 2014.
\newblock Simple and efficient lzw-compressed multiple pattern matching.
\newblock {\em J. Discrete Algorithms}, 25:34--41.

\bibitem[\protect\astroncite{Giacobazzi and Quintarelli}{2001}]{gq01}
Giacobazzi, R. and E.~Quintarelli\leavevmode\nopagebreak\newline 2001.
\newblock Incompleteness, counterexamples, and refinements in abstract
  model-checking.
\newblock In {\em Static Analysis, 8th International Symposium, {SAS} 2001,
  Paris, France, July 16-18, 2001, Proceedings}, P.~Cousot, ed., volume 2126 of
  {\em Lecture Notes in Computer Science}, Pp.~ 356--373. Springer.

\bibitem[\protect\astroncite{Giacobazzi et~al.}{2000}]{GiacobazziRS00}
Giacobazzi, R., F.~Ranzato, and F.~Scozzari\leavevmode\nopagebreak\newline
  2000.
\newblock Making abstract interpretations complete.
\newblock {\em J. {ACM}}, 47(2):361--416.

\bibitem[\protect\astroncite{Ginsburg and Rice}{1962}]{ginsburg}
Ginsburg, S. and H.~G. Rice\leavevmode\nopagebreak\newline 1962.
\newblock Two families of languages related to {ALGOL}.
\newblock {\em J. {ACM}}, 9(3):350--371.

\bibitem[\protect\astroncite{Henriksen et~al.}{1995}]{Klarlund:mona95}
Henriksen, J.~G., J.~L. Jensen, M.~E. J{\o}rgensen, N.~Klarlund, R.~Paige,
  T.~Rauhe, and A.~Sandholm\leavevmode\nopagebreak\newline 1995.
\newblock Mona: Monadic second-order logic in practice.
\newblock In {\em Tools and Algorithms for Construction and Analysis of
  Systems, First International Workshop, {TACAS} '95, Aarhus, Denmark, May
  19-20, 1995, Proceedings}, E.~Brinksma, R.~Cleaveland, K.~G. Larsen,
  T.~Margaria, and B.~Steffen, eds., volume 1019 of {\em Lecture Notes in
  Computer Science}, Pp.~ 89--110. Springer.

\bibitem[\protect\astroncite{Hofman
  et~al.}{2013}]{Hofman:2013:DWS:2591370.2591405}
Hofman, P., R.~Mayr, and P.~Totzke\leavevmode\nopagebreak\newline 2013.
\newblock Decidability of weak simulation on one-counter nets.
\newblock In {\em 28th Annual {ACM/IEEE} Symposium on Logic in Computer
  Science, {LICS} 2013, New Orleans, LA, USA, June 25-28, 2013}, Pp.~ 203--212.
  {IEEE} Computer Society.

\bibitem[\protect\astroncite{Hofman and Totzke}{2018}]{hofman_trace_2018}
Hofman, P. and P.~Totzke\leavevmode\nopagebreak\newline 2018.
\newblock Trace inclusion for one-counter nets revisited.
\newblock {\em Theor. Comput. Sci.}, 735:50--63.

\bibitem[\protect\astroncite{Hofmann and Chen}{2014}]{Hofmann2014}
Hofmann, M. and W.~Chen\leavevmode\nopagebreak\newline 2014.
\newblock Abstract interpretation from b{\"{u}}chi automata.
\newblock In {\em Joint Meeting of the Twenty-Third {EACSL} Annual Conference
  on Computer Science Logic {(CSL)} and the Twenty-Ninth Annual {ACM/IEEE}
  Symposium on Logic in Computer Science (LICS), {CSL-LICS} '14, Vienna,
  Austria, July 14 - 18, 2014}, T.~A. Henzinger and D.~Miller, eds., Pp.~
  51:1--51:10. {ACM}.

\bibitem[\protect\astroncite{Hol{\'{\i}}k and Meyer}{2015}]{Holk2015}
Hol{\'{\i}}k, L. and R.~Meyer\leavevmode\nopagebreak\newline 2015.
\newblock Antichains for the verification of recursive programs.
\newblock In {\em Networked Systems - Third International Conference, {NETYS}
  2015, Agadir, Morocco, May 13-15, 2015, Revised Selected Papers},
  A.~Bouajjani and H.~Fauconnier, eds., volume 9466 of {\em Lecture Notes in
  Computer Science}, Pp.~ 322--336. Springer.

\bibitem[\protect\astroncite{Hopcroft}{1971}]{hopcroft1971}
Hopcroft, J.~E.\leavevmode\nopagebreak\newline 1971.
\newblock An n log n algorithm for minimizing states in a finite automaton.
\newblock In {\em Theory of machines and computations}, Pp.~ 189--196.
\newblock Elsevier.

\bibitem[\protect\astroncite{Hopcroft et~al.}{2001}]{Ullman2003}
Hopcroft, J.~E., R.~Motwani, and J.~D. Ullman\leavevmode\nopagebreak\newline
  2001.
\newblock {\em Introduction to Automata Theory, Languages, and Computation -
  (2. ed.)}, Addison-Wesley Series in Computer Science.
\newblock Addison-Wesley-Longman.

\bibitem[\protect\astroncite{Hopcroft and Ullman}{1979}]{HU79}
Hopcroft, J.~E. and J.~D. Ullman\leavevmode\nopagebreak\newline 1979.
\newblock {\em Introduction to Automata Theory, Languages, and Computation}.
\newblock Addison-Wesley Publishing Company.

\bibitem[\protect\astroncite{Hucke et~al.}{2016}]{Hucke2016Smallest}
Hucke, D., M.~Lohrey, and C.~P. Reh\leavevmode\nopagebreak\newline 2016.
\newblock The smallest grammar problem revisited.
\newblock In {\em String Processing and Information Retrieval - 23rd
  International Symposium, {SPIRE} 2016, Beppu, Japan, October 18-20, 2016,
  Proceedings}, S.~Inenaga, K.~Sadakane, and T.~Sakai, eds., volume 9954 of
  {\em Lecture Notes in Computer Science}, Pp.~ 35--49.

\bibitem[\protect\astroncite{Jancar et~al.}{1999}]{JANCAR1999476}
Jancar, P., J.~Esparza, and F.~Moller\leavevmode\nopagebreak\newline 1999.
\newblock Petri nets and regular processes.
\newblock {\em J. Comput. Syst. Sci.}, 59(3):476--503.

\bibitem[\protect\astroncite{K{\"{a}}rkk{\"{a}}inen
  et~al.}{2003}]{karkkainen2003approximate}
K{\"{a}}rkk{\"{a}}inen, J., G.~Navarro, and
  E.~Ukkonen\leavevmode\nopagebreak\newline 2003.
\newblock Approximate string matching on ziv-lempel compressed text.
\newblock {\em J. Discrete Algorithms}, 1(3-4):313--338.

\bibitem[\protect\astroncite{Kasprzik}{2011}]{Kasprzik2011Inference}
Kasprzik, A.\leavevmode\nopagebreak\newline 2011.
\newblock Inference of residual finite-state tree automata from membership
  queries and finite positive data.
\newblock In {\em Developments in Language Theory - 15th International
  Conference, {DLT} 2011, Milan, Italy, July 19-22, 2011. Proceedings},
  G.~Mauri and A.~Leporati, eds., volume 6795 of {\em Lecture Notes in Computer
  Science}, Pp.~ 476--477. Springer.

\bibitem[\protect\astroncite{Keil and
  Thiemann}{2014}]{keil_et_al:LIPIcs:2014:4841}
Keil, M. and P.~Thiemann\leavevmode\nopagebreak\newline 2014.
\newblock Symbolic solving of extended regular expression inequalities.
\newblock In {\em 34th International Conference on Foundation of Software
  Technology and Theoretical Computer Science, {FSTTCS} 2014, December 15-17,
  2014, New Delhi, India}, V.~Raman and S.~P. Suresh, eds., volume~29 of {\em
  LIPIcs}, Pp.~ 175--186. Schloss Dagstuhl - Leibniz-Zentrum f{\"{u}}r
  Informatik.

\bibitem[\protect\astroncite{Khoussainov and Nerode}{2001}]{Khoussainov2001}
Khoussainov, B. and A.~Nerode\leavevmode\nopagebreak\newline 2001.
\newblock {\em Automata Theory and its Applications}.
\newblock Secaucus, NJ, USA: Birkhäuser Boston.

\bibitem[\protect\astroncite{Kida et~al.}{1998}]{kida1998multipattern}
Kida, T., M.~Takeda, A.~Shinohara, M.~Miyazaki, and
  S.~Arikawa\leavevmode\nopagebreak\newline 1998.
\newblock Multiple pattern matching in {LZW} compressed text.
\newblock In {\em Data Compression Conference, {DCC} 1998, Snowbird, Utah, USA,
  March 30 - April 1, 1998}, Pp.~ 103--112. {IEEE} Computer Society.

\bibitem[\protect\astroncite{Klarlund}{1999}]{Klarlund:mona99}
Klarlund, N.\leavevmode\nopagebreak\newline 1999.
\newblock A theory of restrictions for logics and automata.
\newblock In {\em Computer Aided Verification, 11th International Conference,
  {CAV} '99, Trento, Italy, July 6-10, 1999, Proceedings}, N.~Halbwachs and
  D.~A. Peled, eds., volume 1633 of {\em Lecture Notes in Computer Science},
  Pp.~ 406--417. Springer.

\bibitem[\protect\astroncite{Kunc}{2005}]{kunc2005regular}
Kunc, M.\leavevmode\nopagebreak\newline 2005.
\newblock Regular solutions of language inequalities and well quasi-orders.
\newblock {\em Theor. Comput. Sci.}, 348(2-3):277--293.

\bibitem[\protect\astroncite{Larsson and Moffat}{1999}]{larsson2000off}
Larsson, N.~J. and A.~Moffat\leavevmode\nopagebreak\newline 1999.
\newblock Offline dictionary-based compression.
\newblock In {\em Data Compression Conference, {DCC} 1999, Snowbird, Utah, USA,
  March 29-31, 1999}, Pp.~ 296--305. {IEEE} Computer Society.

\bibitem[\protect\astroncite{Lison and Tiedemann}{2016}]{openSubtitles}
Lison, P. and J.~Tiedemann\leavevmode\nopagebreak\newline 2016.
\newblock Opensubtitles2016: Extracting large parallel corpora from movie and
  {TV} subtitles.
\newblock In {\em Proceedings of the Tenth International Conference on Language
  Resources and Evaluation {LREC} 2016, Portoro{\v{z}}, Slovenia, May 23-28,
  2016}, N.~Calzolari, K.~Choukri, T.~Declerck, S.~Goggi, M.~Grobelnik,
  B.~Maegaard, J.~Mariani, H.~Mazo, A.~Moreno, J.~Odijk, and S.~Piperidis, eds.
  European Language Resources Association {(ELRA)}.

\bibitem[\protect\astroncite{Lohrey}{2012}]{lohrey2012algorithmics}
Lohrey, M.\leavevmode\nopagebreak\newline 2012.
\newblock Algorithmics on slp-compressed strings: {A} survey.
\newblock {\em Groups Complexity Cryptology}, 4(2):241--299.

\bibitem[\protect\astroncite{M{\"{a}}kinen and
  Navarro}{2006}]{navarro2007compressedIndex}
M{\"{a}}kinen, V. and G.~Navarro\leavevmode\nopagebreak\newline 2006.
\newblock Dynamic entropy-compressed sequences and full-text indexes.
\newblock In {\em Combinatorial Pattern Matching, 17th Annual Symposium, {CPM}
  2006, Barcelona, Spain, July 5-7, 2006, Proceedings}, M.~Lewenstein and
  G.~Valiente, eds., volume 4009 of {\em Lecture Notes in Computer Science},
  Pp.~ 306--317. Springer.

\bibitem[\protect\astroncite{Markey and Schnoebelen}{2004}]{markey2004ptime}
Markey, N. and P.~Schnoebelen\leavevmode\nopagebreak\newline 2004.
\newblock A ptime-complete matching problem for slp-compressed words.
\newblock {\em Inf. Process. Lett.}, 90(1):3--6.

\bibitem[\protect\astroncite{Min{\'{e}}}{2017}]{mine17}
Min{\'{e}}, A.\leavevmode\nopagebreak\newline 2017.
\newblock Tutorial on static inference of numeric invariants by abstract
  interpretation.
\newblock {\em Foundations and Trends in Programming Languages},
  4(3-4):120--372.

\bibitem[\protect\astroncite{Moore}{1956}]{moore1956}
Moore, E.~F.\leavevmode\nopagebreak\newline 1956.
\newblock Gedanken-experiments on sequential machines.
\newblock {\em Automata studies}, 23(1):60–60.

\bibitem[\protect\astroncite{Navarro}{2001}]{navarro2001approximateSurvey}
Navarro, G.\leavevmode\nopagebreak\newline 2001.
\newblock A guided tour to approximate string matching.
\newblock {\em {ACM} Comput. Surv.}, 33(1):31--88.

\bibitem[\protect\astroncite{Navarro}{2003}]{navarro2003regular}
Navarro, G.\leavevmode\nopagebreak\newline 2003.
\newblock Regular expression searching on compressed text.
\newblock {\em J. Discrete Algorithms}, 1(5-6):423--443.

\bibitem[\protect\astroncite{Navarro and Tarhio}{2005}]{navarro2005lzgrep}
Navarro, G. and J.~Tarhio\leavevmode\nopagebreak\newline 2005.
\newblock Lzgrep: a boyer-moore string matching tool for ziv-lempel compressed
  text.
\newblock {\em Softw. Pract. Exp.}, 35(12):1107--1130.

\bibitem[\protect\astroncite{Nevill{-}Manning and
  Witten}{1997}]{nevill1997compression}
Nevill{-}Manning, C.~G. and I.~H. Witten\leavevmode\nopagebreak\newline 1997.
\newblock Compression and explanation using hierarchical grammars.
\newblock {\em Comput. J.}, 40(2/3):103--116.

\bibitem[\protect\astroncite{Park}{1969}]{park1969fixpoint}
Park, D.\leavevmode\nopagebreak\newline 1969.
\newblock Fixpoint induction and proofs of program properties.
\newblock {\em Machine Intelligence}, 5.

\bibitem[\protect\astroncite{Plandowski and
  Rytter}{1999}]{plandowski1999complexity}
Plandowski, W. and W.~Rytter\leavevmode\nopagebreak\newline 1999.
\newblock Complexity of language recognition problems for compressed words.
\newblock In {\em Jewels are Forever, Contributions on Theoretical Computer
  Science in Honor of Arto Salomaa}, J.~Karhum{\"{a}}ki, H.~A. Maurer, G.~Paun,
  and G.~Rozenberg, eds., Pp.~ 262--272. Springer.

\bibitem[\protect\astroncite{Ranzato}{2013}]{Ranzato13}
Ranzato, F.\leavevmode\nopagebreak\newline 2013.
\newblock Complete abstractions everywhere.
\newblock In {\em Verification, Model Checking, and Abstract Interpretation,
  14th International Conference, {VMCAI} 2013, Rome, Italy, January 20-22,
  2013. Proceedings}, R.~Giacobazzi, J.~Berdine, and I.~Mastroeni, eds., volume
  7737 of {\em Lecture Notes in Computer Science}, Pp.~ 15--26. Springer.

\bibitem[\protect\astroncite{Rytter}{2004}]{Rytter2004Equivalent}
Rytter, W.\leavevmode\nopagebreak\newline 2004.
\newblock Grammar compression, lz-encodings, and string algorithms with
  implicit input.
\newblock In {\em Automata, Languages and Programming: 31st International
  Colloquium, {ICALP} 2004, Turku, Finland, July 12-16, 2004. Proceedings},
  J.~D{\'{\i}}az, J.~Karhum{\"{a}}ki, A.~Lepist{\"{o}}, and D.~Sannella, eds.,
  volume 3142 of {\em Lecture Notes in Computer Science}, Pp.~ 15--27.
  Springer.

\bibitem[\protect\astroncite{Sakarovitch}{2009}]{Sakarovitch}
Sakarovitch, J.\leavevmode\nopagebreak\newline 2009.
\newblock {\em Elements of Automata Theory}.
\newblock Cambridge University Press.

\bibitem[\protect\astroncite{Schaeffer}{2013}]{schaeffer2013deciding}
Schaeffer, L.\leavevmode\nopagebreak\newline 2013.
\newblock Deciding properties of automatic sequences.
\newblock Master's thesis, University of Waterloo.

\bibitem[\protect\astroncite{Sch{\"{u}}tzenberger}{1963}]{Schutzenberger63}
Sch{\"{u}}tzenberger, M.~P.\leavevmode\nopagebreak\newline 1963.
\newblock On context-free languages and push-down automata.
\newblock {\em Information and Control}, 6(3):246--264.

\bibitem[\protect\astroncite{Tamm}{2015}]{tamm2015generalization}
Tamm, H.\leavevmode\nopagebreak\newline 2015.
\newblock Generalization of the double-reversal method of finding a canonical
  residual finite state automaton.
\newblock In {\em Descriptional Complexity of Formal Systems - 17th
  International Workshop, {DCFS} 2015, Waterloo, ON, Canada, June 25-27, 2015.
  Proceedings}, J.~O. Shallit and A.~Okhotin, eds., volume 9118 of {\em Lecture
  Notes in Computer Science}, Pp.~ 268--279. Springer.

\bibitem[\protect\astroncite{Thompson}{1968}]{thompson1968programming}
Thompson, K.\leavevmode\nopagebreak\newline 1968.
\newblock Programming techniques: Regular expression search algorithm.
\newblock {\em Communications of the ACM}.

\bibitem[\protect\astroncite{To and Libkin}{2008}]{bouajjani2000regular}
To, A.~W. and L.~Libkin\leavevmode\nopagebreak\newline 2008.
\newblock Recurrent reachability analysis in regular model checking.
\newblock In {\em Logic for Programming, Artificial Intelligence, and
  Reasoning, 15th International Conference, {LPAR} 2008, Doha, Qatar, November
  22-27, 2008. Proceedings}, I.~Cervesato, H.~Veith, and A.~Voronkov, eds.,
  volume 5330 of {\em Lecture Notes in Computer Science}, Pp.~ 198--213.
  Springer.

\bibitem[\protect\astroncite{Ullman and Gelder}{1988}]{ullman1988parallel}
Ullman, J.~D. and A.~V. Gelder\leavevmode\nopagebreak\newline 1988.
\newblock Parallel complexity of logical query programs.
\newblock {\em Algorithmica}, 3:5--42.

\bibitem[\protect\astroncite{Welch}{1984}]{welch1984technique}
Welch, T.~A.\leavevmode\nopagebreak\newline 1984.
\newblock A technique for high-performance data compression.
\newblock {\em {IEEE} Computer}, 17(6):8--19.

\bibitem[\protect\astroncite{Wolper and Boigelot}{1995}]{wolper1995automata}
Wolper, P. and B.~Boigelot\leavevmode\nopagebreak\newline 1995.
\newblock An automata-theoretic approach to presburger arithmetic constraints
  (extended abstract).
\newblock In {\em Static Analysis, Second International Symposium, SAS'95,
  Glasgow, UK, September 25-27, 1995, Proceedings}, A.~Mycroft, ed., volume 983
  of {\em Lecture Notes in Computer Science}, Pp.~ 21--32. Springer.

\bibitem[\protect\astroncite{Wulf et~al.}{2006}]{DBLP:conf/cav/WulfDHR06}
Wulf, M.~D., L.~Doyen, T.~A. Henzinger, and
  J.~Raskin\leavevmode\nopagebreak\newline 2006.
\newblock Antichains: {A} new algorithm for checking universality of finite
  automata.
\newblock In {\em Computer Aided Verification, 18th International Conference,
  {CAV} 2006, Seattle, WA, USA, August 17-20, 2006, Proceedings}, T.~Ball and
  R.~B. Jones, eds., volume 4144 of {\em Lecture Notes in Computer Science},
  Pp.~ 17--30. Springer.

\bibitem[\protect\astroncite{Ziv and Lempel}{1977}]{ziv1977compression}
Ziv, J. and A.~Lempel\leavevmode\nopagebreak\newline 1977.
\newblock A universal algorithm for sequential data compression.
\newblock {\em {IEEE} Trans. Inf. Theory}, 23(3):337--343.

\bibitem[\protect\astroncite{Ziv and Lempel}{1978}]{ziv1978compression}
Ziv, J. and A.~Lempel\leavevmode\nopagebreak\newline 1978.
\newblock Compression of individual sequences via variable-rate coding.
\newblock {\em {IEEE} Trans. Inf. Theory}, 24(5):530--536.

\end{thebibliography}
\end{document}